\documentclass[11pt, reqno]{article}
\usepackage[margin=1in]{geometry}
\usepackage{ag} 
\usepackage[toc,page,header]{appendix}
\usepackage{minitoc}

\allowdisplaybreaks
\newcommand\ag[1]{#1}

\begin{document}


\title{Robustness and Efficiency  of Rosenbaum's Rank-based Estimator in  Randomized Trials:~A Design-based Perspective}

\author[1]{Aditya Ghosh}
\affil[1]{Department of Statistics, Stanford University}

\author[2]{Nabarun Deb}
\affil[2]{Booth School of Business, University of Chicago}

\author[3]{Bikram Karmakar \thanks{Supported by NSF grant DMS-2015250}}
\affil[3]{Department of Statistics, University of Wisconsin-Madison}

\author[4]{Bodhisattva Sen \thanks{Supported by NSF grant DMS-2311062}}
\affil[4]{Department of Statistics, Columbia University}

\date{\today}
\maketitle


\begin{abstract}
Mean-based estimators of causal effects in randomized experiments may behave poorly if the potential outcomes have a heavy tail or contain outliers. An alternative estimator proposed by \citet{rosenbaum1993} estimates a constant additive treatment effect by inverting a randomization test using ranks. \ag{We develop a design-based asymptotic theory for this rank-based estimator and study its robustness and efficiency properties. We show that Rosenbaum's estimator is robust against outliers with a breakdown point that uniformly dominates that of any weighted quantile estimator.}
When pretreatment covariates are available, a regression-adjusted version of Rosenbaum's estimator uses an agnostic linear regression on the covariates and bases inference on the ranks of residuals. Under mild integrability conditions, we show that this estimator is at most 13.6\% less efficient, in the worst case, than the commonly used mean-based regression adjustment method proposed by \citet{Lin13}; often outperforming it when the residuals have heavy tails. Moreover, under suitable assumptions, Rosenbaum's regression-adjusted estimator is at least as efficient as the unadjusted one.
Finally, we initiate the study of Rosenbaum's estimator when the constant treatment effect assumption may be violated. To analyze the regression-adjusted estimator, we develop local asymptotics of rank statistics under the design-based framework, which may be of independent interest.
\end{abstract}



\doparttoc 
\faketableofcontents 

\section{Introduction}\label{sec:intro}

When a treatment is randomly assigned to the units in a study,~\cite{Fisher1935} showed that one can use this randomization to provide valid statistical inference about the treatment effect without making strong assumptions regarding the outcome-generating model. In particular, we can calculate confidence intervals for the treatment effect based on an estimator using its randomization distribution. 
\ag{In the design-based (finite population) framework, one considers the potential outcomes \citep{Neyman1923, Rubin1974, Rubin1977, imbens2015} as fixed, making the randomization distribution of the estimator central to inference. In contrast, the model-based (infinite population) approach assumes that units are independently sampled from an infinite superpopulation, introducing randomness through both sampling and treatment assignment.  While the infinite-population approach typically yields simpler mathematical derivations due to the assumption of independent and identically distributed (i.i.d.) samples, this assumption itself has often been criticized as unrealistic and for obscuring the role of the random treatment assignment mechanism \citep{DingLiMiratrix+2017}. We refer the reader to \citet{DingLiMiratrix+2017} for a detailed comparison of these two frameworks.}

A popular choice for an estimator of the average treatment effect in randomized trials is the difference-in-means estimator, which takes the difference of the averages of outcomes for the treated and control units (see, e.g., \citet{Neyman1923}). \citet{Freedman08a} formally derived the design-based asymptotic theory of this estimator under finite fourth moment conditions, showing that the standard Wald confidence interval, which is justified under i.i.d.~sampling from infinite populations, remains asymptotically valid under the randomization distribution. 
\ag{\citet{Li2017} substantially generalized this line of work by deriving finite population central limit theorems under minimal moment conditions, covering a wide range of estimators and experimental designs.}

Despite these advances, mean-based estimators can have volatile standard errors (and thus produce very wide confidence intervals) when the potential outcomes are heavy-tailed or contaminated by outliers. Heavy-tailed or extreme outcomes are common in many modern randomized experiments, including experiments in the digital space (see, e.g., \citet{athey2021}). An important class of estimators used in practice in such situations use the {\it ranks} of the potential outcomes as opposed to their observed values (see, e.g.,~\citet[Section 5.5.4]{imbens2015}). Notably,~\cite{rosenbaum1993} proposed to use a Hodges-Lehmann type point estimate~\citep{HL63} based on the Wilcoxon rank-sum (abbreviated as WRS) test statistic \citep{Wilcoxon1945} for estimating a constant additive treatment effect under randomization inference. \ag{The constant treatment effect hypothesis generalizes Fisher's sharp null hypothesis,} and is often a convenient starting point in answering a causal question \citep{Rosenbaum02, Ho2006, athey2021}. Further, in some situations, the identification of a constant treatment effect has immediate practical use \cite[see][Section 2.4.5]{Rosenbaum2002}.

\ag{
In this paper, we develop a design-based asymptotic theory for Rosenbaum’s estimator based on the WRS statistic (a.k.a.~the Hodges-Lehmann estimator), focusing on its robustness and efficiency properties. Despite its popularity in practice --- especially under heavy-tailed outcomes (see, e.g., \citet{athey2021}) --- the asymptotic behavior of this estimator under finite-population asymptotics has, to our knowledge, not been formally analyzed. While the general finite-population central limit theorems developed by \citet{Li2017} provide a versatile toolkit for establishing the asymptotic normality of many estimators under randomization, their results do not directly apply to nonlinear rank-based statistics such as Rosenbaum's rank-based estimator based on the WRS test statistic. Our work fills this gap and provides a rigorous of Rosenabum's rank-based estimators in this framework.
}

We also consider randomization inference for the treatment effect with regression adjustment of pretreatment covariates. It is increasingly common to collect additional data on covariates in randomized experiments which are then used in an ANCOVA model, aiming to improve the unadjusted difference-in-means type estimator \citep{Fisher1935, Cox2000, Freedman08a, Freedman08b}.~\citet{Lin13} showed that including treatment-covariate interaction in the linear model results in a regression-adjusted estimator which is asymptotically at least as efficient as the unadjusted difference-in-means estimator in completely randomized experiments. Lin's result led to follow-up works on efficient regression-adjusted or model-assisted estimators for other experimental designs
\citep{Fogarty2018, Li2020, Liu2020, su2021model, Zhao2021regression}. 
However, these methods typically use robust standard errors \citep{Huber1967, White1980}, and are still sensitive to heavy-tailed distributions or extreme values of the potential outcomes  \citep{MacKinnon1985, Young2019}. 

The regression-adjusted version of Rosenbaum's rank-based estimator, originally suggested by \cite{Rosenbaum02} for observational studies, is calculated as follows: First, we compute the control potential outcomes under a hypothesized treatment effect and regress it on the covariates using least squares. Then, we invert the WRS test based on the ranks of the residuals from this regression to obtain the point estimator. Although this regression-adjusted rank-based estimator is often used in practice, a thorough understanding of its theoretical properties under the design-based framework was missing from the literature due to technical difficulties\footnote{A recent work by \citet{Cattaneo2024} provides
the first theoretical analysis  for Rosenbaum’s rank-based matching estimator in observational studies. Their results, derived under an i.i.d. superpopulation framework, can be seen as complementary to ours.}.
We undertake a systematic study of the asymptotic properties of this estimator, and derive results that can be directly used by practitioners to make an informed choice of an estimator for their data analysis.

\subsection{Summary of our contributions}

\ag{All results in this paper are derived under the finite population asymptotics framework, where the only source of randomness is the completely randomized assignment of the treatment; the potential outcomes are considered fixed. 
We illustrate the robustness of Rosenbaum's rank-based estimator against extreme potential outcomes or contamination using the concept of breakdown point; see~\cref{ABP}. The notion of breakdown point is well-established in nonparametric statistics (see, e.g., \cite{hettmansperger2010}). However, to our knowledge, its formal application to randomized experiments, particularly in the design-based framework, is novel.
}

\ag{
\cref{ABP} shows that the asymptotic breakdown point (henceforth ABP) of Rosenbaum's estimator uniformly dominates the ABP of any weighted average quantile estimator (see~\eqref{def:waq} for a definition). This class includes the difference-in-means estimator, the $\alpha$-trimmed difference-in-means estimator, the $\alpha$-Winsorized difference-in-means estimator, the difference-in-medians estimator, and the estimators proposed by \citet{athey2021}, among others. 
Our notion of the ABP also has implications for the lengths of corresponding confidence intervals; see \cref{remark:zeroABP-infinit-CI}.}

Although the robustness of Rosenbaum's estimator may come at the cost of a loss in efficiency, we argue that this cost is small. \cref{efflb} establishes that the asymptotic efficiency of Rosenbaum's estimator relative to the difference-in-means estimator is bounded below by 0.864 in the worst case, and is substantially higher when the potential outcomes have heavy tails (cf.~\cref{rem:on_table:eff}).
This result is parallel to the lower bound due to~\citet{HL56}  on the asymptotic efficiency of the WRS test relative to Student's $t$-test under an infinite population model.  We note, in particular, that when the potential outcomes behave like realizations from a normal distribution (cf.~\cref{empirical-assump-1}), 
Rosenbaum's estimator is only 5\% less efficient compared to the difference-in-means estimator (which is most efficient in this case); see also \cref{table:eff}.

When pretreatment covariates are available, we show that the regression-adjusted version of Rosenbaum's estimator offers compelling practical advantages.
First, \cref{Jb.geq.Ib} clarifies that under an asymptotic independence condition, Rosenbaum's regression-adjusted estimator is at least as efficient as Rosenbaum's estimator without regression adjustment. Next, we derive a lower bound parallel to \cref{efflb} that compares the asymptotic efficiency of the regression-adjusted Rosenbaum's estimator relative to Lin's estimator \citep{Lin13}. This finding, stated in \cref{efflb-adj-case}, implies that the confidence intervals obtained using Rosenbaum's regression-adjusted estimator will be at most 7.6\% wider, in the worst case, compared to the same for Lin's estimator. When the residuals are heavy-tailed, Rosenbaum's regression-adjusted estimator is often more efficient than Lin's estimator; we refer the reader to \cref{sec:eff-gain-by-adj} for details.
Thus, among estimators whose asymptotic behavior is well understood under the finite population setting, Rosenbaum's estimators make a strong case for use in practice, both in terms of robustness and efficiency.

We also contribute novel technical tools for studying rank-based estimators under the design-based framework. Since Rosenbaum’s estimator inverts the WRS test, deriving its asymptotic distribution requires the asymptotic distribution of the WRS statistic under a sequence of \emph{local alternatives} \citep{HL63}, which does not follow from the existing literature \citep{Li2017,Li2018,Lei2020,Wu2020}. Classical results from Le Cam's local asymptotic normality theory (see, e.g., \citet[Chapter 7]{vandervaart}) also do not apply in our design-based setting where the observed outcomes are neither independent nor identically distributed. We develop the local asymptotic normality of the WRS statistic in randomized experiments (\cref{propo-asy-tauN}) and thus establish the asymptotic distribution of Rosenbaum's estimator (\cref{tau_HL.CLT}). 
Similarly, for the regression-adjusted estimator, we derive the local asymptotic behavior of the regression-adjusted rank statistic (\cref{tNadj.finalCLT}) by carefully tracking how covariate adjustment affects the joint distribution of the ranks of the residuals. This extension requires new arguments beyond the existing techniques for regression-adjusted estimators \citep{Fogarty2018, Li2020, Zhao2021regression}, which focus on smooth functionals rather than rank-based statistics. Our proof techniques may prove useful in future research on other experimental designs or other rank-based estimators, such as the U-statistics-based estimators of \citet{RosenbaumUstatistics}.

While the main focus of this paper is on the study of Rosenbaum's estimator(s) under the assumption of constant treatment effect, it is natural to ask which estimand it targets when this assumption does not hold. In \cref{sec:HTE}, we initiate a study of this problem by deriving a weak limit of Rosenbaum's unadjusted estimator without any assumption on the treatment effect; see~\cref{thm:treathet}. This leads to a novel estimand of the treatment effect in the context of randomized trials which is robust to outliers and contamination. \ag{A recent work by \cite{lei2024} analyzes the same estimator in the infinite-population framework.}

\ag{Our empirical experiments in \cref{sec:simulations,sec:simulationsApp} demonstrate the benefits of using rank-based estimators in settings with heavy tails or contamination}. The proofs of our main results, additional technical results, and discussions are relegated to 
\cref{app:additional-results,sec:simulationsApp,sec: main tools,main:proofs:unadj,main:proofs:adj,AppendixA,sec:sometechlem}.
An implementation of our proposed methods in \texttt{R} is available from \url{https://github.com/ghoshadi/RRE}.

\section{Inference without regression 
adjustment}\label{sec:w/oAdj}

We work in the Neyman-Rubin potential outcomes framework \citep{Neyman1923, Rubin1974, Rubin1977, imbens2015} \ag{and impose the stable unit treatment value assumption (SUTVA)}. For the $i$-th subject ($i=1,\dots,N$), let $a_i$ and $b_i$ denote the potential outcomes under the treatment and the control, respectively; the observed outcome is given by 
$Y_i = Z_i a_i + (1-Z_i) b_i$, 
where $Z_i$ equals $1$ if the $i$-th subject is treated, and $0$ otherwise. We assume throughout this paper that $\{a_i\}_{i=1}^N$ and $\{b_i\}_{i=1}^N$ are fixed constants, i.e., $Z_i$ is the only source of randomness in $Y_i$. 

\begin{assumption}\label{randomization}
The treatment group is formed by choosing $m=m(N)$ out of the $N$ subjects by {\it simple random sampling without replacement}, and ${m}/{N}\to \lambda\in (0,1)$ as $N\to\infty$. 
\end{assumption}
We now posit the assumption of a \emph{constant additive treatment effect} \citep{rosenbaum1993, Rosenbaum02}: 
\begin{assumption}\label{cte} For each $i=1,\dots,N$, we have
    $a_i - b_i = \cte$, 
for some unknown real number $\cte$.
\end{assumption}
Denote vectors by boldface letters, e.g., $\vec{Z} := (Z_1, \dots, Z_N)^\top$, 
$\vec{b} :=(b_1, \dots, b_N)^\top$ and so forth. We are interested in point estimation and confidence interval for the constant treatment effect $\cte$. 
Following \cite{rosenbaum1993, Rosenbaum02}, we estimate $\cte$ by inverting the following testing problem.
\begin{equation}\label{eq:Hyp-Test}
H_0 : \cte = \tau_0 \qquad \mathrm{versus} \qquad H_1:\cte\neq \tau_0.
\end{equation} 
Under $H_0$, the vector $\vec{Y}-\tau_0\vec{Z}$ (called \emph{adjusted responses}) equals $\vec{b}$, which is non-random. Hence any statistic $t(\vec{Z}, \vec{Y}-\tau_0\vec{Z})$, which is a function of the treatment indicators and the adjusted responses, can be used to make randomization inference about $H_0$, since the null distribution of  $t(\vec{Z}, \vec{Y}-\tau_0\vec{Z})= t(\vec{Z}, \vec{b})$ is completely specified by the randomization distribution of $\vec{Z}$. 
For example, one can consider the difference-in-means test statistic: 
\begin{equation}\label{DM-stat}
t_{\mathrm{dm}}(\vec{Z}, \vec{v}) := \frac{1}{m}\sum_{i\,:\, Z_i\, =\, 1} v_i - \frac{1}{N-m}\sum_{i\,:\, Z_i\, =\, 0} v_i,
\end{equation} 
where $\vec{v}$ is a vector of observations of length $N$, and $m=\sum_{i=1}^N Z_i$.  
Mean-based test statistics like \eqref{DM-stat} above are often not robust to heavy-tailed potential outcomes or the presence of contamination/outliers.
On the other hand, test statistics that use the {\it ranks} of the potential outcomes, as opposed to their exact values, are in general less sensitive to thick-tailed or skewed distributions and hence can lead to more powerful tests (see, e.g.,~\citet[Section 5.5.4]{imbens2015}). 
There are various popular choices for the rank-based statistic $t(\cdot,\cdot)$; \citet{rosenbaum1993} recommended using the WRS  statistic~\citep{Wilcoxon1945}, defined as
\begin{equation}\label{WRS_N}
    t(\vec{Z},\vec{Y}-\tau_0\vec{Z}) := \vec{Z}^\top \vec{q}^{(\tau_0)} = \sum_{j : Z_j = 1} q_j^{(\tau_0)},
\end{equation}
where $q_j^{(\tau_0)}$ is the rank of $Y_j - \tau_0 Z_j$ among $\{Y_i - \tau_0 Z_i\}_{i=1}^N$. Considering the possibility of ties in the data, we take the following definition for the ranks, known as \textit{up-ranks}:
\begin{equation}\label{upranks}
q_{j}^{(\tau_0)} := \sum_{i=1}^N \ind{Y_i - \tau_0 Z_i\le Y_j - \tau_0 Z_j},\qquad \mbox{for }\; 1\le j\le N.
\end{equation}
Note, under $H_0:\cte = \tau_0$, the ranks $\{q_{j}^{(\tau_0)} \}_{j=1}^N$ can equivalently be written as
\begin{equation}\label{null-upranks}
q_{j}^{(\tau_0)}=\sum_{i=1}^N \ind{b_{i}\le b_{j}},\quad 1\le j\le N.
\end{equation}
We present the asymptotic null distribution of the WRS statistic in the following result, which is essentially an adaptation of \citet[Corollary 1]{Li2017}. 

\begin{proposition}[Asymptotic null distribution of $t_N$]\label{propo-null1} 
Let $t_N:=t(\vec{Z}_N, \vec{Y}_N - \tau_0\vec{Z}_N)$ be the WRS statistic for a sample of size $N$, with $t(\cdot,\cdot)$ as defined in \eqref{WRS_N}. Suppose that \cref{cte,randomization} hold, and that none of the ranks in \eqref{null-upranks} dominates all the others, that is,
\begin{equation}\label{assumpq}
\lim_{n\to\infty}\frac{\max_{1\le j\le N} \big(q_{j}^{(\tau_0)} - \overline{q}_N^{(\tau_0)}\big)^2}{\sum_{j=1}^N \big(q_{j}^{(\tau_0)} - \overline{q}_N^{(\tau_0)}\big)^2} = 0,
\end{equation}
where $
\overline{q}_N^{(\tau_0)} := N^{-1}\sum_{j=1}^N q_{j}^{(\tau_0)}$.
Then, under $\cte=\tau_0$,
\begin{equation}\label{asymp.null}
N^{-3/2}\left(t_N - m\overline{q}_N^{(\tau_0)}\right) \dto \mathcal{N}\left(0,\frac{\lambda(1-\lambda)}{12}\right).
\end{equation}
\end{proposition}

The limiting distribution in~\eqref{asymp.null} is identical to the asymptotic null distribution of the WRS statistic under the infinite population framework \cite[][Section 1.3]{Lehmann2006}.~\cref{propo-null1} thus justifies Rosenbaum's proposal (see~\citet[Section 4.6]{Rosenbaum2002}) of deriving a confidence interval for $\cte$ by numerically solving for the hypothesized treatment effects that equate the \emph{standardized} WRS statistic to the $\alpha/2$ and $(1-\alpha/2)$ quantiles of the standard normal distribution.  
\ag{However, Rosenbaum's proposal of constructing confidence sets does not provide much insight about the efficiency of this rank-based approach over the mean-based approach. In the subsequent sections, we develop a design-based asymptotic theory for a point estimator $\rsnu$ of $\cte$ based on the WRS statistic and explicitly characterize its robustness and efficiency properties.}

\begin{remark}[Choice of the tie-breaking method]\label{tie-breaking}
Instead of using up-ranks, one can also break ties in the control potential outcomes by (i) comparing indices, (ii) randomization, or (iii) using average ranks. Since the ranks for (i) or (ii) form a permutation of $\{1,2,\dots,N\}$, they satisfy \eqref{assumpq}, and thus~\eqref{asymp.null} holds without any assumption. For (iii), the asymptotic null distribution in~\eqref{asymp.null} holds as long as these ties do not occur in large blocks; see 
\cref{app:avgranks} 
for a formal result.
\end{remark}

\subsection{Rosenbaum's estimator based on the WRS statistic}\label{sec:WRS}

To obtain a point estimator for $\cte$,~\citet{Rosenbaum02} suggested inverting the testing problem \eqref{eq:Hyp-Test} by equating the test statistic to its expectation under the randomization distribution and solving for the hypothesized treatment effect. Thus, the estimator $\widehat\tau$ advocated by~\citet{Rosenbaum02} is a solution of the equation
\begin{equation}\label{rosenbaum_eqn}
    t(\vec{Z}, \vec{Y}-\widehat{\tau}\vec{Z}) = \E_{\tau_0} t(\vec{Z}, \vec{Y}-\tau_0\vec{Z}),
\end{equation}
where $\E_{\tau_0}$ denotes the expectation under $\tau=\tau_0$.
Note that $\E_{\tau_0} t(\vec{Z}, \vec{Y}-\tau_0\vec{Z})=\E_{\tau_0} t(\vec{Z}, \vec{b})$ does not depend on $\tau_0$ and can be explicitly evaluated for typical choices of $t(\cdot,\cdot)$, as illustrated below. 
\begin{example}
Consider the difference-in-means test statistic $t_{\mathrm{dm}}(\cdot,\cdot)$ defined in \eqref{DM-stat}. 
Since $Z_i Y_i = Z_i a_i$, $(1-Z_i)Y_i = (1- Z_i)b_i$, and $\E Z_i = m/N$, we deduce that $\E_{\tau_0} t_{\mathrm{dm}}(\vec{Z}, \vec{Y}-\tau_0 \vec{Z})=0$. Consequently, the solution to \eqref{rosenbaum_eqn} for $\widehat\tau$ based on the difference-in-means statistic is given by
\begin{equation}\label{dm}
\linu := \frac{1}{m}\sum_{i=1}^N Z_i Y_i - \frac{1}{N-m}\sum_{i=1}^N (1-Z_i) Y_i,
\end{equation}
which is the {difference-in-means} estimator.
\end{example}

Next we define Rosenbaum's estimator based on the  WRS statistic $t(\cdot,\cdot)$ defined in~\eqref{WRS_N}. Assume for the moment that ties are not present in the ranks, so that the right hand side of \eqref{rosenbaum_eqn} becomes $$\mu:=\E_{\tau_0} t(\vec{Z}, \vec{Y}-\tau_0 \vec{Z}) = \E_{\tau_0}  \vec{Z}^\top \vec{q}^{(\tau_0)} =  \frac{m}{N}\sum_{j=1}^N q^{(\tau_0)}_j = \frac{m(N+1)}{2}.$$
Considering the possibility that \eqref{rosenbaum_eqn} may not have an exact solution (e.g., if $m(N+1)$ is odd), \citet[Section 2.7.2]{Rosenbaum2002} suggested  modifying the definition of $\widehat{\tau}$ based on the WRS statistic, in the same style as in \citet{HL63}, as follows. 
Define
\begin{equation}\label{eq:Tau-Unadj}
 \rsnu :=\frac{1}{2}\left(\sup\left\{\tau' : t(\vec{Z}, \vec{Y} - \tau'\vec{Z}) >  \mu\right \}+\inf\left\{\tau' : t(\vec{Z}, \vec{Y} - \tau' \vec{Z}) < \mu\right\}\right).
\end{equation}
We show in the following section that $\rsnu$ is robust to outliers (unlike the difference-in-means estimator $\linu$), and study its asymptotic efficiency relative to $\linu$ in~\cref{subsec: eff}.

\subsection{Robustness of Rosenbaum's estimator}\label{sec:robustRosen}

In classical nonparametrics, a natural way to quantify the robustness of an estimator is via its (asymptotic) breakdown point (see, e.g., \citet{hettmansperger2010}). 
Since in our setting the indices corresponding to the treatment and control groups are not deterministic but are rather chosen by randomization, 
we define the (asymptotic) breakdown point of an estimator of the constant treatment effect $\tau$ in the following fashion.
\begin{definition}
[Breakdown point] \label{def:abp}
Consider any estimator $\widehat{\tau} =\widehat{\tau}(Y_1,\dots,Y_N;Z_1,\dots,Z_N)$ of the constant treatment effect $\tau$.  The finite sample breakdown point of $\widehat{\tau}$ is defined as:
\begin{multline}\label{BP_defn}
\mathrm{BP}_{N}\left(\widehat{\tau}\,\right):=\frac{1}{N}\min\bigg\{1\le k\le N : {\exists}\ z_1,\dots,z_N\in\{0,1\} \text{ such that } \sum_{i=1}^N z_i = m,\ \\ {\max_{1\le i_1<i_2<\dots<i_k\le N}} \sup_{y_{i_1},\dots,y_{i_k}\in\R} |\widehat{\tau}(y_1,\dots,y_N;z_1,\dots,z_N)|=\infty \bigg\}.
\end{multline}
Further, the asymptotic breakdown point of $\widehat{\tau}$ is defined as $
    \mathrm{ABP}(\widehat{\tau}\,) := \lim_{N\to\infty} \mathrm{BP}_N(\widehat{\tau}\,)$.
\end{definition}
Intuitively, the above notion of breakdown point formalizes the following question: ``What is the minimum proportion of {responses} that, if replaced with arbitrarily  extreme values, will cause the estimator to be arbitrarily large (in absolute value), for {some} treatment assignment?" Note the emphasis on the phrase ``for {some} treatment assignment'', which is reflected in the maximum taken over $i_1,i_2,\dots,i_k$ in \eqref{BP_defn}. 

\begin{remark}[Implication of our notion of asymptotic breakdown point for confidence intervals] \label{remark:zeroABP-infinit-CI}
\ag{One practical implication of \cref{def:abp} is that any confidence interval based on an estimator with zero asymptotic breakdown point must either (i) fail to maintain uniform validity under small contamination or (ii) have infinite expected length, even under arbitrarily small contamination rates.
In contrast, any estimator with strictly positive ABP can be used to construct confidence intervals with bounded expected length that maintain uniform validity under contamination, provided the contamination level is below the ABP. We refer the reader to 
\cref{sec:A.1}
for definitions and formal results in this direction.
} 
\end{remark}

\ag{
We derive below the asymptotic breakdown point of Rosenbaum's estimator $\rsnu$ and show that it is uniformly higher than that of the weighted average quantile (WAQ) estimators studied by \citet{athey2021}. For any finite signed measure $\nu$ on $[0,1]$ with $\nu([0,1])=1$, define the corresponding WAQ estimator as
\begin{equation}\label{def:waq}
\begin{split}
        \widehat\tau_{\mathrm{waq}}(\nu)= \sum_{i=1}^m &  \nu\left(\left[\frac{i-1}{m},\frac{i}{m}\right]\right)a_{(i)} - \sum_{i=1}^{N-m}\nu\left(\left[\frac{i-1}{N-m},\frac{i}{N-m}\right]\right)b_{(i)},
        \end{split}
\end{equation}
where $a_{(i)}$ (resp.~$b_{(i)}$) are the order statistics in the treatment (resp.~control) group. Note that $\widehat{\tau}_{\mathrm{waq}}(\nu)$ coincides with the difference-in-means estimator when $\nu$ is the uniform measure, and the difference-in-medians estimator when $\nu$ puts unit mass at $1/2$. Other notable estimators in the above class are the $\alpha$-trimmed difference-in-means, $\alpha$-Winsorized difference-in-means, and the estimators proposed by \citet{athey2021}.
 \begin{theorem}[Asymptotic breakdown point of $\rsnu$]
\label{ABP}
Under \cref{randomization}, the asymptotic breakdown point of Rosenbaum's estimator $\rsnu$ uniformly dominates the asymptotic breakdown point of weighted average quantile estimators of the form \eqref{def:waq}. That is,  
\begin{equation}\label{eq:waq-abp-tauR}
    \sup_{\nu} \mathrm{ABP}(\widehat{\tau}_{\mathrm{waq}}(\nu))\le \frac{1}{2}\min\{\lambda,1-\lambda\} = \mathrm{ABP}(\rsnu),
\end{equation}
where $\lambda\in(0,1)$ is the limiting proportion of treated units as in \cref{randomization}, and the supremum is taken over all finite signed measures $\nu$ on $[0,1]$ with $\nu([0,1])=1$.
\end{theorem}
The inequality in \eqref{eq:waq-abp-tauR} is tight: The difference-in-medians estimator, defined as 
\begin{equation}\label{def:taumed}
    \taumed:=\median\{Y_i:Z_i=1\}-\median\{Y_i:Z_i=0\},
\end{equation}
is a WAQ estimator that achieves this upper bound. The choice between the difference-in-medians estimator $\taumed$ and Rosenbaum's estimator $\rsnu$ thus requires a comparison of their asymptotic relative efficiency (cf.~\cref{table:eff,compare-tauR-and-taumed}).}

\ag{Our proof of \cref{ABP} reveals that the ABP of the novel WAQ estimator(s) proposed by \citet{athey2021} can often be zero. This is expected because their estimator(s) are optimized for efficiency. Thus, \cref{ABP} demonstrates the trade-off between robustness and efficiency.  
}

\subsection{Asymptotic distribution of Rosenbaum's estimator}\label{subsec: asy-distn}

In this section, we establish the asymptotic distribution of Rosenbaum's estimator $\rsnu$ defined in~\eqref{eq:Tau-Unadj}. For notational clarity, we index vectors and matrices by subscript $N$ (the total sample size). 
Before presenting our main results, we briefly
 discuss this method of obtaining the asymptotic distribution of $\rsnu$ from that of the test statistic 
 $t_N := t(\vec{Z}_N,\vec{Y}_N-\tau_0 \vec{Z}_N)$. 
 We consider a sequence of local alternatives $\tau_N := \tau_0 - hN^{-1/2}$ for the testing problem \eqref{eq:Hyp-Test}, where $h\in\R$ is fixed. 
We show in \cref{lem1}
that if
\begin{equation} \label{psa}
\lim_{N\to\infty} \mathbb{P}_{\tau_N} \left(N^{-3/2} (t_N - \E_{\tau_0} t_N) \le x\right) = \Phi \left(\frac{x+hB}{A}\right), \quad \text{ for every } \; x\in\R,
\end{equation}
where $\mathbb{P}_{\tau_N}$ denotes the law under  $\tau=\tau_N$, $\Phi(\cdot)$ is the standard normal distribution function, and $A,B$ are positive constants free of $\tau_0$ and $h$, then  $$
\sqrt{N} (\rsnu - \cte) \dto \mathcal{N}(0, A^2/B^2).$$
In view of the above, it suffices to establish \eqref{psa} for the WRS statistic $t_N:=t(\vec{Z}_N, \vec{Y}_N - \tau_0\vec{Z}_N)$, with $t(\cdot,\cdot)$ as defined in \eqref{WRS_N}. \ Although  \citet[Corollary 1]{Li2017} give the asymptotic null distribution of $t_N$, their proof technique does not work for finding the asymptotic distribution of $t_N$ under local alternatives. 
The classical results from Le Cam's local asymptotic normality theory (cf.~\citet[Chapter 7]{vandervaart}) also do not apply in our fixed design setting.
 We present the asymptotic distribution of $t_N$ under a sequence of local alternatives in~\cref{propo-asy-tauN} below. To concisely state an assumption, we introduce the following notation.
Define, for any $h, x\in\R$,
\begin{equation}\label{Ian}
I_{h,N} (x) :=\begin{cases}  \ind{0\le x <hN^{-1/2}} & \text{if } h \geq 0, \\
- \ind{hN^{-1/2} \le x <0} & \text{if } h< 0,\end{cases} 
\end{equation}
where $\ind{\cdot}$ denotes the indicator of a set. 

\begin{assumption}\label{ACjs}
Let $\{b_{N,\,j}:1\le j\le N\}$ be the  control potential outcomes. We assume that, for $I_{h,N}$ as in \eqref{Ian}, the following holds: \begin{equation}\label{eq:ACjs}\lim_{N\to\infty} N^{-3/2} \sum_{j=1}^N\sum_{i=1}^N I_{h,N} (b_{N,\,j} - b_{N,\,i})= h\mathcal{I}_b, \end{equation} for some fixed constant $\mathcal{I}_b\in (0,\infty)$ and for every $h\in\R$.
\end{assumption}

\begin{remark}[On~\cref{ACjs}]\label{remarkA1} 
Intuitively,~\cref{ACjs} says that the proportion of the pairwise differences $(b_{N,\,j} - b_{N,\,i})$ falling into small intervals (shrinking at the rate of $N^{-1/2}$) scales with the lengths of those intervals. Put differently, it imposes that the control potential outcomes behave, in the limit, as if they were drawn from an absolutely continuous distribution; see also \cref{empirical-assump-1}. Indeed, if $b_{N,\,j}$'s are realizations from a distribution with square-integrable density $f(\cdot)$, then~\cref{ACjs} holds almost surely with $\mathcal{I}_b = \int_{\R} f^2(x)dx$ (see 
\cref{iid.I_C} in \cref{special-cases} 
for a  formal statement). 
Unlike moment assumptions, this condition is satisfied by many heavy-tailed distributions (e.g., Cauchy, Student's $t_\nu$, Pareto and so on.). 
\end{remark}

The following theorem is pivotal; it provides the local asymptotic normality for $t_N$, 
which will aid us to study the limiting behavior of $\rsnu$. 

\begin{theorem}[Local asymptotic normality of $t_N$]\label{propo-asy-tauN} Let $t_N:=t(\vec{Z}_N, \vec{Y}_N - \tau_0\vec{Z}_N)$ be the WRS test statistic, with $t(\cdot,\cdot)$ as in \eqref{WRS_N}. Suppose that \cref{cte,randomization,ACjs} hold. Fix $h\in\R$ and let $\tau_N = \tau_0 - hN^{-1/2}$. Then under the sequence of local alternatives $\tau=\tau_N$, 
\begin{equation*}
N^{-3/2} \left(t_N - \frac{m(N+1)}{2}\right)  \dto \mathcal{N}\left(-h\lambda(1-\lambda) \mathcal{I}_b,\frac{\lambda(1-\lambda)}{12}\right),
\end{equation*}
where $\mathcal{I}_b$ is defined  in~\cref{ACjs}. 
\end{theorem}

Equipped with the above result, we are now prepared to present the asymptotic distribution of the estimator $\rsnu$. This is the content of our next theorem. 

\begin{theorem}[CLT for the estimator $\rsnu$]\label{tau_HL.CLT} 
Let $\rsnu$ be as in \eqref{eq:Tau-Unadj}, with $t(\cdot,\cdot)$ as in \eqref{WRS_N}. Then, under \cref{cte,randomization,ACjs}, it holds that 
\begin{equation*} \sqrt{N}\left(\rsnu - \cte \right) \dto \mathcal{N}\left(0, (12\,\lambda(1-\lambda)\mathcal{I}_b^2)^{-1}\right),
\end{equation*}
where $\mathcal{I}_b$ is defined in~\cref{ACjs}. 
\end{theorem}

\ag{~\cref{tau_HL.CLT} shows that the asymptotic variance of $\rsnu$ is characterized by the quantity $\mathcal{I}_b$ in \cref{ACjs}. In \cref{sec:est-avar}
we propose a method to consistently estimate $\mathcal{I}_b$ that readily yields Wald-type confidence intervals based on $\rsnu$. More importantly, \cref{tau_HL.CLT} allows us to compare the asymptotic efficiency of $\rsnu$ relative to other estimators of the constant treatment effect, which we discuss in the next section.
}

\subsection{Efficiency of Rosenbaum's estimator}\label{subsec: eff}

\ag{In this section, our aim is to assess the asymptotic efficiency of Rosenbaum's estimator $\rsnu$ relative to 
the difference-in-means estimator $\linu$ defined in~\eqref{dm} and the difference-in-medians estimator $\taumed$ defined in~\eqref{def:taumed}.} 
In classical statistics, a framework for comparing two consistent, asymptotically normal estimators is through the ratio of their asymptotic variances. 
A precise definition is given below.

\begin{definition}[Asymptotic relative efficiency]\label{eff_defn}
Let $\widehat{\tau}_{N,\,1}$ and $\widehat{\tau}_{N,\,2}$ be two asymptotically normal estimators of $\tau$, in the sense that there exist positive sequences $\sigma_{N,\,1}^2$ and $\sigma_{N,\,2}^2$ such that $\frac{\widehat{\tau}_{N,\,1} - \tau}{\sigma_{N,\,1}}\dto \mathcal{N}(0,1)$ and $\frac{\widehat{\tau}_{N,\,2} - \tau}{\sigma_{N,\,2}}\dto \mathcal{N}(0,1).$
Then the asymptotic relative efficiency of $\widehat{\tau}_{N,\,1}$ with respect to $\widehat{\tau}_{N,\,2}$ is defined as $$\mathrm{eff}(\widehat{\tau}_{N,\,1}, \widehat{\tau}_{N,\,2}) := \lim_{N\to\infty}{\sigma_{N,\,2}^2}/{\sigma_{N,\,1}^2}.$$
\end{definition}

\ag{We now derive the asymptotic efficiency of $\rsnu$ relative to the difference-in-means estimator $\linu$ under the following regularity assumption, which posits that the empirical distribution of the control potential outcomes has a weak limit with a square-integrable density (note that heavy tails are allowed) and satisfies a smoothness condition at the $N^{-1/2}$ scale.} 

\begin{assumption}
    \label{empirical-assump-1}
    Assume that the distribution $F_N$ putting equal mass on the control potential outcomes $\{b_{N,\,j}\}_{\{1\le j\le N\}}$ converges weakly to a distribution $F$ with density $f$ (w.r.t.~the Lebesgue measure on $\R$) such that (a) $N^{-1} \sum_{j=1}^N f(b_{N,\,j})\to \int_{\R} f^2(x)dx<\infty$, and (b) for every fixed $h$,
$\sup_{x\in\R} \left|\sqrt{N}\, \left(F_N\left(x+\frac{h}{\sqrt{N}}\right) - F_N(x)\right) - hf(x)\right| \to 0$ as $N\to\infty$.
\end{assumption}

\begin{theorem}[Efficiency lower bound]\label{efflb}
Under \cref{cte,randomization,empirical-assump-1}, the asymptotic distribution of $\rsnu$ is given by 
\begin{equation}\label{eq:fsquared}
\sqrt{N}\left(\rsnu - \cte \right)  \dto \mathcal{N}\left(0, \left(12\,\lambda(1-\lambda)\right)^{-1}\left(\int_{\R} f^2 (x) dx\right)^{-2}\right).
\end{equation}
Moreover, if (i) $N^{-1}\sum_{i=1}^N (b_{N,\,j} - \overline{b}_N)^2 \to \sigma^2$, and (ii) $N^{-1}\max_{1\le i\le N} (b_{N,\,j} - \overline{b}_N)^2 \to 0$, the asymptotic efficiency of $\rsnu$ relative to $\linu$ is given by
\begin{equation}\label{eq:are}
\mathrm{eff}(\rsnu,\linu)=12\sigma^2\left(\int_{\R} f^2 (x) dx\right)^2.
\end{equation}
\ag{Furthermore, we have the worst-case lower bound:
$ \mathrm{eff}(\rsnu,\linu)\geq 0.864.$ }
\end{theorem}

\ag{Several remarks are now in order.}

\begin{remark}[On the assumptions of~\cref{efflb}]\label{rem:iid_eff} 
\ag{
When $F_N$ converges weakly to $F$, condition (a) in \cref{empirical-assump-1} holds under mild conditions on $f$, e.g., if $f$ is continuous and $\{f(B_N)\}_{N\ge 1}$ are uniformly integrable where $B_N\sim F_N $.  We show in \cref{iid.I_C2} in \cref{special-cases}
that \cref{empirical-assump-1} holds when $\{b_{N,\,j}\}_{\{1\le j\le N\}}$ are realizations from a distribution with density $f$. We also show in the proof of \cref{efflb} that \cref{empirical-assump-1} implies \cref{ACjs}. Lastly, conditions (i) and (ii) in \cref{efflb} are identical to the assumptions in~\citet[Theorem 5]{Li2017}, and suffice for deriving the asymptotic distribution of $\linu$.
}
\end{remark}

\begin{remark}[The $0.864$ lower bound] 
The efficiency lower bound  in~\cref{efflb} coincides with a celebrated efficiency lower bound due to \citet{HL56} in the context of two sample testing under location shift alternatives. They showed that the Pitman efficiency (see, e.g., \citet[Section 14.3]{vandervaart}) of the WRS test relative to Student's $t$-test never falls below $0.864$.
Here, $\rsnu$ is the estimator that inverts the WRS test, and thus, in a way, our~\cref{efflb} mimics the efficiency result of \citet{HL56}.
\end{remark}

\begin{table}
\caption{\label{table:eff} Asymptotic efficiency of Rosenbaum's estimator $\rsnu$ relative to the difference-in-means estimator $\widehat{\tau}_{\mathrm{dm}}$ and the difference-in-medians estimator $\taumed$, for standard choices of the limiting distribution $F$ (as in~\cref{empirical-assump-1,efflb}).\vspace{1mm}}
\centering
\begin{tabular}{@{}c c c c @{}}
Distribution $(F)$ & Density $f(x)$ &
$\mathrm{eff}(\rsnu,\linu)$ &
$\mathrm{eff}(\rsnu,\taumed)$ \\
\hline\noalign{\vskip 0.5ex}
Normal  
& $(2\pi)^{-1/2}\exp(-x^{2}/2)$
& 0.955 & 1.50 \\
Laplace
& $2^{-1}\,e^{-|x|}$
& 1.50 & 0.75 \\
Cauchy
& ${c(1+x^{2})^{-1}} $
& --- & 0.75 \\
Student's $t_{3}$
& $ c\!\left(1+{x^{2}}/{3}\right)^{-2}$
& 1.90 & 1.17 \\
Student's $t_{5}$
& $ c\!\left(1+{x^{2}}/{5}\right)^{-3}$
& 1.24 & 1.29 \\
Logistic & ${e^{-x}}{(1+e^{-x})^{-2}}$
& 1.10 & 1.33 \\
Hyperbolic secant
& $2^{-1}\,\sech\!\big({\pi x}/{2}\big)$
& 1.22 & 1.22 \\
Exponential
& $ e^{-x}\,\mathbf{1}\{x\ge 0\} $
& 3.00 & 3.00 \\
Pareto$(\alpha)$ 
& $ \alpha x^{-(\alpha+1)}\,\mathbf{1}\{x\ge 1\} $
& $\frac{12\alpha^{5}}{(\alpha-1)^{2}(2\alpha+1)^{2}(\alpha-2)}\ge 3$
& $\frac{3\,\alpha^{2}\,2^{\,2+2/\alpha}}{(2\alpha+1)^{2}}\ge 3$  \\
 &  & ($\nearrow\infty$ as $\alpha\to2$) & ($\nearrow 3.84$ as $\alpha\to2$)
\end{tabular}
\end{table}

\begin{remark}[Comparing $\rsnu$ with difference-in-means]\label{rem:on_table:eff}
\ag{
\cref{efflb} implies that in large samples the confidence intervals obtained using Rosenbaum's estimator $\rsnu$ will be at most 7.6\% wider, in the worst case, compared to the same for the difference-in-means estimator.
\cref{table:eff} records the values of the relative efficiency $\mathrm{eff}(\rsnu,\linu)$ for some standard choices of the limiting distribution $F$ (as in~\cref{empirical-assump-1}). In particular, when the control potential outcomes behave as realizations from a normal distribution, $\rsnu$ only suffers a $5\%$ loss of efficiency relative to $\linu$ (which is most efficient in this case). 
\cref{table:eff} also shows an example where $\rsnu$ exhibits infinite gains over $\linu$: In the Pareto($\alpha$) example, the relative efficiency $\mathrm{eff}(\smash\rsnu,\smash\linu)$ approaches $\infty$ as $\alpha\to 2$. }
\end{remark}

\begin{remark}[Comparing $\rsnu$ with difference-in-medians]\label{compare-tauR-and-taumed}
    To our knowledge, the design-based asymptotic distribution of the difference-in-medians estimator $\taumed$ in randomized experiments has not been worked out explicitly in the literature. However, one can combine the asymptotics of the sample quantiles under simple random sampling with replacement from a finite population (see, e.g., \cite{AChatterjee2011,DeyChaudhuri}) together with the proof techniques of \cite{Li2017} to show that, under \cref{cte,randomization,empirical-assump-1}, the randomization distribution of $\sqrt{N}(\taumed-\tau)$ converges weakly to a zero mean Gaussian with variance $1/(4\lambda(1-\lambda)f^2(\mu))$, where $\mu$ is the median under $f$. We use this formula to compute the relative efficiency of $\rsnu$ and $\taumed$ in \cref{table:eff}. It is important to note that while $\rsnu$ outperforms $\taumed$ in most of the examples we consider in \cref{table:eff}, it also shows examples where $\taumed$ is more efficient than $\rsnu$.
    
\end{remark}

    \cref{eff_defn} above provides a principled framework for comparing estimators under a randomization inference. 
In practice, this framework allows us to use a pilot sample to choose an estimator that is expected to be more efficient in the final study. Pilot sampling is a frequently used tool to plan an empirical study \citep{Wittes1990}. In this context, we can estimate $f$ using a pilot sample, then use \cref{efflb} to choose between the difference-in-means and Rosenbaum's estimator by estimating their relative efficiency. \ag{ In fact, \cref{efflb} serves as a `template' for comparing the efficiency of Rosenbaum's estimator relative to any other estimator of the constant treatment effect with known limiting distribution under finite-population asymptotics.}

\section{Inference with regression adjustment}\label{sec:wAdj}


We now study a method advocated by \citet{Rosenbaum02} for drawing randomization inference on the treatment effect with regression adjustment for pre-treatment covariates, in the setup of completely randomized experiments as in~\cref{sec:w/oAdj}. 
Assume that along with the responses $\{Y_i\}_{1\le i\le N}$ we also collect data on $p$ covariates, denoted by $\vec{x}_i$ for the $i$-th subject. Also denote by $\vec{X}$ the $N\times p$ matrix of covariates. \ag{We assume throughout that $\vec{X}$ is fixed and contains the vector of ones.} 

To motivate our estimation strategy, consider again the testing problem \eqref{eq:Hyp-Test}. The method suggested by Rosenbaum is to apply a randomization test on the residuals, as follows. First we calculate the adjusted responses $\vec{Y} - \tau_0 \vec{Z}$ and regress it on the covariates $\vec{X}$ using the least squares criterion. Define the residuals obtained from this fit as:
\begin{equation}\label{e_0}
    \vec{e}_{\tau_0} := \widehat{\eps}(\vec{Y} - \tau_0 \vec{Z}, \vec{X}) = (\vec{I}-\vec{P}_{\vec{X}})(\vec{Y} - \tau_0 \vec{Z}),
\end{equation}
where $\vec{I}$ is the identity matrix of order $N\times N$, and $\vec{P}_{\vec{X}}$ is the projection matrix onto the column space of $\vec{X}$. 
Now we let the residuals $\vec{e}_{\tau_0}$ play the role of the adjusted responses $\vec{Y} - \tau_0\vec{Z}$ while performing the randomization test. Define the WRS test statistic based on the residuals $\vec{e}_{\tau_0}$ as $t_{\mathrm{adj}} := t(\vec{Z},\vec{e}_{\tau_0}) = \vec{q}^\top \vec{Z}$, with $t(\cdot,\cdot)$ as in \eqref{WRS_N} and the up-ranks $\vec{q}$ in \eqref{upranks} being calculated on the residuals instead of the adjusted responses. That is,
 \begin{equation}\label{WRS_defn2} 
 t_{\mathrm{adj}} :=  t (\vec{Z},\vec{e}_{\tau_0}) =  \sum_{i=1}^N Z_{i} \sum_{j=1}^N \ind{e_{\tau_0,j} \le e_{\tau_0,i}}.
 \end{equation} 
{To find a point estimator for $\tau$ based on the above test procedure, \citet{Rosenbaum02} suggested we invert the test as in \cref{sec:WRS}. We denote Rosenbaum's regression-adjusted estimator by $\rsna$, which is defined using \eqref{eq:Tau-Unadj} in the same manner as $\rsnu$, only with the difference that here $t(\vec{Z},\vec{e}_\tau)$ plays the role of $t(\vec{Z},\vec{Y}-\tau\vec{Z})$. 
We set $\mu_{\mathrm{adj}}:=\E_\tau t(\vec{Z}, \vec{e}_\tau)$
and define
\begin{equation}\label{eq:Tau-adj}
 \rsna :=\frac{\sup\{\tau' : t(\vec{Z}, \vec{e}_{\tau'}) > \mu_{\mathrm{adj}} \}+\inf\{\tau' : t(\vec{Z}, \vec{e}_{\tau'}) < \mu_{\mathrm{adj}}\}}{2}.
 \end{equation}
}

Analyzing $\rsna$ presents significant new challenges compared to $\rsnu$. The chief reason being that $t(\vec{Z},\vec{e}_{\tau'})$, which determines $\rsna$, involves indicators of the form $\mathbf{1}(e_{\tau',j}\leq e_{\tau',i})$, each of which depends on the entire random treatment assignment vector $\vec{Z}$ (unlike the unadjusted case where the corresponding indicator depends only on 
$Z_i$ and $Z_j$). We relegate the discussion on how to resolve this technical challenge to 
\cref{sec: main tools} so as not to impede the flow of the paper. 

\subsection{Asymptotic distribution of Rosenbaum's regression-adjusted estimator}\label{sec:3.1}

Consider now a sequence of completely randomized experiments as in~\cref{subsec: asy-distn}. 
Denote the residuals obtained from the least squares fitting of $\vec{Y}_N - \tau_0 \vec{Z}_N$ on the covariates $ \vec{X}_N$ by $\vec{e}_N :=\widehat{\eps}(\vec{Y}_N - \tau_0 \vec{Z}_N,\vec{X}_N)$, with $\widehat{\eps}(\cdot,\cdot)$ as in \eqref{e_0}. 
With $t(\cdot,\cdot)$ as in \eqref{WRS_defn2}, define the WRS statistic based on the residuals from the regression-adjustment as  \begin{equation}\label{WRS_defn3} 
t_{N,\mathrm{adj}} := t(\vec{Z}_N, \vec{e}_N).\end{equation} Finally, define $\rsna$ as in \eqref{eq:Tau-adj} using the statistic $t_{N,\mathrm{adj}}$. 

As in the without regression adjustment case (see \eqref{psa}), we 
reduce the problem of deriving the asymptotic distribution of the estimator $\rsna$ to the problem of finding the asymptotic distribution of the test statistic $t_{N,\mathrm{adj}}$ under the sequence of local alternatives $\tau_N = \tau_0 - hN^{-1/2}$ for a fixed $h\in\R$. 
 Define the regression-adjusted control potential outcomes as 
\begin{equation}\label{btilde}
\widetilde{b}_{N,\,j} := b_{N,\,j}  - \vec{p}_{N,\,j}^\top \vec{b}_N,\quad 1\le j\le N,
\end{equation} 
where $\vec{p}_{N,\,j}$ is the $j$-th column of the projection matrix $\vec{P}_{\vec{X}_N}$ that projects onto the column space of $\vec{X}_N$. {Note that under $\tau=\tau_0$, these $\widetilde{b}_{N,\,j}$'s are identical to the residuals obtained by regressing $\vec{Y}_N -\tau_0 \vec{Z}_N$ on $\vec{X}_N$.}
The following assumption mimics~\cref{ACjs} of~\cref{subsec: asy-distn}. 

\begin{assumption}\label{AssumpB2} 
Let $\{\widetilde{b}_{N,\,j}\}_{j=1}^N$ be the regression-adjusted control potential outcomes as defined in \eqref{btilde}. 
 We assume that, for $I_{h,N}$ as in \eqref{Ian}, the following holds: $$\lim_{N\to\infty} N^{-3/2}\sum_{j=1}^N \sum_{i=1}^N I_{h,N} (\widetilde{b}_{N,\,j} - \widetilde{b}_{N,\,i}) = h\mathcal{J}_b,$$ for some fixed constant $\mathcal{J}_b\in (0,\infty)$ and for every $h\in\R$.
\end{assumption}

\begin{remark}[On~\cref{AssumpB2}]\label{remarkAadj} 
\ag{The substance of the above assumption is that the residuals behave, in the limit, as realizations from an absolutely continuous distribution; see also \cref{empirical-assump-2}. Indeed, when the control potential outcomes satisfy the regression model $\vec{b}_N = \vec{X}_N\vec{\beta}_N+\vec{\eps}_N$, where $\eps_{N,\,i}$ are i.i.d.~from a distribution with mean zero and square-integrable density $g$ (independent of the design $\vec{Z}_N$), we can show that \cref{AssumpB2} holds with $\mathcal{J}_b=\int_{\mathbb{R}} g^2(x)\,dx$; see \cref{iid.J_C} in \cref{special-cases}
for details.
}
\end{remark}

As a precursor to our result on the local asymptotic normality of $t_{N,\mathrm{adj}}$, we first present its limiting null distribution.

\begin{theorem}[Asymptotic null distribution of $t_{N,\mathrm{adj}}$]\label{propo-null2}  Let $t_{N,\mathrm{adj}}$ be the WRS statistic based on the residuals obtained from the least squares fitting of $\vec{Y}_N-\tau_0\vec{Z}_N$ on $\vec{X}_N$, as defined in \eqref{WRS_defn3}. Suppose that~\cref{cte,randomization,AssumpB2} hold. Then, under $\tau=\tau_0$,
\begin{equation}\label{asymp.null.adj}
N^{-3/2}\left(t_{N,\mathrm{adj}} - \frac{m(N+1)}{2}\right) \dto \mathcal{N}\left(0,\frac{\lambda(1-\lambda)}{12}\right).
\end{equation}
\end{theorem}


It is noteworthy that the limiting null distribution of the regression-adjusted statistic $t_{N,\mathrm{adj}}$ given in \eqref{asymp.null.adj} and that of the unadjusted statistic $t_N$ given in \eqref{asymp.null} are identical. However, this is expected, since both $\vec{b}_N$ and $\vec{e}_N$ are deterministic vectors under the null, the ranks always take values in $\{1,2,\dots,N\}$, and $\vec{Z}_N$ is the only source of randomness.

Our next result provides the local asymptotic normality of the statistic $t_{N,\mathrm{adj}}$ under the sequence of local alternatives $\tau_N = \tau_0 - hN^{-1/2}$. This 
immediately yields the limiting distribution of $\rsna$, which is presented in~\cref{tau_adj.final.CLT} below.  

\begin{theorem}[Local asymptotic normality of $t_{N,\mathrm{adj}}$]\label{tNadj.finalCLT} Assume the setting of~\cref{propo-null2}. Fix $h\in\R$ and let $\tau_N = \tau_0 - hN^{-1/2}$. Then, under $\tau=\tau_N$, 
\begin{equation*}
 N^{-3/2}  \left(t_{N,\mathrm{adj}} - \frac{m(N+1)}{2}\right)  \dto \mathcal{N}\left(-h\lambda(1-\lambda) \mathcal{J}_b,\frac{\lambda(1-\lambda)}{12}\right),
\end{equation*}
where $\mathcal{J}_b$ is defined in~\cref{AssumpB2}. 
\end{theorem}


\begin{theorem}[CLT for the estimator $\rsna$]\label{tau_adj.final.CLT}  Under~\cref{cte,randomization,AssumpB2}, it holds that
\begin{equation*}
    \sqrt{N}\left(\rsna - \cte\right) \dto \mathcal{N}\left(0, \left(12\,\lambda(1-\lambda)\mathcal{J}_b^2\right)^{-1}\right),
\end{equation*}
where $\mathcal{J}_b$ is defined in~\cref{AssumpB2}. 
\end{theorem}

{Although \cref{tNadj.finalCLT,tau_adj.final.CLT} are parallel to their unadjusted counterparts (\cref{propo-asy-tauN,tau_HL.CLT}, respectively), their proofs are substantially different, and more involved; see \cref{sec: main tools}
for a discussion on the proof techniques. Similarly to the unadjusted case, \cref{tau_adj.final.CLT} enables us to compare the asymptotic efficiency of $\rsna$ relative to other estimators of $\tau$, which we discuss in the following section.}


\subsection{Efficiency gain by regression adjustment}\label{sec:eff-gain-by-adj}




As discussed in~\cref{sec:w/oAdj}, appropriate randomization inference can be drawn ignoring the  information available on the covariates. However, it is a popular belief \citep{Rosenbaum02} that regression adjustment may increase the efficiency of the inference based on the rank-based statistic, although any formal result supporting this belief was missing in the literature. \ag{We establish one result in this direction in~\cref{Jb.geq.Ib} below,
under the following assumption.
}
\ag{ 
\begin{assumption}\label{empirical-assump-2}
   Assume that the distribution $G_N$ putting equal mass on the covariate adjusted potential outcomes $\{\widetilde{b}_{N,\,j}\}_{\{1\le j\le N\}}$ converges weakly to a distribution $G$ with density $g$ (w.r.t.~the Lebesgue measure on $\R$) such that (a) $N^{-1} \sum_{j=1}^N g(\widetilde{b}_{N,\,j})\to \int_{\R} g^2(x)dx<\infty$, and (b) for every fixed $h$,
$\sup_{x\in\R} \left|\sqrt{N}\, \left(G_N\left(x+\frac{h}{\sqrt{N}}\right) - G_N(x)\right) - hg(x)\right| \to 0$ as $N\to\infty$. Assume further that,
$$\sup_{x,y}\left|\frac{1}{N}\sum_{i=1}^N \ind{\widetilde{b}_{N,\,i}\le x,\, b_{N,\,i}-\widetilde{b}_{N,\,i}\le y}- G_N(x)\frac{1}{N}\sum_{i=1}^N \ind{ b_{N,\,i}-\widetilde{b}_{N,\,i}\le y}\right|=o(1).$$
\end{assumption}
The above assumption mimics \cref{empirical-assump-1} and posits that the empirical distribution of the residuals has a weak limit with a square-integrable density (note that heavy tails are allowed) and satisfies a smoothness condition at the ${N}^{-1/2}$ scale. Moreover, the last display imposes an asymptotic independence-like condition of the residuals $\widetilde{b}_{N,\,i}$ and the predictions ${b}_{N,\,i}-\widetilde{b}_{N,\,i}$. This holds, in particular, when the regression model $\vec{Y}_N=\tau \vec{Z}_N +\vec{X}_N\vec{\beta}_N+\vec{\eps}_N$  holds for the data, where $\eps_{N,i}$'s are i.i.d.~from a distribution with mean zero and square-integrable density $g$ (independent of the design $\vec{Z}_N$); see \cref{iid.J_C2} in \cref{special-cases}
for details. Admittedly, this asymptotic independence-like condition is more restrictive than the conditions under which standard regression adjustment is known to be beneficial \citep{Lin13}, and bridging this gap remains an open question.
\begin{theorem}[Efficiency gain by covariate adjustment]\label{Jb.geq.Ib}
    Under \cref{cte,randomization,empirical-assump-1,empirical-assump-2}, Rosenbaum's regression-adjusted estimator is at least as efficient (in the sense of \cref{eff_defn}) as the unadjusted estimator, i.e.,
$\mathrm{eff}(\rsna, \rsnu)\ge 1.$
\end{theorem}
The practical implication of \cref{Jb.geq.Ib} is that the regression-adjusted estimator $\rsna$ leads to narrower confidence intervals than the unadjusted estimator $\rsnu$ without compromising on the desired significance level. 
\begin{remark}\label{example:heavy-noise}
\cref{Jb.geq.Ib-old} in \cref{special-cases}
covers a special case of \cref{Jb.geq.Ib} in which only the covariates are heavy-tailed, while the noise remains Gaussian. The following synthetic example illustrates the complementary case where heavy-tailed residuals stem from the noise. Taken together, these demonstrate the efficiency gain from regression adjustment across both sources of heavy tails: the noise and the covariates.
\end{remark}
\begin{example}
    Suppose that we have only one covariate $x_i$ and $b_i = \beta x_i + \eps_i$. Assume that $x_i$ are drawn from $\text{Uniform}\{-1,+1\}$ and $\eps_i$ are drawn from any density $g$. We can show that \cref{empirical-assump-1,empirical-assump-2} hold here, and it follows from the inequality $2(a^2+b^2)\ge (a+b)^2$ that
    $$\mathrm{eff}(\rsna, \rsnu) =\frac{\int g^2(t)dt}{\int (\frac{1}{2}g(t-\beta)+\frac{1}{2}g(t+\beta))^2dt} \ge 1,$$
   where the last inequality is strict for $\beta\neq 0$.
\end{example}
}
\ag{
We conclude this section with the following result which gives the asymptotic efficiency of $\rsna$ with respect to the regression-adjusted estimator $\lini$ proposed by \citet{Lin13}.
\begin{theorem}[Efficiency lower bound after regression adjustment]\label{efflb-adj-case}
Suppose that \cref{cte,randomization,empirical-assump-1,empirical-assump-2} hold true. Then, the asymptotic distribution of Rosenbaum's regression-adjusted estimator $\rsna$ is given by 
\begin{equation*}
\sqrt{N}\left(\rsna - \cte \right)  \dto \mathcal{N}\left(0, \left(12\,\lambda(1-\lambda)\right)^{-1}\left(\int_{\R} g^2 (x) dx\right)^{-2}\right).
\end{equation*}
Moreover, assume that (i) $\smash{\sup_{N\ge1} N^{-1}\sum_{i=1}^N {b}_{N,\,i}^4 <\infty}$, $\smash{\sup_{N\ge 1} N^{-1}\sum_{i=1}^N x_{N,\,i,j}^4 <\infty}$ for each coordinate $j\le p$, (ii) $\smash{N^{-1}\vec{X}_N^\top\vec{X}_N}$ converges to a finite, invertible matrix, and (iii) $\smash{N^{-1}\sum_{i=1}^N b_{N,\,i}}$\,, $\smash{N^{-1}\sum_{i=1}^N b_{N,\,i}^2}$\,, $\smash{N^{-1}\sum_{i=1}^N b_{N,\,i} \vec{x}_{N,\,i}}$ converge to finite limits. Then, the asymptotic efficiency of Rosenbaum's regression-adjusted estimator $\rsna$ relative to Lin's estimator $\lini$ is given by
\begin{equation*}
\mathrm{eff}(\rsna,\lini)=12\sigma^2\left(\int_{\R} g^2 (x) dx\right)^2, \quad \text{where}\quad\sigma^2=\lim_{N\to\infty}\frac{1}{N}\sum_{i=1}^{N} (\widetilde{b}_{N,\,j} - \overline{\widetilde{b}}_N)^2.
\end{equation*}
Furthermore, we have the \ag{worst-case} lower bound:
$ \mathrm{eff}(\rsna,\lini)\geq 0.864.$ 
\end{theorem}
The additional moment assumptions (i)--(iii) in \cref{efflb-adj-case} are identical to the assumptions in~\citet[Theorem 1]{Lin13}, and provide sufficient conditions to derive the asymptotic distribution of Lin's estimator $\lini$.}

\ag{
 \cref{efflb-adj-case} shows that the asymptotic relative efficiency of Rosenbaum's regression-adjusted estimator with respect to Lin's estimator has the same form as $\mathrm{eff}(\rsnu,\linu)$ derived in \cref{efflb} in the unadjusted case. As a result, the comparison of $\rsnu$ and $\linu$ as discussed in \cref{rem:on_table:eff} using \cref{table:eff} also holds for comparing the efficiency of Rosenbaum's regression-adjusted estimator $\rsna$ relative to Lin's estimator $\lini$. For other methods, e.g., those in \citet{athey2021}, the asymptotic behaviors are not known under the design-based framework.
 Thus, among the regression-adjusted estimators whose asymptotic behavior is well understood under the finite population setting, \cref{efflb-adj-case} makes a strong case for using Rosenbaum's estimator in practice. 
}

\section{Numerical experiments}\label{sec:simulations}

In this section we give a preview of our extensive numerical experiments to compare the empirical performance of Rosenbaum's rank-based estimators $\rsnu$ and $\rsna$ with various other estimators of the constant treatment effect $\tau$, namely: The difference-in-means estimator $\linu$ (see \eqref{dm}), the difference-in-medians estimator $\taumed$ (see \eqref{def:taumed}), the $\alpha$-trimmed and $\alpha$-Winsorized difference-in-means estimators (with $\alpha=0.1$; see \citet{athey2021} for definitions), the estimators $\taueif$ and $\widehat\tau_{\mathrm{waq}}$ studied by \citet{athey2021}, the simple regression-adjusted estimator $\lina$ \citep{Freedman08a,Freedman08b}, and the estimator $\lini$ proposed by \citet{Lin13}. We present here the empirical performance of these estimators in two simulation settings, as follows: 
\begin{enumerate}
    \item Setting 1: Generate $x_i$ i.i.d.~from $\text{Unif}(-4, 4)$, and set $a_i = 3x_i + \eps_i$ and $b_i = a_i - 2$.
    \item Setting 2: Same as Setting 1, except that we contaminate $5\%$ of the outcomes with an arbitrary large value $M$; here we use $M=500$.
\end{enumerate}
Setting 2 is designed to illustrate the trade-off between robustness and efficiency by examining how contamination affects the confidence intervals.
 In each of these settings, the i.i.d.~noise $\epsilon_i$'s are drawn from: (a) standard normal, (b) Cauchy, and (c) Student's $t_3$ distribution. The sample size is $N=1000$ and here we only report the results for the balanced design with $m/N = 0.5$. 
 We also report the results across additional simulation settings and unbalanced designs in \cref{subsec:simulationsApp}. The design-based limiting distributions of some of the estimators
we consider here are not known in the literature, so we use permutation-based confidence
intervals for an apples-to-apples comparison.
We repeat the above experiments for $1000$ replications, where each replication draws new potential outcomes, covariates and a random treatment assignment. We then compute the  coverage and average lengths of the approximate 95\% permutation-based confidence intervals for each estimator.  
The results are summarized in~\cref{siml_res1,siml_res2}, and some observations are listed below.
\begin{table}[!ht]
\centering
\caption{\label{siml_res1} Empirical coverage and average length of 95\% CIs for simulation Setting 1.}
\setlength{\tabcolsep}{5pt}
\renewcommand{\arraystretch}{1}
\begin{tabular}{l|cc|cc|cc}
\toprule
\multirow{2}{*}{Estimator} &
\multicolumn{2}{c|}{(a) Gaussian errors} &
\multicolumn{2}{c|}{(b) Cauchy errors} &
\multicolumn{2}{c}{(c) $t_3$ errors} \\
& coverage & length & coverage & length & coverage & length \\
\midrule
Difference-in-Means ($\linu$)                      & $94.8\%$  & $1.75$ & $94.2\%$  & $24.26$ & $94.8\%$  & $1.79$ \\
Difference-in-Medians ($\taumed$)                  & $93.1\%$  & $2.90$ & $94.9\%$  &  $3.07$ & $94.7\%$  & $2.91$ \\
$0.1$-trimmed Diff-in-Means                        & $94.1\%$  & $2.03$ & $95.2\%$  &  $2.19$ & $95.2\%$  & $2.04$ \\
$0.1$-Winsorized Diff-in-Means                     & $94.2\%$  & $1.82$ & $95.6\%$  &  $2.01$ & $94.9\%$  & $1.84$ \\
$\taueif$ \citep{athey2021}                        & $96.3\%$  & $1.30$ & $96.2\%$  &  $1.74$ & $96.7\%$  & $1.39$ \\
$\widehat\tau_{\mathrm{waq}}$ \citep{athey2021}    & $95.7\%$  & $1.32$ & $94.5\%$  &  $2.50$ & $96.2\%$  & $1.44$ \\
Rosenbaum's estimator ($\rsnu$)                    & $94.9\%$  & $1.84$ & $95.6\%$  &  $2.14$ & $95.4\%$  & $1.88$ \\
OLS adjusted ($\lina$)                             & $99.0\%$  & $0.35$ & $94.2\%$  & $24.30$ & $96.5\%$  & $0.49$ \\
Lin's estimator ($\lini$)                          & $99.0\%$  & $0.35$ & $94.2\%$  & $24.30$ & $96.5\%$  & $0.49$ \\
Rosenbaum's adjusted ($\rsna$)                     & $99.4\%$  & $0.37$ & $97.2\%$  &  $1.34$ & $98.7\%$  & $0.43$ \\
\bottomrule
\end{tabular}
\end{table}

\begin{table}[!ht]
\centering
\caption{\label{siml_res2} Empirical coverage and average length of 95\% CIs for simulation Setting 2.}
\setlength{\tabcolsep}{5pt}
\renewcommand{\arraystretch}{1}
\begin{tabular}{l|cc|cc|cc}
\toprule
\multirow{2}{*}{Estimator} &
\multicolumn{2}{c|}{(a) Gaussian errors} &
\multicolumn{2}{c|}{(b) Cauchy errors} &
\multicolumn{2}{c}{(c) $t_3$ errors} \\
& coverage & length & coverage & length & coverage & length \\
\midrule
Difference-in-Means ($\linu$)                      & $100.0\%$ & $27.52$ & $100.0\%$ & $43.71$ & $100.0\%$ & $27.51$ \\
Difference-in-Medians ($\taumed$)                  & $94.5\%$  &  $3.04$ & $94.7\%$  &  $3.22$ & $94.9\%$  &  $3.07$ \\
$0.1$-trimmed Diff-in-Means                        & $95.4\%$  &  $2.15$ & $96.7\%$  &  $2.35$ & $96.9\%$  &  $2.16$ \\
$0.1$-Winsorized Diff-in-Means                     & $96.1\%$  &  $1.95$ & $97.9\%$  &  $2.32$ & $97.4\%$  &  $1.98$ \\
$\taueif$ \citep{athey2021}                        & $92.9\%$  &  $1.63$ & $90.2\%$  &  $9.16$ & $92.1\%$  &  $1.68$ \\
$\widehat\tau_{\mathrm{waq}}$ \citep{athey2021}    & $100.0\%$ & $59.17$ & $94.1\%$  & $659.88$ & $100.0\%$ & $56.97$ \\
Rosenbaum's estimator ($\rsnu$)                    & $95.5\%$  &  $1.97$ & $96.4\%$  &  $2.29$ & $96.0\%$  &  $2.01$ \\
OLS adjusted ($\lina$)                             & $100.0\%$ & $27.77$ & $99.9\%$  & $43.85$ & $100.0\%$ & $27.75$ \\
Lin's estimator ($\lini$)                          & $100.0\%$ & $27.77$ & $99.9\%$  & $43.85$ & $100.0\%$ & $27.75$ \\
Rosenbaum's adjusted ($\rsna$)                     & $97.3\%$  &  $0.95$ & $97.4\%$  &  $1.88$ & $97.8\%$  &  $1.01$ \\
\bottomrule
\end{tabular}
\end{table}

\begin{itemize}
    \item \textit{Comparison between Rosenbaum's estimators and~mean-based estimators:} 
    In each of the settings in~\cref{siml_res1}, the length of the CIs constructed using $\rsna$ (resp.~$\rsnu$) are either smaller than or almost equal to the lengths of the CIs constructed using $\lini$ (resp.~$\linu$), without compromising on coverage. The efficiency losses for rank-based estimators are within a small margin as predicted by~\cref{efflb,efflb-adj-case}. 
 \item \textit{Regression adjustment improves precision:} 
The CIs constructed using the regression-adjusted estimator $\rsna$ are substantially shorter than the CIs constructed using the unadjusted estimator $\rsnu$, which validates~\cref{Jb.geq.Ib}. For example, in Setting 1(a), the average
interval length decreases from 1.84 ($\rsnu$) to 0.37 ($\rsna$), a reduction of approximately 80\%.
\item \textit{Robustness against contamination:} The estimator $\taueif$ \citep{athey2021} consistently provides the shortest CIs in \cref{siml_res1} among the unadjusted estimators, but its performance substantially deteriorates under contamination in \cref{siml_res2}. In contrast, the CIs based on Rosenbaum's estimators, difference-in-medians and $0.1$-trimmed/Winsorized difference-in-means are stable under contamination --- as expected from their high breakdown point (cf.~\cref{ABP}; see also \cref{proof:ABP}). 
\item \textit{Overall performance:} 
Across both tables, Rosenbaum's regression-adjusted estimator $\rsna$ demonstrates the most favorable balance of efficiency and robustness: It achieves near-optimal efficiency under light tails and substantially outperforms mean-based methods under heavy tails (columns (b) and (c)) and contamination (\cref{siml_res2}).
\end{itemize}
Additional numerical experiments, including settings with model misspecification and 
a real-data analysis, are provided in \cref{sec:simulationsApp}.

\section{Rosenbaum's estimator under treatment heterogeneity}\label{sec:HTE}

In this section, we initiate a study of Rosenbaum's unadjusted estimator $\rsnu$ when the constant additive treatment effect assumption (\cref{cte}) does not hold. Under the randomization framework (\cref{randomization}), we aim to understand the estimand targeted by $\rsnu$ without any assumption on the treatment effect. \ag{This investigation is valuable because, in practice, Rosenbaum's estimator and the corresponding CIs are also applied in settings where heterogeneous treatment effects are plausible \citep[see, for example,][]{Silber2017,zubizarreta2013effect}.}


\ag{The definition of the estimator $\rsnu$ in \eqref{eq:Tau-Unadj} does not require \cref{cte}. We show below that $\rsnu$ targets the following quantity:}
\begin{align}\label{eq:medpair}
    \mbox{med}_N:=\median\{a_i-b_j:\ 1\le i\neq j\le N\}.
\end{align}
In words, $\mbox{med}_N$ is the median of the pairwise differences in the potential outcomes under treatment and control situations. Note that $\mathrm{med}_N$ is reminiscent of the popularly used estimand for two sample Wilcoxon statistics from classical nonparametrics \citep[e.g.,][Chapter 2.5]{Lehmann2006}. In the special case where the constant treatment effect assumption holds, $\mbox{med}_N=\tau+\median\{b_i-b_j:\ 1\le i\neq j\le N\}=\tau$. 
The following result shows that under appropriate assumptions,  Rosenbaum's estimator $\rsnu$ targets the estimand $\mbox{med}_N$.

\begin{theorem}\label{thm:treathet}
Suppose that \cref{randomization} holds.
Given $\varepsilon>0$, define 
$$\kappa_N^{(1)}(\varepsilon):=\frac{\sum_{1\le i\neq j}\mathbf{1}(a_i-b_j\le \mathrm{med}_N+\varepsilon)}{N(N-1)}, \qquad \kappa_N^{(2)}(\varepsilon):=\frac{\sum_{1\le i\neq j}\mathbf{1}(a_i-b_j \,\ag{\le}\, \mathrm{med}_N-\varepsilon)}{N(N-1)}.$$
Assume further that for any $\varepsilon>0$,
\begin{equation}\label{eq:separation}
    \sqrt{N}\left(\kappa_N^{(1)}(\varepsilon)-\frac{1}{2}\right)\to \infty, \qquad \sqrt{N}\left(\frac{1}{2}-\kappa_N^{(2)}(\varepsilon)\right)\to\infty.
\end{equation}
Then it holds that $\rsnu-\mathrm{med}_N\overset{P}{\longrightarrow} 0$ as $N\to\infty$. 
\end{theorem}

 A few remarks are now in order.

\begin{remark}[On assumption \eqref{eq:separation}]
By definition of median, $\kappa_N^{(2)}(\varepsilon)\le 1/2\le \kappa_N^{(1)}(\varepsilon)$. Assumption \eqref{eq:separation} can be thought of as a weaker finite sample analogue of the positive density assumption, see, e.g., \citet[Theorem 11.2.8]{tsh4}.  In \cite{mizera1998necessary}, the authors establish necessary and sufficient conditions for convergence of the sample median of independent observations. Condition \eqref{eq:separation} is analogous to their necessary condition \cite[Equation (2.2)]{mizera1998necessary}. 
\end{remark}

\begin{remark}[Robustness of $\mbox{med}_N$]
    Being the  median of pairwise differences between treatment and control outcomes, the causal estimand $\mbox{med}_N$ is clearly more robust than the average treatment effect. In fact, the same framework as in \cref{sec:robustRosen} also shows that $\mbox{med}_N$ has an asymptotic breakdown point of $\frac{1}{2}\min\{\lambda,1-\lambda\}$ where $\lambda$ is as in \cref{randomization}. 
\end{remark}

\cref{thm:treathet} opens up a number of questions on the behavior of Rosenbaum's estimator(s) under heterogeneous treatment effects, e.g., deriving the limiting distribution for $\sqrt{N}(\widehat{\tau}_R-\mbox{med}_N)$, or analogue of the same result for the Rosenbaum's regression-adjusted estimator $\rsna$ (defined in \eqref{eq:Tau-adj}). We hope to explore such questions in future work.

\section{Discussion}

\ag{
Randomization remains the gold standard in causal inference for providing valid statistical inference without modeling assumptions about outcomes \citep{Fisher1935,Neyman1923}. Under this framework, inference relies solely on the treatment assignment mechanism, viewing the potential outcomes as fixed rather than random. This design-based perspective has grown increasingly relevant as large-scale randomized experiments have become more common across diverse fields \citep{Banerjee2016,Deaton2018}. A key question in this framework is how to construct estimators that balance robustness against outliers with statistical efficiency. In this paper, we addressed this question by developing the design-based asymptotic theory for Rosenbaum's rank-based estimators---with and without regression adjustment.
}

\ag{
Our results highlight many practical advantages of Rosenbaum's rank-based estimators when outcomes exhibit heavy tails or contain outliers. We showed that Rosenbaum's unadjusted estimator is uniformly more robust than the class of weighted average quantile estimators, which includes many robust estimators recommended in the literature \citep{athey2021}. While in general, a gain in robustness might come at a cost of efficiency, we showed that this loss is minimal for Rosenbaum's estimators. Specifically, the asymptotic relative efficiency of Rosenbaum's unadjusted (resp.~regression-adjusted) estimator is, in the worst case, only $\sim$13.6\% lower than the difference-in-means estimator (resp.~Lin's estimator) under suitable assumptions, and is often substantially higher when the potential outcomes (resp.~residuals from the linear regression of the outcomes on the covariates) exhibit heavy tails. We also established that regression adjustment gives provable efficiency gains under suitable conditions. 
Our asymptotic results lay the groundwork for comparing Rosenbaum's estimators with any alternative estimator whose design-based asymptotic distribution is known, thereby enabling the practitioner to make data-driven choices between available options such as Lin's estimator and Rosenbaum's regression-adjusted estimator. 
}

\ag{
While the assumption of a constant treatment effect provides a valuable baseline for analysis, it can be restrictive in practice.
 There are statistical methods for testing this assumption (see, e.g., \citet{Ding2016}). It would be interesting to investigate the properties of Rosenbaum's rank-based estimators when this assumption is violated.  We initiated this study by deriving a weak limit of Rosenbaum's unadjusted estimator without any assumption on the treatment effect. This leads to a novel estimand of the treatment effect in the context of randomized trials, which is robust to outliers and contamination. This estimand can be a reasonable alternative to the average treatment effect, for example, when only a fraction of the units have a large positive treatment effect and others have no effect or a negative effect; see, e.g.,~\citet{athey2021} and \citet[][Section 1.3.4]{Rosenbaum2021}.
\citet{lei2024} analyzes the same estimator under the infinite-population framework, providing a complementary perspective to our design-based approach. We plan to further develop the theory for this estimand under heterogeneous treatment effects in future work. 
}

\ag{Finally, our analysis is based on ordinary least squares for regression adjustment, which is a standard choice in practice. However, \citet{Rosenbaum02} observed that rank-based methods can naturally accommodate alternative regression approaches, such as quantile regression or robust M-estimators. Extending our asymptotic theory to these settings would be worthwhile, particularly for cases where covariates exhibit heavy tails. The analytical techniques we develop here may prove helpful for such extensions.
}

\section*{Acknowledgement}

This research was partially supported by the U.S.~National Science Foundation. The authors would like to thank Arun Kuchibhotla for an insightful discussion, and Michael Sobel for his valuable feedback on an early draft.

\bibliographystyle{chicago}
\bibliography{references}

\newpage

\appendixpageoff
\appendixtitleoff
\appendix
\part{Appendices} 

\numberwithin{equation}{section}
\counterwithin{figure}{section}
\counterwithin{table}{section}

\normalsize

\setcounter{assumption}{6}


{ This supplementary material begins with additional results and discussions that are omitted from the main paper for space. This includes:
\begin{itemize}\itemsep0mm
\item[(a)] A formal discussion on the implications of the asymptotic breakdown point for confidence intervals (\cref{sec:A.1}).
\item[(b)] Analogous results that hold for tie-breaking using average ranks, instead of the up-ranks we used in the main paper (\cref{app:avgranks}).
\item[(c)] Interpreting our regularity assumptions from the main paper under sufficient conditions in familiar settings (\cref{special-cases}) 
\item[(d)] Consistent estimators for the asymptotic variances of Rosenbaum's estimators $\rsnu$ and $\rsna$ (\cref{sec:est-avar})
\item[(e)] Numerical experiments to empirically compare the robustness and efficiency  of the rank-based confidence intervals with the same for various other estimators (\cref{sec:simulationsApp})
\item[(f)] Outline of the key ideas involved in the proofs of our main results (\cref{sec: main tools})
\item[(g)] Proofs of all our main results and supporting lemmas   (\cref{main:proofs:unadj,main:proofs:adj,AppendixA,sec:sometechlem}) 
\end{itemize}}

\parttoc 

\section{Additional results}\label{app:additional-results}

\subsection{Implication of breakdown points for confidence intervals}\label{sec:A.1}

The asymptotic breakdown point (ABP) of an estimator characterizes its robustness to outliers or data contamination. In this section, we formalize the connection between breakdown points and the construction of confidence intervals that remain `uniformly valid under contamination'. The following definition formalizes the notion of contamination and uniform validity under contamination. Informally, we define contamination as a function that alters some of the outcomes arbitrarily, and we allow this function to depend on the treatment assignment.

\begin{definition}\label{def:uniform-validity}
    Define $\mathcal{Z}_{N}:=\{z\in\{0,1\}^n:\sum_{i=1}^N z_i = m\}$ where $m=m(N)$ satisfies $m/N\to\lambda\in(0,1)$. For any $\eps\in (0,1)$. The $\eps$-contamination class $\mathscr{C}_{N,\eps}$ is the collection of all measurable maps $\psi:\R^N\times \mathcal{Z}_{N} \to \R^N$ such that $$|\{\psi(y, z)\neq y\}|\le \lfloor \eps N\rfloor,$$ where $\psi(y, z)-y$ can take arbitrary values. We say that a map $\widehat{C}_N:\R^N\times \mathcal{Z}_{N}\to \mathcal{B}(\R)$, that constructs an $(1-\alpha)$-confidence set for the constant treatment effect $\tau$ under \cref{randomization,cte}, is uniformly  valid under $\eps$-contamination if 
    $$\inf_{\psi\in \mathscr{C}_{N,\eps}}\mathbb{P}_\tau\left(\tau\in\widehat{C}_N(\psi(Y, Z), Z)\right)\ge 1-\alpha.$$
    
\end{definition}

We now formally establish the connection between the asymptotic breakdown point and the feasibility of constructing confidence intervals that remain uniformly valid under contamination without being blown up in length. The following result makes this precise; see \cref{proof:thm:CI-contamination} for a proof.

\begin{proposition}\label{thm:CI-contamination} 
Suppose that \cref{randomization,cte} hold true. Consider an estimator $\widehat{\tau}(Y_1,\dots,Y_N;Z_1,\dots,Z_N)$ and a procedure to construct an $(1-\alpha)$-confidence set $\widehat{C}_N$ for the constant treatment effect $\tau$. Suppose that $\widehat{C}_N(\vec{y},\vec{z})$ contains $\widehat{\tau}(\vec{y},\vec{z})$ for all outcomes $\vec{y}\in\R^n$ and for all treatment assignments $\vec{z}\in \mathcal{Z}_N$. Assume further that $\widehat{C}_N$ is uniformly valid under $\eps$-contamination (as in \cref{def:uniform-validity}) for all large $N$, and that $\mathrm{ABP}(\widehat\tau\,)=0$. Then, for all large sample sizes $N$, 
$$\sup_{\psi\in \mathscr{C}_{N,\eps}}\E_\tau \left[\mathrm{Length}\left(\widehat{C}_N(\psi(Y, Z), Z)\right)\right]=\infty.$$
On the other hand, if $\mathrm{ABP}(\widehat\tau\,)=\gamma>0$, then for any $0<\eps<\gamma$ one can construct an $(1-\alpha)$-confidence set $\widehat{C}_N$ for the constant treatment effect $\tau$ which is uniformly valid under $\eps$-contamination for all large $N$, is such that $\widehat{C}_N(\vec{y},\vec{z})$ contains $\widehat{\tau}(\vec{y},\vec{z})$ for all outcomes $\vec{y}\in\R^n$ and for all treatment assignments $\vec{z}\in \mathcal{Z}_N$, and satisfies $$\sup_{\psi\in \mathscr{C}_{N,\eps}}\E_\tau \left[\mathrm{Length}\left(\widehat{C}_N(\psi(Y, Z), Z)\right)\right]<\infty.$$
\end{proposition}

The above result reveals a sharp dichotomy: Estimators with zero ABP are fundamentally unsuitable for constructing confidence intervals that are both uniformly valid and informative under contamination, while estimators with positive ABP can achieve both goals simultaneously, provided that the contamination level is lower than the ABP. \cref{thm:CI-contamination} has immediate consequences for the class of weightd average quantile (WAQ) estimators, including the estimators proposed and studied by \citet{athey2021}, which we explicitly state in the following remark.

\begin{remark}[ABP of WAQ estimators]\label{remark:abp_waq}
We show in the proof of \cref{ABP} (cf.~\cref{proof:ABP}) that
\begin{equation}
    \label{abp-of-waq2}
    \mathrm{ABP}(\widehat{\tau}_{waq}(\nu))=\min\{\alpha_{-}(\nu),\alpha_{+}(\nu)\}\min\{\lambda,1-\lambda\},
\end{equation}
    where 
    $$\alpha_{-}(\nu):=\sup\{\alpha:\nu([0,s])=0\,\forall\, s\le \alpha\},\quad \alpha_{+}(\nu):=\sup\{\alpha:\nu([1-s,1])=0\,\forall\, s\le\alpha\}.$$
A consequence of \eqref{abp-of-waq2} is that
       $\mathrm{ABP}(\widehat{\tau}_{waq}(\nu))=0$ whenever $\alpha_{-}(\nu)=0$ or $\alpha_{+}(\nu)=0$, i.e., when $\nu$ places non-zero mass arbitrarily close to one of the endpoints. In particular, with the weights $w_f(u)$ as defined in \citet[Equation (2.12)]{athey2021}, where $f$ is a density on $\R$, the WAQ estimator (see~\citet[Equation (2.13)]{athey2021} for a precise definition) has zero ABP unless there exists an $a\in\R$ such that $(\log f)''=0$ a.e.~on $[a,\infty)$ or  $(\log f)''=0$ a.e.~on $(-\infty,a]$. In other words, if the control potential outcomes are i.i.d.~from any distribution with a log-density that does not vanish linearly on either tails (e.g., Gaussian, Student-$t_d$, Logistic, Pareto, or Cauchy), then the WAQ estimator proposed by \citet[Equation (2.13)]{athey2021} has zero ABP. On the other hand, if the control potential outcomes are drawn from Laplace, then their estimator reduces to the difference-in-medians estimator, which achieves the highest asymptotic breakdown point in the class of weighted average quantile estimators (cf.~\cref{ABP} in the main paper). This result is not surprising, because \citet{athey2021} design estimators that are optimized for efficiency, and thus the above only highlights the trade-off between robustness and efficiency. 
\end{remark}


The second part of Proposition \ref{thm:CI-contamination} guarantees that when $\mathrm{ABP}(\widehat\tau)>0$, we can construct confidence sets with finite expected length under 
$\eps$-contamination for each fixed sample size $N$. However, this does not control how the expected length behaves as the sample size $N$ increases. The following result establishes that for Rosenbaum's  (unadjusted) estimator, the expected length remains uniformly bounded across all large sample sizes; see \cref{proof:thm:rsnu-CI-contamination} for a proof.

\begin{proposition}
    \label{thm:rsnu-CI-contamination}
    Suppose that \cref{randomization,cte} hold true, and that $0<\eps<\mathrm{ABP}(\rsnu)$. Assume further that the empirical measure on the control potential outcomes $b_{N,1},\dots,b_{N,N}$ is uniformly tight. Then, one can construct an $(1-\alpha)$-confidence set $\widehat{C}_N$ for $\tau$ such that $\widehat{C}_N(\vec{y},\vec{z})$ contains $\widehat{\tau}(\vec{y},\vec{z})$ for all outcomes $\vec{y}\in\R^n$ and for all treatment assignments $\vec{z}\in \mathcal{Z}_N$, $\widehat{C}_N$ is uniformly valid under $\eps$-contamination (as in \cref{def:uniform-validity}), and satisfies $$\limsup_{N\to\infty}\sup_{\psi\in \mathscr{C}_{N,\eps}}\E_\tau \left[\mathrm{Length}\left(\widehat{C}_N(\psi(Y, Z), Z)\right)\right]<\infty.$$
\end{proposition}

We demonstrate the empirical implications of \cref{thm:CI-contamination,thm:rsnu-CI-contamination} in Section \ref{sec:simulations} of the main paper. In particular, our simulation Setting 2 provides a stylized contamination scenario, admittedly artificial, but constructed to expose the vulnerabilities of WAQ estimators (as characterized in \cref{ABP} and \cref{remark:abp_waq}). Under this contamination setting, the confidence intervals constructed using WAQ estimators (including the estimators proposed by \citet{athey2021}) exhibit both degraded coverage and inflated lengths, while Rosenbaum's rank-based intervals maintain both nominal coverage and stable length; see \cref{siml_res2.app,siml_res2.5}.

\subsection{Breaking ties using average ranks}\label{app:avgranks}

In the main paper, we used \emph{up-ranks} to break ties in the responses (see~\eqref{upranks} in the main paper for a definition). We discuss in this section the analogous results when ties are broken using average ranks. Suppose that when we sort the control potential outcomes $\left\{b_{N,j}:1\le j\le N\right\}$ in ascending order, the first $c_{1}$ many are equal, then the next $c_{2}$ many are equal, and so on. Assume that we get $k$ such blocks of sizes $c_{1}, \dots, c_{k}$, where the $c_{i}$'s can equal to 1 as well. Set $c_{0}=0$ and $s_{j}=\sum_{i=1}^{j} c_{i}$. Then for each $1 \leq j \leq k$,  define
$$
q_{N}^{\mathrm{avg}}\left(s_{j-1}+1\right)=q_{N}^{\mathrm{avg}}\left(s_{j-1}+2\right)=\cdots=q_{N}^{\mathrm{avg}}\left(s_{j}\right):=\frac{s_{j-1}+1+s_{j}}{2} = s_{j-1} + \frac{1+c_j}{2}
$$
where the last quantity is simply the average of the ranks $\left\{s_{j-1}+1, \ldots, s_{j-1}+c_j\right\}$. With the above notion of average ranks, we define the Wilcoxon rank-sum (WRS) statistic as $$t_N^{\mathrm{avg}} := \sum_{j=1}^N q_N^{\mathrm{avg}}(j)Z_{N,j}.$$
We make the following assumption on the block-sizes, which essentially tells us that the blocks formed by the ties are not too large.

\begin{assumption}\label{blocksizes} 
The block sizes $c_1,\dots,c_k$ as defined above satisfies
$\max _{1 \leq i \leq k} c_{i}=o\left(N\right)$, as $N \to \infty$.
\end{assumption}

 The following lemma shows that if~\cref{blocksizes} holds, then the statistic $t_N^{\mathrm{avg}}$ has the same asymptotic distribution as the WRS statistic $t_N$ defined using up-ranks that we studied in~\cref{sec:w/oAdj} of the main paper; see \cref{proof:tavgrank} for a proof.

\begin{lemma}\label{tavgrank} 
It holds under~\cref{blocksizes} that $N^{-3/2}(t_N^{\mathrm{avg}} - t_N) = o_p(1)$ as $N\to\infty$. Thus, $t_N^{\mathrm{avg}}$ has the same asymptotic distribution as the WRS statistic $t_N$ defined using up-ranks.
\end{lemma}

Our proof of the above result shows that the above statement holds for any value of the constant treatment effect $\tau$. This implies that the local asymptotic normality we establish in~\cref{propo-asy-tauN} of the main paper also holds if we replace $t_N$ by $t_N^{\mathrm{avg}}$. Consequently, Rosenbaum's estimator for the constant treatment effect $\tau$ constructed by inverting the WRS test defined using $t_N^{\mathrm{avg}}$ instead of $t_N$ must have the same asymptotic distribution as described in~\cref{tau_HL.CLT}. A similar analogy also holds in the regression adjusted case, which we omit here for space.

\subsection{Discussion on our regularity assumptions}\label{special-cases}

 The asymptotic results we present in the main paper rely on regularity conditions (namely, \cref{ACjs,empirical-assump-1,AssumpB2,empirical-assump-2}) stated in terms of empirical processes and convergence properties, that might initially appear as abstract. While this generality is essential for our design-based framework, it naturally raises the question: What do these assumptions substantively require? In this section, we verify these assumptions in familiar scenarios under the infinite-population framework. 
 
\subsubsection{The control potential outcomes are i.i.d.}

In the main paper, we mention that the substance of \cref{ACjs} is  that the sequence $\{b_{N,j}: 1\le j\le N\}$ of potential control outcomes behaves like an i.i.d. sequence in the limit (see also~\cref{empirical-assump-1}). The following result verifies that if the $b_{N,j}$'s are i.i.d.~realizations from a distribution with square-integrable density $f_{}(\cdot)$, then \cref{ACjs} holds; see \cref{proof:iid.I_C} for a proof.

\begin{lemma}\label{iid.I_C} Let the control potential outcomes $\{b_{N,j}: 1\le j\le N\}$ be i.i.d. from a distribution with density $f_{}(\cdot)$ satisfying $\int_{\R} f_{}^2(x)\,dx <\infty$, and $I_{h,N}$ be as in \eqref{Ian} of the main paper. Then
\begin{equation*} N^{-3/2}\sum_{j=1}^N\sum_{i=1}^N I_{h,N} (b_{N,j} - b_{N,i})  \asto h\int_{\R} f_{}^{2} (x)\, dx.\end{equation*} 
\end{lemma}

Our next result shows that by imposing more conditions on the density $f$, we can show \cref{empirical-assump-1} holds when $b_{N,j}$ be i.i.d. from the distribution with density $f(\cdot)$; see \cref{proof:iid.I_C2} for a proof.

\begin{lemma}\label{iid.I_C2} Let the control potential outcomes $\{b_{N,j}: 1\le j\le N\}$ be i.i.d. from a distribution with a square-integrable, continuous density $f_{}(\cdot)$ that has finit limits as $x\to\pm \infty$. Denote by $F_N$ the empirical distribution of $b_{N,1},\dots,b_{N,N}$. Then, (a) $F_N$ converges weakly to the distribution with density $f$, (b)  we have $N^{-1}\sum_{i=1}^N f(b_{N,i})\to \int_\R f^2(x) dx$, and (c) $F_N$ satisfies the following smoothness condition:
$$\sup_{x}\left|\sqrt{N}\left(F_N\left(x+\frac{h}{\sqrt{N}}\right)-F_N(x)\right)-h\,f(x)\right|\asto 0.$$
\end{lemma}

\cref{iid.I_C,iid.I_C2} establish that \cref{ACjs,empirical-assump-1} hold when the control potential outcomes are i.i.d. from a distribution with square-integrable, uniformly continuous density. Notably, these conditions do accommodate heavy-tailed distributions without finite variance. We next turn to the regression-adjusted case, where we verify analogous results under a correctly specified linear model.

\subsubsection{The regression model is correctly specified}

In the main paper, we mention that the substance of \cref{AssumpB2} is  that the sequence $\{b_{N,j}: 1\le j\le N\}$ of potential control outcomes in the limit behaves like realizations of the linear model $\vec{b}_{N}=\vec{X}_N\vec{\beta}_N + \vec{\eps}_N$ where $\eps_{N,i}$ are i.i.d.~realizations from a distribution with zero mean, independent of the design (see also~\cref{empirical-assump-2}).  The following result makes this precise; see \cref{proof:iid.J_C} for a proof.

\begin{lemma}\label{iid.J_C}  Assume that the control potential outcomes $b_{N,i}$ satisfy $\vec{b}_{N}=\vec{X}_N\vec{\beta}_N + \vec{\eps}_N$ where $\eps_{N,i}$ are i.i.d.~from a distribution $G$, independent of the design $\vec{Z}_N$. Assume that $G$ has a square-integrable, continuous density $g$ such that $g(x)$ has finite limits as $x\to\pm \infty$. Also assume that $\eps_{N,i}$ have zero mean and $N^{-1}\vec{X}_N^\top\vec{X}_N\to\Sigma\succ 0$. Then,
\begin{equation*} N^{-3/2}\sum_{j=1}^N\sum_{i=1}^N I_{h,N} (\widetilde{b}_{N,j} - \widetilde{b}_{N,i})  \Pto h\int_{\R} g_{}^{2} (x)\, dx.
\end{equation*} 
Further, if $\|\vec{x}_{N,i}\|$ are uniformly bounded, i.e., $\sup_{N}\max_{i\le N}\|\vec{x}_{N,i}\|_2<\infty$, then \begin{equation*} N^{-3/2}\sum_{j=1}^N\sum_{i=1}^N I_{h,N} (\widetilde{b}_{N,j} - \widetilde{b}_{N,i})  \asto h\int_{\R} g_{}^{2} (x)\, dx.\end{equation*} 
\end{lemma}

It is important to note that the above result allows heavy-tailed errors that possibly do not have any moment beyond the first, and also allows heavy-tailed covariates.
To complete the verification of our regularity conditions in the regression-adjusted case, we must also establish \cref{empirical-assump-2}. Our next result verifies that  \cref{empirical-assump-2} holds under the conditions of \cref{iid.J_C} and further assumptions on the fixed design matrix $\vec{X}_N$; see \cref{proof:iid.J_C2} for a proof. 

\begin{lemma}\label{iid.J_C2} Assume that the control potential outcomes $b_{N,i}$ satisfy $\vec{b}_{N}=\vec{X}_N\vec{\beta}_N + \vec{\eps}_N$ where $\eps_{N,i}$ are i.i.d.~from a distribution $G$, independent of the design $\vec{Z}_N$. Assume that $G$ has a square-integrable, continuous density $g$ such that $g(x)$ has finite limits as $x\to\pm \infty$. Also assume that $\eps_{N,i}$ have zero mean, $N^{-1}\vec{X}_N^\top\vec{X}_N\to\Sigma\succ 0$, and $\sup_{N}\max_{i\le N}\|\vec{x}_{N,i}\|_2<\infty$. Denote by $G_N$ the empirical distribution of the residuals $\widetilde{b}_{N,1},\dots,\widetilde{b}_{N,N}$. Then, (a) $G_N$ converges weakly to $G$ a.s., and $N^{-1}\sum_{i=1}^N g(\widetilde{b}_{N,i})\asto \int_\R g^2(x) dx$, (b) for every fixed $h$, $G_N$ satisfies the following smoothness condition:
$$\sup_{x}\left|\sqrt{N}\left(G_N\left(x+\frac{h}{\sqrt{N}}\right)-G_N(x)\right)-h\,g(x)\right|\asto 0.$$
If, in addition, the empirical distribution on $\{e_i^\top \vec{X}_N\vec{\beta}_N\}_{1\le i\le N}$ converges weakly to a continuous distribution, then we have the following asymptotic independence-like condition:
\begin{equation}\label{eq:indep}
    \sup_{x,y}\left|\frac{1}{N}\sum_{i=1}^N \ind{\widetilde{b}_{N,\,i}\le x,\, b_{N,\,i}-\widetilde{b}_{N,\,i}\le y}- G_N(x)\frac{1}{N}\sum_{i=1}^N \ind{ b_{N,\,i}-\widetilde{b}_{N,\,i}\le y}\right|\asto 0.
\end{equation}

\end{lemma}

\cref{iid.J_C,iid.J_C2} establish that our design-based assumptions hold under correct model specification with minimal regularity conditions. In particular, these results does not require any moment assumptions on the errors beyond the first.

\subsubsection{Heavy-tailed covariates and Gaussian noise}\label{sec:Jb.geq.Ib-old}
The following result is a special case of \cref{Jb.geq.Ib}, in which the residuals $\widetilde{b}_{N,i}$ are possibly heavy-tailed only because the covariates might have heavy-tails, and the errors are Gaussian; see \cref{proof:Jb.geq.Ib-old} for a proof. 

\begin{proposition}\label{Jb.geq.Ib-old}
Suppose that \cref{cte,randomization} hold, and that the potential control outcome sequence $\vec{b}_{N}$ satisfy the regression model $\vec{b}_N=\vec{X}_N\vec{\beta}_N+\vec{\eps}_N$, where $\eps_{N,i}$'s are i.i.d.~from $\mathcal{N}(0,\sigma^2)$.  
Then~\cref{AssumpB2} holds, with $\mathcal{J}_b=(2\sqrt{\pi}\sigma)^{-1}$.
Further, denote $\vec{v}_N :=   \vec{X}_N\vec{\beta}_N$ and assume that $\lim_{N\to\infty} N^{-2}\sum_{j=1}^N\sum_{i=1}^N \exp\left(-(v_{N,j} - v_{N,i})^2/4\sigma^2\right)=\ell$. Then,~\cref{ACjs} holds with $\mathcal{I}_b=\ell\mathcal{J}_b$, and consequently, $$\mathcal{J}_b\ge\mathcal{I}_b,\quad\text{i.e.,}\quad \mathrm{eff}(\rsna, \rsnu)\ge 1.$$ Moreover, if it holds that
$\liminf_{N\to\infty} N^{-1}\sum_{j=1}^N (v_{N,j} - \overline{v}_N)^2 > 0,$ 
where $\overline{v}_N := N^{-1}\sum_{i=1}^N v_{N,i}$, then $\mathcal{J}_b>\mathcal{I}_b$, and consequently $\mathrm{eff}(\rsna, \rsnu)$ is strictly greater than $1$. 
\end{proposition}

The above result, together with \cref{example:heavy-noise} of the main paper, shows that the efficiency gain from regression-adjustment (as shown in \cref{Jb.geq.Ib}) can occur when either the noise or the covariates are heavy-tailed.

\subsection{Consistent estimation of the asymptotic variances}\label{sec:est-avar}

In this section, we propose consistent estimators of the asymptotic variances of Rosenbaum's rank-based estimators $\rsnu$ (without regression adjustment) and $\rsna$ (with regression adjustment) that yield Wald-type confidence intervals for the constant treatment effect $\tau$. 

First consider the case without regression adjustment.~\cref{tau_HL.CLT} in the main paper 
tells us that the problem of estimating the asymptotic variance of $\rsnu$  reduces to the problem of estimating the limiting quantity $\mathcal{I}_b$ defined in \cref{ACjs}, which is unknown in general. We propose an estimator of $\mathcal{I}_b$ and prove its consistency in the following result; see \cref{proof:est.I_C} for a proof.

\begin{proposition}[Consistent estimation of $\mathcal{I}_b$]\label{est.I_C} 
Define \begin{equation}\label{est.basic}\widehat{\mathcal{I}}_{N} := \left(1-\frac{m}{N}\right)^{-2}N^{-3/2} \sum_{j=1}^N \sum_{i=1, i\neq j}^N (1-Z_{N,i})(1-Z_{N,j}) \ind{0\le Y_{N,j} - Y_{N,i}< N^{-1/2}}.\end{equation}
Then, under~\cref{cte,randomization,ACjs},  we have $\widehat{\mathcal{I}}_{N}  \Pto \mathcal{I}_b$ as $N\to\infty$. 
\end{proposition}


The above theorem in conjunction with~\cref{tau_HL.CLT} readily yields an asymptotically valid Wald-type confidence interval for $\tau$, which is formally stated below.

\begin{corollary}[Confidence interval for $\tau$]\label{CI1} Let $\widehat{\mathcal{I}}_{N}$ be as defined in \eqref{est.basic}.
Under the assumptions of~\cref{tau_HL.CLT}, an approximate $100(1-\alpha)\%$ confidence interval for $\tau$ is given by 
$$\rsnu \pm \frac{z_{\alpha/2}}{\sqrt{N}} \left(12\frac{m}{N}\left(1-\frac{m}{N}\right)\widehat{\mathcal{I}}_{N}^{2}\right)^{-1/2},$$
where $z_{\alpha}$ denotes the upper $\alpha$-th quantile of the standard normal distribution.
\end{corollary}

Similarly, \cref{tau_adj.final.CLT} reduces the problem of estimating the asymptotic variance of the regression-adjusted estimator $\rsna$ to estimating the quantity $\mathcal{J}_b$ defined in \cref{AssumpB2}. However, the estimation of $\mathcal{J}_b$ is far more intricate than estimating $\mathcal{I}_b$, since $\mathcal{J}_b$ is the limit of a sum of indicators, each of which involves a linear combination of all the observations. Naturally, it is difficult to find a natural extension of~\cref{est.I_C} in the regression-adjusted case.
 In the following section, we propose an alternate method that applies to both the unadjusted case and the regression-adjusted case.

\subsubsection{Plug-in estimators of asymptotic variances}\label{plug-in est}

 In this section, we propose plug-in estimators of $\mathcal{I}_b$ and $\mathcal{J}_b$ that lead to consistent estimators of the asymptotic variances of $\rsnu$ and $\rsna$ (as derived in \cref{tau_HL.CLT,tau_adj.final.CLT} respectively). Broadly, our approach is to substitute $\tau$ by an estimator $\widehat{\tau}_N$ of $\tau$ in various quantities, such as, the indicators that define $\mathcal{J}_b$ or $\mathcal{I}_b$. 
 
 To set ideas, we consider the unadjusted case first. Using the relation $b_{N,j}=Y_{N,j}-\tau Z_{N,j}$, we define a plug-in estimator of $b_{N,j}$ as 
%
\begin{equation}\label{b.hat}
\widehat{b}_{N,j} := Y_{N,j} - \widehat{\tau}_N  Z_{N,j} =  b_{N,j} - (\widehat{\tau}_N -\tau)Z_{N,j},
\end{equation}
where $\widehat{\tau}_N$ is an estimator of $\tau$.
At this point, the next intuitive step would be to replace $(b_{N,1},\ldots ,b_{N,N})$ with $(\widehat{b}_{N,1},\ldots ,\widehat{b}_{N,N})$ in~\cref{ACjs} of the main paper, i.e., replace the term  $I_{h,N}(b_{N,j}-b_{N,i})$ on the left hand side of~\eqref{eq:ACjs}), with $I_{h,N}(\widehat{b}_{N,j}-\widehat{b}_{N,i})$ to construct the estimator. However this leads to some technical problems in proving consistency, since $I_{h,N}(\cdot)$ is a discontinuous function which itself changes with $N$. This requires us to refine our estimation strategy further and impose a slightly stronger condition than that in~\cref{ACjs}. 
Towards this, for $\nu>0$ and $N\ge 1$, define
\begin{equation}\label{Ian.new}
I_{h,N,\nu} (x) :=\begin{cases}  \ind{0 \le x < hN^{-\nu}} & \text{if } h \geq 0, \\ 
- \ind{hN^{-\nu} \le x  < 0} & \text{if } h< 0.\end{cases} 
\end{equation}


\begin{assumption}\label{AssumpIb.new}
For $h\in\R$ let $I_{h,N,\nu}$ be defined in \eqref{Ian.new}. We assume that there exists $0<\nu< 1/2$, the following holds for every $u \in [\nu, 1/2]$ that
\begin{equation*} N^{-(2-u)}\sum_{j=1}^N\sum_{i=1}^N I_{h,N,u} (b_{N,j} - b_{N,i})  \to h\mathcal{I}_b\end{equation*} 
for some constant $\mathcal{I}_b\in (0,\infty)$, which is the same $\mathcal{I}_b$ as in~\cref{ACjs}.
\end{assumption}

\begin{remark}[On~\cref{AssumpIb.new}] 
Note that if we were to set $\nu=1/2$ in~\cref{AssumpIb.new}, then $I_{h,N,\nu}(\cdot)\equiv I_{h,N,1/2}(\cdot)$ is exactly the same function as $I_{h,N}(\cdot)$ as in~\eqref{Ian} of the main paper, and consequently~\cref{AssumpIb.new} would then imply~\cref{ACjs}. If we define 
\begin{equation}\label{eq:newdef}
T_{h,N,u}:=N^{-(2-u)}\sum_{j=1}^N\sum_{i=1}^N I_{h,N,u} (b_{N,j} - b_{N,i}),
\end{equation}
then~\cref{ACjs} requires $T_{h,N,1/2}\to h\mathcal{I}_b$ whereas~\cref{AssumpIb.new} requires the mildly stronger condition $T_{h,N,u}\to h\mathcal{I}_b$ for $u$ varying in any non-degenerate interval with right endpoint at $1/2$. We firmly believe that this is a reasonable assumption. For instance, one of the ways we justified~\cref{ACjs} was by showing that it is satisfied when $b_{N,i}$'s are randomly sampled from an absolutely continuous distribution (see~\cref{remarkA1,iid.I_C}). The same is also true for~\cref{AssumpIb.new}, 
see~\cref{iid.I_C.new} for a formal result. 
\end{remark}

Finding a consistent estimator of $\mathcal{I}_b$ is now quite intuitive. With $T_{h,N,u}$ defined as in~\eqref{eq:newdef},~\cref{AssumpIb.new} implies that $T_{1,N,\nu}$ converges to $\mathcal{I}_b$. Consequently, we can construct an estimator $\widehat{V}_N$ from $T_{1,N,\nu}$ using the plug-in principle as described earlier in this section, by replacing $b_{N,i}$'s in~\eqref{eq:newdef} with $\widehat{b}_{N,i}$'s (see~\eqref{b.hat}). The following result makes it precise; see \cref{proof:plug-in.Ib} for a proof.

\begin{theorem}[Consistent estimation of $\mathcal{I}_b$]\label{plug-in.Ib}
Define $\widehat{b}_{N,j}$ as in \eqref{b.hat}, with $\widehat{\tau}_N\equiv \rsnu$. Suppose that~\cref{AssumpIb.new} holds for some $\nu\in (0,1/2)$, and define
\begin{equation}\label{Vhat}
    \widehat{V}_N(h) := h^{-1} N^{-(2-\nu)} \sum_{j=1}^N \sum_{i=1}^N  I_{h,N,u} \left(\widehat{b}_{N,j} - \widehat{b}_{N,i}\right).
\end{equation}
Then $\widehat{V}_N(h)  \Pto \mathcal{I}_b \text{ as }N\to\infty$.
\end{theorem}

Combining~\cref{tau_HL.CLT} of the main paper and \cref{plug-in.Ib} above we get an asymptotically valid Wald-type confidence interval for $\tau$. 
This is formally stated in the following corollary. 

\begin{corollary}[Confidence interval for $\tau$ based on $\rsnu$]\label{cor:CIunad} Under \cref{ACjs,AssumpIb.new}, an approximate $100(1-\alpha)\%$ confidence interval for $\tau$ is given by 
\begin{equation}\label{eq_unadjCI}
  \rsnu \pm \frac{z_{\alpha/2}}{\sqrt{N}} \left(12\frac{m}{N}\left(1-\frac{m}{N}\right)\widehat{V}_N^{2}(h)\right)^{-1/2},  
\end{equation}
where $\widehat{V}_N(h)$ is defined in \eqref{Vhat}.
\end{corollary}
 We refer the reader to~\cref{rem:nu} below for a discussion on how the choice of $\nu$ in the above result is inconsequential, and to \cref{choice:h} below for a discussion on the choice of the tuning parameter $h$.


Next, we consider the regression-adjusted case, where we want to estimate the quantity $\mathcal{J}_b$ as defined in \cref{AssumpB2}.
Once again, due to technical reasons we require a mildly stronger condition than~\cref{AssumpB2} of the main paper to come up with a consistent estimator for $\mathcal{J}_b$. To lay the groundwork, define
\begin{equation}
    \widetilde{T}_{h,N,u}:=N^{-(2-u)}\sum_{j=1}^N \sum_{i=1}^N I_{h,N,u} (\widetilde{b}_{N,j} - \widetilde{b}_{N,i}),
\end{equation}
where $\widetilde{b}_{N,j}$'s are defined as in~\eqref{btilde} and $I_{h,N,u}(\cdot)$ is defined as in~\eqref{Ian.new}. Observe the direct correspondence between $\widetilde{T}_{h,N,u}$ and $T_{h,N,u}$ from~\eqref{eq:newdef}. The only difference is that the potential control outcomes $b_{N,j}$'s are replaced by the regression adjusted potential control outcomes $\widetilde{b}_{N,j}$'s. Note also that~\cref{AssumpB2} can be restated as $\widetilde{T}_{h,N,1/2}\to h\mathcal{J}_b$. With this observation, in the same vein as~\cref{AssumpIb.new}, we now state our mildly stronger condition (compared to~\cref{AssumpB2}):

\begin{assumption}\label{AssumpJb.new}
There exists $0<\nu< 1/2$, such that for every $u \in [\nu, 1/2]$ and $h\in\R$, $\widetilde{T}_{h,N,u} \to h\mathcal{J}_b,$
for some constant $\mathcal{J}_b\in (0,\infty)$, which is the same $\mathcal{J}_b$ as in~\cref{AssumpB2}.
\end{assumption}
{One can show that~\cref{AssumpJb.new} holds under the same regression model assumption as in~\cref{Jb.geq.Ib-old} (see the arguments used in the proof of~\cref{Jbexists}).}

We are now ready to construct our consistent estimator for $\mathcal{J}_b$. Towards this, recall the definition of $\widehat{b}_{N,j}$ from~\eqref{b.hat} and set $\widehat{\vec{b}}_N:=(\widehat{b}_{N,1},\ldots ,\widehat{b}_{N,N})$. To carry out the plug-in approach, we define the empirical analogues of $\widetilde{b}_{N,j}$'s, as follows:
\begin{equation}\label{b.hat.adj}
\widehat{\widetilde{b}}_{N,j} :=  \widehat{b}_{N,j} -  \vec{p}_{N,j}^\top \widehat{\vec{b}}_N =  Y_{N,j} - \rsna Z_{N,j} -  \vec{p}_{N,j}^\top(\vec{Y} - \rsna \vec{Z}),\qquad j=1,2,\ldots ,N.
\end{equation}
The plug-in (consistent) estimator for $\mathcal{J}_b$ is now constructed in the same manner as we did for $\mathcal{I}_b$, only with the modification that here $\widehat{\widetilde{b}}_{N,j}$'s will play the role of $\widehat{b}_{N,j}$'s. This is the content of the following theorem; see \cref{proof:plug-in.Jb2} for a proof.
\begin{theorem}[Consistent estimation of $\mathcal{J}_b$]\label{plug-in.Jb2}
Define $\widehat{\widetilde{b}}_{N,j}$ as in \eqref{b.hat.adj}, and suppose that~\cref{AssumpJb.new} holds for some $\nu\in (0,1/2)$. Define
\begin{equation}\label{What}
    \widehat{W}_N(h) := h^{-1}N^{-(2-\nu)} \sum_{j=1}^N \sum_{i=1}^N  I_{h,N,u} \left(\widehat{\widetilde{b}}_{N,j} - \widehat{\widetilde{b}}_{N,i}\right).
\end{equation}
Then $\widehat{W}_N(h)  \Pto \mathcal{J}_b \text{ as }N\to\infty$.
\end{theorem}
~\cref{plug-in.Jb2} in conjunction with~\cref{tau_adj.final.CLT} from the main paper readily yields an asymptotically valid Wald-type confidence interval for $\tau$ based on $\rsna$, which we formally state below.

\begin{corollary}[Confidence interval for $\tau$ based on $\rsna$]\label{adjustedCI} Under \cref{AssumpB2,AssumpJb.new}, an approximate $100(1-\alpha)\%$ confidence interval for $\tau$ is given by 
\begin{equation}\label{eq_adjustedCI}
\rsna \pm \frac{z_{\alpha/2}}{\sqrt{N}} \left(12\frac{m}{N}\left(1-\frac{m}{N}\right)\widehat{W}_N^{2}(h)\right)^{-1/2},
\end{equation}
where $\widehat{W}_N$ is defined in \eqref{What}.
\end{corollary}

We conclude this section with two remarks that highlight that (a) the choice of the tuning parameter $\nu$ is inconsequential for both the unadjusted and regression-adjusted confidence intervals that we proposed above, (b) the choice of the scaling $h$ might be crucial, especially for small to moderately large sample size.

\begin{remark}[Choice of $\nu$]\label{rem:nu}
Our plug-in estimators proposed in  \eqref{Vhat} and \eqref{What} both require a choice of the tuning parameter $\nu\in (0,1/2)$.
We have observed through extensive simulations that the confidence intervals are not very sensitive to the choice of $\nu$.  In fact, it seems that the choice $\nu=1/2$ also works in practice, although our proofs of~\cref{plug-in.Ib,plug-in.Jb2} precludes that choice. 

 For an illustration, consider the simulation Setting 4(a) (misspecified model with Gaussian errors) as described in~\cref{subsec:simulationsApp}. We calculate the plug-in estimator of the asymptotic variances of $\rsnu$ and $\rsna$ for different values of $\nu$, ranging from $1/4$ to $1/2$. Then fixing one particular value of $\nu$ as a reference, say $\nu=1/3$, we calculate the ratios of these estimates for for $\nu = 1/3$ to those for the other values of $\nu$. Repeating this $100$ times, we draw the boxplots of the ratios thus obtained, which is provided in~\cref{ratios}. 
 The plots suggest that the estimates of the asymptotic variance do not differ much with the choice of $\nu$. We observed similar phenomenon for the other simulation settings as well.  Furthermore, in our simulations we also varied $\nu$ while calculating the confidence intervals for $\tau$ proposed above, and found that the results do not vary considerably with the choice of $\nu$. Thus the choice of $\nu$ seems to be unimportant, and in practice one may use any number in $(0,1/2)$ as the value of $\nu$. 

\begin{figure}[!ht]
    \centering
    \includegraphics[width = \textwidth]{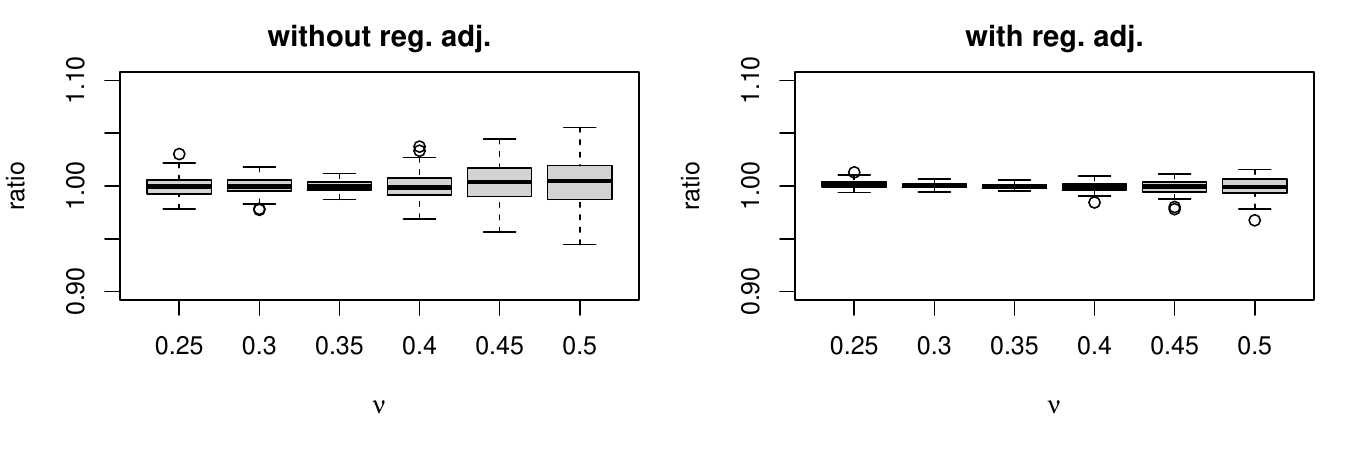}
    \caption{Boxplots of the ratios of the plug-in estimators of asymptotic variances of $\rsnu$ (left panel) and $\rsna$ (right panel) for $\nu = 1/3$ and various other values of $\nu$, in simulation Setting 4(a) (see~\cref{subsec:simulationsApp} for details on the simulation setup).}
    \label{ratios}
\end{figure}
\end{remark}

\begin{remark}[Scaling of the outcome variable]\label{choice:h}
Since \eqref{Vhat} (resp.~\eqref{What}) involves counts of differences in the outcomes (resp.~the residuals) in shrinking intervals, we might not get a reasonable estimate in small samples if the outcomes (resp.~residuals) are typically large in magnitude relative to the sample size. Thus, the variance estimators proposed in \eqref{Vhat} and \eqref{What} are particularly useful in large samples, and can exhibit instability (that stems from the tuning parameter $h$) in small samples. Given this concern, we recommend the practitioners to use our theoretical results to make an informed choice on which estimator they should use for their data analyses, and then obtain permutation-based confidence intervals for the chosen estimator. We do, however, note that the asymptotic intervals proposed in \cref{CI1,cor:CIunad,adjustedCI} provide nominal coverage for all of our numerical experiments.
\end{remark}

\section{Numerical experiments}\label{sec:simulationsApp}

{In this section, we illustrate the empirical performance of Rosenbaum's rank-based estimators $\rsnu$ and $\rsna$, comparing them with various other estimators of the constant treatment effect $\tau$, namely: The difference-in-means estimator $\linu$ (see \eqref{dm}), the difference-in-medians estimator $\taumed$ (see \eqref{def:taumed}), the $\alpha$-trimmed and $\alpha$-Winsorized difference-in-means estimators (with $\alpha=0.1$; see \citet{athey2021} for definitions), the estimators $\taueif$ and $\widehat\tau_{\mathrm{waq}}$ studied by \citet{athey2021}, the simple regression-adjusted estimator $\lina$ \citep{Freedman08a,Freedman08b}, and the estimator $\lini$ proposed by \citet{Lin13}. We present one real data analysis and extensive synthetic experiments.}

\subsection{Progresa data}\label{sec:progresa}
We analyze the data from a randomized trial that aims to study the electoral impact of \emph{Progresa}, Mexico’s conditional cash transfer program (CCT program) \citep{de2013, imai2018}. In this experiment, eligible villages were randomly assigned to receive the program either 21 months (early Progresa, ``treated'') or 6 months (late Progresa, ``control") before the 2000 Mexican presidential election. The data contains $417$ observations, each representing a precinct, and for each precinct we have information about its treatment status, outcomes of interest, socioeconomic indicators, and other precinct characteristics.

\begin{table}
\caption{\label{progresa.0} Different estimates of the effect of early Progresa on PRI support rates with the corresponding standard errors, 95\% (approximate) confidence intervals and their lengths. For the first four estimators and Rosenbaum's estimators we approximate standard error using $10^5$ Bootstrap resamples. For $\lina$ and $\lini$ we use sandwich estimator of variance, and for $\taueif$ and $\widehat\tau_{\mathrm{waq}}$ we directly use the software implemented by \citet{athey2021}.\vspace{2mm}}
\centering
\begin{tabular}{lcccc}
 Estimator & Estimate & Std.~Error & 95\% CI & CI Length \\
  \hline\noalign{\vskip 0.5ex}
 Difference-in-Means ($\linu$)                & $3.62$ & $1.92$ & $[-0.14,\, 7.39]$ & $7.53$ \\
 Difference-in-Medians ($\taumed$)            & $0.69$ & $1.56$ & $[-2.37,\, 3.75]$ & $6.12$ \\
 $0.1$-trimmed Difference-in-Means            & $2.00$ & $1.68$ & $[-1.29,\, 5.28]$ & $6.57$ \\
 $0.1$-Winsorized Difference-in-Means         & $2.59$ & $1.72$ & $[-0.78,\, 5.96]$ & $6.74$ \\
 $\taueif$ \citep{athey2021}                  & $1.95$ & $1.72$ & $[-1.42,\, 5.33]$ & $6.75$ \\
 $\widehat\tau_{\mathrm{waq}}$ \citep{athey2021} & $1.31$ & $1.71$ & $[-2.04,\, 4.65]$ & $6.69$ \\
 Rosenbaum's estimator ($\rsnu$)              & $1.83$ & $1.67$ & $[-1.43,\, 5.10]$ & $6.53$ \\
 OLS adjusted $\lina$                         & $3.67$ & $1.70$ & $[0.34,\, 7.00]$  & $6.67$ \\
 $\lini$ \citep{Lin13}                        & $4.21$ & $1.99$ & $[0.32,\, 8.11]$  & $7.78$ \\
 Rosenbaum's adjusted ($\rsna$)               & $2.19$ & $1.38$ & $[-0.53,\, 4.90]$ & $5.43$ \\
\end{tabular}
\end{table}


Following \citet{de2013}, we use the support rates for the incumbent party as shares of the eligible voting population in the 2000 election (\emph{pri2000s}) as the outcome, and regress it on the following covariates: the average poverty level in a precinct (\emph{avgpoverty}), the total precinct population in 1994 (\emph{pobtot1994}), the total number of voters who turned out in the previous election (\emph{votos1994}), and the total number of votes cast for each of the three main competing parties in the previous election (\emph{pri1994}, \emph{pan1994}, and \emph{prd1994}), and  \emph{villages} (as factors).

The design-based limiting distributions of some of the estimators we consider here are not known in the literature, so we use permutation-based confidence intervals (with $B=10^4$ permutations) for an apples-to-apples comparison.
 We report in \cref{progresa.0} the point estimates and approximate 95\% confidence intervals 
 for all the methods we consider.
 Note that the mean-based estimators $\linu$, $\lina$ and $\lini$ suggest that the CCT program led to a significant positive increase in support for the incumbent party, while all the other methods suggest the treatment effect is not significant. We also note that Rosenbaum's regression-adjusted method provides the shortest confidence interval among the above methods. 
\subsection{Additional simulations}\label{subsec:simulationsApp}

In this section, we present comprehensive numerical experiments to compare the empirical performance of Rosenbaum's rank-based estimators $\rsnu$ and
$\rsna$ with various competing estimators of the constant treatment effect $\tau$, as mentioned above. These experiments extend the preview provided in \cref{sec:simulations} of the main paper, where we reported the performance of these estimators in Settings 1 and 2 with $m/N=0.5$. We consider the following simulation settings:
\begin{enumerate}
    \item Setting 1: Generate $x_i$ i.i.d.~from $\text{Unif}(-4, 4)$, and set $a_i = 3x_i + \eps_i$ and $b_i = a_i - 2$.
    \item Setting 2: Same as Setting 1, except that we contaminate $5\%$ of the outcomes with an arbitrary large value $M$; here we use $M=500$.
    \item Setting 3: Generate $x_i$ i.i.d.~from $\text{Unif}(-4, 4)$ and $v_i$  i.i.d.~from $\mathrm{Exp}(1/10)$ independently, and set $a_i = v_i +\eps_i$ and $b_i=a_i-2$.
    \item Setting 4: Generate $u_i$ i.i.d.~from $\text{Unif}(-4, 4)$, and set $x_i=e^{u_i}$, $a_i = \frac{1}{4}(x_i+\sqrt{x_i}) + \eps_i$ and $b_i = a_i - 2$.
\end{enumerate}
 Setting 1 represents a scenario where the linear model is correctly specified.
 Setting 2 is designed to the illustrate the trade-off between robustness and efficiency by examining how contamination affects the confidence intervals. Setting 3 considers the case where the outcomes are independent of the covariates (we thus expect no efficiency gain by regression adjustment).
Setting 4 is a slight
modification of an example of \citet[Setting 4.2.3]{Lin13}, and helps us illustrate the performance
of the estimators under misspecification of the regression model.

 In each of these settings, the i.i.d.~noise $\epsilon_i$'s are drawn from: (a) standard normal, (b) Cauchy, and (c) Student's $t_3$ distribution. The sample size is $N=1000$ and we report here the results for both the balanced design $m/N = 0.50$ and an unbalanced case $m/N = 0.25$ (the case $m/N=0.75$ is symmetric). The design-based limiting distributions of some of the estimators we consider here are not known in the literature, so we use permutation-based confidence intervals for an apples-to-apples comparison. In each replication, we use $B=10^4$ permutations to approximate the standard errors of the estimators.
We repeat the above experiments for $1000$ replications, and report the coverage and average lengths of the approximate 95\% permutation-based confidence intervals for each estimator.

\begin{table}[!ht]
\centering
\caption{\label{siml_res1.app} Empirical coverage and average length of approximate 95\% permutation CIs for simulation Setting 1 (correctly specified linear model), for $m/N=0.5$.\vspace{1mm}}
\setlength{\tabcolsep}{5pt}
\renewcommand{\arraystretch}{1}
\begin{tabular}{l|cc|cc|cc}
\toprule
\multirow{2}{*}{Estimator} &
\multicolumn{2}{c|}{(a) Gaussian errors} &
\multicolumn{2}{c|}{(b) Cauchy errors} &
\multicolumn{2}{c}{(c) $t_3$ errors} \\
& coverage & length & coverage & length & coverage & length \\
\midrule
Difference-in-Means ($\linu$)                      & $94.8\%$  & $1.75$ & $94.2\%$  & $24.26$ & $94.8\%$  & $1.79$ \\
Difference-in-Medians ($\taumed$)                  & $93.1\%$  & $2.90$ & $94.9\%$  &  $3.07$ & $94.7\%$  & $2.91$ \\
$0.1$-trimmed Diff-in-Means                        & $94.1\%$  & $2.03$ & $95.2\%$  &  $2.19$ & $95.2\%$  & $2.04$ \\
$0.1$-Winsorized Diff-in-Means                     & $94.2\%$  & $1.82$ & $95.6\%$  &  $2.01$ & $94.9\%$  & $1.84$ \\
$\taueif$ \citep{athey2021}                        & $96.3\%$  & $1.30$ & $96.2\%$  &  $1.74$ & $96.7\%$  & $1.39$ \\
$\widehat\tau_{\mathrm{waq}}$ \citep{athey2021}    & $95.7\%$  & $1.32$ & $94.5\%$  &  $2.50$ & $96.2\%$  & $1.44$ \\
Rosenbaum's estimator ($\rsnu$)                    & $94.9\%$  & $1.84$ & $95.6\%$  &  $2.14$ & $95.4\%$  & $1.88$ \\
OLS adjusted ($\lina$)                             & $99.0\%$  & $0.35$ & $94.2\%$  & $24.30$ & $96.5\%$  & $0.49$ \\
Lin's estimator ($\lini$)                          & $99.0\%$  & $0.35$ & $94.2\%$  & $24.30$ & $96.5\%$  & $0.49$ \\
Rosenbaum's adjusted ($\rsna$)                     & $99.4\%$  & $0.37$ & $97.2\%$  &  $1.34$ & $98.7\%$  & $0.43$ \\
\bottomrule
\end{tabular}
\end{table}

\begin{table}[!ht]
\centering
\caption{\label{siml_res1.5} Empirical coverage and average length of approximate 95\% permutation CIs for simulation Setting 1 (correctly specified linear model), for $m/N=0.25$.\vspace{1mm}}
\setlength{\tabcolsep}{5pt}
\renewcommand{\arraystretch}{1}
\begin{tabular}{l|cc|cc|cc}
\toprule
\multirow{2}{*}{Estimator} &
\multicolumn{2}{c|}{(a) Gaussian errors} &
\multicolumn{2}{c|}{(b) Cauchy errors} &
\multicolumn{2}{c}{(c) $t_3$ errors} \\
& coverage & length & coverage & length & coverage & length \\
\midrule
Difference-in-Means ($\linu$)                      & $94.4\%$ & $2.02$  & $94.3\%$ & $41.38$ & $94.0\%$ & $2.06$ \\
Difference-in-Medians ($\taumed$)                  & $93.5\%$ & $3.36$  & $94.8\%$ & $3.55$  & $94.7\%$ & $3.37$ \\
$0.1$-trimmed Diff-in-Means                        & $94.3\%$ & $2.34$  & $95.2\%$ & $2.53$  & $94.5\%$ & $2.36$ \\
$0.1$-Winsorized Diff-in-Means                     & $94.5\%$ & $2.10$  & $95.2\%$ & $2.32$  & $94.6\%$ & $2.13$ \\
$\taueif$ \citep{athey2021}                        & $95.5\%$ & $1.44$  & $94.5\%$ & $1.97$  & $95.9\%$ & $1.56$ \\
$\widehat\tau_{\mathrm{waq}}$ \citep{athey2021}    & $95.1\%$ & $1.48$  & $94.2\%$ & $3.65$  & $95.4\%$ & $1.65$ \\
Rosenbaum's estimator ($\rsnu$)                    & $94.5\%$ & $2.12$  & $95.4\%$ & $2.46$  & $94.4\%$ & $2.16$ \\
OLS adjusted ($\lina$)                             & $98.5\%$ & $0.38$  & $94.7\%$ & $41.31$ & $97.2\%$ & $0.55$ \\
Lin's estimator ($\lini$)                          & $98.7\%$ & $0.38$  & $94.3\%$ & $41.76$ & $97.2\%$ & $0.55$ \\
Rosenbaum's adjusted ($\rsna$)                     & $98.4\%$ & $0.38$  & $96.6\%$ & $1.51$  & $98.6\%$ & $0.45$ \\
\bottomrule
\end{tabular}
\end{table}

\begin{table}[!ht]
\centering
\caption{\label{siml_res2.app} Empirical coverage and average length of approximate 95\% permutation CIs for simulation Setting 2 (contamination), for $m/N=0.5$.\vspace{1mm}}
\setlength{\tabcolsep}{5pt}
\renewcommand{\arraystretch}{1}
\begin{tabular}{l|cc|cc|cc}
\toprule
\multirow{2}{*}{Estimator} &
\multicolumn{2}{c|}{(a) Gaussian errors} &
\multicolumn{2}{c|}{(b) Cauchy errors} &
\multicolumn{2}{c}{(c) $t_3$ errors} \\
& coverage & length & coverage & length & coverage & length \\
\midrule
Difference-in-Means ($\linu$)                      & $100.0\%$ & $27.52$ & $100.0\%$ & $43.71$ & $100.0\%$ & $27.51$ \\
Difference-in-Medians ($\taumed$)                  & $94.5\%$  &  $3.04$ & $94.7\%$  &  $3.22$ & $94.9\%$  &  $3.07$ \\
$0.1$-trimmed Diff-in-Means                        & $95.4\%$  &  $2.15$ & $96.7\%$  &  $2.35$ & $96.9\%$  &  $2.16$ \\
$0.1$-Winsorized Diff-in-Means                     & $96.1\%$  &  $1.95$ & $97.9\%$  &  $2.32$ & $97.4\%$  &  $1.98$ \\
$\taueif$ \citep{athey2021}                        & $92.9\%$  &  $1.63$ & $90.2\%$  &  $9.16$ & $92.1\%$  &  $1.68$ \\
$\widehat\tau_{\mathrm{waq}}$ \citep{athey2021}    & $100.0\%$ & $59.17$ & $94.1\%$  & $659.88$ & $100.0\%$ & $56.97$ \\
Rosenbaum's estimator ($\rsnu$)                    & $95.5\%$  &  $1.97$ & $96.4\%$  &  $2.29$ & $96.0\%$  &  $2.01$ \\
OLS adjusted ($\lina$)                             & $100.0\%$ & $27.77$ & $99.9\%$  & $43.85$ & $100.0\%$ & $27.75$ \\
Lin's estimator ($\lini$)                          & $100.0\%$ & $27.77$ & $99.9\%$  & $43.85$ & $100.0\%$ & $27.75$ \\
Rosenbaum's adjusted ($\rsna$)                     & $97.3\%$  &  $0.95$ & $97.4\%$  &  $1.88$ & $97.8\%$  &  $1.01$ \\
\bottomrule
\end{tabular}
\end{table}

\begin{table}[!ht]
\centering
\caption{\label{siml_res2.5} Empirical coverage and average length of approximate 95\% permutation CIs for simulation Setting 2 (contamination), for $m/N=0.25$. \vspace{1mm}}
\setlength{\tabcolsep}{5pt}
\renewcommand{\arraystretch}{1}
\begin{tabular}{l|cc|cc|cc}
\toprule
\multirow{2}{*}{Estimator} &
\multicolumn{2}{c|}{(a) Gaussian errors} &
\multicolumn{2}{c|}{(b) Cauchy errors} &
\multicolumn{2}{c}{(c) $t_3$ errors} \\
& coverage & length & coverage & length & coverage & length \\
\midrule
Difference-in-Means ($\linu$)                      & $100.0\%$ & $30.47$ & $100.0\%$ & $61.43$  & $100.0\%$ & $30.51$ \\
Difference-in-Medians ($\taumed$)                  & $93.6\%$  & $3.52$  & $94.6\%$  & $3.73$   & $95.8\%$  & $3.55$ \\
$0.1$-trimmed Diff-in-Means                        & $95.8\%$  & $2.47$  & $96.6\%$  & $2.73$   & $96.4\%$  & $2.50$ \\
$0.1$-Winsorized Diff-in-Means                     & $96.4\%$  & $2.25$  & $97.9\%$  & $2.80$   & $96.5\%$  & $2.29$ \\
$\taueif$ \citep{athey2021}                        & $91.3\%$  & $1.83$  & $80.6\%$  & $13.26$  & $92.7\%$  & $1.93$ \\
$\widehat\tau_{\mathrm{waq}}$ \citep{athey2021}    & $100.0\%$ & $82.46$ & $90.9\%$  & $1658.81$& $100.0\%$ & $79.15$ \\
Rosenbaum's estimator ($\rsnu$)                    & $95.3\%$  & $2.28$  & $96.2\%$  & $2.65$   & $96.0\%$  & $2.32$ \\
OLS adjusted ($\lina$)                             & $100.0\%$ & $29.97$ & $100.0\%$ & $61.49$  & $100.0\%$ & $30.02$ \\
Lin's estimator ($\lini$)                          & $100.0\%$ & $30.07$ & $100.0\%$ & $61.86$  & $100.0\%$ & $30.12$ \\
Rosenbaum's adjusted ($\rsna$)                     & $97.8\%$  & $1.09$  & $97.5\%$  & $2.16$   & $97.4\%$  & $1.14$ \\
\bottomrule
\end{tabular}
\end{table}

\begin{table}[!ht]
\centering
\caption{\label{siml_res3} Empirical coverage and average length of approximate 95\% permutation CIs for simulation Setting 3 (covariates are uninformative), for $m/N=0.5$. \vspace{1mm}}
\setlength{\tabcolsep}{5pt}
\renewcommand{\arraystretch}{1}
\begin{tabular}{l|cc|cc|cc}
\toprule
\multirow{2}{*}{Estimator} &
\multicolumn{2}{c|}{(a) Gaussian errors} &
\multicolumn{2}{c|}{(b) Cauchy errors} &
\multicolumn{2}{c}{(c) $t_3$ errors} \\
& coverage & length & coverage & length & coverage & length \\
\midrule
Difference-in-Means ($\linu$)                      & $94.4\%$ & $2.49$ & $96.0\%$ & $25.97$ & $95.3\%$ & $2.52$ \\
Difference-in-Medians ($\taumed$)                  & $94.6\%$ & $2.43$ & $94.0\%$ &  $2.66$ & $94.1\%$ & $2.44$ \\
$0.1$-trimmed Diff-in-Means                        & $95.1\%$ & $2.25$ & $95.2\%$ &  $2.51$ & $95.2\%$ & $2.27$ \\
$0.1$-Winsorized Diff-in-Means                     & $94.8\%$ & $2.37$ & $95.3\%$ &  $2.70$ & $95.2\%$ & $2.40$ \\
$\taueif$ \citep{athey2021}                        & $96.6\%$ & $1.21$ & $95.4\%$ &  $1.66$ & $96.0\%$ & $1.31$ \\
$\widehat\tau_{\mathrm{waq}}$ \citep{athey2021}    & $96.4\%$ & $1.25$ & $95.1\%$ &  $2.46$ & $95.3\%$ & $1.38$ \\
Rosenbaum's estimator ($\rsnu$)                    & $96.3\%$ & $1.67$ & $95.7\%$ &  $2.04$ & $96.7\%$ & $1.73$ \\
OLS adjusted ($\lina$)                             & $94.4\%$ & $2.50$ & $96.1\%$ & $26.51$ & $95.3\%$ & $2.52$ \\
Lin's estimator ($\lini$)                          & $94.4\%$ & $2.50$ & $96.1\%$ & $26.51$ & $95.3\%$ & $2.52$ \\
Rosenbaum's adjusted ($\rsna$)                     & $96.1\%$ & $1.68$ & $95.6\%$ &  $3.02$ & $96.5\%$ & $1.74$ \\
\bottomrule
\end{tabular}
\end{table}

\begin{table}[!ht]
\centering
\caption{\label{siml_res3.5} Empirical coverage and average length of approximate 95\% permutation CIs for simulation Setting 3 (covariates are uninformative), for $m/N=0.25$. \vspace{1mm}}
\setlength{\tabcolsep}{5pt}
\renewcommand{\arraystretch}{1}
\begin{tabular}{l|cc|cc|cc}
\toprule
\multirow{2}{*}{Estimator} &
\multicolumn{2}{c|}{(a) Gaussian errors} &
\multicolumn{2}{c|}{(b) Cauchy errors} &
\multicolumn{2}{c}{(c) $t_3$ errors} \\
& coverage & length & coverage & length & coverage & length \\
\midrule
Difference-in-Means ($\linu$)                      & $95.3\%$  & $2.88$ & $95.6\%$  & $44.30$ & $94.3\%$  & $2.90$ \\
Difference-in-Medians ($\taumed$)                  & $93.3\%$  & $2.82$ & $94.5\%$  &  $3.09$ & $93.9\%$  & $2.83$ \\
$0.1$-trimmed Diff-in-Means                        & $94.8\%$  & $2.59$ & $93.9\%$  &  $2.89$ & $94.2\%$  & $2.62$ \\
$0.1$-Winsorized Diff-in-Means                     & $94.9\%$  & $2.74$ & $93.9\%$  &  $3.10$ & $94.4\%$  & $2.76$ \\
$\taueif$ \citep{athey2021}                        & $95.6\%$  & $1.33$ & $94.7\%$  &  $1.87$ & $94.8\%$  & $1.47$ \\
$\widehat\tau_{\mathrm{waq}}$ \citep{athey2021}    & $96.1\%$  & $1.39$ & $95.4\%$  &  $3.46$ & $94.8\%$  & $1.56$ \\
Rosenbaum's estimator ($\rsnu$)                    & $95.4\%$  & $1.91$ & $94.8\%$  &  $2.34$ & $95.0\%$  & $1.98$ \\
OLS adjusted ($\lina$)                             & $95.2\%$  & $2.88$ & $95.5\%$  & $44.56$ & $94.1\%$  & $2.90$ \\
Lin's estimator ($\lini$)                          & $95.2\%$  & $2.88$ & $95.8\%$  & $45.30$ & $94.1\%$  & $2.91$ \\
Rosenbaum's adjusted ($\rsna$)                     & $95.1\%$  & $1.92$ & $94.7\%$  &  $3.48$ & $94.6\%$  & $1.99$ \\
\bottomrule
\end{tabular}
\end{table}

\begin{table}[!ht]
\centering
\caption{\label{siml_res4} Empirical coverage and average length of approximate 95\% permutation CIs for simulation Setting 4 (model misspecification), for $m/N=0.5$.\vspace{1mm}}
\setlength{\tabcolsep}{5pt}
\renewcommand{\arraystretch}{1}
\begin{tabular}{l|cc|cc|cc}
\toprule
\multirow{2}{*}{Estimator} &
\multicolumn{2}{c|}{(a) Gaussian errors} &
\multicolumn{2}{c|}{(b) Cauchy errors} &
\multicolumn{2}{c}{(c) $t_3$ errors} \\
& coverage & length & coverage & length & coverage & length \\
\midrule
Difference-in-Means ($\linu$)                      & $96.3\%$ & $0.91$ & $94.2\%$ & $24.05$ & $95.9\%$ & $0.98$ \\
Difference-in-Medians ($\taumed$)                  & $98.0\%$ & $0.68$ & $96.1\%$ & $0.84$  & $96.5\%$ & $0.71$ \\
$0.1$-trimmed Diff-in-Means                        & $95.6\%$ & $0.83$ & $95.7\%$ & $1.14$  & $96.4\%$ & $0.88$ \\
$0.1$-Winsorized Diff-in-Means                     & $95.1\%$ & $1.00$ & $95.3\%$ & $1.33$  & $95.8\%$ & $1.02$ \\
$\taueif$ \citep{athey2021}                        & $99.6\%$ & $0.52$ & $98.1\%$ & $0.76$  & $98.8\%$ & $0.60$ \\
$\widehat\tau_{\mathrm{waq}}$ \citep{athey2021}    & $99.4\%$ & $0.52$ & $96.1\%$ & $0.86$  & $97.8\%$ & $0.61$ \\
Rosenbaum's estimator ($\rsnu$)                    & $98.6\%$ & $0.62$ & $97.8\%$ & $0.91$  & $98.5\%$ & $0.69$ \\
OLS adjusted ($\lina$)                             & $99.2\%$ & $0.35$ & $94.7\%$ & $24.24$ & $96.6\%$ & $0.49$ \\
Lin's estimator ($\lini$)                          & $99.2\%$ & $0.35$ & $94.7\%$ & $24.23$ & $96.6\%$ & $0.49$ \\
Rosenbaum's adjusted ($\rsna$)                     & $99.2\%$ & $0.37$ & $97.7\%$ & $0.74$  & $98.5\%$ & $0.43$ \\
\bottomrule
\end{tabular}
\end{table}

\begin{table}[!ht]
\centering
\caption{\label{siml_res4.5} Empirical coverage and average length of approximate 95\% permutation CIs for simulation Setting 4 (model misspecification), for $m/N=0.25$. \vspace{1mm}}
\setlength{\tabcolsep}{5pt}
\renewcommand{\arraystretch}{1}
\begin{tabular}{l|cc|cc|cc}
\toprule
\multirow{2}{*}{Estimator} &
\multicolumn{2}{c|}{(a) Gaussian errors} &
\multicolumn{2}{c|}{(b) Cauchy errors} &
\multicolumn{2}{c}{(c) $t_3$ errors} \\
& coverage & length & coverage & length & coverage & length \\
\midrule
Difference-in-Means ($\linu$)                      & $95.0\%$ & $1.04$ & $94.0\%$ & $41.22$ & $95.9\%$ & $1.12$ \\
Difference-in-Medians ($\taumed$)                  & $98.2\%$ & $0.82$ & $98.2\%$ & $1.12$  & $98.9\%$ & $0.92$ \\
$0.1$-trimmed Diff-in-Means                        & $95.7\%$ & $0.95$ & $96.1\%$ & $1.30$  & $96.0\%$ & $1.00$ \\
$0.1$-Winsorized Diff-in-Means                     & $95.0\%$ & $1.14$ & $94.7\%$ & $1.54$  & $95.5\%$ & $1.17$ \\
$\taueif$ \citep{athey2021}                        & $97.9\%$ & $0.53$ & $97.4\%$ & $0.83$  & $98.4\%$ & $0.64$ \\
$\widehat\tau_{\mathrm{waq}}$ \citep{athey2021}    & $97.8\%$ & $0.54$ & $96.6\%$ & $1.04$  & $97.9\%$ & $0.66$ \\
Rosenbaum's estimator ($\rsnu$)                    & $97.6\%$ & $0.67$ & $97.6\%$ & $1.01$  & $97.9\%$ & $0.75$ \\
OLS adjusted ($\lina$)                             & $98.4\%$ & $0.38$ & $94.3\%$ & $41.29$ & $97.0\%$ & $0.55$ \\
Lin's estimator ($\lini$)                          & $98.5\%$ & $0.38$ & $94.2\%$ & $41.62$ & $96.9\%$ & $0.55$ \\
Rosenbaum's adjusted ($\rsna$)                     & $98.4\%$ & $0.39$ & $97.6\%$ & $0.81$  & $98.9\%$ & $0.46$ \\
\bottomrule
\end{tabular}
\end{table}

We summarize the results  in~\cref{siml_res1.app,siml_res1.5,siml_res2.app,siml_res2.5,siml_res3,siml_res3.5,siml_res4,siml_res4.5}. \cref{siml_res1.app,siml_res2.app} reproduce the results from the main paper for $m/N = 0.5$, while the remaining tables provide additional results for different treatment proportions and settings.
We list some observations below.

\begin{itemize}[leftmargin = *]\setlength\itemsep{0mm}

 \item \textit{Efficiency under light tails:}
    When the errors are Gaussian, the confidence intervals constructed using Rosenbaum's estimator $\rsnu$ (resp.~$\rsna$) have lengths comparable to or only slightly longer than those based on the mean-based estimators $\linu$ (resp.~$\lina$, and $\lini$), while maintaining nominal coverage. The efficiency losses for rank-based estimators remain within the small margin predicted by \cref{efflb,efflb-adj-case}.

\item \textit{Robustness against heavy tails:} 
When errors are heavy-tailed (Cauchy and $t_3$), the confidence intervals based on mean-based estimators $\linu$, $\lina$, and $\lini$ become excessively wide. In contrast, Rosenbaum's rank-based confidence intervals maintain reasonable lengths and provide nominal coverage across all settings.

\item \textit{Robustness against contamination:} 
Setting 2 demonstrates the behavior of different estimators under contamination. The mean-based estimators produce extremely wide confidence intervals, reflecting their sensitivity to outliers. The estimators proposed by \citet{athey2021}, which produces the shortest intervals among unadjusted estimators in all other settings, perform substantially poor under contamination, with undercoverage (in Setting 2b) and inflated interval lengths. In contrast, Rosenbaum's estimators maintain stable performance even under contamination, as expected from \cref{ABP}. The performance of $0.1$-trimmed/Winsorized difference-in-means and the difference-in-medians are also stable, since the contamination level is less than their respective asymptotic breakdown points.

\item \textit{Regression adjustment improves precision:} 
Across Settings 1, 2, and 4, where covariates contain information about the outcomes, the regression-adjusted estimator $\rsna$ produces substantially shorter confidence intervals than the unadjusted estimator $\rsnu$, validating our theoretical result in~\cref{Jb.geq.Ib}. For example, in Setting 1a with $m/N=0.5$, the average interval length decreases from 1.84 ($\rsnu$) to 0.37 ($\rsna$), a reduction of approximately 80\%.

\item \textit{Behavior when covariates are uninformative:} 
In Setting 3, where the potential outcomes are independent of the covariates, the confidence intervals based on $\rsnu$ and $\rsna$ have similar lengths under Gaussian and $t_3$ errors, as expected. Interestingly, under Cauchy errors (Setting 3b), $\rsna$ produces wider intervals than $\rsnu$. This phenomenon occurs because the regression adjustment, while not harmful in terms of coverage, introduces additional variability when the model does not provide any signal and the tails are extremely heavy. Nevertheless, even in this worst-case scenario, $\rsna$ maintains nominal coverage and outperforms all mean-based estimators.

\item \textit{Robustness against model misspecification:} 
In Setting 4, the regression model is misspecified. Yet, the performance of the rank-based confidence intervals 
 is quite satisfactory, 
illustrating that the estimator $\rsna$ can be robust against model misspecification.  Note also that in this setting, despite the regression model being misspecified, the confidence intervals based on $\rsna$ are shorter than those constructed using $\rsnu$.

\item \textit{Overall performance:} 
Across all twelve sub-settings and two treatment proportions, Rosenbaum's regression-adjusted estimator $\rsna$ demonstrates the most favorable balance of efficiency and robustness. It achieves near-optimal efficiency under light tails while substantially outperforming mean-based methods under heavy tails and contamination, maintaining this advantage even under model misspecification.

\end{itemize}

\section{Main technical tools}\label{sec: main tools}

This section outlines the key proof techniques used to establish our main results in~\cref{sec:w/oAdj,sec:wAdj} of the main paper and~\cref{plug-in est}. We organize our discussion around three main challenges: (1) establishing asymptotic normality under local alternatives, (2)  handling the additional complexities introduced by regression adjustment, and (3) constructing consistent variance estimators.

\paragraph{Combinatorial CLT.}
We begin with~\cref{propo-null1} which states the null distribution of the Wilcoxon rank-sum (WRS) test statistic. Under the null $\tau= \tau_0$, the WRS test statistic can be written as: 
$$t_N\equiv t_N^0=m + \sum_{j=1}^N  Z_{N,j}  \sum_{i=1, i\neq j}^N   \ind{b_{N,i} \le b_{N,j}}.$$
Crucially, $t_N^0$ is a weighted linear combination of the treatment indicators $Z_{N,j}$, where the weights are \emph{deterministic}. This structure allows us to apply Hoeffding's combinatorial CLT \citep{Hoeffding51} to establish the asymptotic normality in~\cref{propo-null1}; see \cref{proof:propo-null1} for details. 

\paragraph{Local asymptotics via decomposition.} The case of local alternatives, i.e., where $\tau= \tau_N:= \tau_0-h/\sqrt{N}$, is more challenging. This is the subject of~\cref{propo-asy-tauN} of the main paper. Under $\tau=\tau_N$, the WRS test statistic can be written as:
\begin{equation}\label{tnh}
    t_N\equiv t_N^h=m + \sum_{j=1}^N  Z_{N,j}  \sum_{i=1, i\neq j}^N   \ind{b_{N,i}-hN^{-1/2} Z_{N,i} \le b_{N,j}-hN^{-1/2} Z_{N,j}}.
\end{equation}
Unlike the null case, the weights in this linear combination are now \emph{random} and depend on the treatment assignment itself. This precludes direct application of combinatorial CLTs. Moreover, our fixed design setting lacks the structure needed for standard contiguity arguments à la Le Cam to establish local asymptotic normality (see~\citet[Chapter 7]{vandervaart}).

Our key insight is a novel decomposition that separates the WRS test statistic under local alternatives into a tractable null component plus a shift that contributes to the bias. This decomposition, formalized in the following proposition (see \cref{proof:propo-decomp} for a proof), forms the basis of our asymptotic theory.


\begin{proposition}\label{propo-decomp}  Let $t_N = t_N(\vec{Z}_N, \vec{Y}_N -\tau_0\vec{Z}_N)$  be the Wilcoxon rank-sum statistic based on any random treatment assignment, where $t(\cdot,\cdot)$ is as in \eqref{WRS_N} of the main paper. Fix $h\in\R$ and let $\tau_N = \tau_0 - hN^{-1/2}$.  It holds under $\tau=\tau_N$ that
\begin{equation*}
t_N \equiv t_N^h \dd t_N^0 - \gamma_N^h,\quad \mbox{where}\quad \gamma_N^h:=\sum_{j=1}^N  \sum_{i=1, i\neq j}^N  (1 -Z_{N,i}) Z_{N,j}I_{h,N} (b_{N,j} - b_{N,i})
\end{equation*}
where $I_{h,N}$ is defined in \eqref{Ian} of the main paper.
\end{proposition}

This decomposition allows us to complete the proof of~\cref{propo-asy-tauN} via Slutsky's theorem by establishing two separate results:
 $$N^{-3/2}\left(t_N^0-\frac{m(N+1)}{2}\right)\overset{d}{\longrightarrow}\mathcal{N}\left(0,\frac{\lambda(1-\lambda)}{12}\right)\quad \mbox{and}\quad N^{-3/2}\gamma_N^h\overset{P}{\longrightarrow}-h\lambda(1-\lambda)\mathcal{I}_b,$$
 where $\mathcal{I}_b$ is defined in~\cref{ACjs} of the main paper. The asymptotic fluctuation of $t_N^0$ in the above display follows directly from~\cref{propo-null1} as discussed above. To prove the limit of $\gamma_N^h$, we use the standard Markov's inequality coupled with the convergence condition in~\cref{ACjs}; see \cref{SN} for details. 

\paragraph{Classical tools due to \citet{HL56,HL63}.} Once the asymptotic distribution of $t_N\equiv t_N^h$ is obtained, the same for $\rsnu$ as stated in~\cref{tau_HL.CLT}, follows as a direct consequence of a classical argument due to Hodges and Lehmann~\citep{HL63}; see also~\cref{HLtypeCLT}.

Our efficiency lower bound comparing $\rsnu$ with the difference-in-means estimator $\linu$ (\cref{efflb}) builds on the asymptotic normality result in~\cref{tau_HL.CLT}, combined with the classical efficiency theory of \citet[Theorem 1]{HL56} and the finite-population perspective of \citet[Theorem 5]{Li2017}. The key technical step involves verifying that under~\cref{empirical-assump-1}, the quantity $\mathcal{I}_b$ in~\cref{ACjs} equals $\int_\R f^2(x)\,dx$, where $f$ is the limiting density of the control potential outcomes; see~\cref{proof:efflb} for details.

\paragraph{Handling global dependence of indicators.}
We now move on to our results from~\cref{sec:wAdj} of the main paper. The proof of~\cref{tNadj.finalCLT} requires substantially different techniques compared to its unadjusted counterpart~\cref{propo-asy-tauN}. To understand why,
{recall from \eqref{WRS_defn3} of the main paper} that under $\tau=\tau_0 - hN^{-1/2}$, the regression adjusted WRS test statistic is given by:
\begin{align*}
t_{N,\textrm{adj}}\equiv t_{N,\textrm{adj}}^h=\sum_{i,j=1}^N Z_{N,j}\mathbf{1}\left(hN^{-1/2}(\vec{p}_{N,i}-\vec{p}_{N,j})^{\top}\vec{Z}_N\leq \tilde{b}_{N,j}-\tilde{b}_{N,i}-hN^{-1/2}\right) ,
\end{align*}
where $\vec{p}_{N,j}$ is the $j$-th row of the projection matrix $\vec{P}_{\vec{X}_N}$ that projects onto the column space of $\vec{X}_N$, and the residuals are defined as $\tilde{b}_{N,i}:=Y_{N,i}-\tau_0 Z_{N,i}-\vec{p}_{N,i}^{\top}\vec{b}_N$. The critical obstacle is that, for every pair $(i,j)$, the indicators in $t_{N,\textrm{adj}}$ depend on the entire random vector $(Z_{N,1},\ldots , Z_{N,N})$, unlike in \eqref{tnh} where for every pair $(i,j)$, the indicators depend only on $(Z_{N,i},Z_{N,j})$. This global dependence structure makes standard combinatorial calculations intractable. 
We circumvent this by leveraging properties of the projection matrix and showing that the term $(\vec{p}_{N,i}-\vec{p}_{N,j})^{\top}\vec{Z}_N$ appearing in the indicators is asymptotically negligible, allowing us to replace $t_{N,\textrm{adj}}^h$ with a simpler statistic. To be more specific, we define a quantity $\widetilde{A}_{N,\textrm{adj}}^h$ in \eqref{tNadj.decomp2}, which is essentially given by
$$\widetilde{A}_{N,\textrm{adj}}^h=\sum_{i,j=1}^N Z_{N,j}\mathbf{1}\left(0\leq \tilde{b}_{N,j}-\tilde{b}_{N,i}-hN^{-1/2}\right)+ \frac{m(m+1)}{2} +o_p(1).$$
Our main contribution lies in showing that 
$$N^{-3/2}|t_{N,\textrm{adj}}^h-\widetilde{A}_{N,\textrm{adj}}^h|=o_p(1).$$
The proof of this step is technical and we refer the reader to~\cref{theta_ij,gamma_ij,D_NQ_N} for details. Once the above display is established, it only remains to obtain the asymptotic distribution of $N^{-3/2}\widetilde{A}_{N,\textrm{adj}}$. Note that the indicators in $\widetilde{A}_{N,\textrm{adj}}^h$ are deterministic and do not depend on $\vec{Z}_N$. This puts us in a similar situation to the analysis of $t_N^0$, which we have already discussed above.
Once~\cref{tNadj.finalCLT} is established, the proof of~\cref{tau_adj.final.CLT} follows using ideas similar to the proof of~\cref{tau_HL.CLT}.

\paragraph{Fourier analysis.} We establish the efficiency gain by regression adjustment, as in~\cref{Jb.geq.Ib}, by leveraging tools from Fourier analysis. In view of \cref{tau_HL.CLT,tau_adj.final.CLT}, this requires showing that $\mathcal{I}_b \leq \mathcal{J}_b$, where these quantities are defined in \cref{ACjs,AssumpB2}, respectively. To derive this inequality, we impose regularity conditions on the empirical distribution of the potential outcomes (\cref{empirical-assump-1}) and the residuals (\cref{empirical-assump-2}). It follows under~\cref{empirical-assump-1} that $\mathcal{I}_b=\int_Rf^2(x)\,dx$ and under~\cref{empirical-assump-2} that $\mathcal{J}_b = \int_\R g^2(x)\,dx$, where $f$ and $g$ are the limiting densities of the empirical distributions on the control potential outcomes and the residuals, respectively. Moreover, \cref{empirical-assump-2} posits an asymptotic independence-like condition on the joint empirical distribution of the residuals and the fitted values, which implies that the densities $f$ and $g$ are related via a convolution. 
To finish the proof, we apply the Parseval–Plancherel identity from Fourier analysis; see \cref{proof:Jb.geq.Ib} for details.

\paragraph{Gaussian mollifiers.} We characterize the asymptotic variances of Rosenbaum's estimators $\rsnu$ and $\rsna$ using the unknown limiting quantities $\mathcal{I}_b$ and $\mathcal{J}_b$, respectively. A natural question is how to estimate these quantities. Here, we explain the core idea in the unadjusted case, and refer the reader to \cref{plug-in est} for details. By~\cref{AssumpIb.new}, $\mathcal{I}_b$ is the limit of $u_N$ where  
\begin{multline*}
u_N:=N^{-3/2}\sum_{j=1}^N \sum_{i=1}^N \left(\mathbf{1}( Y_{N,j}-Y_{N,i}-\tau (Z_{N,j}-Z_{N,i})\geq 0)\right.\\ -\left.\mathbf{1}(Y_{N,j}-Y_{N,i}-\tau (Z_{N,j}-Z_{N,i})\geq N^{-\nu})\right).
\end{multline*}
The natural plug-in estimator can then be constructed by replacing $\tau$ by $\rsnu$ above, to get a plug-in analogue of $u_N$, say $\widehat{u}_N^{\mathrm{R}}$. While $\widehat{u}_N^{\mathrm{R}}$ should intuitively be consistent, there are two technical issues: (a) indicators are not continuous functions, and (b) the sets in the indicators used above are not fixed, but effectively shrinking with $N$ (they are of the form $[0,N^{-\nu}]$). These technical issues preclude the possibility of using a continuous mapping type argument to establish consistency. We circumvent this by approximating the  indicators in $u_N$ with \emph{Gaussian mollifiers} with an appropriately chosen variance parameter. In particular, setting $\Phi(\cdot)$ as the standard Gaussian distribution  function and $\sigma_{N,\nu}:=(\log{N})/N^{\nu}$, define
\begin{multline*}
\hat{u}_N^{\mathrm{G}}:=\frac{1}{N^{3/2}}\sum_{i,j=1}^N  \left(\Phi(\sigma_{N,\nu}^{-1} (Y_{N,j}-Y_{N,i}-\rsnu (Z_{N,j}-Z_{N,i}))\right.\\ -\left.\Phi(\sigma_{N,\nu}^{-1}(Y_{N,j}-Y_{N,i}-\rsnu (Z_{N,j}-Z_{N,i}-N^{-\nu}))\right).
\end{multline*}
Using some careful truncation arguments, we then show that $\hat{u}_N^{\mathrm{G}}-\hat{u}_N^{\mathrm{R}}\overset{P}{\longrightarrow}0$ and $\hat{u}_N^{\mathrm{G}}-u_N\overset{P}{\longrightarrow}0$, which when combined, yield $\hat{u}_N^{\mathrm{R}}\overset{P}{\longrightarrow}\mathcal{I}_b$.

\section{Proofs of main results in the unadjusted case}\label{main:proofs:unadj}

In the sequel, if $\{x_n\}_{n\ge1}$ and $\{y_n\}_{n\ge1}$ are two sequence of positive real numbers, we write $x_n \sim y_n$ to denote that $\lim_{n\to\infty} x_n/y_n = 1$, and $x_n \lesssim y_n$ to denote that $x_n \leq C y_n$ holds for all sufficiently large $n$, for some constant $C>0$.


\subsection{Proof of~\texorpdfstring{\cref{propo-null1}}{Proposition 2.1}}\label{proof:propo-null1}

\begin{proof}
We begin by recalling that 
$t_N := \widehat{\vec{q}}_N^\top \vec{Z}_N$, where $$\widehat{q}_{N,j} := \sum_{i=1}^N \ind{Y_{N,i} - \tau_0 Z_{N,i} \le Y_{N,j} - \tau_0 Z_{N,j}}\ ,\ 1\le j\le N.$$
It follows under $\tau=\tau_0$ that
\begin{equation}\label{null-dd} 
t_N \dd \vec{q}_N^\top \vec{Z}_N,\quad q_{N,j} = \sum_{i=1}^N \ind{b_{N,i}  \le b_{N,j} }.
\end{equation} 
Further, under $\tau=\tau_0$,
\begin{equation*}\label{permu}
\vec{q}_N^\top \vec{Z}_N \dd \sum_{i=1}^m q_{N,\Pi_N(i)} = \sum_{i=1}^N c_{N,i}\cdot q_{N,\Pi_N(i)}
\end{equation*}
where $c_{N,i} = \ind{i\le m}$ and $\Pi_N$ is a random permutation of $\{1,2,\dots,N\}$. 
Note, $\bar{c}_N := N^{-1} \sum_{i=1}^N c_{N,i} = m/N$, and thus \begin{align*}
 \lim_{N\to\infty} N\cdot \frac{\max_{1\le i\le N} (c_{N,i} - \bar{c}_N)^2}{\sum_{i=1}^N (c_{N,i} - \bar{c}_N)^2}
&= \lim_{N\to\infty}\frac{\max\{\bar{c}_N^2, (1 - \bar{c}_N)^2\}}{\frac{m}{N}(1- \bar{c}_N)^2 + \left(1 - \frac{m}{N}\right) \bar{c}_N^2}
\\
&= \frac{\max\{\lambda^2, (1 -\lambda)^2\}}{\lambda(1-\lambda)}.
\end{align*}
Consequently,
\begin{align*}
 \lim_{N\to\infty} N\cdot \frac{\max_{1\le i\le N} (c_{N,i} - \bar{c}_N)^2}{\sum_{i=1}^N (c_{N,i} - \bar{c}_N)^2}\cdot\frac{\max_{1\le i\le N} (q_{N,i} - \overline{q}_N)^2}{\sum_{i=1}^N (q_{N,i} - \overline{q}_N)^2}
= 0.
\end{align*}
In view of the above, Hoeffding's combinatorial CLT \citet[Theorem 4]{Hoeffding51} implies that under $\tau=\tau_0,$ 
\begin{equation*}
\frac{t_N - \E( t_N)}{\sqrt{\mathrm{Var}(t_N)}} \dto \mathcal{N}(0,1)\text{ as $N\to\infty$.}
\end{equation*}
Now invoking~\cref{null-var}, we have  $\var_{\tau_0}(t_N) \sim  \frac{\lambda(1-\lambda)}{12}N^{3}$, which completes the proof.
\end{proof}

\subsection{Proof of~\texorpdfstring{\cref{ABP}}{Theorem 2.2}}\label{proof:ABP}

\begin{proof}
 It follows from the algebraic manipulations provided in \citet[Section 4]{HL63} that Rosenbaum's estimator $\rsnu$ as defined in~\eqref{eq:Tau-Unadj} of the main paper, is given by\begin{equation*}
 \rsnu = \median\{Y_i - Y_j: Z_i = 1, Z_j = 0, 1\le i, j\le N\}.
 \end{equation*}
Now suppose we change $k$ responses arbitrarily. If $k_1=k_1(z_1,\dots,z_N)$ denotes how many of those units fall in the treatment group, the number of pairs $(i, j)$ with $z_i=1$ and $z_j=0$ that remain unchanged is given by $I(k_1)=(m-k_1)(N-m-(k-k_1))$. The estimator $\rsnu$ remains bounded as long as $I(k_1)\ge m(N-m)/2$, and $\rsnu$ can be arbitrarily large (in absolute value) for some treatment assignment if $I(k_1)<m(N-m)/2$ holds for some treatment assignment.

Since the function $I(x)=(m-x)(N-m-(k-x))$ is concave, the smallest possible $I(k_1)$ is obtained at one of the endpoints $k_1\in\{0,k\}$. By placing all contaminated units in one arm (either treated or control), we deduce that $$\min_{k_1} I(k_1)=\min\{m(N-m-k), (m-k)(N-m)\}=m(N-m)-k \max\{m, N-m\},$$ and note that this minimum is achieved provided $k<\min\{m, N-m\}$. We can therefore continue from the last paragraph to deduce that $\rsnu$ can be arbitrarily large in absolute value for some treatment assignment iff
\begin{align*}
    m(N-m) - k\max\{m, N-m\} <\frac{1}{2}m(N-m)\\
    \iff k > \frac{m(N-m)}{\max\{m,N-m\}}=\frac{1}{2}\min\{m,N-m\}.
\end{align*}
Therefore,
$$\mathrm{BP}(\rsnu)=\frac{1}{N}\left(\left\lfloor\frac{m(N-m)}{2}\right\rfloor+1\right).$$
Now letting $N\to\infty$ and invoking \cref{randomization}, we conclude that
\begin{equation}
    \label{abp-tauR}
    \mathrm{ABP}(\rsnu)=\frac{1}{2}\min\{\lambda,1-\lambda\}.
\end{equation}
Next, we derive the ABP of the weighted average quantile estimators. Let $\nu$ be any finite signed (Borel) measure on $[0,1]$ with $\nu([0,1])=1$. Consider the weighted average quantile estimator $\widehat{\tau}_{waq}(\nu)$ defined in \eqref{def:waq} of the main paper which we reproduce here for convenience
\begin{equation*}
        \widehat\tau_{\mathrm{waq}}(\nu)= \sum_{i=1}^m  \nu\left(\left[\frac{i-1}{m},\frac{i}{m}\right]\right)a_{(i)} - \sum_{i=1}^{N-m}\nu\left(\left[\frac{i-1}{N-m},\frac{i}{N-m}\right]\right)b_{(i)}.
\end{equation*}
We will show that, under \cref{randomization},
\begin{equation}
    \label{abp-of-waq}
    \mathrm{ABP}(\widehat{\tau}_{waq}(\nu))=\min\{\alpha_{-}(\nu),\alpha_{+}(\nu)\}\min\{\lambda,1-\lambda\},
\end{equation}
    where 
    $$\alpha_{-}(\nu):=\sup\{\alpha:\nu([0,s])=0\,\forall\, s\le \alpha\},\quad \alpha_{+}(\nu):=\sup\{\alpha:\nu([1-s,1])=0\,\forall\, s\le\alpha\}.$$
    To show this, note that increasing each of the largest $k$ order statistics in the treated group by $M_1$ changes their contribution in $\widehat{\tau}_{\mathrm{waq}}(\nu)$ by $M_1\cdot \nu([1-k/m,1])$. Align the sign of $M_1$ to that of $\nu([1-k/m,1])$. We can align the sign of $M_1$ to that of $\nu([1-k/m,1])$, so that this product is positive. 
    Similarly, decreasing each of the smallest $k$ order statistics in the treated group by $-M_2$ changes their contribution in $\widehat{\tau}_{\mathrm{waq}}(\nu)$ by $-M_2\cdot \nu([0,k/m])$. Align the sign of $M_2$ to that of $\nu([0,k/m])$ so that this product is negative. These change in tail masses do not make $\widehat{\tau}_{\mathrm{waq}}(\nu)$ arbitrarily large in magnitude as long as the contaminated fraction is less than $\alpha(\nu):=\min\{\alpha_{-}(\nu),\alpha_+(\nu)\}$, and the smallest $k$ that alters this is given by $k^*_1 =\lfloor\alpha(\nu)m\rfloor+1$.
    
     Similarly, the smallest $k$ such that altering $k$ outcomes in the control group pushes $\widehat{\tau}_{\mathrm{waq}}(\nu)$ to take arbitrarily large values in magnitude is given by $k_0^* = \lfloor\alpha(\nu)(N-m)\rfloor+1$. We therefore conclude that
     $$\mathrm{BP}(\widehat{\tau}_{\mathrm{waq}}(\nu))=\frac{\min\{k_0^*,k_1^*\}}{N}=\min\left\{\frac{\lfloor\alpha(\nu)m\rfloor+1}{N},\frac{\lfloor\alpha(\nu)(N-m)\rfloor+1}{N}\right\}.$$
     As $N\to\infty$ such that $m/N\to\lambda\in(0,1)$, the above quantity converges to $\alpha(\nu)\min\{\lambda,1-\lambda\}$. 
    Finally, note that $$\alpha_{-}(\nu)+\alpha_{+}(\nu)>1\implies \nu([0,1])\le \nu([0, \alpha_{-}(\nu)])+\nu([1-\alpha_{+}(\nu), 1])=0,$$ which contradicts $\nu([0,1])=1$. Therefore, $\alpha_{-}(\nu)+\alpha_{+}(\nu)\le 1$, which implies that
 $\alpha(\nu)=\min\{\alpha_{-}(\nu),\alpha_{+}(\nu)\}\le 1/2$. This, combined with \eqref{abp-tauR} and \eqref{abp-of-waq}, finishes the proof.
\end{proof}

%
%

\subsection{Proof of~\texorpdfstring{\cref{propo-asy-tauN}}{Theorem 2.3}}\label{proof:propo-asy-tauN}

\begin{proof}
We start with the decomposition stated in~\cref{propo-decomp}. Under $\tau=\tau_N,$ it holds that
 \begin{equation*}
t_N  \dd \underbrace{m + \sum_{j=1}^N  Z_{N,j}  \sum_{i=1, i\neq j}^N   \ind{b_{N,i} \le b_{N,j}}}_{\text{call this }W_N} -  \underbrace{\sum_{j=1}^N  \sum_{i=1, i\neq j}^N  (1 -Z_{N,i}) Z_{N,j} I_{h,N} (b_{N,j} - b_{N,i})}_{\text{call this }S_N},
\end{equation*}
where $I_{h,N}(\cdot)$ is as in \eqref{Ian} of the main paper.
We now observe that the randomization distribution of $W_N$ is identical to the null distribution of $t_N$. Thus, invoking~\cref{propo-null1} and~\cref{lem001}, we deduce that 
$$
N^{-3/2}\left(W_N - \frac{m(N+1)}{2}\right)\dto \mathcal{N}\left(0,\frac{\lambda(1-\lambda)}{12}\right).
$$
On the other hand,~\cref{SN} tells us that $
N^{-3/2} S_N \Pto h\lambda(1-\lambda)\mathcal{I}_b$. Combining these using  Slutsky's theorem we conclude that under $\tau=\tau_N$,
\begin{equation}
\begin{split}
N^{-3/2} \left(t_N - \frac{m(N+1)}{2}\right) &\dd N^{-3/2} \left(W_N - \frac{m(N+1)}{2}\right)   - N^{-3/2} S_N \\ &\dto \mathcal{N}\left(-h\lambda(1-\lambda)\mathcal{I}_b,\frac{\lambda(1-\lambda)}{12}\right),
\end{split}
\end{equation}
which completes the proof. 
\end{proof}

%
%

\subsection{Proof of~\texorpdfstring{\cref{tau_HL.CLT}}{Theorem 2.4}}\label{proof:tau_HL.CLT}

\begin{proof}
Recall from~\cref{propo-asy-tauN} that under $\tau=\tau_0+hN^{-1/2}$ we have
\begin{equation*}
N^{-3/2} \left(t_N - \frac{m(N+1)}{2}\right)  \dto \mathcal{N}\left(-h\lambda(1-\lambda) \mathcal{I}_b,\frac{\lambda(1-\lambda)}{12}\right),
\end{equation*}
where $\mathcal{I}_b$ is defined  in~\cref{ACjs}. 
Now we can invoke~\cref{HLtypeCLT} to complete the proof.
\end{proof}

%

\subsection{Proof of~\texorpdfstring{\cref{efflb}}{Theorem 2.5}}\label{proof:efflb}

\begin{proof}
Let $B_{N,1},B_{N,2}$ be two independent samples from $F_{N}$. Then the quantity on the LHS of~\cref{ACjs} reads 
\begin{equation}\label{pf:efflb_1}
\sqrt{N}\, \mathbb{P}\left(0\le B_{N,2}- B_{N,1}<\frac{h}{\sqrt{N}} \right) =\sqrt{N} \, \E \left(F_N\left(B_{N,1}+\frac{h}{\sqrt{N}}\right) -F_N(B_{N,1})\right).
\end{equation}
Since $$\sup_{x\in \R} \left|\sqrt{N}\left(F_N\left(x+\frac{h}{\sqrt{N}}\right) - F_N(x)\right) - h f_{}(x)\right| \to 0,$$  the quantity in \eqref{pf:efflb_1} is equal to  $h\,\E[ f_{}(B_{N,1})] + o(1)$ which tends to $h \int_\R f^2(x)dx$ as $N\to\infty$ by assumption. Thus,~\cref{ACjs} holds with $\mathcal{I}_b=\int_\R f_{}^2(x)dx$, and hence~\cref{tau_HL.CLT} gives \begin{equation*}
\sqrt{N}\left(\rsnu - \tau \right)  \dto \mathcal{N}\left(0, \left(12\lambda(1-\lambda)\right)^{-1}\left(\int_{\R} f_{}^2 (x) dx\right)^{-2}\right).
\end{equation*} (Note that, since $F_N\dto F$, the assumption $\E f(B_{N,1}) \to \int_\R f_{}^2(x)dx<\infty$ holds under mild conditions on $f_{}$; e.g., when $f$ is continuous and $\{f_{}(B_N)\}_{N\ge 1}$ is uniformly integrable where $B_N\sim \mathbb{P}_N $.  For bounded continuous $f_{}$, this is trivial.)

Next, we consider the difference-in-means estimator $\linu$.  When $N^{-1}\sum_{j=1}^N (b_{N,j} - \overline{b}_N)^2\to \sigma_{}^2$ and  $N^{-1}\max_{1\le j\le N} (b_{N,j} - \overline{b}_N)^2\to 0$,
it follows from \citet[Theorem 5]{Li2017} that $$\sqrt{N}(\linu - \tau) \dto \mathcal{N}\left(0, {\sigma_{}^2}/{\lambda(1-\lambda)}\right).$$  
Therefore, the asymptotic efficiency of $\rsnu$ relative to $\linu$ {(see~\cref{subsec: eff} of the main paper for definition)} is given by
$$\mathrm{eff}(\rsnu,\widehat{\tau}_{\mathrm{dm}})= 12\sigma_{}^2\left(\int_{\R} f_{}^{2} (x) dx\right)^2.$$
The desired lower bound then follows from \citet[Theorem 1]{HL56}. (Note that since $F_N\dto F$, $\sigma_{}^2$ is indeed the variance of the distribution with density $f_{}$ under mild regularity conditions; e.g., when $\{B_N^2\}_{N\ge 1}$ is uniformly integrable where $B_N\sim \mathbb{P}_N$.)   
\end{proof}

\subsection{Proof of~\texorpdfstring{\cref{thm:treathet}}{Theorem 5.1}}\label{proof:thm:treathet}

\begin{proof}
    It follows from the algebraic manipulations provided in \citet[Section 4]{HL63} that our modified estimator $\rsnu$ in~\eqref{eq:Tau-Unadj}, with $t(\cdot,\cdot)$ as in~\eqref{WRS_N} of the main paper, is given by\begin{align*}
 \rsnu &= \median\{Y_i - Y_j: Z_i = 1, Z_j = 0, 1\le i, j\le N\} \\ &= \median\{a_i-b_j: Z_i=1, Z_j=0, 1\le i,j\le N\}.
 \end{align*}
 For simplicity of presentation, we restrict to the case of $m$ is odd and $N$ is even. Given $(z_1,z_2)\in \{0,1\}^2$, we define 
    $$\phi(z_1,z_2):=\begin{cases} 1 & \mbox{if}\ z_1=1, z_2=0,\\ \infty & \mbox{otherwise}.\end{cases}$$
    Furthermore, given any $x\in\R$, define 
    $$K_{m,N}(x):=\frac{1}{m(N-m)}\sum_{1\le i\neq j\le N} \mathbf{1}((a_i-b_j)\phi(Z_i,Z_j)\le x).$$
    We will first prove the following two facts.
    \begin{itemize}
        \item For any $x\in\R$, we have: 
        \begin{equation}\label{eq:dec1}
            \mathbb{E}\left[K_{m,N}(x)\right]=\frac{1}{N(N-1)}\sum_{1\le i\neq j\le N} \mathbf{1}(a_i-b_j\le x).
        \end{equation}
        \item There exists a constant $C>0$ (free of $N$) such that 
        \begin{equation}\label{eq:dec2}
            \sup_{x\in \R}\mbox{Var}\left[K_{m,N}(x)\right]\le C N^{-1}.
        \end{equation}
\end{itemize}

        \emph{Proof of \eqref{eq:dec1}.} For $i\neq j$, note that $\mathbb{P}(Z_i=1, Z_j=0)=\frac{m(N-m)}{N(N-1)}$. Therefore, 
        \begin{align*}
            \mathbb{E}[K_{m,N}(x)]&=\frac{1}{m(N-m)}\sum_{1\leq i\neq j\le N}\mathbf{1}(a_i-b_j\le x)\frac{m(N-m)}{N(N-1)}\\ &=\frac{1}{N(N-1)}\sum_{1\le i\neq j\le N}\mathbf{1}(a_i-b_j\le x).
        \end{align*}

\vspace{0.1in}

        \emph{Proof of \eqref{eq:dec2}.} 
        For notational convenience, define $\Theta_{i,j}:=(a_i-b_j)\phi(Z_i,Z_j)$. Note that for $i,j,k,l$ distinct, we have 
        \begin{align*}
        \mathbb{P}(Z_i=Z_k=1,Z_j=Z_l=0)=\frac{m(m-1)(N-m)(N-m-1)}{N(N-1)(N-2)(N-3)}.
        \end{align*}
        We also observe that 
        \begin{align}\label{eq:dec21}
        &\;\;\;\;\frac{1}{m^2(N-m)^2}\sum_{(i,j,k,l)\ \mbox{distinct}} \mathbb{E}[\mathbf{1}(\Theta_{i,j}\le x)\mathbf{1}(\Theta_{k,l}\le x)]\nonumber \\ &=\frac{1}{m^2(N-m)^2}\sum_{(i,j,k,l)\ \mbox{distinct}} \mathbf{1}(a_i-b_j\le x)\mathbf{1}(a_k-b_l\le x)\frac{m(m-1)(N-m)(N-m-1)}{N(N-1)(N-2)(N-3)}\nonumber \\&=\frac{(m-1)(N-m-1)}{m(N-m)N(N-1)(N-2)(N-3)}\left(\sum_{1\le i\neq j\le N}\mathbf{1}(a_i-b_j\le x)\right)^2+O(N^{-1})\nonumber \\&=\left(\mathbb{E}[K_{m,N}(x)]\right)^2+O(N^{-1}).
        \end{align}
        Here all the $O(N^{-1})$ terms are uniform bounds in $x$ which holds as the indicators are all uniformly bounded by $1$ irrespective of $x$. Note that we need $m/N$ to be bounded away from $0$ and $1$ for the above to hold. 

        Next, we observe that 
        \begin{align*}
            &\mathbb{E}[K^2_{m,N}(x)]\\ &=O(N^{-1})+\frac{1}{m^2(N-m)^2}\sum_{(i,j,k,l)\ \mbox{distinct}} \mathbb{E}[\mathbf{1}(\Theta_{i,j}\le x)\mathbf{1}(\Theta_{k,l}\le x)]\\ &=\left(\mathbb{E}[K_{m,N}(x)]\right)^2+O(N^{-1}).
        \end{align*}
        In the last equality above, we have used \eqref{eq:dec21}. This completes the proof of \eqref{eq:dec2}.
\vspace{0.1in}

        We are now in position to complete the proof of \cref{thm:treathet} of the main paper.   
        Recall that $\mbox{med}_N=\median\{a_i-b_j: 1\le i\neq j\le N\}$ (see \eqref{eq:medpair} of the main paper). Given any $\varepsilon>0$, we then have 
        \begin{align*}
            &\;\;\;\;\mathbb{P}(\rsnu-\mbox{med}_N>\varepsilon)\\ &=\mathbb{P}\left(K_{m,N}(\mbox{med}_N+\varepsilon)\le \frac{1}{2}\right)\\ &=\mathbb{P}\bigg(\frac{1}{N(N-1)}\sum_{1\le i\neq j\le N} \mathbf{1}(a_i-b_j\le \mbox{med}_N+\varepsilon)-K_{m,N}(\mbox{med}_N+\varepsilon)\\ &\qquad\qquad \ge \frac{1}{N(N-1)}\sum_{1\le i\neq j\le N} \mathbf{1}(a_i-b_j\le \mbox{med}_N+\varepsilon)-\frac{1}{2}\bigg)\\ &\overset{(a)}{\le} \mathbb{P}\left(\bigg|\mathbb{E}[K_{m,N}(\mbox{med}_N+\varepsilon)]-K_{m,N}(\mbox{med}_N+\varepsilon)\bigg|\ge \frac{\sum_{1\le i\neq j\le N} \mathbf{1}(a_i-b_j\le \mbox{med}_N+\varepsilon)}{N(N-1)}-\frac{1}{2}\right)\\ &\overset{(b)}{\le} \frac{\mbox{Var}(K_{m,N}(\mbox{med}_N+\varepsilon))}{\left(\frac{\sum_{1\le i\neq j\le N} \mathbf{1}(a_i-b_j\le \mbox{med}_N+\varepsilon)}{N(N-1)}-\frac{1}{2}\right)^2}\\ &\overset{(c)}{\le} C \left(\sqrt{N}\left(\kappa_{N}^{(1)}(\varepsilon)-\frac{1}{2}\right)\right)^{-2}\to 0.
        \end{align*}
        Here (a) follows from \eqref{eq:dec1} and the fact that the definition of $\mbox{med}_N$ implies $\frac{\sum_{1\le i\neq j\le N} \mathbf{1}(a_i-b_j\le \mbox{med}_N+\varepsilon)}{N(N-1)}\ge \frac{1}{2}$, (b) follows from Markov's inequality, and (c) is immediate from \eqref{eq:dec2}, the definition of $\kappa_n^{(1)}(\varepsilon)$, and assumption \eqref{eq:separation} of the main paper. 
        A similar calculation shows that $\mathbb{P}(\rsnu-\mbox{med}_N<-\varepsilon)\to 0$. This completes the proof.
\end{proof}

\section{Proofs of main results in the regression-adjusted case}\label{main:proofs:adj}

%

\subsection{Proof of~\texorpdfstring{\cref{propo-null2}}{Theorem 3.1}}\label{proof:propo-null2}

\begin{proof}
Denote by $\vec{p}_{N,i}$ the $i$-th row of $\vec{P}_{\vec{X}_N}$. Observe that
$$e_{N,i}\le e_{N,j}\iff Y_{N,i} - \tau_0 Z_{N,i} - \vec{p}_{N,i}^\top (\vec{Y}_N - \tau_0 \vec{Z}_N)\le Y_{N,j} - \tau_0 Z_{N,j} - \vec{p}_{N,j}^\top (\vec{Y}_N - \tau_0 \vec{Z}_N).$$ 
So under $\tau=\tau_0,$ we have $t_{N,\mathrm{adj}} \dd \sum_{j=1}^N q_{N,j} Z_{N,j}$, where 
\begin{equation*}
q_{N,j} := \sum_{i=1}^N \ind{b_{N,i} - \vec{p}_{N,i}^\top \vec{b}_N \leq  b_{N,j} - \vec{p}_{N,j}^\top \vec{b}_N },\ j=1,2,\dots,N.
\end{equation*}
Since  the ranks $q_{N,j}$'s are deterministic, the asymptotic normality of $\sum_{j=1}^N q_{N,j} Z_{N,j}$ can be derived in the same way as for the without regression adjustment  case. {A closer look at the proofs of \cref{lem001,propo-null1} in the without regression adjustment  case reveals that the following results hold in this case as well, since \cref{AssumpB2} plays the role of \cref{ACjs}.}
\begin{enumerate}[label = (\alph*)]
\item As $N\to\infty,$ $\sum_{j=1}^N \left(q_{N,j} - \overline{q}_N\right)^2= \frac{1}{12}{N(N^2-1)} + o(N^3)$.
\item  $\lim_{N\to\infty} {\max_{1\le j \le N}  \left(q_{N,j} - \overline{q}_N\right)^2}/{\sum_{j=1}^N  \left(q_{N,j} - \overline{q}_N\right)^2} = 0$.
\item Under $\tau=\tau_0$,
$\var(t_{N,\mathrm{adj}}) \sim \frac{1}{12}\lambda(1-\lambda)N^3$ as $N\to\infty$.
\end{enumerate}
Equipped with (b) above, we apply Hoeffding's combinatorial CLT \citep[Theorem 4]{Hoeffding51} to say that under $\tau=\tau_0$, 
\begin{equation*}
\frac{t_{N,\mathrm{adj}} - \E( t_{N,\mathrm{adj}})}{\sqrt{\var(t_{N,\mathrm{adj}})}} \dto \mathcal{N}(0,1)\text{ as $N\to\infty$.}
\end{equation*}
This, in conjunction with (c) above, completes the proof.
\end{proof}

\subsection{Proof of~\texorpdfstring{\cref{tNadj.finalCLT}}{Theorem 3.2}}\label{proof:tNadj.final.CLT}

\begin{proof}
Recall the notation
$\widetilde{b}_{N,i} = b_{N,i}- \vec{p}_{N,i}^\top \vec{b}_N$ ($1\le i\le N$) from \eqref{btilde} of the main paper.
Under $\tau=\tau_{N},$ we have
$$\vec{Y}_N - \tau_0 \vec{Z}_N \dd \vec{b}_N + (\tau_N - \tau_0) \vec{Z}_N =  \vec{b}_N - \frac{h}{\sqrt{N}} \vec{Z}_N,$$
which implies that
\begin{align*}t_{N,\mathrm{adj}} &\dd  \sum_{j=1}^N Z_{N,j} \sum_{i=1}^N \ind{ \widetilde{b}_{N,i} - \frac{h}{\sqrt{N}} Z_{N,i} + \frac{h}{\sqrt{N}} \vec{p}_{N,i}^\top \vec{Z}_N \leq   \widetilde{b}_{N,j} - \frac{h}{\sqrt{N}} Z_{N,j} + \frac{h}{\sqrt{N}}\vec{p}_{N,j}^\top  \vec{Z}_N} \\
&= \sum_{j=1}^N Z_{N,j}\sum_{i=1}^N (1- Z_{N,i}) \ind{ \widetilde{b}_{N,i} + \frac{h}{\sqrt{N}}\vec{p}_{N,i}^\top \vec{Z}_N \leq  \widetilde{b}_{N,j} + \frac{h}{\sqrt{N}}\vec{p}_{N,j}^\top \vec{Z}_N - \frac{h}{\sqrt{N}} } + C_N,
\end{align*}
where 
\begin{equation}\label{C_N.def}
C_N := \sum_{j=1}^N \sum_{i=1}^N Z_{N,j} Z_{N,i} \ind{ \widetilde{b}_{N,i} + \frac{h}{\sqrt{N}}\vec{p}_{N,i}^\top \vec{Z}_N  \leq  \widetilde{b}_{N,j} + \frac{h}{\sqrt{N}}\vec{p}_{N,j}^\top \vec{Z}_N }.
\end{equation}
Thus under $\tau=\tau_N$ we have 
$t_{N,\mathrm{adj}} \dd \mathbf{I}_N - \mathbf{II}_N + C_N$,
where
\begin{equation}\label{I.II.N}
\mathbf{I}_N := \sum_{j=1}^N  \sum_{i=1}^N Z_{N,j}\xi_{N,i,j},\quad \text{and}\quad \mathbf{II}_N :=  \sum_{j=1}^N \sum_{i=1}^N Z_{N,j} Z_{N,i}\xi_{N,i,j},
\end{equation}
and
\begin{align*}
\xi_{N,i,j}  &:=\ind{ b_{N,i} - \vec{p}_{N,i}^\top (  \vec{b}_N - hN^{-1/2} \vec{Z}_N) \leq    b_{N,j}  - \vec{p}_{N,j}^\top (  \vec{b}_N - hN^{-1/2} \vec{Z}_N)- hN^{-1/2} }\\ 
&= \ind{hN^{-1/2} (\vec{p}_{N,i} - \vec{p}_{N,j})^\top  \vec{Z}_N \leq  \widetilde{b}_{N,j} - \widetilde{b}_{N,i} - hN^{-1/2} }.
\end{align*}
The indicators $\xi_{N,i,j} $ are quite complicated to handle, since it depends on the entire random vector $(Z_{N,1}, \dots , Z_{N,N})$, for every pair $(i, j)$. To circumvent this technical hurdle, we replace $\xi_{N,i,j}$ with $\widetilde{\xi}_{N,i,j}$, where
 $$ \widetilde{\xi}_{N,i,j} := \ind{0 \leq   \widetilde{b}_{N,j} - \widetilde{b}_{N,i} - hN^{-1/2} }.$$
Define 
\begin{equation}\label{I.II.tilde.N}
\widetilde{\mathbf{I}}_N := \sum_{j=1}^N Z_{N,j} \sum_{i=1}^N \widetilde{\xi}_{N,i,j},\quad \text{and}\quad
\widetilde{\mathbf{II}}_N :=  \sum_{j=1}^N \sum_{i=1}^N Z_{N,j} Z_{N,i} \widetilde{\xi}_{N,i,j}.\end{equation}
Thus we decompose $t_{N,\mathrm{adj}}$ under $\tau=\tau_N$ as
\begin{equation}\label{tNadj.decomp2}
t_{N,\mathrm{adj}} \dd \underbrace{\left(\widetilde{\mathbf{I}}_N - \widetilde{\mathbf{II}}_N+ C_N\right)}_{\widetilde{A}_{N,\textrm{adj}}^h}+ \underbrace{(\mathbf{I}_N - \widetilde{\mathbf{I}}_N)}_{D_n}- \underbrace{(\mathbf{II}_N -  \widetilde{\mathbf{II}}_N )}_{Q_n}.
\end{equation}
We first focus on  $\widetilde{A}_{N,\textrm{adj}}^h$. Note that in \eqref{C_N.def} we are summing up the ranks of the $m$ numbers $$\{\widetilde{b}_{N,j}  +hN^{-1/2} \vec{p}_{N,j}^\top \vec{Z}_N : 1\le j\le N, Z_{N,j} = 1\}$$ within this set. Unfortunately, due to possibility of ties, we cannot directly equate it with $m(m+1)/2$. 
However, invoking~\cref{AssumpB2} we can show that $N^{-3/2}(C_N - m(m+1)/2) = o_p(1)$ as $N\to\infty$. Towards that, observe that for any $\delta>0$,
\begin{align*}
    &\left|C_N - \frac{m(m+1)}{2}\right| = \sum_{j=1}^N\sum_{i=1}^N \ind{\widetilde{b}_{N,j} -\widetilde{b}_{N,i} + hN^{-1/2}(\vec{p}_{N,j}-\vec{p}_{N,i})^\top \vec{Z}_N = 0}\\
    &\leq \sum_{j=1}^N\sum_{i=1}^N \ind{0\le \widetilde{b}_{N,j} -\widetilde{b}_{N,i} + hN^{-1/2}(\vec{p}_{N,j}-\vec{p}_{N,i})^\top \vec{Z}_N < \delta N^{-1/2}}\\
    &\leq \sum_{(i,j): \left|(\vec{p}_{N,j}-\vec{p}_{N,i})^\top \vec{Z}_N\right|<\delta} \ind{0\le \widetilde{b}_{N,j} -\widetilde{b}_{N,i} + hN^{-1/2}(\vec{p}_{N,j}-\vec{p}_{N,i})^\top \vec{Z}_N < \delta N^{-1/2}}\\
    &\qquad\qquad + \left|\left\{(i,j): \left|(\vec{p}_{N,j}-\vec{p}_{N,i})^\top \vec{Z}_N\right|\ge \delta\right\}\right|\\
    &\leq \sum_{(i,j): \left|(\vec{p}_{N,j}-\vec{p}_{N,i})^\top \vec{Z}_N\right|<\delta} \ind{-h\delta N^{-1/2}\le \widetilde{b}_{N,j} -\widetilde{b}_{N,i} < (h+1) \delta N^{-1/2}}\\
    &\qquad\qquad + \left|\left\{(i,j): \left|(\vec{p}_{N,j}-\vec{p}_{N,i})^\top \vec{Z}_N\right|\ge \delta\right\}\right|\\
    &\leq \sum_{j=1}^N\sum_{i=1}^N\ind{ \left| \widetilde{b}_{N,j} -\widetilde{b}_{N,i}\right| \le (h+1) \delta N^{-1/2}} + \delta^{-2}\sum_{j=1}^N\sum_{i=1}^N \left((\vec{p}_{N,j}-\vec{p}_{N,i})^\top \vec{Z}_N\right)^2.\numberthis\label{C_N}
\end{align*}
Now we invoke~\cref{PiPj} to say that 
\begin{equation*}
    N^{-3/2}\sum_{j=1}^N\sum_{i=1}^N \E \left[\left((\vec{p}_{N,j}-\vec{p}_{N,i})^\top \vec{Z}_N\right)^2\right] = O(N^{-1/2})=o(1),
\end{equation*}
for every fixed $\delta >0$.
This, in conjunction with \eqref{C_N} and Markov inequality, tells us that for any fixed $\delta>0$,
$$N^{-3/2}\left|\left\{(i,j): \left|(\vec{p}_{N,j}-\vec{p}_{N,i})^\top \vec{Z}_N\right|\ge \delta\right\}\right|=o_p(1).$$
On the other hand, we can use~\cref{AssumpB3} to conclude that under~\cref{AssumpB2}, 
\begin{equation*}
    \lim_{N\to\infty} N^{-3/2}\sum_{j=1}^N\sum_{i=1}^N\ind{ \left| \widetilde{b}_{N,j} -\widetilde{b}_{N,i}\right| \le (h+1) \delta N^{-1/2}} = 2(h+1)\delta.
\end{equation*}
Appealing to \eqref{C_N} we can now conclude, by letting $N\to\infty$ first, and then $\delta\to 0$, that $$N^{-3/2} (C_N - m(m+1)/2) = o_p(1),\text{ as }N\to\infty.$$
This allows us to write \begin{equation*}\label{anorm2} \widetilde{A}_{N,\textrm{adj}}^h =\sum_{j=1}^N Z_{N,j} \sum_{i=1}^N \ind{0\le \widetilde{b}_{N,j} - \widetilde{b}_{N,i} - hN^{-1/2} } + \frac{m(m+1)}{2} +o_p(1).\end{equation*}
In view of the above display, the asymptotic normality of $\widetilde{A}_{N,\textrm{adj}}^h$  under $\tau=\tau_N$ can be derived in the same manner as we proved the local asymptotic normality of $t_N$ in the without regression adjustment case. To be precise, it follows by mimicking the proof of~\cref{propo-asy-tauN} (in the same manner as we proved~\cref{propo-null2} by mimicking the proof of~\cref{propo-null1}) that under $\tau=\tau_N$, 
\begin{equation}\label{CLT_Atilde}
 N^{-3/2}  \left(\widetilde{A}_{N,\textrm{adj}}^h- \frac{m(N+1)}{2}\right)  \dto N\left(-h\lambda(1-\lambda) \mathcal{J}_b,\frac{\lambda(1-\lambda)}{12}\right),
\end{equation}
where $\mathcal{J}_b$ is as defined in~\cref{AssumpB2}.
The proof of the fact that $D_N$ and $Q_N$ defined in \eqref{tNadj.decomp2} are asymptotically negligible, is split into a couple of lemmas in~\cref{sec:sometechlem}.~\cref{theta_ij,gamma_ij} give upper bounds on the second moments of $D_N$ and $Q_N$, respectively. Then~\cref{D_NQ_N} shows that under~\cref{AssumpB2}, $N^{-3/2}D_N = o_p(1)$ and   $N^{-3/2}Q_N = o_p(1)$ as $N\to\infty$. This, in conjunction with \eqref{CLT_Atilde} completes the proof of~\cref{tNadj.finalCLT}.
\end{proof}

\subsection{Proof of~\texorpdfstring{\cref{tau_adj.final.CLT}}{Theorem 3.3}}\label{proof:tau_adj.final.CLT}

%
%

\begin{proof}
Recall from~\cref{tNadj.finalCLT} that under $\tau=\tau_0+hN^{-1/2}$, we have
\begin{equation*}
N^{-3/2}  \left(t_{N,\mathrm{adj}} -\frac{m(N+1)}{2}\right)  \dto \mathcal{N}\left(-h\lambda(1-\lambda) \mathcal{J}_b,\frac{\lambda(1-\lambda)}{12}\right),
\end{equation*}
where $\mathcal{J}_b$ is defined  in~\cref{AssumpB2}. 
Now we invoke~\cref{HLtypeCLT} to complete the proof.
\end{proof}


%
%

\subsection{Proof of~\texorpdfstring{\cref{Jb.geq.Ib}}{Theorem 3.4}}\label{proof:Jb.geq.Ib}

\begin{proof} We follow the same generic approach as in the proof of \cref{efflb}. There we show that \cref{empirical-assump-1} implies \cref{ACjs} with the quantity $\mathcal{I}_b$ explicitly given by\begin{equation}\label{eq:Ib-special-formula}
    \mathcal{I}_b=\int_{\R} f^2(x)\, dx,
\end{equation} 
and hence~\cref{tau_HL.CLT} gives 
\begin{equation}\label{eq:thm8-eq1}
\sqrt{N}\left(\rsnu - \tau \right)  \dto \mathcal{N}\left(0, \left(12\lambda(1-\lambda)\right)^{-1}\left(\int_{\R} f_{}^2 (x) dx\right)^{-2}\right).
\end{equation} 
A similar argument also applies to the residuals $\widetilde{b}_{N,i}$, as follows. Let $\widetilde{B}_{N,1},\widetilde{B}_{N,2}$ be two independent samples from the empirical distribution $G_{N}$ of the residuals $\widetilde{b}_{N,i}$. Then the quantity on the LHS of~\cref{AssumpB2} reads 
\begin{equation}\label{pf:efflb_cov_adj}
\sqrt{N}\, \mathbb{P}\left(0\le \widetilde{B}_{N,2}- \widetilde{B}_{N,1}<\frac{h}{\sqrt{N}} \right) =\sqrt{N} \, \E \left(G_N\left(\widetilde{B}_{N,1}+\frac{h}{\sqrt{N}}\right) -G_N(B_{N,1})\right).
\end{equation}
Since $$\sup_{x\in \R} \left|\sqrt{N}\left(G_N\left(x+\frac{h}{\sqrt{N}}\right) - G_N(x)\right) - h g_{}(x)\right| =o(1),$$  the quantity in \eqref{pf:efflb_cov_adj} is equal to  $h\,\E [g_{}(\widetilde{B}_{N,1})] + o(1)$ which tends to $h \int_\R g^2(x)dx$ as $N\to\infty$ by assumption. Thus,~\cref{AssumpB2} holds with 
\begin{equation}\label{eq:Jb-special-formula}
   \mathcal{J}_b=\int_\R g^2(x)\,dx,
\end{equation} and hence~\cref{tau_adj.final.CLT} gives \begin{equation}\label{eq:thm8-eq2}
\sqrt{N}\left(\rsna - \tau \right)  \dto \mathcal{N}\left(0, \left(12\lambda(1-\lambda)\right)^{-1}\left(\int_{\R} g_{}^2 (x)\, dx\right)^{-2}\right).
\end{equation} 
Next, denote by $H_N$ the empirical distribution of the predictions $b_{N,i}-\widetilde{b}_{N,i}=\vec{p}_{N,i}^\top\vec{b}_N$ from the linear regression of $\vec{b}_N$ on $\vec{X}_N$, i.e., $$H_N(A):=\frac{1}{N}\sum_{i=1}^N \ind{b_{N,i}-\widetilde{b}_{N,i}\in A}.$$
Also denote by $\Pi_N$ the empirical joint distribution of the residuals $\widetilde{b}_{N,i}$ and the predictions $b_{N,i}-\widetilde{b}_{N,i}$, i.e., $$\Pi_N(A\times B)=\frac{1}{N}\sum_{i=1}^N \ind{\widetilde{b}_{N,i}\in A}\ind{b_{N,i}-\widetilde{b}_{N,i}\in B}.$$
With the above notation, the asymptotic independence-like condition in \cref{empirical-assump-2} now reads 
\begin{equation}\label{eq:Product}
    \sup_{x,\,y\,\in\,\R}\left|\Pi_N((-\infty,x]\times(-\infty, y])-G_N((-\infty, x]) H_N((-\infty, y])\right|={o}(1).
\end{equation}
The tightness of the empirical distribution $F_N$ of the control potential outcomes $b_{N,i}$ (as in \cref{empirical-assump-1}) and the empirical distribution $G_N$ of the residuals $\widetilde{b}_{N,i}$ (as in \cref{empirical-assump-2}) imply that the empirical distribution $H_N$ of $b_{N,i}-\widetilde{b}_{N,i}$ is also tight. 
It then follows that along a subsequence $\{N_k\}$, $H_{N_k}$ converges weakly to a distribution $H$. Next, we use \eqref{eq:Product} to deduce that along the same subsequence $\{N_k\}$, the joint empirical distribution $\Pi_{N_k}$ converges to the produce measure: $$\Pi_{N_k}\dto G\otimes H,$$
where $G$ is the weak limit of $G_N$ as in \cref{empirical-assump-2}.
Since $b_{N,i}=\widetilde{b}_{N,i}+(b_{N,i}-\widetilde{b}_{N,i})$, the above implies that the empirical distribution $F_{N}$ of the control potential outcomes along the subsequence $\{N_k\}$ converges weakly to $G * H$, where $*$ denotes convolution. Combining this with \cref{empirical-assump-1}, we deduce that 
\begin{equation}\label{eq:convolution}
    F = G * H,
\end{equation}
where $F$ is the weak limit of $F_N$ as in \cref{empirical-assump-1}.
Denote by $\phi_F$, $\phi_G$ and $\phi_H$ the characteristic functions of the weak limits $F$, $G$ and $H$, respectively. We can now use the Parseval–Plancherel identity to deduce the following.
\begin{align*}
    \mathcal{I}_b &= \int_{\R} f^2(x)\,dx \tag{from \eqref{eq:Ib-special-formula}}\\
    &= \frac{1}{2\pi}\int_{\R} \left|\phi_F(t)\right|^2 \,dt \tag{Parseval–Plancherel identity} \\
    &= \frac{1}{2\pi}\int_{\R} \left|\phi_G(t)\right|^2\,\left|\phi_H(t)\right|^2 \,dt\tag{using the convolution in \eqref{eq:convolution}}\\
    &\le \frac{1}{2\pi}\int_{\R} \left|\phi_G(t)\right|^2 \,dt\tag{since $|\phi_H(t)|\le 1$}\\
    &=\int_{\R} g^2(x)\,dx  \tag{Parseval–Plancherel identity} \\
    &= \mathcal{J}_b. \tag{from \eqref{eq:Jb-special-formula}}
\end{align*}
This completes the proof, in light of \eqref{eq:thm8-eq1} and \eqref{eq:thm8-eq2}.
\end{proof}

\subsection{Proof of~\texorpdfstring{\cref{efflb-adj-case}}{Theorem 3.5}}\label{proof:efflb-adj-case}

\begin{proof}
    It follows from \citet[Theorem 1]{Lin13} and  \cref{cte,randomization} that
    $$\sqrt{N}\left(\lini-\tau\right)\dto \mathcal{N}\left(0,\,\frac{\sigma^2}{\lambda(1-\lambda)}\right).$$
    The rest of the proof essentially follows by arguments analogous to the proof of \cref{efflb}.
\end{proof}

\section{Proofs of other main results}\label{AppendixA}

%

%

\subsection{Proof of~\texorpdfstring{\cref{thm:CI-contamination}}{Proposition A.1}}\label{proof:thm:CI-contamination}

\begin{proof}
    Fix an arbitrary $B>|\tau|$. Suppose that $N$ is large enough so that $\lfloor\eps N\rfloor>1$. Since $\mathrm{ABP}(\widehat\tau\,)=0$, it follows that for any $y$ there exists $\psi_{B,y,z}\in\mathscr{C}_{N,\eps}$ such that $$|\widehat\tau(\psi_{B,y,z}(y, z), z)|>|\tau|+B.$$ Since $\widehat{C}_N(y,z)$ contains $\widehat{\tau}(y,z)$ for all $y,z$, it follows that 
    \begin{align*}
        \mathrm{Length}\left(\widehat{C}_N\left(\psi_{B,y,z}(y, z)\right)\right)&\ge \left|\widehat\tau(\psi_{B,y,z}(y, z), z)-\tau\right|\ind{\tau\in \widehat{C}_n(\psi_{B,y,z}(y,z),z)}\\&\ge B\cdot\ind{\tau\in \widehat{C}_n\left(\psi_{B,y,z}(y,z),z\right)},
    \end{align*}
   using triangle inequality and the definition of $\psi_{B,y,z}$. Taking expectations under $\mathbb{P}_\tau$, we conclude that
        \begin{align*}
        \sup_{\psi\in \mathscr{C}_{N,\eps}}\E_\tau\left[\mathrm{Length}\left(\widehat{C}_N\left(\psi(Y, Z)\right)\right)\right]&\ge \E_\tau\left[\mathrm{Length}\left(\widehat{C}_N\left(\psi_{B,Y,Z}(Y, Z)\right)\right)\right]\\
        &\ge B\cdot \mathbb{P}_\tau\left(\tau\in \widehat{C}_n\left(\psi_{B,Y,Z}(Y,Z),Z\right)\right)\\
        &\ge B\inf_{\psi\in \mathscr{C}_{N,\eps}} \mathbb{P}_\tau\left(\tau\in \widehat{C}_n\left(\psi(Y,Z),Z\right)\right)\\
        &\ge B\cdot (1-\alpha).
    \end{align*}
    Since $B>|\tau|$ is arbitrary, this completes the proof of the first part.

    Next, we turn to showing the second conclusion by construction. The definition of the asymptotic breakdown point combined with $\lfloor\eps N\rfloor\le \lfloor\gamma N\rfloor$ implies that the estimator $\widehat{\tau}(y_1,\dots,y_N; z_1,\dots,z_N)$ remains bounded when we arbitrarily change $\lfloor\eps N\rfloor$ many outcomes (preserving the constant treatment effect assumption). In other words,
    the set $\mathcal{R}(z) = \left\{\widehat{\tau}(\psi(y, z),z):\psi\in \mathscr{C}_{N,\eps}\right\}$ is bounded. Define $L(z):=\inf \mathcal{R}(z)$, $U(z):=\sup \mathcal{R}(z)$, which are both finite. Thus, $B(z):=\max\{|\tau-L(z)|,|\tau-U(z)|\}$ is finite for all $z\in\mathcal{Z}_N$, and therefore $B_N:=\sup_{z\in \mathcal{Z}_N} B(z)$ is finite as well. Now construct a confidence set as $$\widehat{C}_N(y,z):=[\widehat{\tau}(y,z)-2B_N, \widehat{\tau}(y,z)+2B_N].$$
    Note that $\widehat{C}_N(y,z)$ contains $\widehat{\tau}(y,z)$ by construction. Next, for any $\psi\in\mathscr{C}_{N,\eps}$, $|\tau-\widehat{\tau}(\psi(y,z), z)|\le B(z)\le B_N$ pointwise, implying that $\tau\in \widehat{C}_N(y,z)$ almost surely. Finally, the expected length of $\widehat{C}_N(\psi(Y,Z), Z)$ under $\mathbb{P}_\tau$ is uniformly bounded by $4B_N$, which completes the proof. 
\end{proof}

\subsection{Proof of~\texorpdfstring{\cref{thm:rsnu-CI-contamination}}{Proposition A.2}}\label{proof:thm:rsnu-CI-contamination}

\begin{proof}
    Denote by $F_N$ the empirical measure on the control potential outcomes $b_{N,1},\dots,b_{N,N}$. \cref{ABP} tells us that $\mathrm{ABP}(\rsnu)=\frac{1}{2}\min\{\lambda,1-\lambda\}$. Since $0<\eps<\mathrm{ABP}(\rsnu)=\frac{1}{2}\min\{\lambda,1-\lambda\}$, we deduce that $\phi(\eps)\in (0,1)$, where $$\phi(\eps):= \frac{1}{2} + \eps \frac{\max\{\lambda,1-\lambda\}}{\lambda(1-\lambda)}.$$
    For any $\eta> \sqrt{\phi(\eps)}$, we  invoke tightness of $F_N$ (uniform in $N$) to say that there exists $B_\eta>0$ such that $\inf_N F_N([-B_\eta,B_\eta])\ge \eta$ . 
    Now define $$X_T(B):=|\{i: |b_{N,i}|\le B\}\cap \{i: Z_{N,i}=1\}|,\quad X_C(B):=|\{i: |b_{N,i}|\le B\}\cap \{i: Z_{N,i}=0\}|.$$
    It is straightforward to note that $X_T(B)\sim \text{Hypergeometric}(N, NF_N([-B,B]), m)$ and $X_C(B)=NF_N([-B,B])-X_T(B)$. It follows from standard concentration inequality arguments that
    $$\Pr\left(\frac{X_T(B)}{m}\ge F_N([-B,B])- t\right)\ge 1-\exp\left(-2t^2 m\frac{N-m}{N-1}\right),$$
    and similarly $$\Pr\left(\frac{X_C(B)}{N-m}\ge F_N([-B,B])- t\right)\ge 1-\exp\left(-2t^2 m\frac{N-m}{N-1}\right).$$
    Combining the above using union bound, and using $\inf_N F_N([-B_\eta, B_\eta])\ge \eta$ for every $\eta>\sqrt{\phi(\eps)}$ we conclude that 
    $$\Pr\left(\min\left\{\frac{X_T(B_\eta)}{m}, \frac{X_C(B_\eta)}{N-m}\right\}\ge \eta- t\right)\ge 1-2\exp\left(-2t^2 m\frac{N-m}{N-1}\right).$$
    Therefore, with $\phi(\eps)$ as defined earlier, there exists $B>0$ (in particular, $B_\eta$ where $\eta=\sqrt{\phi}+t$ fits the bill) such that the event $$\mathcal{E}(t)=\{|X_T(B)|\ge m\sqrt{\phi(\eps)} \text{ and }|X_C(B)|\ge (N-m)\sqrt{\phi(\eps)}\}$$ satisfies $$\Pr(\mathcal{E}(t))\ge 1 - 2\exp(-2\lambda(1-\lambda)Nt^2).$$ Note that here $t<1-\sqrt{\phi(\eps)}$, and that the same $B$ works uniformly for all large $N$, because of the tightness assumption. On the event $\mathcal{E}(t)$, the number of cross-group pairs $(i, j)$ with $Z_{N,i}=1$, $Z_{N,j}=0$ and both $b_{N,i}$ and $b_{N,j}$ in $[-B,B]$ is bounded below by $m(N-m)\phi(\eps)$. On the other hand, it follows from \citet{HL63} that we can write $$\rsnu = \tau + \median\{b_{N,i} - b_{N,j}:Z_i=1, Z_i=0\}.$$ Note that contaminating at most $k=\lfloor\eps N\rfloor$ responses arbitrarily can change at most $\lfloor\eps N\rfloor\max\{m, n-m\}$ of the cross-group differences. Therefore, on the event $\mathcal{E}(t)$, the proportion of cross-group differences bounded by $2B$ that remain unchanged when we arbitrarily alter $k=\lfloor\eps N\rfloor$  outcomes is bounded below by
    $$\phi(\eps) - \frac{\lfloor\eps N\rfloor\max\{m, n-m\}}{m(N-m)} > \frac{1}{2}.$$
    This implies that $|\rsnu(\psi(y, z),z)-\tau|\le 2B$ on $\mathcal{E}(t)$, for any $\psi\in \mathscr{C}_{N,\eps}$. As a consequence, the interval $\widehat{C}_N(y,z)=[\rsnu(y,z)-2B,\rsnu(y,z)+2B]$ satisfies
    $$\inf_{\psi\in\mathscr{C}_{N,\eps}}\mathbb{P}_\tau(\tau\in \widehat{C}_{N,\eps}(\psi(Y, Z), Z))\ge 1-2\exp\left(-2\lambda(1-\lambda)Nt^2\right),$$
    for any $t<1-\sqrt{\phi(\eps)}$ and for all large $N$. For any such $t$, the above probability is greater than or equal to $1-\alpha$ for all sufficiently large $N$. Moreover, $\widehat{C}_n(y,z)$ contains $\rsnu(y, z)$ by construction, and satisfies $$\limsup_{N\to\infty}\sup_{\psi\in \mathscr{C}_{N,\eps}}\E_\tau \left[\mathrm{Length}\left(\widehat{C}_N(\psi(Y, Z), Z)\right)\right]\le 4B<\infty.$$
    This completes the proof.
\end{proof}

\subsection{Proof of~\texorpdfstring{\cref{tavgrank}}{Lemma A.3}}\label{proof:tavgrank}

\begin{proof}
Recall that $t_N = \sum_{j=1}^N q_N(j) Z_{N,j}$ where $q_N(j) = \sum_{i=1}^N \ind{b_{N,i}\le b_{N,j}}.$ Observe that when $s_{j-1}< j \le s_j,$ $$0\le q_N(j) - q_N^{\mathrm{avg}}(j)  \le c_j - \frac{c_j+1}{2}= \frac{c_j-1}{2}.$$ 
Hence $t_N\ge t_N^{\mathrm{avg}}$ a.s., and therefore it suffices to show that  the following converges to $0$, as $N\to\infty$.
\begin{align*}
    \E N^{-3/2} (t_N - t_N^{\mathrm{avg}})
    = \frac{m}{N} N^{-3/2}\sum_{j=1}^N (q_N(j)-q_N^{\mathrm{avg}}(j)) = \frac{m}{N} N^{-3/2}\sum_{i=1}^k \frac{c_i(c_i-1)}{2}.
\end{align*}
We split the above sum into two parts, as follows. Fix any $\epsilon>0$ and define $\mathcal{S}_\epsilon = \{1\le j\le k: c_j - 1\ge \epsilon N^{1/2}\}$. Then
$$\frac{m}{N} N^{-3/2}\sum_{i\not\in \mathcal{S}_\epsilon} \frac{c_i(c_i-1)}{2}  
\le \epsilon N^{-1}\sum_{i\not\in \mathcal{S}_\epsilon} c_i \le \epsilon.$$
On the other hand, for $i\in \mathcal{S}_\epsilon$, and $s_{i-1}<j\le s_j$ we have $q_N(j) - q_N^{\mathrm{avg}}(j) = c_i(c_i-1)/2 \ge \epsilon^2 N/2.$
Hence, if $J$ be an index chosen uniformly at random from $\{1,2,\dots,N\}$ (independent of everything else), then
\begin{align*}
    N^{-3/2}\sum_{i\in \mathcal{S}_\epsilon} \frac{c_i(c_i-1)}{2} &=N^{-3/2}\sum_{i\in \mathcal{S}_\epsilon} \sum_{j=s_{i-1}+1}^{s_i} (q_N(j)-q_N^{\mathrm{avg}}(j))\\
    &= N^{-1/2} \Pr\left(q_N(J) - q_N^{\mathrm{avg}}(J) \ge \epsilon^2 N/2\right)\\
    &\le (\epsilon^2/2)^{-2} N^{-1/2} N^{-2} \E_J\left[(q_N(J) - q_N^{\mathrm{avg}}(J))^2\right]\\
    &= (\epsilon^2/2)^{-2} N^{-7/2} \sum_{j=1}^N (q_N(j) - q_N^{\mathrm{avg}}(j))^2.
\end{align*}
Therefore it suffices to show that as $N\to\infty$,
\begin{equation}\label{avglast}
   \sum_{j=1}^{N}\left(q_{N}^{\mathrm{avg}}(j)-{q}_{N}(j)\right)^{2}=O\left(N^{3}\right). 
\end{equation}
Note that for each $1\le j\le k$, $\sum_{i=s_{j-1}+1}^{s_{j}}\left(q_{N}^{\mathrm{avg}}(i)-i\right)=0$. Thus $\overline{q}_{N}^{\mathrm{avg}}=N^{-1} \sum_{i=1}^{N} i=(N+1) / 2$, and we obtain the following.
$$
\begin{aligned}
&\sum_{i=s_{j-1}+1}^{s_{j}}\left(\left(q_{N}^{\mathrm{avg}}(i)-\overline{q}_{N}^{\mathrm{avg}}\right)^{2}-\left(i-\overline{q}_{N}^{\mathrm{avg}}\right)^{2}\right) \\
&=\sum_{i=s_{j-1}+1}^{s_{j}}\left(q_{N}^{\mathrm{avg}}(i)^{2}-2 \overline{q}_{N}^{\mathrm{avg}}\left(q_{N}^{\mathrm{avg}}(i)-i\right)-i^{2}\right) \\
&=-\sum_{i=s_{j-1}+1}^{s_{j}}\left(i^{2}-\left(\frac{s_{j-1}+1+s_{j}}{2}\right)^{2}\right) \\
&=-\sum_{i=s_{j-1}+1}^{s_{j}}\left(i-\frac{s_{j-1}+1+s_{j}}{2}\right)^{2} \\
&=-\frac{1}{12}c_{j}\left(c_{j}^{2}-1\right).
\end{aligned}
$$
Therefore,
$$
    \sum_{i=1}^{N}\left(q_{N}^{\mathrm{avg}}(i)-\overline{q}_{N}^{\mathrm{avg}}\right)^{2}=\sum_{i=1}^{N}\left(i-\overline{q}_{N}^{\mathrm{avg}}\right)^{2}-\frac{1}{12}\sum_{j=1}^{k} c_{j}\left(c_{j}^{2}-1\right).
$$
Notice that
\begin{equation}\label{eq:E2}
    0 \leq \frac{1}{12}\sum_{j=1}^{k} c_{j}\left(c_{j}^{2}-1\right) \leq \frac{1}{12}\left(\sum_{i=1}^{k} t_{i}\right) \max _{1 \leq i \leq k} t_{i}^{2}=\frac{N}{12} \max _{1 \leq i \leq k} t_{i}^{2},
\end{equation}
and therefore~\cref{blocksizes} tells us that
\begin{equation}\label{eq:E3}
\sum_{i=1}^{N}\left(q_{N}^{\mathrm{avg}}(i)-\overline{q}_{N}^{\mathrm{avg}}\right)^{2}=\frac{N(N^2-1)}{12}+o(N^3).
\end{equation}
On the other hand, it follows from the proof of~\cref{lem001} that $$\sum_{i=1}^{N}\left(q_{N}(i)-\frac{N+1}{2}\right)^{2}=\frac{N(N^2-1)}{12}+o(N^3)$$
Combining the above display with \eqref{eq:E2} and \eqref{eq:E3}, and doing some algebraic manipulations, \eqref{avglast} follows. This completes the proof.
\end{proof}

%

\subsection{Proof of~\texorpdfstring{\cref{iid.I_C}}{Lemma A.4}}\label{proof:iid.I_C}

\begin{proof}
Let us denote $S_N = S_N(b_{N,1},\dots,b_{N,N}) = \sum_{j=1}^N \sum_{i=1}^N I_{h,N} (b_{N,j} - b_{N,i})$.
Observe that for any fixed $i\neq j,$ and $h > 0,$ $$\E\left(I_{h,N} (b_{N,j} - b_{N,i})\right) = P\left(0\le b_{N,2} - b_{N,1} < hN^{-1/2}\right) = g(hN^{-1/2}) - g(0)$$
where $$g(x) = P(b_{N,2} - b_{N,1} \le x) = \int_{-\infty}^{x} \int_{\R}  f_{}(u+t) f_{}(u)\,du\,dt, \ x\in\R.$$
Using the DCT for integrals, we argue that $g'(x) = \int_{\R} f_{}(u+x) f_{}(u)du$. 
Hence 
\begin{align*}  \lim_{N\to\infty} N^{-3/2} \E(S_N) &= \lim_{N\to\infty}  N^{-3/2} N^2 (g(hN^{-1/2}) - g(0))\\
&= \lim_{N\to\infty}  h\cdot\frac{g(hN^{-1/2}) - g(0)}{hN^{-1/2}} \\
&= hg'(0) \\
&= h\int_{\R} f_{}(u)^2du.\numberthis\label{Glim}
\end{align*}
Now we bound $\E(S_N - \E S_N)^2$ using the Efron-Stein inequality \citep{ES81}. For each $1\le k \le N,$ let $b_{N,k}'$ be an i.i.d. copy of $b_{N,k},$ independent of everything else, and define $$S_N^{(k)} = S_N(b_{N,1},\dots,b_{N,k-1},b_{N,k}',b_{N,k+1},\dots,b_{N,N}), \ 1\le k\le N.$$ 
Note that 
\begin{align*}
S_N - S_N^{(k)}
&=  \sum_{j=1}^N \left(I_{h,N} (b_{N,j} - b_{N,k}) -  I_{h,N} (b_{N,j} - b_{N,k}')\right)\\
&\quad \quad+\sum_{j=1}^N \left(I_{h,N} (b_{N,k} - b_{N,j}) -  I_{h,N} (b_{N,k}' - b_{N,j})\right).
\end{align*}
An application of the Cauchy-Schwarz inequality yields the following.
\begin{align*}
&\E\left[ \bigg(\sum_{j=1}^N I_{h,N} (b_{N,j} - b_{N,k}) -  I_{h,N} (b_{N,j} - b_{N,k}')\bigg)^2 \ \Big|\ b_{N,1},\dots,b_{N,N}\right] \\
&\leq N\E\left[ \sum_{j=1}^N \left(I_{h,N} (b_{N,j} - b_{N,k}) -  I_{h,N} (b_{N,j} - b_{N,k}')\right)^2 \ \Big|\ b_{N,1},\dots,b_{N,N}\right] \\
&\leq 2N\E\left[ \sum_{j=1}^N \left(I_{h,N} (b_{N,j} - b_{N,k}) + I_{h,N} (b_{N,j} - b_{N,k}')\right) \ \Big|\ b_{N,1},\dots,b_{N,N}\right] \\
&= 2N \sum_{j=1}^N \left(I_{h,N} (b_{N,j} - b_{N,k}) +\E\left[I_{h,N} (b_{N,j} - b_{N,k})\mid b_{N,j}\right]\right).
\end{align*}
Applying the same argument to the second part yields
\begin{align*}
&\E\left[ \bigg(\sum_{j=1}^N I_{h,N} (b_{N,k} - b_{N,j}) -  I_{h,N} (b_{N,k}' - b_{N,j})\bigg)^2 \ \Big|\ b_{N,1},\dots,b_{N,N}\right]\\
&\le 2N \sum_{j=1}^N \left(I_{h,N} (b_{N,k} - b_{N,j}) +\E\left[I_{h,N} (b_{N,k} - b_{N,j})\mid b_{N,j}\right]\right).
\end{align*}
Using the above bounds and a generalized Efron-Stein inequality (see~\cref{GenES}),
\begin{align*}
& \E(N^{-3/2}(S_N - \E S_N))^{6} \\
&\lesssim
 N^{-9} \E \left[\sum_{k=1}^N \E\left[(S_N - S_N^{(k)})^2 \ \big|\ b_{N,1},\dots,b_{N,N}\right] \right]^3\\
&\lesssim N^{-6}\, \E \left[\sum_{k=1}^N \sum_{j=1}^N \left(I_{h,N} (b_{N,j} - b_{N,k}) +\E\left[I_{h,N} (b_{N,j} - b_{N,k})\mid b_{N,j}\right]\right)\right]^3\\
&\quad + N^{-6}\, \E \left[\sum_{k=1}^N \sum_{j=1}^N \left( I_{h,N} (b_{N,k} - b_{N,j}) +\E\left[I_{h,N} (b_{N,k} - b_{N,j})\mid b_{N,j}\right]\right)\right]^3\\
 &\lesssim N^{-6}\, \E \left[\sum_{k=1}^N \sum_{j=1}^N I_{h,N} (b_{N,j} - b_{N,k})\right]^3 + N^{-6}\,\E \left[\sum_{k=1}^N \sum_{j=1}^N \E\left[I_{h,N} (b_{N,j} - b_{N,k})\mid b_{N,j}\right]\right]^3\\
&\quad + N^{-6}\,\E \left[\sum_{k=1}^N \sum_{j=1}^N \E\left[I_{h,N} (b_{N,k} - b_{N,j})\mid b_{N,j}\right]\right]^3.
\end{align*}
We now show that the contribution of each of last three sums is at most of the order of $N^{-3/2},$ which would complete the proof, with the aid of Borel-Cantelli lemma. We start with
\begin{align*}
&N^{-6}\, \E \left[\sum_{k=1}^N \sum_{j=1}^N I_{h,N} (b_{N,j} - b_{N,k})\right]^3 \\
&= N^{-6}\sum_{1\le i_1,j_1,i_2,j_2,i_3,j_3\le N} \E\left[I_{h,N} (b_{N,j_1} - b_{N,i_1})I_{h,N} (b_{N,j_2} - b_{N,i_2})I_{h,N} (b_{N,j_3} - b_{N,i_3})\right]
\end{align*}
When $|\{i_1,j_1,i_2,j_2,i_3,j_3\}|\le 4$, the number of ways to choose such a set of indices would be $O(N^4),$ and the summands being bounded above by $1,$ the contributions from these terms is $O(N^{-2})$. If $|\{i_1,j_1,i_2,j_2,i_3,j_3\}|=6,$ i.e., the indices are all distinct, we use the independence of the $b_{N,i}$'s to split the joint probability as the product of marginal probabilities, and thus the contribution from such terms becomes
\begin{align*}
& N^{-6}\sum_{i_1,j_1,i_2,j_2,i_3,j_3 \text{ distinct}} \E I_{h,N} (b_{N,j_1} - b_{N,i_1}) \E I_{h,N} (b_{N,j_2} - b_{N,i_2}) \E I_{h,N} (b_{N,j_3} - b_{N,i_3}) \\
&\le N^{-6}\left(\sum_{1\le i, j\le N} \E I_{h,N} (b_{N,j} - b_{N,i})\right)^3 = N^{-3/2} \left(\E N^{-3/2} S_N\right)^3 = O(N^{-3/2}).
\end{align*}
Finally, consider the case where $|\{i_1,j_1,i_2,j_2,i_3,j_3\}|=5,$ i.e., exactly one index is repeated. Then at least two of the sets $\{i_1,j_1\},\{i_2,j_2\},\{i_3,j_3\}$ will make four distinct indices altogether. If the set left-out is $\{i_3,j_3\}$, we can just upper bound $I_{h,N}(b_{N,j_3} -b_{N,i_3})$ by $1$. Also, the number of ways of choosing the indices $\{i_3,j_3\}$ will be $O(N),$ since one of them is repeated within $\{i_1,j_1,i_2,j_2\}$. Using this idea, 
\begin{align*}
& N^{-6}\sum_{|i_1,j_1,i_2,j_2,i_3,j_3|=5} \E\left[ I_{h,N} (b_{N,j_1} - b_{N,i_1})  I_{h,N} (b_{N,j_2} - b_{N,i_2})  I_{h,N} (b_{N,j_3} - b_{N,i_3})\right] \\
&\lesssim N^{-5}\sum_{i_1,j_1,i_2,j_2 \text{ distinct}} \E I_{h,N} (b_{N,j_1} - b_{N,i_1})\E I_{h,N} (b_{N,j_2} - b_{N,i_2}) \\
&\le N^{-5}\left(\sum_{1\le i, j\le N} \E I_{h,N} (b_{N,j} - b_{N,i})\right)^2 = N^{-2} \left(\E N^{-3/2} S_N\right)^2 = O(N^{-2}).
\end{align*}
Combining the above cases we conclude that $$N^{-6}\, \E \left[\sum_{k=1}^N \sum_{j=1}^N I_{h,N} (b_{N,j} - b_{N,k})\right]^3 = O(N^{-3/2}).$$
We next look at
\begin{align*}
& N^{-6}\,\E \left[\sum_{k=1}^N \sum_{j=1}^N \E\left[I_{h,N} (b_{N,j} - b_{N,k})\mid b_{N,j}\right]\right]^3\\
&= N^{-6}\sum_{1\le i_1,j_1,i_2,j_2,i_3,j_3\le N} \E\left[I^*(b_{N,j_1}, b_{N,i_1})I^*(b_{N,j_2}, b_{N,i_2})I^*(b_{N,j_3}, b_{N,i_3})\right]
\end{align*}
where $I^*(b_{N,j}, b_{N,i}) := \E \left[I_{h,N} (b_{N,j} - b_{N,i}) \mid b_{N,j}\right]$, which is really a function of $b_{N,j}$. Thus, when all the $6$ indices are distinct, the joint probability is split into marginals, and once again we can see that the contribution of these terms is $O(N^{-3/2})$. When $|\{i_1,j_1,i_2,j_2,i_3,j_3\}|\le 4$, the same argument applies as we gave earlier, and tells us that contribution from these terms  would be $O(N^{-2})$. Finally, when $|\{i_1,j_1,i_2,j_2,i_3,j_3\}|= 5$ we play the same trick applied in the previous case to conclude that contribution from these terms is
\begin{align*}
& N^{-6}\sum_{|i_1,j_1,i_2,j_2,i_3,j_3|=5} \E\left[I^*(b_{N,j_1}, b_{N,i_1})I^*(b_{N,j_2}, b_{N,i_2})I^*(b_{N,j_3}, b_{N,i_3})\right]\\
&\lesssim N^{-5}\sum_{i_1,j_1,i_2,j_2 \text{ distinct}} \E\left[I^*(b_{N,j_1}, b_{N,i_1})I^*(b_{N,j_2}, b_{N,i_2})\right] \\
&= N^{-5}\sum_{i_1,j_1,i_2,j_2 \text{ distinct}} \E\left[I_{h,N} (b_{N,j_1} - b_{N,i_1}) \right]\E\left[I_{h,N} (b_{N,j_2} - b_{N,i_2}) \right] \\
&\le N^{-5}\left(\sum_{1\le i, j\le N} \E I_{h,N} (b_{N,j} - b_{N,i})\right)^2 = N^{-2} \left(\E N^{-3/2} S_N\right)^2 = O(N^{-2}).
\end{align*}
Thus, $$N^{-6}\,\E \left[\sum_{k=1}^N \sum_{j=1}^N \E\left[I_{h,N} (b_{N,j} - b_{N,k})\mid b_{N,j}\right]\right]^3=O(N^{-3/2}).$$ In a similar manner one can show that $$N^{-6}\,\E \left[\sum_{k=1}^N \sum_{j=1}^N \E\left[I_{h,N} (b_{N,k} - b_{N,j})\mid b_{N,j}\right]\right]^3=O(N^{-3/2}).$$ Hence the proof follows.
\end{proof}


\subsection{Proof of~\texorpdfstring{\cref{iid.I_C2}}{Lemma A.5}}\label{proof:iid.I_C2}

\begin{proof}
Since $f(x)$ is continuous and has finite limits as $x\to\pm \infty$, it follows that $f$ is uniformly continuous and bounded. Next, by using the continuous mapping theorem and the strong law of large numbers, we obtain
$$\frac{1}{N}\sum_{i=1}^N f(b_{N,i})\asto \int_\R f^2(x)\,dx.$$
Next, define $$\Delta_N:=\sqrt{N} \sup_{x}\left|\left(F_N-F\right)\left(\left(x, x+\frac{h}{\sqrt{N}}\right]\right)\right|,\quad \omega_f(\eps)=\sup_x\sup_{|t-x|\le \eps}|f(t)-f(x)|.$$
    The modulus of continuity $\omega_f(\eps)$ is finite and goes to $0$ as $\eps\to 0$, thanks to the uniform continuity of $f$.
    Observe now that
    \begin{align*}
        &\sup_{x}\left|\sqrt{N}\left(F_N\left(x+\frac{h}{\sqrt{N}}\right)-F_N(x)\right)-h\,f(x)\right|\\
        &\le \Delta_N+\sup_{x}\left|\sqrt{N}\left(F\left(x+\frac{h}{\sqrt{N}}\right)-F(x)\right)-h\,f(x)\right|\\
        &\le \Delta_N + \sqrt{N}\sup_{x}\left|\int_x^{x+\frac{h}{\sqrt{N}}}|f(t)-f(x)|\,dt\right|\\
        &\le \Delta_N + h\,\omega_f\left(\frac{h}{\sqrt{N}}\right).
    \end{align*}
    The uniform continuity of $f$ implies that $\omega_f({h}/{\sqrt{N}})\to 0$ as $N\to\infty$. It only remains to show now that $\Delta_N\asto 0$. The class $\mathcal{F}_N :=\{\mathbf{1}(x, x+h/\sqrt{N}]:x\in\R\}$ is a VC-class with envelope $F_\text{env}\equiv 1$, and $\sigma^2 = \sup_{f\in \mathcal{F}_N} P_F f^2=\sup_{x} |F(x+h/\sqrt{N})-F(x)|\le \|f\|_\infty h N^{-1/2}$. It follows from \citet[Corollary 5.1]{Chernozhukov2014} that
    \begin{equation}\label{F.8}
        \E \Delta_N= \E\sup_{f\in\mathcal{F}_N}\sqrt{N}|(F_N - F)f|\lesssim N^{-1/4} \sqrt{\log N} + N^{-1/2}\sqrt{\log N}=o(1),
    \end{equation}
    which immediately gives $\Delta_N\Pto 0$. Moreover, we can apply \citet[Theorem 5.1]{Chernozhukov2014} with $F_\text{env}\equiv 1$ and $M=1$ to deduce that for any $q\ge 2$ and $t\ge 1$,
$$
\operatorname{Pr}\left\{\Delta_N>(1+\alpha) \E \Delta_N+K(q)\left[\left(\sigma+N^{-1 / 2}\right) \sqrt{t}+\alpha^{-1} N^{-1 / 2} t\right]\right\} \leq t^{-q / 2} .
$$
We choose $t=N^\kappa$ with any $\kappa \in(0,1 / 2)$ and $\alpha=1$. Since $\sigma=O\left(N^{-1 / 4}\right)$,
$$
\left(\sigma+N^{-1 / 2}\right) \sqrt{t}+\alpha^{-1} N^{-1 / 2} t=O\left(N^{-1 / 4+\kappa / 2}\right)+O\left(N^{-1 / 2+\kappa}\right)=o(1).
$$
We choose $q>2 / \kappa$ so that $\sum_N t^{-q / 2}=\sum_N N^{-\kappa q / 2}<\infty$. By the first Borel-Cantelli lemma,
$$
\Delta_N\leq 2 \mathbb{E}\Delta_N+o(1)=O\left(N^{-1 / 4} \sqrt{\log N}\right) \text { a.s. }
$$    
   This combined with \eqref{F.8} completes the proof.
\end{proof}

\subsection{Proof of~\texorpdfstring{\cref{iid.J_C}}{Lemma A.6}}\label{proof:iid.J_C}

\begin{proof}
    Denote by $G_N$ the empirical distribution of the residuals $\widetilde{b}_{N,1},\dots, \widetilde{b}_{N,N}$, and denote by $G_N^{(\eps)}$ the empirical distribution of the noises $\eps_{N,1},\dots,\eps_{N,N}$. Define
       \begin{equation*}
        \Delta_N^{(1)}:=\sup_{x}\left|\sqrt{N}\left(G_N^{(\eps)}\left(x+\frac{h}{\sqrt{N}}\right)-G_N^{(\eps)}(x)\right)-hg(x)\right|,
    \end{equation*}
    and
       \begin{equation*}
        \Delta_N^{(2)}:=\sqrt{N}\sup_{x}\left|G_N\left(\left(x,x+\frac{h}{\sqrt{N}}\right]\right)-G_N^{(\eps)}\left(\left(x,x+\frac{h}{\sqrt{N}}\right]\right)\right|.
    \end{equation*}
    Note that
    \begin{align*}
      & \left| N^{-3/2}\sum_{j=1}^N\sum_{i=1}^N I_{h,N} (\widetilde{b}_{N,j} - \widetilde{b}_{N,i}) - h\int_\R g^2(t)\,dt\right| \\
       &= \left|\frac{1}{N} \sum_{i=1}^N \sqrt{N}\left(G_N\left(\widetilde{b}_{N,i}+\frac{h}{\sqrt{N}}\right)-G_N(\widetilde{b}_{N,i})\right)-h\int_\R g^2(t)\,dt\right|\\
       &\le \Delta_N^{(1)} \,+\,\Delta_N^{(2)} \,+\, |h|\cdot\left|\frac{1}{N}\sum_{i=1}^N g({\eps}_{N,i}) - \int_\R g^2(t)\,dt\right|.\numberthis\label{eq:Delta123}
    \end{align*}
    Since $g(x)$ is continuous and has finite limits as $x\to\pm \infty$, it follows that $g$ is uniformly continuous and bounded. it follows from the continuous mapping theorem and the strong law of large numbers that $$\frac{1}{N}\sum_{i=1}^N g(\eps_{N,i})\asto \int_\R g^2(t)\,dt.$$
    The same argument as in the proof of \cref{iid.I_C2} implies that $\Delta_N^{(1)}\asto 0$.  
    In view of \eqref{eq:Delta123}, it suffices to show now that  $\Delta_{N}^{(2)}$ converges to $0$.
    Denote by $\vec{P}_N$ the projection matrix that projects onto the column space of $\vec{X}_N$. We can write $$\widetilde{b}_{N,i}=((I-\vec{P}_N)\vec{b}_N)_i=((I-\vec{P}_N)\vec{\eps}_N)_i=\eps_{N,i} - r_{N,i},\quad\text{where}\quad r_{N,i} :=  \vec{p}_{N,i}^\top\,\eps_N=\vec{x}_{N,i}^\top\, \vec{\beta}_N^{(0)},$$
    where $\vec{\beta}_N^{(0)}:=(\vec{X}_N^\top\vec{X}_N)^{-1}\vec{X}_N^\top\vec{\eps}_N$.  Define  $h_{N,i}:=(\vec{P}_{N})_{i,i}$, $V_{N,i} := \sum_{j\neq i} (\vec{P}_{N})_{i,j} \eps_{N,j}$, and $u_{N,i} :=(\vec{X}_N\vec{\beta}_N)_i$. Also denote by $\mathcal{F}_{N,-i}$ the $\sigma$-algebra generated by $\{\eps_{N,j}, j\neq i\}$. We can then write $$\widetilde{b}_{N,i} = (1-h_{N,i})\eps_{N,i} - V_{N,i},\quad b_{N,i}-\widetilde{b}_{N,i}=u_{N,i} + h_{N,i}\,\eps_{N,i}+V_{N,i}.$$
Finally, write
$$G_N\left(\left(x,x+\frac{h}{\sqrt{N}}\right]\right)-G_N^{(\eps)}\left(\left(x,x+\frac{h}{\sqrt{N}}\right]\right)=\frac{1}{N}\sum_{i=1}^N \Xi_{N,i}(x),$$
    where \begin{align*}
        \Xi_{N,i}(x) &= \ind{x<\eps_{N,i}-r_{N,i}\le x+\frac{h}{\sqrt{N}}}-\ind{x<\eps_{N,i}\le x+\frac{h}{\sqrt{N}}}\\
        &= \ind{\frac{x+V_{N,i}}{1-h_{N,i}}<\eps_{N,i} \le \frac{x+V_{N,i}+\frac{h}{\sqrt{N}}}{1-h_{N,i}}}-\ind{x<\eps_{N,i}\le x+\frac{h}{\sqrt{N}}}.
    \end{align*}
Note that\begin{align*}
    \Delta_N^{(2)}&=\sup_{x}\left|\frac{1}{\sqrt{N}}\sum_{i=1}^N\Xi_{N,i}(x)\right|\le\sup_{x}|A_N(x)|+\sup_{x}|B_N(x)|
\end{align*}
where \begin{align*}
        A_N(x)&:=\frac{1}{\sqrt{N}}\sum_{i=1}^N \left(\Xi_{N,i}(x)-\E\left[\Xi_{N,i}(x)\mid \mathcal{F}_{N,-i}\right]\right), \quad
    B_N(x):=\frac{1}{\sqrt{N}}\sum_{i=1}^N \E\left[\Xi_{N,i}(x)\mid \mathcal{F}_{N,-i}\right].
\end{align*}
It is straightforward to bound $B_N(x)$ using the uniform continuity of $g$, as follows.
\begin{align*}
  \sup_x | B_N(x)|&= \sup_x \left|\frac{1}{\sqrt{N}}\sum_{i=1}^N G\left(\frac{x+V_{N,i}+\frac{h}{\sqrt{N}}}{1-h_{N,i}}\right)-G\left(\frac{x+V_{N,i}}{1-h_{N,i}}\right)-G\left(x+\frac{h}{\sqrt{N}}\right)+G(x)\right|\\
   &\le \sup_x  \frac{1}{\sqrt{N}}\sum_{i=1}^N \int_{x}^{x+h/\sqrt{N}}\left|g\left(\frac{t+V_{N,i}}{1-h_{N,i}}\right)-g(t)\right|dt\\
   &\le \frac{h}{N}\sum_{i=1}^N \omega_g(|V_{n,i}|)+\eta_g\left(\frac{h_{N,i}}{1-h_{N,i}}\right),\numberthis\label{eq:BN}
\end{align*}
where $\omega_g$ is the modulus of continuity of $g$ and $\eta_g(t)=\sup_{u}|g((1+t)u)-g(u)|$. It follows from the uniform continuity of $g$ that $\eta_g(t)\to 0$ as $t\to 0$. Recall that $\sum_{i=1}^N h_{N, i}=\Tr\left(\vec{P}_N\right)=p$ and $h_{N, i} \in[0,1)$. Therefore, for any $\delta>0$, \begin{align*}
\frac{1}{N} \sum_{i=1}^N \eta\left(\frac{h_{N, i}}{1-h_{N, i}}\right) &\leq \eta\left(\frac{\delta}{1-\delta}\right)+\sup _{t \in[0,1)} \eta(t) \cdot \frac{1}{N} \sum_{i=1}^N 1\left\{h_{N, i}>\delta\right\}\\
&\leq \eta\left(\frac{\delta}{1-\delta}\right)+\frac{p}{\delta N} \rightarrow \eta\left(\frac{\delta}{1-\delta}\right).
\end{align*}
Since $\delta>0$ is arbitrary, it follows that
\begin{align*}
    \frac{1}{N} \sum_{i=1}^N \eta\left(\frac{h_{N, i}}{1-h_{N, i}}\right)\to 0.
\end{align*}
Next we show that $N^{-1}\sum_{i=1}^N \omega_g(|V_{N,i}|)\to 0$ in probability / almost surely whenever $\vec\beta_N^{(0)}\to 0$ in probability / almost surely (cf.~\cref{beta0goesto0,beta0goesto0a.s.}).
Recall that $V_{N,i}=\vec{p}_{N,i}^\top\,\eps_{N,i}-(\vec{P}_{N})_{i,i}\,\eps_{N,i}=\vec{x}_{N,i}^\top\,\vec{\beta}_N^{(0)}-h_{N,i}\,\eps_{N,i}$. Therefore,
\begin{align*}
&    \frac{1}{N}\sum_{i=1}^N \omega_g(|V_{N,i}|)\le \omega_g(\delta) + 2\|g\|_\infty\frac{1}{N}\sum_{i=1}^N\ind{|V_{N,i}|>\delta}\\
    &\le \omega_g(\delta) + 2\|g\|_\infty\frac{1}{N}\sum_{i=1}^N\ind{|\vec{x}_{N,i}^\top\,\vec{\beta}_N^{(0)}|>\frac{\delta}{2}}+\ind{h_{N,i}\,|\eps_{N,i}|>\frac{\delta}{2}}\\
      &\le \omega_g(\delta) + 2\|g\|_\infty\frac{1}{N}\sum_{i=1}^N\frac{2}{\delta}|\vec{x}_{N,i}^\top\,\vec{\beta}_N^{(0)}|+\ind{h_{N,i}>\frac{\delta}{2K}}+\ind{|\eps_{N,i}|>K}\\
         &\le \omega_g(\delta) + 2\|g\|_\infty\frac{1}{N}\sum_{i=1}^N\frac{2}{\delta}\|\vec{x}_{N,i}\|_2\|\vec{\beta}_N^{(0)}\|_2+\ind{h_{N,i}>\frac{\delta}{2K}}+\ind{|\eps_{N,i}|>K}\\
         &\le \omega_g(\delta) + \frac{4\|g\|_\infty}{\delta}\left(\frac{1}{N}\sum_{i=1}^N\|\vec{x}_{N,i}\|_2^2\right)^{1/2}\|\vec{\beta}_N^{(0)}\|_2+\frac{4K\|g\|_\infty}{\delta}\frac{p}{N}+\frac{2\|g\|_\infty}{N}\sum_{i=1}^N\ind{|\eps_{N,i}|>K}.
\end{align*}
Letting $N\to\infty$ in the above display and using \cref{beta0goesto0} (resp.~\cref{beta0goesto0a.s.}), then $K\to\infty$ and $\delta\to 0$ we arrive at the conclusion that the RHS of \eqref{eq:BN} converges to $0$ in probability (resp.~almost surely, under the stronger assumption of uniformly bounded covariates).

Finally, since $|\Xi_{N,i}(x)|\le 2$ uniformly and has jumps at at most $4N$ points, a standard symmetrization argument applies to $A_{N}(x)$ and tells us that $\E\sup_x|A_N(x)|\le 2\E\sup_x |C_N(x)|$ where $C_N(x)=\frac{1}{N}\sum_{i=1}^N \sigma_i \Xi_{N,i}(x)$, $\sigma_i$ being i.i.d.~Radamacher random variables. This yields, via 
Hoeffding's bound, $\Pr(\sup_x|C_N(x)|>t)\le 8N \exp(-Nt^2/2)$. Using $t=\sqrt{4\log N}/\sqrt{N}$ and using the first Borel-Cantelli lemma, we obtain the desired almost sure convergence for $\sup_{x}|A_N(x)|$. This completes the proof, in light of \eqref{eq:Delta123}.
\end{proof}

\subsection{Proof of~\texorpdfstring{\cref{iid.J_C2}}{Lemma A.7}}\label{proof:iid.J_C2}

\begin{proof}
Denote by $G_N$ the empirical distribution of the residuals $\widetilde{b}_{N,1},\dots,\widetilde{b}_{N,N}$, and denote by $G_N^{(\eps)}$ the empirical distribution of the noises $\eps_{N,1},\dots,\eps_{N,N}$.  With a slight abuse of notation, we also denote by $G$ the CDF of the distribution with density $g$.
As in the proof of \cref{iid.J_C} (cf.~\cref{proof:iid.J_C}), we define $\vec{\beta}_N^{(0)}=(\vec{X}_N^\top\vec{X}_N)^{-1}\vec{X}_N^\top\vec{\eps}_N$. Then, $$\widetilde{b}_{N,i}-\eps_{N,i}=-\vec{x}_{N,i}^\top\, \vec{\beta}_N^{(0)}.$$
Since $g(x)$ is continuous and has finite limits as $x\to\pm \infty$, it follows that $g$ is uniformly continuous and bounded. Recall that $\vec{x}_{N,i}$ denotes the $i$-th row of the matrix $\vec{X}_N$. 
For any bounded $L$-Lipschitz function $\psi$, we have
\begin{align*}
   &\left| \int \psi dG_N - \int \psi dG\right|\\
   &\le \frac{1}{N}\sum_{i=1}^N |\psi(\widetilde{b}_{N,i})-\psi(\eps_{N,i})| + \left| \int \psi dG_N^{(\eps)} - \int \psi dG\right|\\
   &\le \frac{1}{N}\sum_{i=1}^N L\left|\vec{x}_{N,i}^\top\vec{\beta}_N^{(0)}\right| + \left| \int \psi dG_N^{(\eps)} - \int \psi dG\right|\\
   &\le L \left(\frac{1}{N}\sum_{i=1}^N \|\vec{x}_{N,i}\|_2^2\right)^{1/2}\left\|\vec{\beta}_N^{(0)}\right\|_2 + \left| \int \psi dG_N^{(\eps)} - \int \psi dG\right|.\numberthis\label{eq:L}
\end{align*}
Note that $N^{-1}\vec{X}_N^\top\vec{X}_N\to\Sigma\succ 0$ yields $\sup_{N\ge 1} N^{-1}\sum_{i=1}^N \|\vec{x}_{N,i}\|_2^2 = \sup_{N\ge 1} \Tr(N^{-1}\vec{X}_N^\top\vec{X}_N) <\infty$. Hence the second term in \eqref{eq:L} converges to 0 almost surely by the Glivenko-Cantelli lemma. On the other hand, we show in \cref{beta0goesto0a.s.} that $\vec\beta_N^{(0)}\asto \vec{0}$, which implies that the first term in  \eqref{eq:L}  converges to 0 almost surely. This proves that $G_N$ converges weakly to $G$ almost surely, using the Portmanteau theorem. Next, using the fact that $g$ is uniformly continuous and bounded, we deduce that
\begin{align*}
    \left|\frac{1}{N}\sum_{i=1}^N g(\widetilde{b}_{N,i})- \int_\R g^2(t)\,dt\right|
    &\le \left|\frac{1}{N}\sum_{i=1}^N g(\eps_{N,i})- \int_\R g^2(t)\,dt\right| +  \frac{1}{N}\sum_{i=1}^N \omega_g(|\widetilde{b}_{N,i}- \eps_{N,i}|),\numberthis
    \label{omega_g_bound}
\end{align*}
where $\omega_g(\delta)=\sup_{x}\sup_{|t-x|\le\delta}|g(t)-g(x)|$ denotes the modulus of continuity of $g$, which is finite and converges to $0$ as $\delta\to 0$ thanks to the uniform continuity of $g$.
Using the continuous mapping theorem and the strong law of large numbers, we conclude that the first part in the above display converges to $0$ almost surely. 
Since $\omega_g$ is bounded by $2\|g\|_\infty$, and $\sup_{N\ge 1} N^{-1}\sum_{i=1}^N\|\vec{x}_{N,i}\|^2 <\infty$, we can choose $T=T(\eps)$ large enough such that
$$\frac{1}{\sqrt{N}}\sum_{i\,:\,\|\vec{x}_{N,i}\|> T} \omega_g(|r_{N,i}|)=\frac{1}{N}\sum_{i\,:\,\|\vec{x}_{N,i}\|> T} \omega_g(|\vec{x}_{N,i}^\top \beta_{N}|)\le \frac{2\|g\|_\infty}{T^2}\frac{1}{N}\sum_{i=1}^N \|\vec{x}_{N,i}\|^2<\eps.$$
    On the other hand,
    \begin{align*}
        \frac{1}{N}\sum_{i\,:\,\|x_{N,i}\|\le T}\omega_g(|r_{N,i}|)\le \omega_g(T \|\vec\beta_{N}^{(0)}\|).
    \end{align*}
We show in \cref{beta0goesto0a.s.} that $\vec\beta_{N}^{(0)}\asto 0$. This, in conjunction with the last two displays implies that the RHS of \eqref{omega_g_bound} converges to $0$ almost surely. This completes the proof of part (a).

Next, we turn to showing part (b). It follows from our proof of \cref{iid.J_C} (cf.~\cref{proof:iid.J_C}) that for any fixed $h$,

       \begin{equation*}
        \label{def:Delta-A.5}
        \Delta_N^{(1)}:=\sup_{x}\left|\sqrt{N}\left(G_N^{(\eps)}\left(x+\frac{h}{\sqrt{N}}\right)-G_N^{(\eps)}(x)\right)-hg(x)\right|\asto 0,
    \end{equation*}
    and 
       \begin{equation*}
        \label{def:Delta2-A.5}
        \Delta_N^{(2)}:=\sqrt{N}\sup_{x}\left|G_N\left(\left(x,x+\frac{h}{\sqrt{N}}\right]\right)-G_N^{(\eps)}\left(\left(x,x+\frac{h}{\sqrt{N}}\right]\right)\right|\asto 0.
    \end{equation*}
    Consequently,
    \begin{align*}
        \sup_{x}\left|\sqrt{N}\left(G_N\left(x+\frac{h}{\sqrt{N}}\right)-G_N(x)\right)-h\,g(x)\right|\le \Delta_N^{(1)}+\Delta_{N}^{(2)}\asto 0.
    \end{align*}
    This completes the proof of part (b).

    Toward showing the last conclusion, 
    first we justify replacing $G_N$ with $G$ in \eqref{eq:indep}, as follows.
\begin{equation}\label{A3}
\begin{split}    
       &\sup_{x,y}\left| (G_N(x)-G(x))\frac{1}{N}\sum_{i=1}^N \ind{ b_{N,\,i}-\widetilde{b}_{N,\,i}\le y}\right|\\
        &\le \sup_{x}\left| G_N(x)-G(x)\right|\\
        &\le \sup_{x}\left| G_N(x)-G_N^{(\eps)}(x)\right|+\sup_{x}\left| G_N^{(\eps)}(x)-G(x)\right|.
        \end{split}
    \end{equation}
    It follows from the Glivenko-Cantelli lemma that $\sup_{x}\left| G_N^{(\eps)}(x)-G(x)\right|\asto 0$. 
Next, we show that
\begin{equation}\label{GminusGeps}
    \sup_{x}\left|G_N(x)-G_N^{(\eps)}(x)\right|\asto 0.
\end{equation}
    Observe that
    \begin{align*}
     \left|G_N(x) - G_N^{(\eps)}(x)\right| &=\frac{1}{N}\sum_{i=1}^N \left|\ind{\eps_{N,i}\le x+r_{N,i}} - \ind{\eps_{N,i}\le x}\right|\\
        &\le \frac{2}{N}|\{i:|r_{N,i}|>\delta\}| +\frac{1}{N}\sum_{i:|r_{N,i}|\le \delta} \ind{x-\delta<\eps_{N,i}\le x+\delta} \\
        &\le \frac{2}{\delta N}\sum_{i=1}^N |r_{N,i}| +\frac{1}{N}\sum_{i=1}^N \ind{x-\delta<\eps_{N,i}\le x+\delta} \\
        &\le \frac{2}{\delta}\left(\frac{1}{N}\sum_{i=1}^N \|\vec{x}_{N,i}\|_2^2\right)^{1/2}\|\vec{\beta}_N^{(0)}\|_2 +\frac{1}{N}\sum_{i=1}^N \ind{x-\delta<\eps_{N,i}\le x+\delta}.
    \end{align*}
We show in \cref{beta0goesto0a.s.} that $\vec\beta_{N}^{(0)}\asto 0$. Since $N^{-1}\sum_{i=1}^N \|\vec{x}_{N,i}\|_2^2=\Tr(N^{-1}\vec{X}_N^\top\vec{X}_N)\to \Tr(\Sigma)$, the above display implies that
$$\limsup_{N\to\infty}\left|G_N(x) - G_N^{(\eps)}(x)\right| \le \sup_x|G(x+\delta)-G(x-\delta)|\le 2\|g\|_\infty\delta.$$
Since $\delta>0$ is arbitrary, this completes the proof of \eqref{GminusGeps}.
In view of \eqref{eq:indep}, \eqref{A3} and \eqref{GminusGeps}, it suffices to show that
\begin{equation}\label{eq:indep2}
    \sup_{x,y}\left|\frac{1}{N}\sum_{i=1}^N \ind{\widetilde{b}_{N,\,i}\le x,\, b_{N,\,i}-\widetilde{b}_{N,\,i}\le y}- G(x)\frac{1}{N}\sum_{i=1}^N \ind{ b_{N,\,i}-\widetilde{b}_{N,\,i}\le y}\right|\asto 0.
\end{equation}
We now introduce some more notation. Define  $h_{N,i}:=(\vec{P}_{N})_{i,i}$, $V_{N,i} := \sum_{j\neq i} (\vec{P}_{N})_{i,j} \eps_{N,j}$, and $u_{N,i}=(\vec{X}_N\vec{\beta}_N)_i$. Also denote by $\mathcal{F}_{N,-i}$ the $\sigma$-algebra generated by $\{\eps_{N,j}, j\neq i\}$. We can then write $$\widetilde{b}_{N,i}=((\vec{I}-\vec{P}_N)\eps)_i = (1-h_{N,i})\eps_{N,i} - V_{N,i},\quad b_{N,i}-\widetilde{b}_{N,i}=u_{N,i} + h_{N,i}\,\eps_{N,i}+V_{N,i}.$$
Define $$\xi_{N,\,x,\,y}:=\ind{\widetilde{b}_{N,\,i}\le x,\, b_{N,\,i}-\widetilde{b}_{N,\,i}\le y}- G(x)\ind{ b_{N,\,i}-\widetilde{b}_{N,\,i}\le y}.$$
Since $|\xi_{N,\,x,\,y}|\le 2$ and $\xi_{N,\,x,\,y}$ has jumps at finitely many points (more precisely, $O(N^2)$ points), the same concentration inequality argument using symmetrization, Hoeffding's lemma and Massart's finite class lemma as in the proof of \cref{iid.J_C} (cf.~\cref{proof:iid.J_C}) implies that
$$\sup_{x,y}\left|\frac{1}{N}\sum_{i=1}^N\left(\xi_{N,x,y}-\E\left[\xi_{N,x,y}\mid \mathcal{F}_{N,-i}\right]\right)\right|\asto 0.$$
It only remains to show now that
$\sup_{x,y}\left|\frac{1}{N}\sum_{i=1}^N \E\left[\xi_{N,x,y}\mid \mathcal{F}_{N,-i}\right]\right|\asto 0$.
Note that
    \begin{align*}
&    \frac{1}{N}\sum_{i=1}^N \E\left[\xi_{N,x,y}\mid \mathcal{F}_{N,-i}\right]\\
      &= \frac{1}{N}\sum_{i=1}^N \E\left[ \ind{\widetilde{b}_{N,\,i}\le x,\, b_{N,\,i}-\widetilde{b}_{N,\,i}\le y}- G(x)\ind{ b_{N,\,i}-\widetilde{b}_{N,\,i}\le y}\mid \mathcal{F}_{N,-i}\right]\\
       &= \frac{1}{N}\sum_{i=1}^N \E\left[\ind{\eps_{N,\,i}\le \frac{x+V_{N,i}}{1-h_{N,i}},\, \eps_{N,i}\le \frac{y-u_{N,i}-V_{N,i}}{h_{N,i}}}\right.\\
       &\qquad\qquad\qquad\qquad\qquad\qquad- \left.G(x)\ind{ \eps_{N,i}\le \frac{y-u_{N,i}-V_{N,i}}{h_{N,i}}} \mid \mathcal{F}_{N,-i}\right]\\
       &=\frac{1}{N}\sum_{i=1}^N \left[G\left(\frac{x+V_{N,i}}{1-h_{N,i}}\wedge  \frac{y-u_{N,i}-V_{N,i}}{h_{N,i}}\right)-G(x)\,G\left(\frac{y-u_{N,i}-V_{N,i}}{h_{N,i}}\right)\right].\numberthis\label{eq:copula}
    \end{align*}
    Note that
    \begin{align*}
    &\sup_{x,y} \left|\frac{1}{N}\sum_{i=1}^N \left(G\left(\frac{x+V_{N,i}}{1-h_{N,i}}\wedge \frac{y-u_{N,i}-V_{N,i}}{h_{N,i}}\right)-G\left(x\wedge \frac{y-u_{N,i}-V_{N,i}}{h_{N,i}}\right)\right)\right| \\
  &\le \sup_x \frac{1}{N}\sum_{i=1}^N \left|G\left(\frac{x+V_{N,i}}{1-h_{N,i}}\right)-G(x)\right|\\
   &\le \sup_x  \frac{1}{N}\sum_{i=1}^N \int_{-\infty}^{x}\left|g\left(\frac{t+V_{N,i}}{1-h_{N,i}}\right)-g(t)\right|dt\\
   &\le \sup_x  \frac{1}{N}\sum_{i=1}^N \int_{-\infty}^{x}\left|g(t+V_{N,i})-g(t)-g\left(\frac{t+V_{N,i}}{1-h_{N,i}}\right)+g(t+V_{N,i})\right|dt\\
   &\le \frac{h}{N}\sum_{i=1}^N \omega_g(|V_{n,i}|)+\eta_g\left(\frac{h_{N,i}}{1-h_{N,i}}\right),\numberthis\label{eq:BN2}
\end{align*}
where $\omega_g$ is the modulus of continuity of $g$ and $\eta_g(t)=\sup_{u}|g((1+t)u)-g(u)|$. We show in the proof of \cref{iid.J_C} (cf.~\eqref{eq:BN}) that the RHS of \eqref{eq:BN2} goes to $0$ a.s. Therefore it suffices to show now that
$$\sup_{x,y}\left|\frac{1}{N}\sum_{i=1}^N \left[G\left(x\wedge  w_{N,i}(y)\right)-G\left(x\right)\,G\left(w_{N,i}\right)\right]\right|\asto 0,\quad w_{N,i}:=\frac{y-u_{N,i}-V_{N,i}}{h_{N,i}}.$$
Towards this, first note that for any fixed $y$,
\begin{align*}
    &\sup_{x}\left|\frac{1}{N}\sum_{i=1}^N \left[G\left(x\wedge  w_{N,i}(y)\right)-G\left(x\right)\,G\left(w_{N,i}\right)\right]\right|\\
     &=\sup_{x}\left|\frac{1}{N}\sum_{i=1}^N G\left(x\wedge  w_{N,i}(y)\right)\left[1-G\left(x \vee w_{N,i}\right)\right]\right|\\
    &\le  \frac{1}{N}\sum_{i=1}^N G(w_{N,i}(y))(1-G(w_{N,i}(y))).\numberthis\label{eq:wNi}
\end{align*}
Note that $0\le G(t)(1-G(t))\le {1}/{4}$. We will show that $\sup_y  N^{-1}|{i:|w_{N,i}(y)|\le M}|  \asto 0$.  Write
\begin{equation*}
    \begin{split}
        \frac{1}{N}|{i:|w_{N,i}(y)|\le M}| \le \frac{|i:h_{N,i}>\sqrt{\delta_N}|}{N} &+ \frac{|i:|V_{N,i}|>\delta|}{N}\\&+
        \frac{1}{N}\sum_{i\,:\,|V_{N,i}|\le \delta,\, h_{N,i}\le \sqrt{\delta_N}} \ind{|y-u_{N,i}-V_{N,i}|\le Mh_{N,i}},
    \end{split}
\end{equation*}
where $\delta_N=\max_{i\le N}h_{N,i}$ denotes the max-leverage. Note that $\delta_N\ge \frac{1}{N}\sum_{i=1}^N h_{N,i}=\frac{p}{N}$, which implies that $$\frac{|i: h_{N,i}>\sqrt{\delta_N}|}{N} \le \delta_N^{-1/2}\frac{p}{N}\le \sqrt{\frac{p}{N}}\to 0.$$
On the other hand, we show in the proof of \cref{iid.J_C} that $\frac{1}{N}|i:|V_{N,i}|>\delta|\asto 0$. Finally, note that
\begin{align*}
     &\sup_y \frac{1}{N}\sum_{i\,:\,|V_{N,i}|\le \delta,\, h_{N,i}\le \sqrt{\delta_N}} \ind{|y-u_{N,i}-V_{N,i}|\le Mh_{N,i}}\\
      &\le \sup_y \frac{1}{N}\sum_{i=1}^N \ind{|y-u_{N,i}|\le M\sqrt{\delta_N}+\delta}\\
      &\le \sup_y \left|H\left(y+M\sqrt{\delta}_N+\delta\right)-H\left(y-M\sqrt{\delta}_N-\delta\right)\right|+\sup_t |H_N(t)-H(t)|,
\end{align*}
where $H_N$ is the empirical distribution of the (deterministic) $u_{N,i}$'s and $H$ is the weak limit of $H_N$. Since $H$ is continuous, we conclude that $\sup_t|H_N(t)-H(t)|\to 0$.
Since $\delta_N\le \frac{1}{N}\sum_{i=1}^N h_{N,i}=\frac{p}{N}\to 0$, we conclude from the last display that $$\limsup_{N\to\infty}\sup_y \frac{1}{N}\sum_{i\,:\,|V_{N,i}|\le \delta,\, h_{N,i}\le \sqrt{\delta_N}} \ind{|y-u_{N,i}-V_{N,i}|\le Mh_{N,i}}\le \sup_y (H(y+\delta)-H(y-\delta)).$$
Since $\delta>0$ is arbitrary and $H$ is continuous, we are through. This, in turn, completes the proof, because we can deduce from \eqref{eq:wNi} that $$\limsup_{N\to\infty} \frac{1}{N}\sum_{i=1}^N G(w_{N,i}(y))(1-G(w_{N,i}(y)))\le \sup_{|t|>M}G(t)(1-G(t)\to 0,$$ as $M\to\infty$.
\end{proof}

\subsection{Proof of~\texorpdfstring{\cref{Jb.geq.Ib-old}}{Proposition A.8}}\label{proof:Jb.geq.Ib-old}

\begin{proof}
Under the assumptions of this theorem,~\cref{Jbexists} tells us that~\cref{AssumpB2} holds in probability, with $\mathcal{J}_b = (2\sqrt{\pi}\sigma)^{-1}$.
Next, fix any $h\in\R$ and define $$\mathcal{I}_{h,N} := N^{-3/2} \sum_{j=1}^N \sum_{i=1}^N I_{h,N}(b_j -b_i),$$
where $I_{h,N}$ is defined in \eqref{Ian} of the main paper. Without loss of generality, we take $h>0$ (the other case will be similar). For simplicity in notation, we shall omit the index $N$ in this proof for $\vec{b}_N,\vec{X}_N$ etc. From $\vec{b}=\vec{X}\vec{\beta}+\vec{\eps}$ we write $b_j - b_i = \eps_j - \eps_i + v_j - v_i$ where $v_i$ denotes the $i$-th coordinate of  $\vec{v}=\vec{X}\vec{\beta}$. It then follows that 
\begin{align*}
\E \mathcal{I}_{h,N} 
&= N^{-3/2} \sum_{j=1}^N \sum_{i=1}^N P(0\le \eps_j - \eps_i + v_j - v_i < hN^{-1/2})\\
&= N^{-3/2} \sum_{j=1}^N \sum_{i=1}^N \left[\Phi\left(\frac{v_i-v_j + hN^{-1/2}}{\sqrt{2}\sigma}\right) - \Phi\left(\frac{v_i-v_j}{\sqrt{2}\sigma}\right)\right]\\
&= N^{-3/2} \sum_{j=1}^N \sum_{i=1}^N \left[\frac{h}{\sqrt{2}\sigma\sqrt{N}}\phi\left(\frac{v_i-v_j}{\sqrt{2}\sigma}\right) + \frac{h^2}{4N\sigma^2} \phi'(\xi_{i,j})\right]\\
&= \frac{h}{2\sqrt{\pi}\sigma}N^{-2} \sum_{j=1}^N \sum_{i=1}^N e^{-(v_j-v_i)^2/4\sigma^2} + (4\sigma^2)^{-1}h^2 N^{-5/2} \sum_{j=1}^N \sum_{i=1}^N\phi'(\xi_{i,j}),
\end{align*}
where for each $(i,j)$, $\xi_{i,j}$ is a number between $v_i-v_j$ and $v_i-v_j+N^{-1/2}$
Since $\phi'$ is bounded, we can show (by proceeding in the same manner as in the proof of~\cref{Jbexists}) that the second sum in the above display is asymptotically negligible. Hence 
\begin{equation}\label{JbIb-0}
\lim_{N\to\infty} \left(\E \mathcal{I}_{h,N} -\frac{h}{2\sqrt{\pi}\sigma} N^{-2} \sum_{j=1}^N \sum_{i=1}^N e^{-(v_j-v_i)^2/4\sigma^2}\right)=0.\end{equation}
Since $0\le N^{-2} \sum_{j=1}^N \sum_{i=1}^N e^{-(v_j-v_i)^2/4\sigma^2}\le 1$, it follows that for any $h>0$,
\begin{equation*}
    \limsup_{N\to\infty} \E \mathcal{I}_{h,N}/h \le (2\sqrt{\pi}\sigma)^{-1}=\mathcal{J}_b.
\end{equation*}
Next we show that $\var(\mathcal{I}_{h,N})\to 0$. Towards that, we first write
$$\E \mathcal{I}_{h,N}^2 = N^{-3}\sum_{i,j,k,l} P\left(0\le b_j-b_i<hN^{-1/2}, 0\le b_l-b_k<hN^{-1/2}\right).$$
{Now, as mentioned earlier,~\cref{AssumpB2} holds in this setting (cf.~\cref{Jbexists}); hence appealing to arguments similar to those given in~\cref{proof:est.I_C}, we argue that the contribution from the terms with repeated indices are negligible.} So we are left with terms with distinct indices $i,j,k,l$. For such indices, note that we have $P(0\le b_j-b_i<hN^{-1/2}, 0\le b_l-b_k<hN^{-1/2})=P(0\le b_j-b_i<hN^{-1/2})P( 0\le b_l-b_k<hN^{-1/2})$, since the errors $\eps_i$'s are independent and $v_i$'s non-stochastic. We can thus write 
$$\E \mathcal{I}_{h,N}^2 =  \left(N^{-3/2}\sum_{i,j \text{ distinct }} P\left(0\le b_j-b_i<hN^{-1/2}\right)\right)^2 + o(1) = (\E\mathcal{I}_{h,N})^2 + o(1),$$
from which the conclusion follows.
Moreover, \eqref{JbIb-0} tells us that~\cref{ACjs} holds, with $\lim_{N\to\infty}\E \mathcal{I}_{h,N}/h=\ell (2\sqrt{\pi} \sigma)^{-1}=\ell \mathcal{J}_b.$
Finally, for any $x\in\R,$ 
\begin{equation} e^{-x^2} \le \frac{1}{1+x^2} \le \max\left\{1-\frac{x^2}{2}, \frac{1}{2}\right\}.
\end{equation}
Hence 
\begin{align*}
    N^{-2} \sum_{j=1}^N \sum_{i=1}^N e^{-(v_j-v_i)^2/4\sigma^2} 
    &\le \max\left\{1-\frac{1}{2N^2}\sum_{j=1}^N \sum_{i=1}^N (v_j-v_i)^2, \frac{1}{2}\right\} \\
    &= \max\left\{1-\frac{1}{N}\sum_{j=1}^N (v_j - \overline{v})^2, \frac{1}{2}\right\}.
\end{align*}
Therefore, if $ \liminf_{N\to\infty} \frac{1}{N}\sum_{j=1}^N (v_j - \overline{v})^2>0$, then
\begin{align*}
\limsup_{N\to\infty}\E\mathcal{I}_{h,N}/h 
& \leq \frac{1}{2\sqrt{\pi}\sigma} \limsup_{N\to\infty} \max\left\{1-\frac{1}{N}\sum_{j=1}^N (v_j - \overline{v})^2, \frac{1}{2}\right\}  \\
& \leq \frac{1}{2\sqrt{\pi}\sigma} \max\left\{1-\liminf_{N\to\infty} \frac{1}{N}\sum_{j=1}^N (v_j - \overline{v})^2, \frac{1}{2}\right\}  \\
&< \frac{1}{2\sqrt{\pi}\sigma} = \mathcal{J}_b.
\end{align*}
Hence the proof is complete.
\end{proof}

\subsection{Proof of~\texorpdfstring{\cref{est.I_C}}{Proposition A.9}}\label{proof:est.I_C}

\begin{proof}
First note that under any value of $\tau,$ $$B_N \dd  \sum_{j=1}^N \sum_{i=1, i\neq j}^N (1-Z_{N,i})(1-Z_{N,j}) I_{1,N} (b_{N,j} - b_{N,i}).$$ 
Invoking~\cref{ACjs}, we get
$$ \E N^{-3/2} B_N = \frac{(N-m)(N-m-1)}{N(N-1)} N^{-3/2}\sum_{j=1}^N  \sum_{i=1, i\neq j}^N I_{1,N} (b_{N,j} - b_{N,i})
\to (1-\lambda)^2\mathcal{I}_b.$$
Consequently, $\E ((1-m/N)^{-2} N^{-3/2}  B_N) \Pto \mathcal{I}_b$. Thus it suffices to show that
$\var(B_N) =o(N^3)$.
For brevity, we denote by $I_N(i,j)$ the indicator $I_{1,N}(b_{N,j} - b_{N,i})$ in the rest of the proof.
Note that $$\text{Var}\left(B_N\right) =
\sum_{i=1}^N  \sum_{j=1, j\neq i}^N \sum_{k=1}^N  \sum_{l=1, l\neq k}^N  \text{Cov}\left((1-Z_{N,i}) (1-Z_{N,j}),(1-Z_{N,k}) (1-Z_{N,l})\right) I_N (i,j) I_N (k, l).$$ 
Since $2\le |\{i, j, k, l\}| \le 4$, we consider the following cases. 
\begin{enumerate}
\item[(a)] $|\{i, j, k, l\}|=2,$ i.e., $(i, j) = (k, l)$ or $(l,k)$. Since $\text{Var}((1-Z_{N,i}) (1-Z_{N,j})) = p_N(1-p_N)$ where $p_N = P(Z_{N,i} = 0, Z_{N,j}=0) = \frac{(N-m)(N-m-1)}{N(N-1)} \sim (1-\lambda)^2,$ the contribution of these terms in $\text{Var}(B_N)$ is given by $$\sum_{i=1}^N  \sum_{j=1, j\neq i}^N 2p_N(1-p_N) I^2_N (i,j) \lesssim 2 (1-\lambda)^2 \sum_{i=1}^N  \sum_{j=1, j\neq i}^N I_N (i,j)  \lesssim N^{3/2}.$$

\item[(b)] $|\{i, j, k, l\}|=4,$ i.e., all $4$ indices are distinct. Note that in this case 
\begin{align*} &\cov\left((1-Z_{N,i}) (1-Z_{N,j}),(1-Z_{N,k}) (1-Z_{N,l})\right)\\ 
&= P(Z_{N,i}=0, Z_{N,j}=0, Z_{N,k} = 0, Z_{N,l} =0) - p_N^2\\
& =  p_N \left(\frac{(N-m-2)(N-m-3)}{(N-2)(N-3)}-  \frac{(N-m)(N-m-1)}{N(N-1)} \right)\\
& =  - p_N \frac{2 m (2 N^2 - 2 m N - 6 N  + 3 + 3 m)}{N(N-1)(N-2)(N-3)}\sim -4\lambda(1-\lambda)^3 N^{-1}.
\end{align*}
Hence the contribution $u_N$ of these terms in $\text{Var}(B_N)$ satisfies the following. \begin{align*}u_N &\lesssim N^{-1} \sum_{i=1}^N  \sum_{j=1, j\neq i}^N \sum_{k=1}^N  \sum_{l=1, l\neq k}^N I_N (i,j)I_N (k,l) \lesssim N^{-1+3/2+3/2} =N^{2}.\end{align*}

\item[(c)] $|\{i, j, k, l\}|=3$. Here we have 4 sub-cases: $i = k, j\neq l$; $i \neq k, j = l$; $i = l, j \neq k$; $i \neq l, j = k$. In each of the subcases, 
\begin{align*} 
&\cov\left((1-Z_{N,i}) (1-Z_{N,j}),(1-Z_{N,k}) (1-Z_{N,l})\right)\\
&= P(Z_{N,1}=Z_{N,2}=Z_{N,3} = 0) - p_N^2 \\
&\sim (1-\lambda)^3 - (1-\lambda)^4 = \lambda(1-\lambda)^3.
\end{align*}
Hence, if $v_n$ be the contribution of these terms in $\text{Var}(B_N)$, then 
\begin{align*}
v_N  & \lesssim  \sum_{i, j, l \text{ distinct}}  I_N (i,j)I_N (i,l)  + \sum_{i, j, k \text{ distinct}}  I_N (i,j) I_N (k,j)  \\ 
&\quad + \sum_{i, j, k \text{ distinct}}  I_N (i,j) I_N (k,i)  + \sum_{i, j, l \text{ distinct}} I_N (i,j)I_N (j,l) \\
&\le 4N\sum_{i, j\neq i}  I_N (i,j)  \lesssim N^{1+3/2} =N^{5/2}.
\end{align*}
\end{enumerate}
Combining the above cases, we conclude that $N^{-3} \text{Var}(B_N) \lesssim N^{-3} (N^{3/2} + N^2 + N^{5/2}) = o(1),$
which completes the proof.
\end{proof}

\subsection{Proof of~\texorpdfstring{\cref{plug-in.Ib}}{Theorem A.10}}\label{proof:plug-in.Ib}

\begin{proof}
Recall the notation from \cref{plug-in est}. We present the proof for $h=1$; the same proof will work for any fixed $h$.
 In this proof, we replace $\rsnu$ by $\widehat{\tau}_N$, where $\widehat{\tau}_N$ is any $\sqrt{N}$-consistent estimator of $\tau$.
For brevity, we shall skip the index $N$ for $b_{N,i}$'s and $Z_{N,i}$'s in this proof.
Define $$V_N := N^{-(2-\nu)} \sum_{j=1}^N \sum_{i=1}^N \ind{0\le b_j - b_i < N^{-\nu}},$$
and note that~\cref{AssumpIb.new} yields that $V_N\to \mathcal{I}_b$ as $N\to\infty$. So, it suffices to show that $\widehat{V}_N - V_N \Pto 0$ as $N\to\infty$.
We write
\begin{align*}
\widehat{V}_N - V_N &= N^{-(2-\nu)} \sum_{j=1}^N \sum_{i=1}^N \left( \ind{0\le \widehat{b}_j - \widehat{b}_i} - \ind{0\le b_j - b_i}\right)\\
&\quad\quad - N^{-(2-\nu)} \sum_{j=1}^N \sum_{i=1}^N \left( \ind{N^{-\nu}\le \widehat{b}_j - \widehat{b}_i} - \ind{N^{-\nu}\le b_j - b_i}\right).\numberthis\label{Vhat-V}
\end{align*}
Since $\widehat{b}_j - \widehat{b}_i = b_j - b_i -(\widehat{\tau}_N -\tau) (Z_j - Z_i)$, we can write
\begin{align*} 
\Delta_{1,N} &:= N^{-(2-\nu)} \sum_{j=1}^N \sum_{i=1}^N\left( \ind{0\le \widehat{b}_j - \widehat{b}_i} - \ind{0\le b_j - b_i}\right)\\
&= N^{-(2-\nu)} \sum_{j=1}^N \sum_{i=1}^N\left( \ind{(\widehat{\tau}_N -\tau) (Z_j - Z_i)\le b_j- b_i} - \ind{0\le b_j - b_i}\right)\\
&= N^{-(2-\nu)} \sum_{j=1}^N \sum_{i=1}^N\left( \ind{(\widehat{\tau}_N -\tau) (Z_j - Z_i)\le b_j- b_i < 0}\right.\\
&\qquad\qquad\qquad\qquad\qquad\qquad\left. - \ind{0\le b_j - b_i < (\widehat{\tau}_N -\tau) (Z_j - Z_i)}\right).
\end{align*}
Now pick $\nu<\nu'<1/2$ and define an event $$E_N := \{ \left|\widehat{\tau}_N  - \tau\right| > N^{-\nu'}\}.$$ Since $\widehat{\tau}_N  - \tau=O_P(N^{-1/2})$ and $\nu'<1/2$, we get $\widehat{\tau}_N  - \tau=o_p(N^{-\nu'}),$ implying that $P(E_N) \to 0$ as $N\to\infty$.  Therefore for any $\epsilon>0$, $$P(\left|\Delta_{1,N} \inds{E_N} \right|> \epsilon) \le P(\inds{E_N}=1) = P(E_N) \to 0.$$
On the other hand, we have $\left| (\widehat{\tau}_N -\tau) (Z_j - Z_i)\right| \le N^{-\nu'}$ on $E_N^c$, hence
\begin{align*}
&N^{-(2-\nu)} \sum_{j=1}^N \sum_{i=1}^N \ind{(\widehat{\tau}_N -\tau) (Z_j - Z_i)\le b_j- b_i < 0}\inds{E_N^c} \\
&\le N^{-(2-\nu)} \sum_{j=1}^N \sum_{i=1}^N \ind{-N^{-\nu'} \le b_j- b_i < 0}\inds{E_N^c}  \\
&\le N^{-(\nu'-\nu)} \left(N^{-(2-\nu')} \sum_{j=1}^N \sum_{i=1}^N \ind{-N^{-\nu'} \le b_j- b_i < 0}\right).
\end{align*}
\cref{AssumpIb.new} tells us that the term in the above parentheses converges to $\mathcal{I}_b$. Since $\nu'>\nu$, we conclude that $$N^{-(2-\nu)} \sum_{j=1}^N \sum_{i=1}^N \ind{(\widehat{\tau}_N -\tau) (Z_j - Z_i)\le b_j- b_i < 0}\inds{E_N^c}\Pto 0.$$
Similarly, $$N^{-(2-\nu)} \sum_{j=1}^N \sum_{i=1}^N \ind{0\le b_j - b_i < (\widehat{\tau}_N -\tau) (Z_j - Z_i)}\inds{E_N^c}\Pto 0.$$
Thus, $\Delta_{1,N} \inds{E_N^c}\Pto 0$ and consequently $\Delta_{1,N} \Pto 0,\text{ as }N\to\infty$.

We next focus on the second summand in the RHS of \eqref{Vhat-V}.
The key idea to deal with this sum is to replace the indicators with smooth functions, such as a Gaussian CDF. Define 
$\sigma_N :=N^{-\nu}(\log N)^{-1}$, $r_N(i,j) := b_j - b_i - N^{-\nu}$, and $\widehat{r}_N(i,j) := \widehat{b}_j - \widehat{b}_i - N^{-\nu}$.
Then \begin{align*}\widehat{V}_N - V_N - \Delta_{1,N} &= N^{-(2-\nu)} \sum_{j=1}^N \sum_{i=1}^N\left( \ind{N^{-\nu}\le \widehat{b}_j - \widehat{b}_i} - \ind{N^{-\nu}\le b_j - b_i}\right)\\ &= \Delta_{3,N} - \Delta_{2,N} + \Delta_{4,N} + \Delta_{5,N},\end{align*}
where
$$\Delta_{2,N} := 
N^{-(2-\nu)} \sum_{j=1}^N \sum_{i=1}^N\left( \ind{r_N(i,j)>0 }- \Phi\left(\frac{r_N(i,j)}{\sigma_N}\right)\right),$$
$$\Delta_{3,N} := 
N^{-(2-\nu)} \sum_{j=1}^N \sum_{i=1}^N\left( \ind{\widehat{r}_N(i,j)>0 }- \Phi\left(\frac{\widehat{r}_N(i,j)}{\sigma_N}\right)\right),$$
$$ \Delta_{4,N} := N^{-(2-\nu)} \sum_{j=1}^N \sum_{i=1}^N\left(\Phi\left(\frac{\widehat{r}_N(i,j)}{\sigma_N}\right)- \Phi\left(\frac{r_N(i,j)}{\sigma_N}\right)\right),$$
and $$ \Delta_{5,N} := N^{-(2-\nu)} \sum_{j=1}^N \sum_{i=1}^N\left(\ind{\widehat{r}_N(i,j)=0}- \ind{r_N(i,j)=0}\right).$$
Showing $\Delta_{2,N} \to0$, $\Delta_{k,N}\Pto 0$ for $k=3,4,5,$ will complete the proof. 
\begin{enumerate}
\item[(a)] First we deal with  $\Delta_{5,N}$. Note that
$$N^{-(2-\nu)} \sum_{j=1}^N \sum_{i=1}^N \ind{r_N(i,j)=0} \le N^{-(2-\nu)} \sum_{j=1}^N \sum_{i=1}^N \ind{N^{-\nu} \le b_j - b_i < (1+\delta)N^{-\nu}} \to \delta \mathcal{I}_b.$$
Now let $\delta\to 0$ to get $$\lim_{N\to\infty} N^{-(2-\nu)} \sum_{j=1}^N \sum_{i=1}^N \ind{r_N(i,j)=0}=0.$$ For $\widehat{r}_N(i,j)$ we  proceed just as in the proof of $\Delta_{1,N}\to 0$. Define $E_N = \{ \left|\widehat{\tau}_N  - \tau\right| > N^{-\nu'}\}$ for some $\nu<\nu'<1/2$. Then
\begin{align*}
&N^{-(2-\nu)} \sum_{j=1}^N \sum_{i=1}^N \ind{\widehat{r}_N(i,j)=0}\inds{E_N^c} \\
&\le N^{-(2-\nu)} \sum_{j=1}^N \sum_{i=1}^N \ind{N^{-\nu} \le \widehat{b}_j - \widehat{b}_i < (1+\delta)N^{-\nu}}\inds{E_N^c} \\
&= N^{-(2-\nu)} \sum_{j=1}^N \sum_{i=1}^N \ind{N^{-\nu} \le b_j -b_i  - (\widehat{\tau}_N -\tau) (Z_j - Z_i)< (1+\delta)N^{-\nu}}\inds{E_N^c} \\
&\le N^{-(2-\nu)} \sum_{j=1}^N \sum_{i=1}^N \ind{-\left|\widehat{\tau}_N -\tau\right| +N^{-\nu} \le b_j -b_i < (1+\delta)N^{-\nu}+\left|\widehat{\tau}_N -\tau\right| }\inds{E_N^c} \\
&\le N^{-(2-\nu)} \sum_{j=1}^N \sum_{i=1}^N \ind{-N^{-\nu'} +N^{-\nu} \le b_j -b_i < (1+\delta)N^{-\nu}+N^{-\nu'}}.
\end{align*}
Since $\nu'>\nu,$ it holds for all sufficiently large $N$ that $N^{-\nu'} \le \delta N^{-\nu}$, and consequently
\begin{align*}
&\limsup_{N\to\infty} N^{-(2-\nu)} \E\sum_{j=1}^N \sum_{i=1}^N \ind{\widehat{r}_N(i,j)=0}\inds{E_N^c}\\
&\le \limsup_{N\to\infty} N^{-(2-\nu)} \sum_{j=1}^N \sum_{i=1}^N \ind{(1-\delta)N^{-\nu} \le b_j -b_i < (1+2\delta)N^{-\nu}} = 3\delta \mathcal{I}_b.
\end{align*}
Letting $\delta\to 0$ here, and using $P(E_N) \to 0,$ we conclude that $$N^{-(2-\nu)} \sum_{j=1}^N \sum_{i=1}^N \ind{\widehat{r}_N(i,j)=0}\Pto 0, \text{ as }N\to\infty.$$

\item[(b)] Next we deal with $\Delta_{2,N},$ for which the bound provided by~\cref{Phi-bound} is crucial.
Recall that $$\Delta_{2,N} = N^{-(2-\nu)} \sum_{j=1}^N \sum_{i=1}^N\left( \ind{r_N(i,j)>0 }- \Phi\left(\frac{r_N(i,j)}{\sigma_N}\right)\right),$$
where $\sigma_N =  N^{-\nu}(\log N)^{-1}$ and 
$r_N(i,j) = b_j - b_i - N^{-\nu}$.
The key idea is to use the bound in~\cref{Phi-bound} only for those $i,j$ for which $r_N(i,j)$ is at least as large as $\delta N^{-\nu}$. Towards that, fix $\delta>0$ and define 
$$\mathcal{S}_{N,\delta}  = \{(i, j) : \left| r_N(i,j) \right| \le \frac{\delta}{N^\nu}, 1\le i,j\le N\}.$$
Note that
\begin{align*} 
(i,j)\in\mathcal{S}_{N,\delta}  &\implies \delta N^{-\nu} \ge \left| r_N (i,j) \right| \ge \left|\left|b_j - b_i\right| -N^{-\nu}\right| \\
&\implies \left|b_j -b_i\right| \in\left[(1-\delta)N^{-\nu},(1+\delta)N^{-\nu}\right].
\end{align*}
Therefore~\cref{Ib.even.more} implies that 
\begin{equation}\label{S_N,delta.new}
\lim_{\delta\downarrow 0}\limsup_{N\to\infty} N^{-(2-\nu)} \left|\mathcal{S}_{N,\delta} \right| = 0.
\end{equation}
On the other hand, $(i,j)\not\in\mathcal{S}_{N,\delta} \implies \left| r_N (i,j) \right| >  \delta N^{-\nu}$, and then~\cref{Phi-bound} yields that
\begin{align*}&\left| \ind{r_N(i,j)>0 }- \Phi\left(\frac{r_N(i,j)}{\sigma_N}\right)\right| \\
&\le \sigma_N \left|r_N(i,j)\right|^{-1} \exp(-r_N^2(i,j)/2\sigma_N^2)\\
&\le \sigma_N \delta^{-1} N^\nu \exp(-\delta^2N^{-2\nu}/2\sigma_N^2)\\
&= (\log N)^{-1} \delta^{-1} \exp\left(-\delta^2(\log N)^2/2\right)
\end{align*}
Combining the above bounds, the following chain of inequalities emerges.
\begin{align*}
\left|\Delta_{2,N}\right| &\le  N^{-(2-\nu)} \sum_{j=1}^N  \sum_{i=1}^N \left| \ind{r_N(i,j)>0 }- \Phi\left(\frac{r_N(i,j)}{\sigma_N}\right)\right|  \\
&\le N^{-(2-\nu)}  \sum_{(i,j)\in \mathcal{S}_{N,\delta} } 2 + N^{-(2-\nu)}  \sum_{(i,j)\not\in \mathcal{S}_{N,\delta} } \left| \ind{r_N(i,j)>0 }- \Phi\left(\frac{r_N(i,j)}{\sigma_N}\right)\right| \\
&\le 2N^{-(2-\nu)}  \left| \mathcal{S}_{N,\delta} \right|  + N^{-(2-\nu)}  N^2   (\log N)^{-1} \delta^{-1} \exp\left(-\delta^2(\log N)^2/2\right)\\[2mm]
&=2N^{-(2-\nu)}  \left| \mathcal{S}_{N,\delta} \right|  +  (\log N)^{-1} \delta^{-1} \exp\left(\nu \log N-\delta^2(\log N)^2/2\right).
\end{align*}
Therefore, $$\limsup_{N\to\infty} \left|\Delta_{2,N}\right| \le 2\limsup_{N\to\infty} N^{-(2-\nu)}  \left| \mathcal{S}_{N,\delta} \right|.$$
Letting $\delta\downarrow 0$ and invoking \eqref{S_N,delta.new} we get the desired conclusion.

\item[(c)] Recall that 
$$\Delta_{3,N} = 
N^{-(2-\nu)} \sum_{j=1}^N \sum_{i=1}^N\left( \ind{\widehat{r}_N(i,j)>0 }- \Phi\left(\frac{\widehat{r}_N(i,j)}{\sigma_N}\right)\right).$$
We can proceed just as in the previous proof. First we argue that for any fixed $\delta>0,$ $$N^{-(2-\nu)} \sum_{(i,j): \left|\widehat{r}_N(i,j)\right| > \delta N^{-\nu}} \left|\ind{\widehat{r}_N(i,j)>0 }- \Phi\left(\frac{\widehat{r}_N(i,j)}{\sigma_N}\right)\right| \Pto 0,\text{ as }N\to\infty.$$
Proof of the above part is exactly same as what we did in part (b). For the other part, it suffices to show that for any $\epsilon>0,$
\begin{equation}\label{plug-in.Ib.last}
\lim_{\delta\downarrow 0 }\limsup_{N\to\infty} P\left(N^{-(2-\nu)} \sum_{j=1}^N\sum_{i=1}^N \ind{\left|\widehat{b}_j -  \widehat{b}_i\right| \in \left[(1-\delta)N^{-\nu},(1+\delta)N^{-\nu}\right]} > \epsilon\right) = 0\end{equation}
because the above implies that  $$\lim_{\delta\downarrow 0 }\limsup_{N\to\infty} P\left(N^{-(2-\nu)} \left|\{(i,j): \left|\widehat{r}_N(i,j)\right| \le \delta N^{-\nu}\}\right| > \epsilon\right) = 0,$$
which further implies that $$\lim_{\delta\downarrow 0 }\limsup_{N\to\infty} P\left(\left|\Delta_{3,N}\right| > \epsilon\right) = 0.$$
To prove \eqref{plug-in.Ib.last}, we  again use the event $\widetilde{E}_{N,K} = \{ \left|\widehat{\tau}_N  - \tau\right| > N^{-\nu'}\}$ where $\nu<\nu'<1/2$.
\begin{align*}
& N^{-(2-\nu)} \sum_{j=1}^N\sum_{i=1}^N \ind{\left|\widehat{b}_j -  \widehat{b}_i\right| \in \left[(1-\delta)N^{-\nu},(1+\delta)N^{-\nu}\right]}\inds{E_N^c}\\
&\le N^{-(2-\nu)} \sum_{j=1}^N \sum_{i=1}^N \ind{(1-\delta)N^{-\nu} \le \widehat{b}_j - \widehat{b}_i < (1+\delta)N^{-\nu}}\inds{E_N^c} \\
&= N^{-(2-\nu)} \sum_{j=1}^N \sum_{i=1}^N \ind{(1-\delta)N^{-\nu} \le b_j -b_i  - (\widehat{\tau}_N -\tau) (Z_j - Z_i)< (1+\delta)N^{-\nu}}\inds{E_N^c} \\
&\le N^{-(2-\nu)} \sum_{j=1}^N \sum_{i=1}^N \ind{-\left|\widehat{\tau}_N -\tau\right| +(1-\delta)N^{-\nu} \le b_j -b_i < (1+\delta)N^{-\nu}+\left|\widehat{\tau}_N -\tau\right| }\inds{E_N^c} \\
&\le N^{-(2-\nu)} \sum_{j=1}^N \sum_{i=1}^N \ind{-N^{-\nu'} +(1-\delta)N^{-\nu} \le b_j -b_i < (1+\delta)N^{-\nu}+N^{-\nu'}}.
\end{align*}
Since $\nu'>\nu,$ it holds for all sufficiently large $N$ that $N^{-\nu'} \le \delta N^{-\nu}$, and consequently
\begin{align*}
& \limsup_{N\to\infty} N^{-(2-\nu)} \E\sum_{j=1}^N\sum_{i=1}^N \ind{\left|\widehat{b}_j -  \widehat{b}_i\right| \in \left[(1-\delta)N^{-\nu},(1+\delta)N^{-\nu}\right]}\inds{E_N^c}\\
&\le \limsup_{N\to\infty}  N^{-(2-\nu)} \sum_{j=1}^N \sum_{i=1}^N \ind{(1-2\delta)N^{-\nu} \le b_j -b_i < (1+2\delta)N^{-\nu}}  =4\delta \mathcal{I}_b.
\end{align*}
Letting $\delta\to 0$ here, we conclude that
\begin{equation}\label{Delta3Nlast}
\lim_{\delta\downarrow 0 }\limsup_{N\to\infty} N^{-(2-\nu)} \E\sum_{j=1}^N\sum_{i=1}^N \ind{\left|\widehat{b}_j -  \widehat{b}_i\right| \in \left[(1-\delta)N^{-\nu},(1+\delta)N^{-\nu}\right]}\inds{E_N^c}= 0.
\end{equation}
Finally, for any $\epsilon>0,$
\begin{align*}
&P\left(N^{-(2-\nu)} \sum_{j=1}^N\sum_{i=1}^N \ind{\left|\widehat{b}_j -  \widehat{b}_i\right| \in \left[(1-\delta)N^{-\nu},(1+\delta)N^{-\nu}\right]} > \epsilon\right)\\
&\le P\left(N^{-(2-\nu)} \sum_{j=1}^N\sum_{i=1}^N \ind{\left|\widehat{b}_j -  \widehat{b}_i\right| \in \left[(1-\delta)N^{-\nu},(1+\delta)N^{-\nu}\right]}\inds{E_N^c} > \epsilon/2\right)+ P(E_N)\\
&\le (\epsilon/2)^{-1} N^{-(2-\nu)} \E\sum_{j=1}^N\sum_{i=1}^N \ind{\left|\widehat{b}_j -  \widehat{b}_i\right| \in \left[(1-\delta)N^{-\nu},(1+\delta)N^{-\nu}\right]}\inds{E_N^c} + P(E_N).
\end{align*}
Invoking \eqref{Delta3Nlast} and $\lim_{N\to\infty} P(E_N)=0$, we complete the proof of \eqref{plug-in.Ib.last}.
\item[(d)] Next let us focus on $\Delta_{4,N}$. We first write $$\left|\Delta_{4,N}\right| \le N^{-(2-\nu)} \sum_{j=1}^N \sum_{i=1}^N\left|\Phi\left(\frac{\widehat{r}_N(i,j)}{\sigma_N}\right)- \Phi\left(\frac{r_N(i,j)}{\sigma_N}\right)\right|.$$
Now we break the last summation into three parts: (i) $(i,j)$ such that $|r_N(i,j)| \le \delta N^{-\nu},$ i.e., $(i,j)\in \mathcal{S}_{N,\delta};$ (ii) $(i,j)$ such that $|\widehat{r}_N(i,j)| \le \delta N^{-\nu};$ and (iii) remaining $(i,j)$'s, for which $|r_N(i,j)| \wedge |\widehat{r}_N(i,j)| > \delta N^{-\nu}$. For indices $(i,j)$ in the cases (i) and (ii), we can use the crude bound $\left|\Phi(x)- \Phi(y)\right|\le 2$, since  equations \eqref{S_N,delta.new} and \eqref{plug-in.Ib.last} tell us that these contributions will be asymptotically negligible. For the case (iii), we use the following bound: for any $x,y\in\R,$ $$\left|\Phi(x) - \Phi(y)\right| \le |x-y| \sup_{z\in [x\wedge y, x\vee y]} \phi(z) \le |x-y| \exp\left(-\frac{x^2\wedge y^2}{2}\right).$$ Thus, if $\mathcal{T}_{N,\delta}$ denote the (random) set of indices in case (iii), we have
\begin{align*}
& N^{-(2-\nu)} \sum_{(i,j)\in \mathcal{T}_{N,\delta}} \left|\Phi\left(\frac{\widehat{r}_N(i,j)}{\sigma_N}\right)- \Phi\left(\frac{r_N(i,j)}{\sigma_N}\right)\right|\\
&\le \sigma_N^{-1} N^{-(2-\nu)} \sum_{(i,j)\in \mathcal{T}_{N,\delta}}\left|\widehat{r}_N(i,j) - r_N(i,j)\right| \exp\left(-\frac{r_N(i,j)^2\wedge \widehat{r}_N(i,j)^2}{2\sigma_N^{2} }\right) \\
&\le \sigma_N^{-1} N^{-(2-\nu)} \sum_{(i,j)\in \mathcal{T}_{N,\delta}}\left|\widehat{r}_N(i,j) - r_N(i,j)\right| \exp\left(-\frac{\delta^2 N^{-2\nu}}{2\sigma_N^{2} }\right) \\
&\le  \sigma_N^{-1} N^{-(2-\nu)}\exp\left(-\frac{\delta^2 N^{-2\nu}}{2\sigma_N^{2} }\right) \sum_{j=1}^N \sum_{i=1}^N\left|(\widehat{b}_j - \widehat{b}_i)-(b_j - b_i)\right|\\
&= \sigma_N^{-1} N^{-(2-\nu)} \exp\left(-\frac{\delta^2}{2}(\log N)^2\right)\sum_{j=1}^N \sum_{i=1}^N\left|(\widehat{\tau}_N -\tau) (Z_j - Z_i)\right|\\
&\le \sigma_N^{-1} N^{\nu} \left|\widehat{\tau}_N -\tau\right|\exp\left(-\frac{\delta^2}{2}(\log N)^2\right)\\
&=  N^{2\nu-1/2} \log N \cdot \sqrt{N} \left|\widehat{\tau}_N -\tau\right|\exp\left(-\frac{\delta^2}{2}(\log N)^2\right).
\end{align*}
Since $\widehat{\tau}_N -\tau=O_P(N^{-1/2})$, we get the desired conclusion that $\Delta_{4,N}\Pto 0$ as $N\to\infty$.
\end{enumerate}
Combining the above steps, the proof is now complete.
\end{proof}

\subsection{Proof of~\texorpdfstring{\cref{plug-in.Jb2}}{Theorem A.11}}\label{proof:plug-in.Jb2}

\begin{proof}
Recall the notation from \cref{plug-in est}. We present the proof for $h=1$; the same proof will work for any fixed $h$.
In this proof, we replace $\rsna$ by $\widehat{\tau}_N$, where $\widehat{\tau}_N$ is any $\sqrt{N}$-consistent estimator of $\tau$. For brevity, we shall skip the index $N$ for $b_{N,i}$, $Z_{N,i}$, $\vec{X}_N$, and $\vec{p}_{N,i}$ in this proof.
To start with, we define $$W_N := N^{-(2-\nu)} \sum_{j=1}^N \sum_{i=1}^N \ind{0\le \widetilde{b}_j -  \widetilde{b}_i < N^{-\nu}},$$
and note that $W_N\to \mathcal{J}_b,$ by~\cref{AssumpJb.new}. We are to show that $\widehat{W}_N - W_N \Pto 0$.  
Let $I$ and $J$ be two random indices chosen with replacement from $\{1,2,\dots,N\}$.  We replace the event $E_N = \{\left|\widehat{\tau}_N - \tau \right| > N^{-\nu'}\}$ in~\cref{proof:plug-in.Ib} with  the following event: 
$$\widetilde{E}_{N,K} = \{\left|\widehat{\tau}_N - \tau \right| > N^{-\nu'}\} \cup \{\sqrt{N}\|\vec{p}_J -  \vec{p}_I\| > K\},$$
where $\nu<\nu'<1/2$, and $K>0$ is a fixed positive real number. Observe that 
\begin{align*}
P(\sqrt{N}\|\vec{p}_J -  \vec{p}_I\| > K) 
&\le K^{-2} N\cdot \E_{J,I} \|\vec{p}_J -  \vec{p}_I\|^2 \\
&= K^{-2} N^{-1} \sum_{j=1}^N \sum_{i=1}^N \|\vec{p}_j -  \vec{p}_i\|^2 \\
&= 2(\rank(\vec{X})-1) K^{-2},
\end{align*}
where the last step follows from the fact that $\vec{p}_j$'s are columns of an idempotent matrix (see~\eqref{norm} for a proof). Thus,
$$P(\widetilde{E}_{N,K}) \leq P(\sqrt{N}\|\vec{p}_J -  \vec{p}_I\| > K) \lesssim 2(p-1)K^{-2},$$
where $p$ is the number of covariates (i.e., number of columns of $\vec{X}$).
Also note that,
$$\widehat{\widetilde{b}}_j  =  Y_j - \widehat{\tau}_N Z_j -  \vec{p}_j^\top(\vec{Y} - \widehat{\tau}_N \vec{Z})=\widetilde{b}_j - (\widehat{\tau}_N - \tau) \left( Z_j - \vec{p}_j^\top \vec{Z}\right).$$
Thus $\widehat{\widetilde{b}}_j - \widehat{\widetilde{b}}_i = \widetilde{b}_j - \widetilde{b}_i -  (\widehat{\tau}_N - \tau)u_N(i,j)$, where  
$u_N(i,j):= Z_j -  Z_i- (\vec{p}_j - \vec{p}_i)^\top \vec{Z}.$
Now write \begin{align*}
\widehat{W}_N - W_N &= N^{-(2-\nu)} \sum_{j=1}^N \sum_{i=1}^N\Big( \ind{0\le \widehat{\widetilde{b}}_j - \widehat{\widetilde{b}}_i} - \ind{0\le \widetilde{b}_j - \widetilde{b}_i}\Big)\\
&\quad\quad - N^{-(2-\nu)} \sum_{j=1}^N \sum_{i=1}^N\Big( \ind{N^{-\nu}\le \widehat{\widetilde{b}}_j - \widehat{\widetilde{b}}_i} - \ind{N^{-\nu}\le \widetilde{b}_j - \widetilde{b}_i}\Big) .\end{align*}
For the first part, we imitate the proof of $\Delta_{1,N} \Pto 0, $  as follows. \begin{align*} 
\widetilde{\Delta}_{1,N} &= N^{-(2-\nu)} \sum_{j=1}^N \sum_{i=1}^N\Big( \ind{0\le \widehat{\widetilde{b}}_j - \widehat{\widetilde{b}}_i} - \ind{0\le \widetilde{b}_j - \widetilde{b}_i}\Big)\\
&= N^{-(2-\nu)} \sum_{j=1}^N \sum_{i=1}^N\Big( \ind{(\widehat{\tau}_N -\tau) u_N(i,j)\le \widetilde{b}_j - \widetilde{b}_i} - \ind{0\le \widetilde{b}_j - \widetilde{b}_i}\Big)\\
&= N^{-(2-\nu)} \sum_{j=1}^N \sum_{i=1}^N\left( \ind{(\widehat{\tau}_N -\tau) u_N(i,j)\le \widetilde{b}_j - \widetilde{b}_i < 0}\right.\\
&\qquad\qquad\qquad\qquad\qquad\qquad\left. - \ind{0\le \widetilde{b}_j - \widetilde{b}_i < (\widehat{\tau}_N -\tau) u_N(i,j)}\right).
\end{align*}
Now fix any $\epsilon>0$ and note that on  $\widetilde{E}_{N,K}^c$ we have $\left|\widehat{\tau}_N -\tau\right|\le N^{-\nu'}$,  and 
\begin{align*}
\left|u_N(i,j)\right| &\le  \left|Z_j -  Z_i\right| + \left|(\vec{p}_j - \vec{p}_i)^\top \vec{Z}\right|\\
&\le 1 + \|\vec{p}_j - \vec{p}_i\| \cdot \|\vec{Z}\| \\
&\le 1 + \sqrt{N} \|\vec{p}_j - \vec{p}_i\|\\
&\le 1+K.
\end{align*}
Hence
\begin{align*}
&N^{-(2-\nu)} \sum_{j=1}^N \sum_{i=1}^N \ind{(\widehat{\tau}_N -\tau) u_N(i,j)\le \widetilde{b}_j - \widetilde{b}_i < 0}\inds{\widetilde{E}_{N,K}^c} \\
&\le N^{\nu} \cdot N^{-2} \sum_{j=1}^N \sum_{i=1}^N \ind{-N^{-\nu'}(1 + K) \le \widetilde{b}_j - \widetilde{b}_i   < 0}\inds{\widetilde{E}_{N,K}^c} \\
&\le N^{-(\nu'-\nu)} \left(N^{-(2-\nu')} \sum_{j=1}^N \sum_{i=1}^N \ind{-N^{-\nu'}(1+K) \le \widetilde{b}_j - \widetilde{b}_i  < 0}\right).
\end{align*}
Now~\cref{AssumpJb.new} tells us that the term in the above parentheses converges to $(1+K)\mathcal{J}_b$. Since $\nu'>\nu,$ we conclude that $$N^{-(2-\nu)} \sum_{j=1}^N \sum_{i=1}^N \ind{(\widehat{\tau}_N -\tau) u_N(i,j)\le \widetilde{b}_j - \widetilde{b}_i< 0}\inds{\widetilde{E}_{N,K}^c}\Pto 0.$$
Similarly, $$N^{-(2-\nu)} \sum_{j=1}^N \sum_{i=1}^N \ind{0\le \widetilde{b}_j - \widetilde{b}_i< (\widehat{\tau}_N -\tau) u_N(i,j)}\inds{\widetilde{E}_{N,K}^c}\Pto 0.$$
Thus $\widetilde{\Delta}_{1,N} \inds{\widetilde{E}_{N,K}^c}\Pto 0$. Hence
\begin{align*}
\limsup_{N\to\infty} P\left(\Big|\widetilde{\Delta}_{1,N} \Big| > \epsilon\right) &\le \limsup_{N\to\infty} P\left(\Big|\widetilde{\Delta}_{1,N} \inds{\widetilde{E}_{N,K}} \Big|> \epsilon/2\right)\\
&\le \limsup_{N\to\infty} P\left(\inds{\widetilde{E}_{N,K}}=1\right) \\
&= \limsup_{N\to\infty} P(\widetilde{E}_{N,K}) \le 2(p-1)K^{-2}.\end{align*}
Letting $K\to\infty$ here, we conclude that
$\widetilde{\Delta}_{1,N} \Pto 0,\text{ as }N\to\infty$.
Let us next focus on the second part of $\widehat{W}_N - W_N$, namely, $$ N^{-(2-\nu)} \sum_{j=1}^N \sum_{i=1}^N\left( \ind{N^{-\nu}\le \widehat{\widetilde{b}}_j - \widehat{\widetilde{b}}_i} - \ind{N^{-\nu}\le \widetilde{b}_j - \widetilde{b}_i}\right) .$$
Once again we approximate the indicators in the above display using the CDF of $\mathcal{N}(0,\sigma_N^2)$ where  $\sigma_N :=N^{-\nu}(\log N)^{-1}$. 
Define $$\widetilde{r}_N(i,j) := \widetilde{b}_j - \widetilde{b}_i - N^{-\nu},\quad\text{and}\quad \widehat{\widetilde{r}}_N(i,j) := \widehat{\widetilde{b}}_j - \widehat{\widetilde{b}}_i - N^{-\nu}.$$
The rest of the proof is essentially same as the proof of~\cref{plug-in.Ib}. We write $$N^{-(2-\nu)} \sum_{j=1}^N \sum_{i=1}^N\left( \ind{N^{-\nu}\le \widehat{\widetilde{b}}_j - \widehat{\widetilde{b}}_i} - \ind{N^{-\nu}\le \widetilde{b}_j - \widetilde{b}_i}\right)  = \widetilde{\Delta}_{3,N} - \widetilde{\Delta}_{2,N} + \widetilde{\Delta}_{4,N} + \widetilde{\Delta}_{5,N},$$
where
$$\widetilde{\Delta}_{2,N} = 
N^{-(2-\nu)} \sum_{j=1}^N \sum_{i=1}^N\left( \ind{\widetilde{r}_N(i,j)>0 }- \Phi\left(\frac{\widetilde{r}_N(i,j)}{\sigma_N}\right)\right),$$
$$\widetilde{\Delta}_{3,N} = 
N^{-(2-\nu)} \sum_{j=1}^N \sum_{i=1}^N\left( \ind{\widehat{\widetilde{r}}_N(i,j)>0 }- \Phi\left(\frac{\widehat{\widetilde{r}}_N(i,j)}{\sigma_N}\right)\right),$$
$$ \widetilde{\Delta}_{4,N} = N^{-(2-\nu)} \sum_{j=1}^N \sum_{i=1}^N\left(\Phi\left(\frac{\widehat{\widetilde{r}}_N(i,j)}{\sigma_N}\right)- \Phi\left(\frac{\widetilde{r}_N(i,j)}{\sigma_N}\right)\right),$$
and $$ \widetilde{\Delta}_{5,N} = N^{-(2-\nu)} \sum_{j=1}^N \sum_{i=1}^N\left(\ind{\widehat{\widetilde{r}}_N(i,j)=0}- \ind{\widetilde{r}_N(i,j)=0}\right).$$
The proof of $\widetilde{\Delta}_{2,N}\to 0$ is exactly same as the proof of $\Delta_{2,N}\to 0$ given in~\cref{proof:plug-in.Ib}, hence omitted. 
Next, in order to prove $\widetilde{\Delta}_{3,N}\Pto 0,$ we first argue that for any fixed $\delta>0,$  $$N^{-(2-\nu)} \sum_{(i,j): \left|\widehat{\widetilde{r}}_N(i,j)\right| > \delta N^{-\nu}} \left|\ind{\widehat{\widetilde{r}}_N(i,j)>0 }- \Phi\left(\frac{\widehat{\widetilde{r}}_N(i,j)}{\sigma_N}\right)\right| \Pto 0,\text{ as }N\to\infty.$$
Proof of the above display is omitted, for being completely analogous to the proof of part (b) in the proof of~\cref{plug-in.Ib}. For the other part of $\widetilde{\Delta}_{3,N}$, it suffices to show that for any $\epsilon>0,$
\begin{equation}\label{plug-in.Jb.last}
\lim_{\delta\downarrow 0 }\limsup_{N\to\infty} P\left(N^{-(2-\nu)} \sum_{j=1}^N\sum_{i=1}^N \ind{\big|\widehat{\widetilde{b}}_j -  \widehat{\widetilde{b}}_i\big| \in \left[(1-\delta)N^{-\nu},(1+\delta)N^{-\nu}\right]} > \epsilon\right) = 0\end{equation}
because the above implies that  $$\lim_{\delta\downarrow 0 }\limsup_{N\to\infty} P\left(N^{-(2-\nu)} \left|\{(i,j): \left|\widehat{\widetilde{r}}_N(i,j)\right| \le \delta N^{-\nu}\}\right| > \epsilon\right) = 0,$$
which further implies that $$\lim_{\delta\downarrow 0 }\limsup_{N\to\infty} P\left(\left|\widetilde{\Delta}_{3,N}\right| > \epsilon\right) = 0.$$
In order to prove \eqref{plug-in.Jb.last}, we once again use the event $\widetilde{E}_{N,K}$ and the trick with $I$ and $J$. For brevity, we use the notation $$\kappa_{j,i}:=(1+\sqrt{N}\|\vec{p}_j - \vec{p}_i\|) \left|\widehat{\tau}_N - \tau\right|.$$ Observe that,
\begin{align*}
& N^{-(2-\nu)} \sum_{j=1}^N\sum_{i=1}^N \ind{\big|\widehat{\widetilde{b}}_j -  \widehat{\widetilde{b}}_i\big| \in \left[(1-\delta)N^{-\nu},(1+\delta)N^{-\nu}\right]}\inds{\widetilde{E}_{N,K}^c}\\
&\le N^{-(2-\nu)} \sum_{j=1}^N \sum_{i=1}^N \ind{(1-\delta)N^{-\nu} \le \widehat{\widetilde{b}}_j - \widehat{\widetilde{b}}_i< (1+\delta)N^{-\nu}}\inds{\widetilde{E}_{N,K}^c} \\
&= N^{-(2-\nu)} \sum_{j=1}^N \sum_{i=1}^N \ind{(1-\delta)N^{-\nu} \le \widetilde{b}_j-\widetilde{b}_i - (\widehat{\tau}_N -\tau) u_N(i,j)< (1+\delta)N^{-\nu}}\inds{\widetilde{E}_{N,K}^c} \\
&\le N^{-(2-\nu)} \sum_{j=1}^N \sum_{i=1}^N \ind{-\kappa_{j,i} +(1-\delta)N^{-\nu} \le \widetilde{b}_j-\widetilde{b}_i< (1+\delta)N^{-\nu}+\kappa_{j,i} }\inds{\widetilde{E}_{N,K}^c} \\
&= N^{\nu}\cdot \E_{J,I} \left[\ind{-\kappa_{J,I} +(1-\delta)N^{-\nu} \le \widetilde{b}_J -\widetilde{b}_I < (1+\delta)N^{-\nu}+\kappa_{J,I} }\inds{\widetilde{E}_{N,K}^c}\right]\\
\begin{split} &\le N^{\nu}\cdot \E_{J,I} \left[\mathbf{1}\Big\{-(1+K)\left|\widehat{\tau}_N -\tau\right| +(1-\delta)N^{-\nu} \le \widetilde{b}_J -\widetilde{b}_I  \right.  \\
&\left. \qquad\qquad\qquad\qquad\qquad\qquad\qquad\qquad\qquad\qquad\qquad < (1+\delta)N^{-\nu}+(1+K)\left|\widehat{\tau}_N -\tau\right| \Big\}\inds{\widetilde{E}_{N,K}^c}\right]
\end{split}\\
&\le N^{-(2-\nu)} \sum_{j=1}^N \sum_{i=1}^N \mathbf{1}\Big\{-(1+K)\left|\widehat{\tau}_N -\tau\right| +(1-\delta)N^{-\nu} \le \widetilde{b}_j-\widetilde{b}_i  \\
&\qquad\qquad\qquad\qquad\qquad\qquad\qquad\qquad\qquad\qquad\qquad\qquad < (1+\delta)N^{-\nu}+(1+K)\left|\widehat{\tau}_N -\tau\right| \Big\} \inds{\widetilde{E}_{N,K}^c}\\
&\le N^{-(2-\nu)} \sum_{j=1}^N \sum_{i=1}^N \ind{-(1+K)N^{-\nu'} +(1-\delta)N^{-\nu} \le \widetilde{b}_j-\widetilde{b}_i< (1+\delta)N^{-\nu}+(1+K)N^{-\nu'}}.
\end{align*}
Since $\nu'>\nu,$ it holds for all sufficiently large $N$ that $(1+K)N^{-\nu'} \le \delta N^{-\nu}$, and consequently
\begin{align*}
& \limsup_{N\to\infty} N^{-(2-\nu)} \E\left(\sum_{j=1}^N\sum_{i=1}^N \ind{\big|\widehat{\widetilde{b}}_j -  \widehat{\widetilde{b}}_i\big| \in \left[(1-\delta)N^{-\nu},(1+\delta)N^{-\nu}\right]}\inds{\widetilde{E}_{N,K}^c}\right)\\
&\le \limsup_{N\to\infty} N^{-(2-\nu)} \sum_{j=1}^N \sum_{i=1}^N \ind{(1-2\delta)N^{-\nu} \le \widetilde{b}_j-\widetilde{b}_i< (1+2\delta)N^{-\nu}} = 4\delta \mathcal{J}_b.
\end{align*}
Letting $\delta\to 0$ here, we conclude that
\begin{equation}\label{tildeDelta3Nlast}
\lim_{\delta\downarrow 0 }\limsup_{N\to\infty} N^{-(2-\nu)} \E\left(\sum_{j=1}^N\sum_{i=1}^N \ind{\big|\widehat{\widetilde{b}}_j -  \widehat{\widetilde{b}}_i\big| \in \left[(1-\delta)N^{-\nu},(1+\delta)N^{-\nu}\right]}\inds{\widetilde{E}_{N,K}^c}\right)= 0.
\end{equation}
Finally, for any $\epsilon>0,$
\begin{align*}
&P\left(N^{-(2-\nu)} \sum_{j=1}^N\sum_{i=1}^N \ind{\big|\widehat{\widetilde{b}}_j -  \widehat{\widetilde{b}}_i\big| \in \left[(1-\delta)N^{-\nu},(1+\delta)N^{-\nu}\right]} > \epsilon\right)\\
&\le P\left(N^{-(2-\nu)} \sum_{j=1}^N\sum_{i=1}^N \ind{\big|\widehat{\widetilde{b}}_j -  \widehat{\widetilde{b}}_i\big| \in \left[(1-\delta)N^{-\nu},(1+\delta)N^{-\nu}\right]}\inds{\widetilde{E}_{N,K}^c} > \epsilon/2\right)+ P(\widetilde{E}_{N,K})\\
&\le (\epsilon/2)^{-1} N^{-(2-\nu)} \E\left(\sum_{j=1}^N\sum_{i=1}^N \ind{\big|\widehat{\widetilde{b}}_j -  \widehat{\widetilde{b}}_i\big| \in \left[(1-\delta)N^{-\nu},(1+\delta)N^{-\nu}\right]}\inds{\widetilde{E}_{N,K}^c}\right) + P(\widetilde{E}_{N,K}).
\end{align*}
Invoking \eqref{tildeDelta3Nlast} and using  $\lim_{K\to\infty}\limsup_{N\to\infty} P(\widetilde{E}_{N,K})=0$, we complete the proof of \eqref{plug-in.Jb.last}, which, in turn, proves that $\widetilde{\Delta}_{3,N}\Pto 0$.

For $k=4$ or $5,$ the proof of $\widetilde{\Delta}_{k,N}\Pto 0$ follows by making similar changes in the proof of  $\Delta_{k,N}\Pto 0$ in~\cref{proof:plug-in.Ib}, just as we did for $k=1$ or $3$.
\end{proof}

%

\subsection{Proof of~\texorpdfstring{\cref{propo-decomp}}{Proposition C.1}}\label{proof:propo-decomp}

\begin{proof}
The statistic $t_N$ is defined as $$t_N = \widehat{\vec{q}}_N^\top \vec{Z}_N,\quad \widehat{q}_N(j) = \sum_{i=1}^N \ind{Y_{N,i} - \tau_0 Z_{N,i} \le Y_{N,j} - \tau_0 Z_{N,j}},\ 1\le j\le N.$$ Note, the distribution of $t_N$ under $\tau=\tau_N$ depends on $\tau_N$ only through the vector $\widehat{\vec{q}}_N$. Under $\tau=\tau_N$ we can write $$(Y_{N,1} - \tau_0 Z_{N,1}, \dots, Y_{N,N} - \tau_0 Z_{N,N}) \dd (b_{N,1} + (\tau_N - \tau_0)Z_{N,1}, \dots, b_{N,N} + (\tau_N - \tau_0)Z_{N,N})$$ and hence the distribution of $\widehat{\vec{q}}_N^\top \vec{Z}_N$ under $\tau=\tau_N$ is same as the randomization distribution of $\vec{q}_N^\top \vec{Z}_N$ where $$q_{N,j}=\sum_{i=1}^N \ind{b_{N,i} + (\tau_N -\tau_0) Z_{N,i} \le b_{N,j}+ (\tau_N -\tau_0) Z_{N,j}},\ 1\le j\le N.$$ 
In order words, we have  $t_N \dd t^*_N$ under $\tau=\tau_N,$ where 
\begin{align*}
t^*_N  &= \sum_{j=1}^N Z_{N,j} \sum_{i=1}^N \ind{b_{N,i} + (\tau_N -\tau_0) Z_{N,i} \le b_{N,j}+ (\tau_N -\tau_0)Z_{N,j}}\\
&= \sum_{j=1}^N Z_{N,j} \sum_{i=1}^N \left(Z_{N,i}\ind{b_{N,i}  \le b_{N,j}} + (1-Z_{N,i})\ind{hN^{-1/2} \le b_{N,j} - b_{N,i}} \right) \\
&=\sum_{j=1}^N Z_{N,j} + \sum_{j=1}^N \sum_{i=1, i\neq j}^N Z_{N,j} Z_{N,i}\ind{b_{N,i}  \le b_{N,j}} \\
&\qquad\qquad\qquad+  \sum_{j=1}^N \sum_{i=1, i\neq j}^N Z_{N,j} (1-Z_{N,i})\ind{hN^{-1/2} \le b_{N,j} - b_{N,i}}. 
\end{align*}
Now when $h\ge 0,$  $$\ind{hN^{-1/2} \le b_{N,j} - b_{N,i}} = \ind{0 \le b_{N,j} - b_{N,i}} - \ind{0\le b_{N,j} - b_{N,i}<hN^{-1/2}}$$ and thus
\[t^*_N = m +  \sum_{j=1}^N \sum_{i=1, i\neq j}^N Z_{N,j}\ind{b_{N,i}  \le b_{N,j}} -  \sum_{j=1}^N \sum_{i=1, i\neq j}^N(1-Z_{N,i})Z_{N,j}\ind{0 \le b_{N,j} - b_{N,i} < hN^{-1/2}} .\]
Similarly, for $h < 0,$  $$\ind{hN^{-1/2} \le b_{N,j} - b_{N,i}} = \ind{0 \le b_{N,j} - b_{N,i}} + \ind{hN^{-1/2}\le b_{N,j} - b_{N,i}<0}$$ which gives us the desired expression. 
\end{proof}

\section{Some technical results}\label{sec:sometechlem}


%
%


\begin{lemma} \label{HLtypeCLT} 
Let $t(\vec{Z}, \vec{Y}-\tau_0\vec{Z})$ be a test statistic for testing $H_0:\tau=\tau_0$ and $\rsnu$ be the estimator of $\tau$ based on $t(\cdot,\cdot)$, as defined in~\eqref{eq:Tau-Unadj} of the main paper. Assume that for  all values of $\vec{y}$ and $\vec{z}$, $t(\vec{z}, \vec{y} -\tau \vec{z})$ is a non-increasing function of $\tau$.
Let $\{c_N\}_{N\ge 1}$ and $\{d_N\}_{N\ge 1}$ be sequences of positive real numbers and $\mu_N$ be the mean of $t_N = t(\vec{Z}_N, \vec{Y}_N - \tau_0 \vec{Z}_N)$ under $\tau=\tau_0$.  Fix $h\in\R$ and define $\tau_N := \tau_0 - h/c_N$. Suppose that 
\begin{equation} \label{psa0}
\lim_{N\to\infty} \mathbb{P}_{\tau_N} \left(d_N (t_N - \mu_N) \le x\right) = G((x+hB)/A)
\end{equation}
holds for every $x\in\R$, where $G$ is a  distribution function of a continuous random variable with mean $0$ and variance $1$, and $A, B>0$ are constants. Then it holds that
\begin{equation} \label{CLT0}
\lim_{N\to\infty} \mathbb{P}_{\tau_0} \left(c_N (\rsnu - \tau_0) \le h\right) = G(hB/A).
\end{equation}
\end{lemma}


\begin{proof}
Denote $\tau'_N = \tau_0 + h/c_N$. We use~\cref{lem1} to write the following:
\begin{equation}\label{sand}
\begin{split}
\mathbb{P}_{\tau_0}\left(t_N(\vec{Z}, \vec{Y} - \tau'_N \vec{Z}) < \mu_N\right) &\le \mathbb{P}_{\tau_0}\left(\widehat{\tau}^* \leq \tau'_N\right) \\
& \le \mathbb{P}_{\tau_0}\left(\rsnu \leq \tau'_N\right) \\
& \le \mathbb{P}_{\tau_0}\left(\widehat{\tau}^{**} \leq \tau'_N\right)\le \mathbb{P}_{\tau_0}\left(t_N(\vec{Z}, \vec{Y} -\tau'_N \vec{Z}) \le \mu_N\right).
\end{split}
\end{equation}
In view of~\cref{lem2}, the distribution of $t_N(\vec{Z}, \vec{Y} - \tau'_N \vec{Z})$ under $\tau=\tau_0$ is same as the distribution of $t_N = t_N(\vec{Z}, \vec{Y} - \tau_0 \vec{Z})$ under $\tau=\tau_N$, which we denote by $G_N$. Then equation \eqref{psa0} tells us that $G_N(x/d_N + \mu_N) \to G((x+hB)/A)$ as $n\to\infty,$ for every $x\in\R$. In particular, for $x=0$ we can say that $G_N(\mu_N) \to G(hB/A)$ and that $G_N(\mu_N-)\to G(hB/A - )$ which is also equal to $G(hB/A)$ since $G$ is continuous. Therefore both the extreme sides of  \eqref{sand} converge to $G(hB/A)$ and thus by sandwich principle, equation \eqref{CLT0} holds.

\end{proof}

\begin{lemma}\label{lem1}
Let $t(\cdot, \cdot)$ be a test statistic such that for all values of $\vec{y}$ and $\vec{z}$, $t(\vec{z}, \vec{y} -\tau \vec{z})$ is a non-increasing function of $\tau$. Then for any $h\in\R,$ 
\begin{equation*}
\begin{split}
\mathbb{P}_\tau\left(t(\vec{Z}, \vec{Y} -a \vec{Z}) < \mu\right) \le \mathbb{P}_\tau\left(\widehat{\tau}^* \leq a\right) 
& \le \mathbb{P}_\tau\left(\rsnu \leq a\right) \\
& \le \mathbb{P}_\tau\left(\widehat{\tau}^{**} \leq a\right)\le \mathbb{P}_\tau\left(t(\vec{Z}, \vec{Y} -a \vec{Z}) \le \mu\right).
\end{split}
\end{equation*}
\end{lemma}

\begin{proof}
The proof is straightforward from the definitions of $\widehat{\tau}^*$ and $\widehat{\tau}^{**}$. First observe that $\tau^*\le \rsnu\le \tau^{**}$ holds almost surely. Now, if $a$ is such that $ t(\vec{Z}, \vec{Y} -a \vec{Z}) < \mu,$ then $$h \ge \sup \{\tau : t(\vec{Z}, \vec{Y} - \tau \vec{Z}) < \mu\} = \widehat{\tau}^{**} \implies h \ge \widehat{\tau}^*,$$ so the left-most inequality  follows. Next, $$\widehat{\tau}^*  \leq a \implies \rsnu \leq a \implies \widehat{\tau}^{**}\leq a.$$ Finally, suppose that $a$ is such that $ t(\vec{Z}, \vec{Y} -a \vec{Z}) > \mu$. We split it into two cases:\\ (a) when $\widehat{\tau}^{*} < \widehat{\tau}^{**},$  we have $$t(\vec{Z}, \vec{Y} -a \vec{Z}) > \mu \implies a \le \sup \{\tau : t(\vec{Z}, \vec{Y} - \tau \vec{Z}) > \mu \} = \widehat{\tau}^{*} < \widehat{\tau}^{**}.$$
(b) When $\widehat{\tau}^{*} = \widehat{\tau}^{**},$ we have $$t(\vec{Z}, \vec{Y} -a \vec{Z}) > \mu \implies h< \sup \{\tau : t(\vec{Z}, \vec{Y} - \tau \vec{Z}) > \mu \} = \widehat{\tau}^{*} = \widehat{\tau}^{**}.$$
This finishes the proof.
\end{proof}

\begin{lemma}\label{lem2}
The distribution of $t(\vec{Z}, \vec{Y} - (\tau + \delta)\vec{Z})$ under $\tau=\tau_0$ is identical to the distribution of $t(\vec{Z}, \vec{Y} - \tau \vec{Z})$ under $\tau=\tau_0 - \delta$.
\end{lemma}

\begin{proof}
For any $x\in\R,$
\begin{align*}
 \mathbb{P}_{\tau_0} \left(t(\vec{Z}, \vec{Y} - (\tau + \delta)\vec{Z})\le x\right)
&= \mathbb{P}_{\tau_0} \left(t(\vec{Z}, \tau_0 \vec{Z} + \vec{b} - (\tau + \delta)\vec{Z})\le x\right) \\
&= P\left(t(\vec{Z}, (\tau_0  - \delta)\vec{Z} + \vec{b} - \tau\vec{Z})\le x\right) \\
&= \mathbb{P}_{\tau_0-\delta} \left(t(\vec{Z}, \vec{Y} - \tau\vec{Z})\le x\right),
\end{align*}
where the second step follows from the fact that the randomization distribution of $\vec{Z}$ is free of  $\tau$. 
\end{proof}




\begin{lemma}\label{lem001}
Suppose that the potential control outcomes $\{b_{N,j} \}_{1\le j\le N}$ satisfy~\cref{ACjs}. Define $$q_{N,j}=\sum_{i=1}^N \ind{b_{N,i}\le b_{N,j}},\quad 1\le j\le N.$$
Then, as $N\to\infty$,
\begin{equation*}
\overline{q}_N := \frac{1}{N}\sum_{j=1}^N q_{N,j} = \frac{N+1}{2} + o(N^{1/2}),\quad\text{ and }\quad
\frac{1}{N}\sum_{j=1}^N \left(q_{N,j} - \overline{q}_N\right)^2 = \frac{N^2-1}{12}+o(N^2).
\end{equation*}
\end{lemma}

\begin{proof}
We first break the ties in an arbitrary manner, and obtain $b_{N(1)}\le b_{N(2)} \le \dots \le b_{N(N)}$. Define $r_{N,j} = \sum_{i=1}^N \ind{b_{N,i}\le b_{N(j)}}$. Note that for each $1\le j\le N,$ 
\begin{equation*}\left|r_{N,j} - j\right| \leq \sum_{i=1}^N \ind{b_{N,i} = b_{N(j)}}.\end{equation*} 
The last equation in conjunction with~\cref{ACjs} yields the following bounds.
\begin{align*}\sum_{j=1}^N \left|r_{N,j} - j\right| &\leq \sum_{j=1}^N \sum_{i=1}^N \ind{b_{N,i} = b_{N,j}} \\
&\leq  \sum_{j=1}^N \sum_{i=1}^N \ind{0\le b_{N,j} - b_{N,i}< N^{-1/2}}\\&
\lesssim N^{3/2}, \end{align*} 
and
\begin{align*}\sum_{j=1}^N \left(r_{N,j} - j\right)^2 &\leq \sum_{j=1}^N N \sum_{i=1}^N \ind{b_{N,i} = b_{N,j}} \\
&\leq  N\sum_{j=1}^N \sum_{i=1}^N \ind{0\le b_{N,j} - b_{N,i}< N^{-1/2}}\\&\lesssim N^{5/2}.\end{align*} 
 Now,
 \begin{align*} 
 \left|\sum_{i=1}^N r_{N,i}^2 - \sum_{i=1}^N i^2 \right| &= \sum_{i=1}^N (r_{N,i} - i)^2 + \left|\sum_{i=1}^N 2i(r_{N,i} - i)\right|\\ 
 &\le\sum_{i=1}^N (r_{N,i} - i)^2 + 2\left(\sum_{i=1}^N i^2 \sum_{i=1}^N (r_{N,i} - i)^2\right)^{1/2}\\
 &\lesssim N^{5/2} + N^{11/4}\\[2mm]
 &= o(N^3).
 \end{align*}
 On the other hand, 
 \begin{align*} 
N\left|\overline{r}_N^2 - ((N+1)/2)^2\right| &= \frac{1}{N}\left|\sum_{i=1}^N (r_{N,i} - i)\sum_{i=1}^N (r_{N,i} + i)\right|\\
&\le\frac{1}{N}\sum_{i=1}^N |r_{N,i} - i|\sum_{i=1}^N (|r_{N,i} - i| + 2i)\\
&\lesssim N^{-1} N^{3/2} \left(N^{3/2} + N^2\right)\\&= o(N^3).
 \end{align*}
Combining the above with the fact that $\{r_{N,1},\dots,r_{N,N}\}$ is same as $\{q_{N,1},\dots,q_{N,N}\}$, 
\begin{align*}\sum_{i=1}^N \left(q_{N,i} - \overline{q}_N\right)^2 &= \sum_{i=1}^N r_{N,i}^2 - N \bar{r}_N^2\\ &= \sum_{i=1}^N \left(i - \frac{N+1}{2}\right)^2 + o(N^3)\\ &= \frac{N(N^2-1)}{12} + o(N^3).\end{align*}
This finishes the proof of the second assertion. To prove the first assertion, we fix any $\delta>0$ and note that
\begin{align*}
\left|\sum_{j=1}^N q_{N,j} - \frac{N(N+1)}{2}\right| &= \left|\sum_{j=1}^N (r_{N,j} - j)\right| \\
&\le \sum_{j=1}^N \left|r_{N,j} - j\right| \\
&\leq \sum_{j=1}^N \sum_{i=1}^N \ind{b_{N,i} = b_{N,j}} \\
&\leq  \sum_{j=1}^N \sum_{i=1}^N \ind{0\le b_{N,j} - b_{N,i}< \delta N^{-1/2}}.
\end{align*} 
Invoking~\cref{ACjs} we can say that
$$\limsup_{N\to\infty}\left|\sum_{j=1}^N q_{N,j} - \frac{N(N+1)}{2}\right|\le \delta \mathcal{I}_b.$$
Now letting $\delta\to 0$ completes the argument.
\end{proof}


\begin{lemma}\label{null-var} Suppose that ranks $\{q_{N,j}\}$ satisfy 
$$\sum_{j=1}^N \left(q_{N,j} - \overline{q}_N\right)^2 = \frac{N(N^2-1)}{12}+o(N^3),\text{ as }N\to\infty.$$
Let $t_N$ be the Wilcoxon rank-sum statistic when the treatments are assigned by an $m$-out-of-$N$ SRSWOR sample where $m/{N}\to \lambda\in (0,1)$ as $N\to\infty$.  Then, under $\tau=\tau_0$, $$\var \left(t_N\right) \sim\frac{ \lambda(1-\lambda)}{12}N^3\text{ as $N\to \infty$.}$$ 
\end{lemma}

\begin{proof}
Recall from \eqref{null-dd} that under $\tau=\tau_0,$ we can write $t_N \dd \sum_{j=1}^N  q_{N,j}Z_{N,j},$ where $q_{N,j}=\sum_{i=1}^N \ind{b_{N,i} \le b_{N,j}}$.
Note that each $Z_{N,j}$ is a Bernoulli($m/N$) random variable, and for $i\neq j$ we have \begin{align*}\text{Cov}(Z_{N,i}, Z_{N,j}) &= P(Z_{N,i}=1,Z_{N,j}=1) - P(Z_{N,i}=1)P(Z_{N,j}=1)\\ 
&= \frac{m(m-1)}{N(N-1)} - \left(\frac{m}{N}\right)^2\\ &=-\frac{1}{N-1} \cdot\frac{m}{N} \left(1-\frac{m}{N}\right).\end{align*} Hence we deduce that, under $\tau=\tau_0$, \begin{align*}
\var\left(t_N\right) &= \sum_{j=1}^N\text{Var}\left(q_{N,j} Z_{N,j}\right) +2\sum_{1\le i<j\le N}\text{Cov}\left(q_{N,i} Z_{N,i}, q_{N,j} Z_{N,j}\right) \\
&= \sum_{j=1}^N q_{N,j}^2 \frac{m}{N}\left(1-\frac{m}{N}\right) -2\sum_{1\le i<j\le N} q_{N,i}q_{N,j} \frac{m}{N} \left(1-\frac{m}{N}\right)\frac{1}{N-1} \\
&=  \frac{m}{N}\left(1-\frac{m}{N}\right)\frac{N}{N-1} \sum_{j=1}^N \left(q_{N,j} - \overline{q}_N\right)^2.
\end{align*}
In light of the given condition on the ranks, the above implies that
\begin{align*}
\var\left(t_N\right) &=  \frac{m}{N}\left(1-\frac{m}{N}\right)\frac{N}{N-1}\left( \frac{N(N^2-1)}{12} + o(N^3)\right)\sim \lambda(1-\lambda)\frac{N^3}{12},
\end{align*}
which completes the proof.
\end{proof}

\begin{lemma}\label{SN}  Let $\vec{Z}_N$ be the vector of treatment indicators when the treatments are assigned by $m$-out-of-$N$ SRSWOR where ${m(N)}/{N}\to \lambda\in (0,1)$ as $N\to\infty$. Define $$S_N = \sum_{j=1}^N  \sum_{i=1, i\neq j}^N  (1 -Z_{N,i}) Z_{N,j} I_{h,N} (b_{N,j} - b_{N,i})$$ where  $I_{h,N}$ is defined in \eqref{Ian} of the main paper.  Then under~\cref{ACjs} it holds that $$N^{-3/2} S_N\Pto h\lambda(1-\lambda)\mathcal{I}_b.$$
\end{lemma}

\begin{proof}
We imitate the proof of~\cref{est.I_C}. The first step is to note that $$N^{-3/2} \E S_N = \frac{m}{N} \left(1-\frac{m-1}{N-1}\right) N^{-3/2}\sum_{j=1}^N  \sum_{i=1, i\neq j}^N I_{h,N} (b_{N,j} - b_{N,i})
\to h\lambda(1-\lambda)\mathcal{I}_b.$$ It therefore remains to show that $N^{-3}\text{Var}\left(S_N\right) \to 0$.
For brevity, let us abuse the notation $I_{h,N} (b_{N,j}-b_{N,i})$ and write $I_{h,N} (i, j)$ instead. We have $$\text{Var}\left(S_N\right) =
\sum_{i=1}^N  \sum_{j=1, j\neq i}^N \sum_{k=1}^N  \sum_{l=1, l\neq k}^N  \text{Cov}\left((1 -Z_{N,i}) Z_{N,j},(1 -Z_{N,k}) Z_{N,l}\right) I_{h,N} (i, j)I_{h,N} (k, l).$$ 
Note, $2\le |\{i, j, k, l\}| \le 4$. Consider the following cases. 
\begin{enumerate}
\item[(a)] $|\{i, j, k, l\}|=2,$ i.e., $(i, j) = (k, l)$. Note that $\text{Var}((1 -Z_{N,i}) Z_{N,j}) = p_N(1-p_N)$ where $$p_N = P(Z_{N,i} = 0, Z_{N,j}=1) = \frac{m}{N}\frac{N-m}{N-1} \sim \lambda (1-\lambda).$$ Hence the contribution of these terms in $\text{Var}(S_N)$ is given by $$\sum_{i=1}^N  \sum_{j=1, j\neq i}^N p_N(1-p_N)I^2_{h,N} (i,j) \lesssim   \lambda(1-\lambda) \sum_{i=1}^N  \sum_{j=1, j\neq i}^N I_{h,N} (i,j)  \lesssim N^{3/2}.$$

\item[(b)] $|\{i, j, k, l\}|=4,$ i.e., all $4$ indices are distinct. Note that in this case 
\begin{align*} &\text{Cov}\left((1 -Z_{N,i}) Z_{N,j},(1 -Z_{N,k}) Z_{N,l}\right)\\ 
&= P(Z_{N,i}=0, Z_{N,j}=1, Z_{N,k} = 0, Z_{N,l} =1) - p_N^2\\
& = \frac{\binom{N-4}{m-2}}{\binom{N}{m}} - \left(\frac{m}{N}\frac{N-m}{N-1} \right)^2\\
& = \frac{m}{N}\frac{N-m}{N-1} \left(\frac{(m-1)(N-m-1)}{(N-2)(N-3)} - \frac{m}{N}\frac{N-m}{N-1} \right)\\
&= \frac{m}{N}\frac{N-m}{N-1} \frac{(4 mN^2 - 4 m^2 N  - N^3 + 6 m^2  - 6 m N + 2 N^2 - N)}{N(N-1)(N-2)(N-3)}\\
&\sim -\lambda(1-\lambda)(1-2\lambda)^2 N^{-1}.
\end{align*}
Hence the contribution $u_N$ of these terms in $\text{Var}(S_N)$ satisfies the following. $$u_N \lesssim N^{-1} \sum_{i=1}^N  \sum_{j=1, j\neq i}^N \sum_{k=1}^N  \sum_{l=1, l\neq k}^N I_{h,N} (i,j)I_{h,N} (k,l)  \lesssim N^{-1+3/2+3/2} =N^{2}.$$
\item[(c)] $|\{i, j, k, l\}|=3$. Here we have 4 sub-cases:
\begin{center}
\fbox{%
\begin{tabular}{c c r}
\vspace{1mm}
      sub-case & $\E(1 -Z_{N,i}) Z_{N,j}(1 -Z_{N,k}) Z_{N,l}$ & $\text{Cov}\left((1 -Z_{N,i}) Z_{N,j},(1 -Z_{N,k}) Z_{N,l}\right)$ \\
      \hline\noalign{\vskip 1ex}
$i = k, j\neq l$ & $P(Z_{N,i} = 0, Z_{N,j} =1, Z_{N,l}=1)$  & $\frac{m(m-1)(N-m)}{N(N-1)(N-2)}   - p_N^2\sim \lambda^3(1-\lambda) $\\[2mm]
$i \neq k, j = l$ & $P(Z_{N,i} = 0, Z_{N,j}=1, Z_{N,k} =0)$  & $\frac{m(N-m)(N-m-1)}{N(N-1)(N-2)} - p_N^2\sim\lambda(1-\lambda)^3 $\\[1mm]
$i = l, j \neq k$ & $0$  & $- p_N^2 \sim - \lambda^2(1-\lambda)^2$\\[2mm]
$i \neq l, j = k$ & $0$ & $-p_N^2 \sim - \lambda^2(1-\lambda)^2$\\
\end{tabular}}
\end{center}

Hence, if $v_n$ be the contribution these terms in $\text{Var}(S_N)$, then 
\begin{align*}
v_N  & \lesssim  \sum_{i, j, l \text{ distinct}}  I_{h,N} (i,j)I_{h,N} (i,l)  + \sum_{i, j, k \text{ distinct}}  I_{h,N} (i,j)I_{h,N} (k,j)  \\ 
&\quad + \sum_{i, j, k \text{ distinct}}  I_{h,N} (i,j)I_{h,N} (k,i)  + \sum_{i, j, l \text{ distinct}} I_{h,N} (i,j)I_{h,N} (j,l) \\
&\le 4N\sum_{i, j\neq i}  I_{h,N} (i,j) \lesssim N^{1+3/2} =N^{5/2}.
\end{align*}
\end{enumerate}
Combining the three cases, we can say that $$N^{-3} \text{Var}(S_N) \lesssim N^{-3} (N^{3/2} + N^2 + N^{5/2}) = o(1),$$ as desired to show.
 \end{proof}

\begin{lemma}\label{beta0goesto0}
    Assume that $\eps_{N,i}$ are i.i.d.~with mean zero, and $\vec{X}_N\in\R^{N\times p}$ is deterministic and satisfies  $N^{-1}\vec{X}_N^\top\vec{X}_N\to\Sigma\succ 0$. Then, $$\vec{\beta}_N^{(0)}:=(\vec{X}_N^\top \vec{X}_N)^{-1} \vec{X}_N^\top \vec{\eps}_N\Pto \vec{0}.$$
\end{lemma}
\begin{proof}
    Since $N^{-1}\vec{X}_N^\top\vec{X}_N\to \Sigma\succ 0$, it suffices to show that
$$
\frac{1}{N} \sum_{i=1}^N \vec{x}_{N, i}\, \eps_{N, i}\Pto \vec{0},
$$
where we recall that $\vec{x}_{N,i}$ denotes the $i$-th row of the matrix $\vec{X}_N$. We define truncation at level $K$ as $$\varepsilon^{(K)}=\varepsilon \mathbf{1}\{|\varepsilon| \leq K\}-\mathbb{E}[\varepsilon \mathbf{1}\{|\varepsilon| \leq K\}],\quad\text{and}\quad \eps^{(>K)}=\varepsilon-\varepsilon^{(K)}.$$ Then $\mathbb{E} \left[\varepsilon^{(K)}\right]=\mathbb{E} \left[\eps^{(>K)}\right]=0$, and $\left|\varepsilon^{(K)}\right| \leq 2 K$. Moreover, it follows from $\E|\eps|<\infty$ and the DCT that $\mathbb{E}\left|\eps^{(>K)}\right| \rightarrow 0$ as $K\to \infty$. Fix $\delta>0$ and choose $K$ large enough so that $\E\left|\eps^{(>K)}\right|<\delta$.
Next, decompose $N^{-1}\sum_{i=1}^N\vec{x}_{N,i}\,\eps_{N,i}$ into the following parts:
$$U_N^{(K)}=\frac{1}{N} \sum_{i=1}^N \vec{x}_{N, i}\, \varepsilon_{N, i}^{(K)},\quad \text{and} \quad V_N^{(K)}=\frac{1}{N}  \sum_{i=1}^N \vec{x}_{N, i}\, \eps_{N, i}^{(>K)}.$$
Note that
$$
\mathbb{E}\left\|U_N^{(K)}\right\|_2^2=\frac{\operatorname{Var}\left(\varepsilon^{(K)}\right)}{N^2} \sum_{i=1}^N \|\vec{x}_{N,i}\|_2^2 \rightarrow 0\implies V_N^{(K)} \Pto \vec{0}.
$$
On the other hand,
$$
\mathbb{E}\left\|V_N^{(K)}\right\|_1\le\frac{1}{N}\sum_{i=1}^N\sum_{j=1}^p |(\vec{x}_{N,i})_{j}|\E|\eps^{(>K)}| \leq\sqrt{p}\left(\frac{1}{N} \sum_{i=1}^N\left\|x_{N, i}\right\|_2^2\right)^{1 / 2} \delta.
$$
Since $\delta>0$ is arbitrary, we are through.
\end{proof}

\begin{lemma}\label{beta0goesto0a.s.}
    Assume that $\eps_{N,i}$ are i.i.d.~with mean zero, and $\vec{X}_N\in\R^{N\times p}$ is deterministic, satisfies  $N^{-1}\vec{X}_N^\top\vec{X}_N\to\Sigma\succ 0$ and has uniformly bounded row-norms: $\sup_N\max_{i\le N}\|\vec{x}_{N,i}\|_2<\infty$. Then, $$\vec{\beta}_N^{(0)}:=(\vec{X}_N^\top \vec{X}_N)^{-1} \vec{X}_N^\top \vec{\eps}_N\asto \vec{0}.$$
\end{lemma}

\begin{proof}
It suffices to show that for each coordinate $1\le j\le p$,
$$W_{N,\,j}:=\frac{1}{N} \sum_{i=1}^N (\vec{x}_{N, i})_{j}\, \eps_{N, i}\asto 0,$$
where $(\vec{x}_{N,i})_j$ denotes the $j$-th entry in the $i$-th row of $\vec{X}_N$. We define truncation at level $K$ as follows:
$$\varepsilon^{(K)}=\varepsilon \mathbf{1}\{|\varepsilon| \leq K\}-\mathbb{E}[\varepsilon \mathbf{1}\{|\varepsilon| \leq K\}],\quad\text{and}\quad \eps^{(>K)}=\varepsilon-\varepsilon^{(K)}.$$ Then $\mathbb{E} \left[\varepsilon^{(K)}\right]=\mathbb{E} \left[\eps^{(>K)}\right]=0$, and $\left|\varepsilon^{(K)}\right| \leq 2 K$. Moreover, it follows from $\E|\eps|<\infty$ and the DCT that $\mathbb{E}\left|\eps^{(>K)}\right| \rightarrow 0$ as $K\to \infty$. 
Next, decompose $W_{N,\,j} =N^{-1}\sum_{i=1}^N (\vec{x}_{N,i})_{j}\,\eps_{N,i}$ into the following parts:
$$U_N^{(K)}=\frac{1}{N} \sum_{i=1}^N (\vec{x}_{N, i})_{j}\, \varepsilon_{N, i}^{(K)},\quad \text{and} \quad V_N^{(K)}=\frac{1}{N}  \sum_{i=1}^N (\vec{x}_{N, i})_{j}\, \eps_{N, i}^{(>K)}.$$
It follows from Hoeffding's inequality that
\begin{align*}
\PP\left(\left|U_N^{(K)}\right|>t\right)\le 2\exp\left(\frac{-t^2}{8K^2\sum_{i=1}^N (\vec{x}_{N,i})_{j}^2/N^2}\right)\le 2\exp\left(\frac{-Nt^2}{8K^2C_j}\right),
\end{align*}
where $C_j=\sup_{N\ge 1}N^{-1}\sum_{i=1}^N (\vec{x}_{N,i})_{j}^2<\infty$. It therefore follows from the first Borel-Cantelli lemma that
$U_N^{(K)}\asto 0$ as $N\to\infty$, for any fixed $K$.
On the other hand,
\begin{align*}
        \sup_{N\ge 1} \left|V_N^{(K)}\right|&\le\sup_{N\ge 1}\frac{1}{N}\sum_{i=1}^N  |(\vec{x}_{N,i})_{j}||\eps_{N,i}^{(>K)}| \\
        &\leq \left(\sup_{N\ge 1}\max_{i\le N} |(\vec{x}_{N,i})_{j}| \right) \frac{1}{N} \sum_{i=1}^N|\eps_{N,i}^{(>K)}|\\
        &\asto \left(\sup_{N\ge 1}\max_{i\le N} |(\vec{x}_{N,i})_{j}| \right) \E|\eps^{(>K)}|.
\end{align*}
Consequently,
$$\limsup_{N\to\infty} |W_{N,\,j}|\le \left(\sup_{N\ge 1}\max_{i\le N} |(\vec{x}_{N,i})_{j}| \right) \E|\eps^{(>K)}|\quad\text{almost surely.}$$
Since $\E|\eps^{(>K)}|\to 0$ as $K\to\infty$, this finishes the proof.
\end{proof}

\begin{lemma}\label{Jbexists}
 Suppose that $\vec{b}_N = \vec{X}_N\vec{\beta}_N +\vec{\eps}_N$, where $\eps_{N,1},\dots,\eps_{N,N}$ are i.i.d$.$ from $\mathcal{N}(0,\sigma^2)$. Define, for any fixed $h$, 
 \begin{equation}
     J_N := N^{-3/2}\sum_{j=1}^N \sum_{i=1}^N I_{h,N} (\widetilde{b}_{N,j} - \widetilde{b}_{N,i}),
 \end{equation}
 where $\widetilde{b}_{N,j}$ is defined in \eqref{btilde}, and $I_{h,N}$ is defined in \eqref{Ian} of the main paper.
 Then $\E J_N\to h(2\sqrt{\pi}\sigma)^{-1}$, and $\var(J_N)\to 0$, implying that~\cref{AssumpB2} holds in probability, with $\mathcal{J}_b=(2\sqrt{\pi}\sigma)^{-1}$.
\end{lemma}

\begin{proof}
We do the proof for $h>0$, the other case will be similar. For simplicity in notation, we omit the index $N$ in $\widetilde{b}_{N,j}, Z_{N,j}, \vec{p}_{N,j}$, etc. throughout this proof. Observe that under the linear model, we have $$\vec{\widetilde{b}} = (\vec{I} - \vec{P}_{\vec{X}})\vec{b} = (\vec{I} - \vec{P}_{\vec{X}})\vec{\eps}.$$ Hence for any pair $(i,j)$ of distinct indices, $\widetilde{b}_{j} - \widetilde{b}_i \sim \mathcal{N}(0, \sigma_{ij}^2)$, where $\sigma_{ij}^2 :=  \sigma^2(2-\|\vec{p}_j - \vec{p}_i\|^2)$ (this follows from the fact that $\vec{P}_{\vec{X}}$ is a projection matrix). 
Thus
\begin{align*}
\E J_N &= N^{-3/2} \sum_{j=1}^N \sum_{i=1}^N P\left(0\le  \widetilde{b}_{j}-\widetilde{b}_{i} < hN^{-1/2}\right)\\
&= N^{-3/2} \sum_{j=1}^N \sum_{i=1}^N \left(\Phi\left(\sigma_{i,j}^{-1}hN^{-1/2}\right) - \Phi(0)\right).
 \end{align*}
 For any $\delta>0$ and $N\in\mathbb{N}$, define \begin{equation}\label{SNd}
     S_{N,\delta} := \{(i,j) : 1\le i,j\le N, \|\vec{p}_j - \vec{p}_i\| \le \delta\}.
 \end{equation}
 Invoking~\cref{PiPj}, we get
$$|S_{N,\delta}^c|=|\{(i,j) : 1\le i,j\le N, \|\vec{p}_j - \vec{p}_i\| > \delta\}|\le \delta^{-2} \sum_{j=1}^N\sum_{i=1}^N \|\vec{p}_j - \vec{p}_i\|^2 = O(N).$$
 Consequently, $$N^{-3/2} \sum_{(i,j)\in S_{N,\delta}^c} \left(\Phi\left(\sigma_{i,j}^{-1}hN^{-1/2}\right) - \Phi(0)\right) = O(N^{-1/2}).$$
 For $(i,j)\in S_{N,\delta}$ we apply the Taylor theorem to conclude that there exists $\psi_{i,j}\in (0, \sigma_{i,j}^{-1}hN^{-1/2})$ such that
 \begin{align*}
 \Phi\left(\sigma_{i,j}^{-1}hN^{-1/2}\right) - \Phi(0)
  &=  \sigma_{i,j}^{-1}hN^{-1/2}\phi\left(0\right) +  \sigma_{i,j}^{-2}h^2 N^{-1} \phi'(\psi_{i,j}).
\end{align*}
Since $\phi'$ is  bounded on $\R,$ we can say that 
\begin{align*}
\Delta_N &:= \left|\E J_N - N^{-3/2} \sum_{(i,j)\in S_{N,\delta}} \sigma_{i,j}^{-1}hN^{-1/2}\phi\left(0\right) \right| \\
 &\lesssim N^{-1/2} + N^{-3/2} \sum_{(i,j)\in S_{N,\delta}} \sigma_{i,j}^{-2}h^2 N^{-1} \\
&\lesssim N^{-1/2} + N^{-5/2} \sum_{(i,j)\in S_{N,\delta}} (2- \|\vec{p}_j - \vec{p}_i\|^2)^{-1}\\
&\lesssim N^{-1/2} + N^{-5/2} \sum_{(i,j)\in S_{N,\delta}} (2- \delta^2)^{-1} \lesssim N^{-1/2}.
\end{align*}
Hence $\Delta_N\to 0$ as $N\to\infty$, and thus 
\begin{equation*}
    \limsup_{N\to\infty} \E J_N = \limsup_{N\to\infty} N^{-3/2} \sum_{(i,j)\in S_{N,\delta}} \sigma_{i,j}^{-1}hN^{-1/2}\phi\left(0\right) \le h\sigma^{-1}\phi(0)(2-\delta^2)^{-1/2}.
\end{equation*}
Since the above holds for every $\delta>0$, we can say that $\limsup_{N\to\infty} \E J_N\le  \frac{h}{2\sqrt{\pi}\sigma}.$
On the other hand, the fact that $\Delta_{N}\to 0$ also yields the following:
\begin{align*}
    \liminf_{N\to\infty} \E J_N &= \liminf_{N\to\infty} N^{-3/2} \sum_{(i,j)\in S_{N,\delta}} (2-\|\vec{p}_j-\vec{p}_i\|)^{-1/2}h\sigma^{-1} N^{-1/2}\phi\left(0\right)\\
    &\ge \liminf_{N\to\infty} N^{-2} |S_{N,\delta}| \cdot h\sigma^{-1}\phi(0)\cdot 2^{-1/2}=h\sigma^{-1}\phi(0)\cdot 2^{-1/2},
\end{align*}
where in the last step we again used the fact that $|S_{N, \delta}^c| = O(N)$. We thus conclude that $$\lim_{N\to\infty} \E J_N =  \frac{h}{2\sqrt{\pi}\sigma}.$$
Next we show that $\var(J_N)\to 0$ as $N\to\infty$. First,
$$\E J_N^2 = N^{-3}\E \sum_{i,j,k,l} \ind{0\le  \widetilde{b}_{j}-\widetilde{b}_{i} < hN^{-1/2}, 0\le  \widetilde{b}_{l}-\widetilde{b}_{k} < hN^{-1/2}}.$$
Now, as in the proof of~\cref{est.I_C}, observe that the contribution of the terms with repeated indices in the above summation is negligible (since~\cref{AssumpB2} holds and we already have shown that $\E J_N$ converges). To analyze the terms with distinct indices, note that  for any $4$ distinct indices $i,j,k,l$,
\begin{equation*}
    \begin{pmatrix}
    \widetilde{b}_{j}-\widetilde{b}_{i} \\ \widetilde{b}_{l} - \widetilde{b}_{k}
    \end{pmatrix} \sim \mathcal{N}_2\left(\begin{pmatrix}
    0\\ 0
    \end{pmatrix}, \sigma^2\begin{pmatrix}
    2-\|\vec{p}_j - \vec{p}_i\|^2 & - (\vec{p}_j - \vec{p}_i)^\top(\vec{p}_l - \vec{p}_k) \\ 
   - (\vec{p}_j - \vec{p}_i)^\top(\vec{p}_l - \vec{p}_k)  & 2-\|\vec{p}_l - \vec{p}_k\|^2
    \end{pmatrix}\right).
\end{equation*}
Let $\rho_{i,j,k,l}:=\text{corr}(\widetilde{b}_{j}-\widetilde{b}_{i},\widetilde{b}_{l} - \widetilde{b}_{k})$.
Now we again play the trick of splitting the sum into two groups, such that in one group $|\rho_{i,j,k,l}|$ is small, whereas for the other group the number of summands is small. Fix any $\delta\in(0,1)$, and consider the set $S_{N,\delta}$ defined in \eqref{SNd}. If $(i,j)$ or $(k,l)$ does not belong to $S_{N,\delta}$, then that brings down the count for such summands. To be precise,
\begin{align*}
    & N^{-3}\sum_{(i,j)\in S_{N,\delta}^c \text{ or } (k,l) \in S_{N,\delta}^c} P(0\le   \widetilde{b}_{j}-\widetilde{b}_{i} \le hN^{-1/2}, 0\le   \widetilde{b}_{l}-\widetilde{b}_{k} \le hN^{-1/2})\\
    &\le 2N^{-3}\sum_{1\le j,i\le N}\sum_{(k,l) \in S_{N,\delta}^c} P(0\le   \widetilde{b}_{j}-\widetilde{b}_{i} \le hN^{-1/2})\\
    &\lesssim N^{-2}\sum_{1\le j,i\le N} P(0\le   \widetilde{b}_{j}-\widetilde{b}_{i} \le hN^{-1/2}) = N^{-1/2} \E J_N.
\end{align*}
Since $\E J_N$ converges, the above shows that the contribution from such indices are also negligible. Finally, for $(i,j),(k,l)\in S_{N,\delta}$, we have $$|\rho_{i,j,k,l}|=\frac{\left|(\vec{p}_j - \vec{p}_i)^\top(\vec{p}_l - \vec{p}_k)\right|}{\sqrt{2-\|\vec{p}_j - \vec{p}_i\|^2 }\sqrt{2-\|\vec{p}_l - \vec{p}_k\|^2}} \leq \frac{\delta}{2-\delta^2}\le \delta,$$
and hence~\cref{bvn-estimate} tells us that
\begin{align*}
&\big|P(0\le   \widetilde{b}_{j}-\widetilde{b}_{i} \le hN^{-1/2}, 0\le   \widetilde{b}_{l}-\widetilde{b}_{k} \le hN^{-1/2}) \\
&\qquad \qquad - 
P(0\le   \widetilde{b}_{j}-\widetilde{b}_{i} \le hN^{-1/2})P(0\le   \widetilde{b}_{l}-\widetilde{b}_{k} \le hN^{-1/2})\big|\\[2mm]
&\le |\rho_{i,j,k,l}|\left((1-\rho_{i,j,k,l}^2)^{-1/2} + (1-|\rho_{i,j,k,l}|)^{-2} \right) h^2 N^{-2} \sigma_{i,j}^{-1}\sigma_{k,l}^{-1}\\[2mm]
&\le h^2 N^{-2} \sigma^{-2} \delta \left((1-\delta^2)^{-1/2}+(1-\delta)^{-2}\right).
\end{align*}
Thus we have shown that
\begin{equation*}
\begin{split}
\left|\E J_N^2 - \left(N^{-3/2}\sum_{(i,j)\in S_{N,\delta}} P\left(0\le \widetilde{b}_{j}-\widetilde{b}_{i} \le hN^{-1/2}\right)\right)^2\right|\\ 
\leq h^2\sigma^{-2} \delta \left((1-\delta^2)^{-1/2}+(1-\delta)^{-2}\right)+O(N^{-1/2}).
\end{split}
\end{equation*}
Now for every $\delta\in (0,1)$, we have $|S_{N,\delta}^c|=O(N)$, so it follows that $$\limsup_{N\to\infty}\var(J_N) \leq h^2\sigma^{-2} \delta \left((1-\delta^2)^{-1/2}+(1-\delta)^{-2}\right).$$ Since $\delta\in(0,1)$ is arbitrary here, letting $\delta\to 0$ completes the proof. 
\end{proof}


\begin{lemma}\label{AssumpB3} Suppose that~\cref{AssumpB2} holds. Then for any $h\ge \delta>0,$  
\begin{equation*}
\lim_{N\to\infty} N^{-3/2}\sum_{j=1}^N \sum_{i=1}^N \ind{\big|\widetilde{b}_{N,j}- \widetilde{b}_{N,i}\big| \in \left[(h-\delta)N^{-1/2},(h+\delta)N^{-1/2}\right]} = 4\delta\mathcal{J}_b
\end{equation*}
where $\mathcal{J}_b$ is defined in~\cref{AssumpB2} of the main paper.
\end{lemma}

\begin{proof}
First we show that for any $h\in\R,$ 
\begin{equation}\label{equal.a}
\lim_{N\to\infty} N^{-3/2}\sum_{j=1}^N \sum_{i=1}^N \ind{\widetilde{b}_{N,j}- \widetilde{b}_{N,i} =hN^{-1/2}} =0.\end{equation}
To show this, fix $h\geq 0$. Note that for any $\delta>0,$
\begin{align*}
0 &\le N^{-3/2}\sum_{j=1}^N \sum_{i=1}^N \ind{\widetilde{b}_{N,j}- \widetilde{b}_{N,i} =hN^{-1/2}} \\
&\le N^{-3/2}\sum_{j=1}^N \sum_{i=1}^N \ind{hN^{-1/2}\le \widetilde{b}_{N,j}- \widetilde{b}_{N,i} < \frac{h+\delta}{\sqrt{N}}} \\
&\le N^{-3/2}\sum_{j=1}^N \sum_{i=1}^N \ind{0\le \widetilde{b}_{N,j}- \widetilde{b}_{N,i} < \frac{h+\delta}{\sqrt{N}}} - N^{-3/2}\sum_{j=1}^N \sum_{i=1}^N \ind{0\le \widetilde{b}_{N,j}- \widetilde{b}_{N,i} < hN^{-1/2}}.
\end{align*}
As $N\to\infty,$ the above RHS converges to $\delta \mathcal{J}_b$. Then letting $\delta\to 0$ we finish the proof of \eqref{equal.a} for $h\geq 0$. The case $h < 0$ is similar.

Next, fix any $h\ge \delta>0$. We write
\begin{align*}
&N^{-3/2}\sum_{j=1}^N \sum_{i=1}^N \ind{\left|\widetilde{b}_{N,j}- \widetilde{b}_{N,i}\right| \in \left[\frac{h-\delta}{\sqrt{N}},\frac{h+\delta}{\sqrt{N}}\right]} \\
&= N^{-3/2}\sum_{j=1}^N \sum_{i=1}^N \left(\ind{\left|\widetilde{b}_{N,j}- \widetilde{b}_{N,i}\right| \le\frac{h+\delta}{\sqrt{N}}} - \ind{\left|\widetilde{b}_{N,j}- \widetilde{b}_{N,i}\right| \le\frac{h-\delta}{\sqrt{N}}} \right)\\
&= N^{-3/2}\sum_{j=1}^N \sum_{i=1}^N \left(\ind{0\le \widetilde{b}_{N,j}- \widetilde{b}_{N,i} < \frac{h+\delta}{\sqrt{N}}}  + \ind{- \frac{h+\delta}{\sqrt{N}}\le \widetilde{b}_{N,j}- \widetilde{b}_{N,i} <0} 
 \right)\\
 &\quad -  N^{-3/2}\sum_{j=1}^N \sum_{i=1}^N \left(\ind{0\le \widetilde{b}_{N,j}- \widetilde{b}_{N,i} < \frac{h-\delta}{\sqrt{N}}}  + \ind{- \frac{h-\delta}{\sqrt{N}}\le \widetilde{b}_{N,j}- \widetilde{b}_{N,i} <0}  \right)\\
 &\quad + N^{-3/2}\sum_{j=1}^N \sum_{i=1}^N \left(\ind{\widetilde{b}_{N,j}- \widetilde{b}_{N,i} = \frac{h+\delta}{\sqrt{N}}} - \ind{\widetilde{b}_{N,j}- \widetilde{b}_{N,i} = \frac{h-\delta}{\sqrt{N}}}\right).
\end{align*}
Appealing to~\cref{AssumpB2} and \eqref{equal.a}, we conclude that as $N\to\infty$, the above display converges to $2(h+\delta)\mathcal{J}_b - 2(h-\delta)\mathcal{J}_b = 4\delta\mathcal{J}_b$.
\end{proof}

\begin{lemma}\label{AssumpIb.more} Suppose that~\cref{ACjs} holds. Then for any $h\ge \delta>0,$  
\begin{equation*}
\lim_{N\to\infty} N^{-3/2}\sum_{j=1}^N \sum_{i=1}^N \ind{\left|b_{N,j}- b_{N,i}\right| \in \left[\frac{h-\delta}{\sqrt{N}},\frac{h+\delta}{\sqrt{N}}\right]} = 4\delta\mathcal{I}_C
\end{equation*}
where $\mathcal{I}_C$ is defined in~\cref{ACjs} of the main paper.
\end{lemma}

\begin{proof}
The proof is essentially same as the proof of~\cref{AssumpB3}, hence omitted.
\end{proof}

\begin{lemma}\label{theta_ij}
Define $D_N$ as in \eqref{tNadj.decomp2}, and let 
\begin{equation}\label{theta_ij_defn}
    \theta_{N,i,j} := \E \left[|\xi_{N,i,j} - \widetilde{\xi}_{N,i,j}| \ \Big|\  Z_{N,j} =1\right].
\end{equation}
Then the following holds:
\begin{equation*}
\sqrt{\E(N^{-3/2} D_N^2)}  \leq N^{-3/2} \sum_{j=1}^N  \sum_{i=1}^N\sqrt{\theta_{N,i,j}}.
\end{equation*}
\end{lemma}

\begin{proof}
Recall the notations $\mathbf{I}_N$ and $\widetilde{\mathbf{I}}_N$ from \eqref{I.II.N} and \eqref{I.II.tilde.N}, respectively.
We use the simple result $\sqrt{\E (\sum_{j=1}^N V_j)^2} \leq  \sum_{j=1}^N \sqrt{\E(V_j^2)}$ to derive the following.
\begin{align*}
\sqrt{\E(  D_N)^2} &=    \sqrt{\E\left(\mathbf{I}_N - \widetilde{\mathbf{I}}_N\right)^2}\\
&=  \left[\E\bigg( \sum_{j=1}^N Z_{N,j}  \sum_{i=1}^N (\xi_{N,i,j} - \widetilde{\xi}_{N,i,j})\bigg)^2 \right]^{1/2}\\
&\le   \sum_{j=1}^N  \left[\E\bigg( Z_{N,j} \sum_{i=1}^N (\xi_{N,i,j} - \widetilde{\xi}_{N,i,j})\bigg)^2\right]^{1/2}\\
&=   \sum_{j=1}^N  \Bigg[\E\left(Z_{N,j} \E\left(\bigg(\sum_{i=1}^N (\xi_{N,i,j} - \widetilde{\xi}_{N,i,j})\bigg)^2\ \Big|\ Z_{N,j}\right)\right)\Bigg]^{1/2}\\
&=   \sum_{j=1}^N  \Bigg[\E\left(Z_{N,j} \E\left(\bigg(\sum_{i=1}^N (\xi_{N,i,j} - \widetilde{\xi}_{N,i,j})\bigg)^2\ \Big|\ Z_{N,j}=1\right)\right)\Bigg]^{1/2}\\
&=   \sum_{j=1}^N  \sqrt{\E(Z_{N,j})} \Bigg[\E\left(\bigg(\sum_{i=1}^N (\xi_{N,i,j} - \widetilde{\xi}_{N,i,j})\bigg)^2\ \Big|\ Z_{N,j}=1\right)\Bigg]^{1/2}\\
&\le   \sum_{j=1}^N  \sqrt{\E(Z_{N,j})} \sum_{i=1}^N \left[\E\left((\xi_{N,i,j} - \widetilde{\xi}_{N,i,j})^2\ \Big|\ Z_{N,j}=1\right)\right]^{1/2}\\
&=   \sum_{j=1}^N  \sqrt{\E(Z_{N,j})} \sum_{i=1}^N \left[\E\left(|\xi_{N,i,j} - \widetilde{\xi}_{N,i,j}|\ \Big|\ Z_{N,j}=1\right)\right]^{1/2}\\
&\le   \sum_{j=1}^N  \sum_{i=1}^N \sqrt{\theta_{N,i,j}}.
\end{align*}
This completes the proof.
\end{proof} 

\begin{lemma}\label{gamma_ij} Define $Q_N$ as in \eqref{tNadj.decomp2}, and let 
\begin{equation}\label{gamma_ij_defn}
    \gamma_{N,i,j} := \E \left[|\xi_{N,i,j} - \widetilde{\xi}_{N,i,j}| \ \Big|\  Z_{N,i} =1,Z_{N,j} =1\right] 
\end{equation}
Then it holds that,
\begin{equation*}
\sqrt{\E(N^{-3/2} Q_N)^2}  \leq N^{-3/2} \sum_{j=1}^N  \sum_{i=1}^N\sqrt{\gamma_{N,i,j}}.
\end{equation*}
\end{lemma}
\begin{proof}
This proof mimics the proof of~\cref{theta_ij}. Recall the notations $\mathbf{II}_N$ and $\widetilde{\mathbf{II}}_N$ from \eqref{I.II.N} and \eqref{I.II.tilde.N}, respectively. Note that,
\begin{align*}
\sqrt{\E( Q_N^2)} &=    \sqrt{\E\left(\mathbf{II}_N - \widetilde{\mathbf{II}}_N\right)^2}\\
&=  \left[\E\bigg( \sum_{j=1}^N Z_{N,j}  \sum_{i=1}^N Z_{N,i}(\xi_{N,i,j} - \widetilde{\xi}_{N,i,j})\bigg)^2 \right]^{1/2}\\
&\le   \sum_{j=1}^N  \left[\E\bigg( Z_{N,j} \sum_{i=1}^N Z_{N,i}(\xi_{N,i,j} - \widetilde{\xi}_{N,i,j})\bigg)^2\right]^{1/2}\\
&=   \sum_{j=1}^N  \Bigg[\E\left(Z_{N,j} \E\left(\bigg(\sum_{i=1}^N Z_{N,i}(\xi_{N,i,j} - \widetilde{\xi}_{N,i,j})\bigg)^2\ \Big|\ Z_{N,j}\right)\right)\Bigg]^{1/2}\\
&=   \sum_{j=1}^N  \Bigg[\E\left(Z_{N,j} \E\left(\bigg(\sum_{i=1}^N Z_{N,i}(\xi_{N,i,j} - \widetilde{\xi}_{N,i,j})\bigg)^2\ \Big|\ Z_{N,j}=1\right)\right)\Bigg]^{1/2}\\
&=   \sum_{j=1}^N  \sqrt{\E(Z_{N,j})} \Bigg[\E\left(\bigg(\sum_{i=1}^N Z_{N,i}(\xi_{N,i,j} - \widetilde{\xi}_{N,i,j})\bigg)^2\ \Big|\ Z_{N,j}=1\right)\Bigg]^{1/2}\\
&\le   \sum_{j=1}^N  \sqrt{\E(Z_{N,j})} \sum_{i=1}^N \left[\E\left(Z_{N,i}(\xi_{N,i,j} - \widetilde{\xi}_{N,i,j})^2\ \Big|\ Z_{N,j}=1\right)\right]^{1/2}\\
&= \sum_{j=1}^N  \sqrt{\E(Z_{N,j})} \sum_{i=1}^N \left[\E\left(Z_{N,i}\E\left((\xi_{N,i,j} - \widetilde{\xi}_{N,i,j})^2\ \Big|\ Z_{N,i}, Z_{N,j}=1\right)  \Big|\ Z_{N,j}=1\right)\right]^{1/2}\\
&= \sum_{j=1}^N  \sqrt{\E(Z_{N,j})} \sum_{i=1}^N \left[\E\left(Z_{N,i}\E\left((\xi_{N,i,j} - \widetilde{\xi}_{N,i,j})^2\ \Big|\ Z_{N,i}=1, Z_{N,j}=1\right)  \Big|\ Z_{N,j}=1\right)\right]^{1/2}\\
&=   \sum_{j=1}^N  \sqrt{\E(Z_{N,j})} \sum_{i=1}^N \sqrt{\E(Z_{N,i} \mid Z_{N,j}=1)}  \left[\E\left(|\xi_{N,i,j} - \widetilde{\xi}_{N,i,j}|\ \Big|\ Z_{N,i}=1,Z_{N,j}=1\right)\right]^{1/2}\\
&\le   \sum_{j=1}^N   \sum_{i=1}^N  \left[\E\left(|\xi_{N,i,j} - \widetilde{\xi}_{N,i,j}|\ \Big|\ Z_{N,i}=1,Z_{N,j}=1\right)\right]^{1/2}\\
&=  \sum_{j=1}^N  \sum_{i=1}^N \sqrt{\gamma_{N,i,j}}.
\end{align*}
This completes the proof.
\end{proof} 

\begin{lemma}\label{D_NQ_N} Define $D_N$ and $Q_N$ as in \eqref{tNadj.decomp2}. It holds under~\cref{AssumpB2} that $N^{-3/2} D_N=o_p(1)$ and $N^{-3/2} Q_N=o_p(1)$ as $N\to\infty$.\end{lemma}

\begin{proof}
We first focus on $D_N$. In view of~\cref{theta_ij}, it suffices to show that 
\begin{equation*}
\lim_{N\to\infty} N^{-3/2} \sum_{j=1}^N  \sum_{i=1}^N\sqrt{\theta_{N,i,j}} =0,
\end{equation*}
where $\theta_{N,i,j}$ is defined in \eqref{theta_ij_defn}. 
Fix $h$ for the moment, and introduce the notation $$r_{N} (i,j) := \widetilde{b}_{N,j} - \widetilde{b}_{N,i} - hN^{-1/2}, \ 1\le i,j\le N.$$
Note, 
\begin{align*}
\xi_{N,i,j} - \widetilde{\xi}_{N,i,j} &= \ind{hN^{-1/2} (\vec{p}_{N,i} - \vec{p}_{N,j})^\top  \vec{Z}_N \leq  r_{N} (i,j) } -  \ind{0 \leq  r_{N} (i,j)}\\
&= \ind{hN^{-1/2} (\vec{p}_{N,i} - \vec{p}_{N,j})^\top  \vec{Z}_N \leq  r_{N} (i,j)  <  0 } \\
&\qquad\qquad -  \ind{hN^{-1/2} (\vec{p}_{N,i} - \vec{p}_{N,j})^\top  \vec{Z}_N >  r_{N} (i,j)  \ge  0 }\\
&= \ind{hN^{-1/2} \left|(\vec{p}_{N,i} - \vec{p}_{N,j})^\top  \vec{Z}_N\right|  \ge  \left|r_{N} (i,j)\right|}\left(\ind{ r_{N} (i,j)  <  0 } -  \ind{r_{N} (i,j)  \ge  0}\right).
\end{align*}
Now applying the Markov inequality, we get
\begin{align*}
\sqrt{\theta_{N,i,j}}  &= \E^{1/2} \left[|\xi_{N,i,j} - \widetilde{\xi}_{N,i,j}| \ \Big|\  Z_{N,j} =1\right]\\
&= \left[P\left(hN^{-1/2} \left|(\vec{p}_{N,i} - \vec{p}_{N,j})^\top  \vec{Z}_N\right| \ \ge \ \left|r_{N} (i,j)\right| \ \Big|\ Z_{N,j}=1\right)\right]^{1/2}\\
&\leq   h^2 N^{-1}  r_{N} (i,j)^{-2}\E^{1/2}\left[\left(\left(\vec{p}_{N,i} - \vec{p}_{N,j}\right)^\top  \vec{Z}_N\right)^4 \ \Big|\ Z_{N,j}=1\right].
\end{align*}
To upper bound the above RHS, we use~\cref{PiPj}.
The key idea is to use this bound only for those $i,j$ for which $r_{N} (i,j)$ is at least as large as $\delta N^{-1/2}$. Towards that, fix $\delta>0$ and define 
$$\mathcal{S}_{N,\delta}  = \{(i, j) : \left| r_{N} (i,j) \right| \le \delta N^{-1/2}, 1\le i,j\le N\}.$$
Then
\begin{align*} 
(i,j)\not\in\mathcal{S}_{N,\delta} 
\implies \sqrt{\theta_{N,i,j}} \leq a^2 \delta^{-2}\E^{1/2}\left[\left(\left(\vec{p}_{N,i} - \vec{p}_{N,j}\right)^\top  \vec{Z}_N\right)^4 \ \Big|\ Z_{N,j}=1\right].
\end{align*}
Therefore
\begin{align*}
& N^{-3/2} \sum_{j=1}^N  \sum_{i=1}^N\sqrt{\theta_{N,i,j}}  \\
&\leq N^{-3/2} \sum_{(i,j)\in \mathcal{S}_{N,\delta} } \sqrt{\theta_{N,i,j}}  + N^{-3/2}  \sum_{(i,j)\not\in \mathcal{S}_{N,\delta} }^N\sqrt{\theta_{N,i,j}}  \\
&\leq N^{-3/2} \sum_{(i,j)\in \mathcal{S}_{N,\delta} } 1 + N^{-3/2}  \sum_{(i,j)\not\in \mathcal{S}_{N,\delta} }^N a^2 \delta^{-2}\E^{1/2}\left[\left(\left(\vec{p}_{N,i} - \vec{p}_{N,j}\right)^\top  \vec{Z}_N\right)^4 \ \Big|\ Z_{N,j}=1\right] \\
&\leq N^{-3/2} \left|\mathcal{S}_{N,\delta}  \right| + a^2 \delta^{-2} N^{-3/2}  \sum_{j=1}^N\sum_{i=1}^N \E^{1/2}\left[\left(\left(\vec{p}_{N,i} - \vec{p}_{N,j}\right)^\top  \vec{Z}_N\right)^4 \ \Big|\ Z_{N,j}=1\right]\\
&= N^{-3/2} \left|\mathcal{S}_{N,\delta}  \right| + a^2 \delta^{-2} (\rank(\vec{X}_N)-1)\cdot O(N^{-1/2}).
\end{align*}
In the last step we used~\cref{PiPj}. Letting $N\to\infty$ and then $\delta\downarrow 0$, it follows that $$\limsup_{N\to\infty}N^{-3/2} \sum_{j=1}^N  \sum_{i=1}^N\sqrt{\theta_{N,i,j}} \leq \lim_{\delta\downarrow 0}\limsup_{N\to\infty} N^{-3/2} \left|\mathcal{S}_{N,\delta}\right|.$$
Now
\begin{align*} 
(i,j)\in\mathcal{S}_{N,\delta}  &\implies \frac{\delta}{\sqrt{N}} \ge \left| r_{N} (i,j) \right| \ge \left|\left|(b_{N,j} -\vec{p}_{N,j}^\top \vec{b}_N)- (b_{N,i}-\vec{p}_{N,i}^\top \vec{b}_N)\right| - \frac{|a|}{\sqrt{N}}\right| \\
&\implies \left|(b_{N,j} -\vec{p}_{N,j}^\top \vec{b}_N)- (b_{N,i}-\vec{p}_{N,i}^\top \vec{b}_N)\right| \in\left[\frac{|a|-\delta}{\sqrt{N}},\frac{|a|+\delta}{\sqrt{N}}\right].
\end{align*}
We now invoke~\cref{AssumpB3} to arrive at
\begin{equation*}
\lim_{\delta\downarrow 0}\limsup_{N\to\infty} N^{-3/2} \left|\mathcal{S}_{N,\delta} \right| = 0,
\end{equation*}
which finishes the proof for $D_N$. 
The proof for $Q_N$ can be done in the same manner, using~\cref{gamma_ij}, Markov inequality, and \cref{PiPj}.
\end{proof} 

\begin{lemma}\label{Ib.even.more} Suppose that~\cref{AssumpIb.new} holds. Then for any $h\ge \delta>0,$  
\begin{equation*}
\lim_{N\to\infty} N^{-(2-\nu)}\sum_{j=1}^N \sum_{i=1}^N \ind{\left|b_{N,j}- b_{N,i}\right| \in \left[\frac{h-\delta}{N^\nu},\frac{h+\delta}{N^\nu}\right]} = 4\delta\mathcal{I}_b
\end{equation*}
where $\mathcal{I}_b$ is defined in~\cref{ACjs} of the main paper.
\end{lemma}

\begin{proof}
The proof is analogous to those of~\cref{AssumpB3,AssumpIb.more}, and hence omitted.
\end{proof}

\begin{lemma}\label{Jb.even.more}  Suppose that~\cref{AssumpJb.new} holds. Then for any $h\ge \delta>0,$  
\begin{equation*}
\lim_{N\to\infty} N^{-(2-\nu)}\sum_{j=1}^N \sum_{i=1}^N \ind{\left|\widetilde{b}_{N,j}- \widetilde{b}_{N,i}\right| \in \left[\frac{h-\delta}{N^\nu},\frac{h+\delta}{N^\nu}\right]} = 4\delta\mathcal{J}_b
\end{equation*}
where $\mathcal{J}_b$ is defined in~\cref{AssumpB2} of the main paper.
\end{lemma}

\begin{proof}
The proof is analogous to those of~\cref{AssumpB3,AssumpIb.more}, and hence omitted.
\end{proof}



\begin{lemma}\label{iid.I_C.new} For each $N\ge 1,$ let $b_{N,1},\dots,b_{N,N}$ be i.i.d. from a  distribution with density $f_{}(\cdot)$, and $I_{h,N,\nu}$ be defined in \eqref{Ian.new} of the main paper,  then for any $0<\nu\le 1/2,$
\begin{equation*} N^{-(2-\nu)}\sum_{j=1}^N\sum_{i=1}^N I_{h,N,\nu} (b_{N,j} - b_{N,i})  \Pto h\int_{\R} f_{}^{2} (x) dx.\end{equation*} 
\end{lemma}


\begin{proof}
We proved the above for $\nu=1/2$ in~\cref{iid.I_C}, so assume now that $\nu\in (0,1/2)$. Define $$S_N := S_N(b_{N,1},\dots,b_{N,N}) := \sum_{j=1}^N \sum_{i=1}^N I_{h,N,\nu} (b_{N,j} - b_{N,i}).$$
Observe that for any fixed $i\neq j,$ and $h > 0,$ $$\E\left(I_{h,N,\nu} (b_{N,j} - b_{N,i})\right) = P\left(0\le b_{N,2} - b_{N,1} < hN^{-\nu}\right) = g(hN^{-\nu}) - g(0)$$
where $$g(x) = P(b_{N,2} - b_{N,1} \le x) =\int_{-\infty}^x \int_{\R} f_{}(u + t) f_{}(u)\,du\,dt, \ x\in\R.$$
Using the DCT for integrals, we argue that $$g'(x) = \int_{\R} f_{}(u+x) f_{}(u)du.$$ 
Hence 
\begin{align*} \numberthis\label{Glim2} \lim_{N\to\infty} N^{-(2-\nu)} \E(S_N) &= \lim_{N\to\infty}  N^{-(2-\nu)} N^2 (g(hN^{-\nu}) - g(0))\\
&= \lim_{N\to\infty}  h\cdot\frac{g(hN^{-\nu}) - g(0)}{hN^{-\nu}} \\
&= hg'(0)\\
&= h\int_{\R} f_{}(u)^2du.
\end{align*}
The case $h < 0$ can be handled similarly, and the case $h = 0$ is straight-forward.
Next we bound $\E(S_N - \E S_N)^2$ using the Efron-Stein inequality \citep{ES81}. For each $1\le k \le N,$ let $b_{N,k}'$ be an i.i.d. copy of $b_{N,k},$ independent of everything else, and define $$S_N^{(k)} = S_N(b_{N,1},\dots,b_{N,k-1},b_{N,k}',b_{N,k+1},\dots,b_{N,N}), \ 1\le k\le N.$$ 
Note,
\begin{align*}
S_N - S_N^{(k)}
&=  \sum_{j=1}^N \left(I_{h,N,\nu} (b_{N,j} - b_{N,k}) -  I_{h,N,\nu} (b_{N,j} - b_{N,k}')\right)\\
&\quad \quad+\sum_{j=1}^N \left(I_{h,N,\nu} (b_{N,k} - b_{N,j}) -  I_{h,N,\nu} (b_{N,k}' - b_{N,j})\right).
\end{align*}
Hence $\big| S_N - S_N^{(k)}\big| \leq 4N$ almost surely. The Efron-Stein inequality tells us that
\begin{align*}
 \E(N^{-(2-\nu)}(S_N - \E S_N))^2 
&\leq
 N^{-(4-2\nu)} \E \left[\sum_{k=1}^N \E\left[(S_N - S_N^{(k)})^2 \ \big|\ b_{N,1},\dots,b_{N,N}\right] \right]\\
&\le N^{-(4-2\nu)}\cdot 16 N^3= 16N^{-(1-2\nu)}.
\end{align*}
Since $\nu<1/2,$ we can now invoke \eqref{Glim2} to get the desired conclusion.
\end{proof}


\begin{lemma}\label{PiPj} As $N\to\infty,$ the following results hold, for $r=1,2$.
\begin{align}
 \label{PiPjMu}\sum_{j=1}^N \sum_{i=1}^N \E^{1/r}\left[\left(\left(\vec{p}_{N,i} - \vec{p}_{N,j}\right)^\top  \vec{Z}_N\right)^{2r} \right]  &= (\rank(\vec{X}_N)-1)\cdot O(N),\\ 
 \label{PiPjM}\sum_{j=1}^N \sum_{i=1}^N \E^{1/r}\left[\left(\left(\vec{p}_{N,i} - \vec{p}_{N,j}\right)^\top  \vec{Z}_N\right)^{2r} \ \Big|\ Z_{N,j}=1\right]  &= (\rank(\vec{X}_N)-1)\cdot O(N),\\
 \label{PiPjMdc}\sum_{j=1}^N \sum_{i=1}^N \E^{1/r}\left[\left(\left(\vec{p}_{N,i} - \vec{p}_{N,j}\right)^\top  \vec{Z}_N\right)^{2r} \ \Big|\ Z_{N,i}=1,Z_{N,j}=1\right]  &= (\rank(\vec{X}_N)-1)\cdot O(N).
\end{align}
\end{lemma}
\begin{proof}
Using the fact that $\vec{P}_{\vec{X}_N}$ is idempotent, we derive the following identity.
\begin{align*}\sum_{j=1}^N \sum_{i=1}^N  \|\vec{p}_{N,i} - \vec{p}_{N,j}\|^2
&= \sum_{j=1}^N \sum_{i=1}^N \left(\vec{P}_{\vec{X}_N}(i,i)+\vec{P}_{\vec{X}_N}(j,j)-2\vec{P}_{\vec{X}_N}(i,j)\right)\\
&= 2N\sum_{i=1}^N  \vec{P}_{\vec{X}_N}(i,i)-2\sum_{i=1}^N \vec{p}_{N,i}^\top \vec{1}\\
&= 2N \left(\Tr\left(\vec{P}_{\vec{X}_N}\right)-1\right) = 2N \left(\rank(\vec{X}_N)-1\right).\numberthis\label{norm}
\end{align*}
Now $\vec{1}$ belongs to the column space of $\vec{X}_N,$ so  $\E\left[(\vec{p}_{N,i} - \vec{p}_{N,j})^\top  \vec{Z}_N\right]=0$ for each $i,j$. Also note that $\var(\vec{Z}_N)$ is of the form $\alpha\vec{I} + \beta\vec{1}\vec{1}^\top$, where $\alpha=\frac{m}{N}(1-\frac{m}{N})$. 
Consequently, 
\begin{align*}
 \var\left[(\vec{p}_{N,i} - \vec{p}_{N,j})^\top  \vec{Z}_N \right] &= (\vec{p}_{N,i} - \vec{p}_{N,j})^\top \left(\alpha\vec{I} + \beta\vec{1}\vec{1}^\top\right)(\vec{p}_{N,i} - \vec{p}_{N,j})\\
 &= \frac{m}{N}\left(1-\frac{m}{N}\right) \|\vec{p}_{N,i} - \vec{p}_{N,j}\|^2.
\end{align*}
Hence
\begin{align*}
 \sum_{j=1}^N\sum_{i=1}^N \E\left[\left((\vec{p}_{N,i} - \vec{p}_{N,j})^\top  \vec{Z}_N\right)^2 \right] \le \sum_{j=1}^N\sum_{i=1}^N \|\vec{p}_{N,i} - \vec{p}_{N,j}\|^2.
\end{align*}
Thus, \eqref{PiPjMu} follows for $r=1$. 
To prove it for $r=2,$  fix $i$ and $j$ for the moment and write $\left(\vec{p}_{N,i} - \vec{p}_{N,j}\right)^\top  \vec{Z}_N = \sum_{k=1}^N v_k Z_k$. Observe now that
\begin{align*}
\E \bigg( \sum_{k=1}^N v_k Z_k \bigg)^4
= \sum_{k_1=1}^N v_{k_1}^4 \E Z_{1} &+ \sum_{k_1,k_2 \text{ distinct}} ( 3v_{k_1}^2 v_{k_2}^2+4v_{k_1}^3 v_{k_2}) \E Z_{1}Z_{2} \\
&+ \sum_{k_1,k_2,k_3 \text{ distinct}} 6v_{k_1}^2 v_{k_2}v_{k_3} \E Z_{1}Z_{2} Z_{3}\\
&+ \sum_{k_1,k_2,k_3,k_4  \text{ distinct}} v_{k_1} v_{k_2}v_{k_3}v_{k_4} \E Z_{1}Z_{2} Z_{3}Z_{4}.
\end{align*}
Now $\sum_{k=1}^N v_k = (\vec{p}_{N,i} - \vec{p}_{N,j})^\top \vec{1}=0,$ so we obtain 
\begin{align*}
\sum_{k_1,k_2,k_3,k_4  \text{ distinct}} v_{k_1} v_{k_2}v_{k_3}v_{k_4} &= \sum_{k_1,k_2,k_3 \text{ distinct}} v_{k_1} v_{k_2}v_{k_3}(-v_{k_1}-v_{k_2}-v_{k_3})\\
&= - 3\sum_{k_1,k_2,k_3 \text{ distinct}} v_{k_1}^2 v_{k_2}v_{k_3},\\
\sum_{k_1,k_2,k_3 \text{ distinct}} v_{k_1}^2 v_{k_2}v_{k_3} &= \sum_{k_1,k_2 \text{ distinct}} v_{k_1}^2 v_{k_2}(-v_{k_1}-v_{k_2})\\ 
&= - \sum_{k_1,k_2 \text{ distinct}} v_{k_1}^3 v_{k_2} - \sum_{k_1,k_2 \text{ distinct}}v_{k_1}^2 v_{k_2}^2,\\
\sum_{k_1,k_2 \text{ distinct}} v_{k_1}^3 v_{k_2} &= \sum_{k_1=1}^N v_{k_1}^3 (-v_{k_1}) = - \sum_{k_1=1}^N v_{k_1}^4.
\end{align*}
Since $\E \left|Z_1\cdots Z_v\right|\le 1$ for $1\le v\le 4$, it follows from the above identities that 
\begin{align*}
\E^{1/2}\left(\left(\vec{p}_{N,i} - \vec{p}_{N,j}\right)^\top  \vec{Z}_N\right)^4 &\le \sqrt{D'}  \left( \sum_{k=1}^N v_k^4 + \sum_{k\neq l} v_k^2 v_l^2\right)^{1/2} \\
&\le \sqrt{2D'} \sum_{k=1}^N v_k^2= \sqrt{2D'} \|\vec{p}_{N,i} - \vec{p}_{N,j}\|^2
\end{align*}
for some constant $D'>0$ (which is free of $i,j$ and $N$). The rest follows again from \eqref{norm}.

Next, we deal with the identities where we condition upon $Z_{N,j}=1$. Observe that
\begin{align*}
 \E\left[\left((\vec{p}_{N,i} - \vec{p}_{N,j})^\top  \vec{Z}_N\right)^2 \ \Big|\ Z_{N,j}=1\right] 
&= (\vec{p}_{N,i} - \vec{p}_{N,j})^\top \var\left(\vec{Z}_N \ \big|\ Z_{N,j}=1\right) (\vec{p}_{N,i} - \vec{p}_{N,j}).
\end{align*}
For any $k\neq l$, $$\var(Z_{N,k} \mid Z_{N,j} =1) =  (\alpha_1-\alpha_1^2)\ind{k\neq j},$$ and $$\cov(Z_{N,k},Z_{N,l} \mid Z_{N,j} =1) = (\alpha_2-\alpha_1^2)\ind{k\neq j},$$ where $\alpha_1=\frac{m-1}{N-1},\  \alpha_2=\frac{m-1}{N-1}\cdot\frac{m-2}{N-2}$.
We thus obtain 
\begin{multline*}
\E\left[\left((\vec{p}_{N,i} - \vec{p}_{N,j})^\top  \vec{Z}_N\right)^2 \ \Big|\ Z_{N,j}=1\right] \\= \Big[\alpha_1(1-\alpha_1) \|\vec{p}_{N,i} - \vec{p}_{N,j}\|^2  - (\alpha_1-2\alpha_2+\alpha_1^2)(\vec{P}_{\vec{X}_N}(i,j)-\vec{P}_{\vec{X}_N}(j,j))^2\Big].
\end{multline*}
But $(\vec{P}_{\vec{X}_N}(i,j)-\vec{P}_{\vec{X}_N}(j,j))^2\le  \|\vec{p}_{N,i} - \vec{p}_{N,j}\|^2$. So it follows that $$\E\left[\left((\vec{p}_{N,i} - \vec{p}_{N,j})^\top  \vec{Z}_N\right)^2 \ \Big|\ Z_{N,j}=1\right] \leq D  \|\vec{p}_{N,i} - \vec{p}_{N,j}\|^2$$ for some $D>0$ which is free of $i,j$ and $N$. Invoking \eqref{norm} again, we finish the proof of \eqref{PiPjM} for $r=1$. 
To prove \eqref{PiPjM} for $r=2,$ note the following.
\begin{align*}
&\sqrt{\E\left[\left(\left(\vec{p}_{N,i} - \vec{p}_{N,j}\right)^\top  \vec{Z}_N\right)^4 \ \Big|\ Z_{N,j}=1\right] } \\
&\leq \sqrt{\var\left[\left(\left(\vec{p}_{N,i} - \vec{p}_{N,j}\right)^\top  \vec{Z}_N\right)^2 \ \Big|\ Z_{N,j}=1\right] } + \E\left[\left(\left(\vec{p}_{N,i} - \vec{p}_{N,j}\right)^\top  \vec{Z}_N\right)^2 \ \Big|\ Z_{N,j}=1\right] \\
&\leq \sqrt{\frac{N}{m}}\sqrt{\E\left[\var\left(\left(\left(\vec{p}_{N,i} - \vec{p}_{N,j}\right)^\top  \vec{Z}_N\right)^2 \ \Big|\ Z_{N,j}\right)\right] } + \E\left[\left(\left(\vec{p}_{N,i} - \vec{p}_{N,j}\right)^\top  \vec{Z}_N\right)^2 \ \Big|\ Z_{N,j}=1\right] \\
&\leq \sqrt{\frac{N}{m}}\sqrt{\var\left[\left(\left(\vec{p}_{N,i} - \vec{p}_{N,j}\right)^\top  \vec{Z}_N\right)^2\right] } + \E\left[\left(\left(\vec{p}_{N,i} - \vec{p}_{N,j}\right)^\top  \vec{Z}_N\right)^2 \ \Big|\ Z_{N,j}=1\right] \\
&\leq \sqrt{\frac{N}{m}}\sqrt{\E\left[\left(\left(\vec{p}_{N,i} - \vec{p}_{N,j}\right)^\top  \vec{Z}_N\right)^4\right] } + \E\left[\left(\left(\vec{p}_{N,i} - \vec{p}_{N,j}\right)^\top  \vec{Z}_N\right)^2 \ \Big|\ Z_{N,j}=1\right].
\end{align*}
Since $m/N\to\lambda\in (0,1),$ the desired conclusion follows from \eqref{PiPjMu} (with $r=2$) and \eqref{PiPjM} (with $r=1$). 
The proof of \eqref{PiPjMdc} is completely analogous to the proofs of \eqref{PiPjMu}  and \eqref{PiPjM}, hence omitted.
\end{proof} 

\begin{lemma}\label{Phi-bound}
For $x\neq 0$ it holds that $\left|\Phi(x/\sigma) - \ind{x\ge 0}\right| \le \sigma \left|x\right|^{-1} \exp(-x^2/2\sigma^2),$
where $\Phi$ is the standard Normal CDF.
\end{lemma}
\begin{proof} Let $\phi$ be the density of standard Normal. For $x>0$, we have $$\left|\Phi(x/\sigma) - \ind{x\ge 0}\right| = 1 - \Phi(x/\sigma) \le \sigma \cdot x^{-1} \phi(x/\sigma),$$ using a standard inequality. For $x<0$, we can write $$\left|\Phi(x/\sigma) - \ind{x\ge 0}\right| =  \Phi(x/\sigma) = 1 - \Phi(-x/\sigma) \le \sigma \cdot (-x)^{-1} \phi(-x/\sigma).$$ Hence the result.
\end{proof}

\begin{lemma}\label{bvn-estimate}
Suppose that $(X,Y)$ follows the bivariate Normal distribution with $\E X = \E Y = 0$,  $\var(X) = \var(Y) = 1$, and $\text{corr}(X,Y) = \rho\in (-1,1)$. Then for any $h_1,h_2\in (0,1)$, the following bound holds:
\begin{equation*}
    \left|P(0\le X \le h_1, 0\le Y \le h_2) - P(0\le X \le h_1)P(0\le Y \le h_2)\right|\le C_\rho h_1h_2,
\end{equation*}
where $C_\rho =  |\rho|\left((1-\rho^2)^{-1/2} + (1-|\rho|)^{-2} \right)$.
\end{lemma}

\begin{proof}
 Define $$g_{x,y}(\rho):= \exp\left(-\frac{1}{2(1-\rho^2)}(x^2+y^2-2\rho xy)\right).$$
Note that for any fixed $x,y\in [0,1]$,
\begin{align*}
    \left|\frac{\partial g_{x,y}(\rho)}{\partial \rho}\right| = g_{x,y}(\rho) \left|\frac{xy}{1-\rho^2} - \frac{\rho}{1-\rho^2} \frac{x^2+y^2-2\rho xy}{1-\rho^2}\right|\le  \frac{1}{(1-|\rho|)^2}.
\end{align*}
Hence
\begin{align*}
    &\left|P(0\le X \le h_1, 0\le Y \le h_2) - P(0\le X \le h_1)P(0\le Y \le h_2)\right|\\
    &=\left|\int_{[0,\ h_1]\times[0,\ h_2]} \left((2\pi\sqrt{1-\rho^2})^{-1} g_{x,y}(\rho) -(2\pi)^{-1} g_{x,y}(0)\right)~dx~dy\right|\\
    &\le(2\pi)^{-1}\left|(1-\rho^2)^{-1/2}-1\right|\int_{[0,\ h_1]\times[0,\ h_2]} g_{x,y}(\rho)~dx~dy \\
    &\hspace{5cm}+ (2\pi)^{-1}\int_{[0,\ h_1]\times[0,\ h_2]} \left|g_{x,y}(\rho) - g_{x,y}(0)\right|~dx~dy\\
    &\le \rho^2 (1-\rho^2)^{-1/2} h_1h_2 + \int_{[0,\ h_1]\times[0,\ h_2]} \left|\rho\right|\int_0^1 \left|\frac{\partial g_{x,y}(t\rho)}{\partial \rho} \right|dt~dx~dy\\
    &\le  |\rho|\left((1-\rho^2)^{-1/2} + (1-|\rho|)^{-2} \right) h_1h_2.
\end{align*}
This completes the proof.
\end{proof}

\begin{proposition}[A generalized Efron-Stein inequality]\label{GenES} Given a sequence of independent real-valued random variables $W_{1}, W_{2}, \ldots, W_{n}$ and $F: \R^{n} \rightarrow \R$ be a measurable function. For each $1\le i\le n,$ let $W'_{i}$ be an independent copy of $W_{i}$, independent of the other $W_j$'s. Define $S:=F\left(W_{1}, \ldots, W_{n}\right)$ and $S_{i}:=F\left(W_{1}, \ldots, W_{i-1}, W'_{i}, W_{i+1}, \ldots, W_{n}\right)$. Then for all integers $q \geq 2$, there exists a constant $c_{q}$ (depending only on $q$) such that 
$$
\E\left|S-\mathbb{E}S\right|^q \leq c_{q}\E\left|\mathbb{E}\left[\sum_{i=1}^{n}\left(S-S_{i}\right)^{2} \mid\left(W_{1}, W_{2}, \ldots, W_{n}\right)\right]\right|^{q/2}.
$$
\end{proposition}

\begin{proof}
{See {\citet[Theorem 2]{GES05}} for a proof.}
\end{proof}
The special case $q=2$ yields the Efron-Stein inequality (see \citet{ES81}).

\end{document}